%% file: APSP-p3-v4.tex
\newif\ifabstract
\abstracttrue
\newif\iffull
\ifabstract \fullfalse \else \fulltrue \fi

\documentclass[11pt]{article}
\usepackage{amsfonts}
\usepackage{amssymb}
\usepackage{amstext}
\usepackage{amsmath}
\usepackage{xspace}
\usepackage{ntheorem}
\usepackage{graphicx}
\usepackage{url}
\usepackage{graphics}
\usepackage{colordvi}
\usepackage{hyperref}
\usepackage{subfigure}
\usepackage{xcolor}
\usepackage{cleveref}

\advance \oddsidemargin by -0.8in
\newcommand{\myparskip}{3pt}
\parskip \myparskip

\newcommand{\heap}{\mathsf{Heap}}
\newcommand{\pred}{\mathsf{pred}}

\newcommand{\epsx}{\epsilon}

\newcommand{\Otilde}{\widetilde O}
\newcommand{\otilde}{\widetilde O}

\newcommand{\vertexlist}{\mbox{\sf{VertexList}}}
\newcommand{\clusterlist}{\mbox{\sf{ClusterList}}}

\newcommand{\coveringcluster}{\mbox{\sf{CoveringCluster}}}
\newcommand{\clustercover}{\mbox{\sf{CoveringCluster}}}
\newcommand{\apxdist}{\tilde{dist}}
\newcommand{\MMF}{\mbox{\sf{Maximum Multicommodity Flow}}\xspace}
\newcommand{\MM}{\mbox{\sf{Minimum Multicut}}\xspace}
\newcommand{\SSSP}{\mbox{\sf{SSSP}}\xspace}
\newcommand{\APSP}{\mbox{\sf{APSP}}\xspace}
\newcommand{\NC}{\mbox{\sf{Neighborhood Cover}}\xspace}

\newcommand{\delvertex}{\ensuremath{\operatorname{DeleteVertex}}\xspace}
\newcommand{\addsupernode}{\ensuremath{\operatorname{AddSuperNode}}\xspace}
\newcommand{\csplit}{\ensuremath{\operatorname{ClusterSplit}}\xspace}
\newcommand{\algtransformpath}{\ensuremath{\operatorname{AlgTransformPath}}\xspace}
\newcommand{\alg}{\ensuremath{\operatorname{Alg}}\xspace}

\newcommand{\optmcf}{\mathsf{OPT}_{\mathsf{MCF}}}
\newcommand{\optmm}{\mathsf{OPT}_{\mathsf{MM}}}

\newcommand{\ceil}[1]{\ensuremath{\left\lceil#1\right\rceil}}
\newcommand{\floor}[1]{\ensuremath{\left\lfloor#1\right\rfloor}}

\newcommand{\opt}{\mathsf{OPT}}

\newcommand{\set}[1]{\left\{ #1 \right\}}

\newcommand{\iset}{{\mathcal{I}}}
\newcommand{\pset}{{\mathcal{P}}}
\newcommand{\qset}{{\mathcal{Q}}}

\newcommand{\bset}{{\mathcal{B}}}

\newcommand{\cset}{{\mathcal{C}}}

\newcommand{\mset}{{\mathcal M}}

\newcommand{\sset}{{\mathcal{S}}}

\newcommand{\uset}{{\mathcal{U}}}
\newcommand{\dset}{{\mathcal{D}}}
\newcommand{\oset}{{\mathcal{O}}}

\newcommand{\notdset}{\overline{\mathcal{D}}}

\newcommand{\tw}{w}
\newcommand{\tW}{W}

\newcommand{\be}{\begin{enumerate}}
\newcommand{\ee}{\end{enumerate}}
\newcommand{\bd}{\begin{description}}
\newcommand{\ed}{\end{description}}
\newcommand{\bi}{\begin{itemize}}
\newcommand{\ei}{\end{itemize}}

\newtheorem{theorem}{Theorem}[section]
\newtheorem{lemma}[theorem]{Lemma}
\newtheorem{observation}[theorem]{Observation}
\newtheorem{corollary}[theorem]{Corollary}
\newtheorem{claim}[theorem]{Claim}

\newtheorem*{definition}{Definition}
\newenvironment{proof}{\par \smallskip{\bf Proof:}}{\hfill\stopproof}
\def\stopproof{\square}
\def\square{\vbox{\hrule height.2pt\hbox{\vrule width.2pt height5pt \kern5pt
\vrule width.2pt} \hrule height.2pt}}

\newenvironment{proofof}[1]{\noindent{\bf Proof of #1.}}%
        {\hfill\stopproof}

\newenvironment{prog}[1]{
\begin{minipage}{5.8 in}
\begin{center}
{\sc #1}
\end{center}
}
{
\end{minipage}
}

\newcommand{\program}[3]{\begin{figure} \fbox{\vspace{2mm}\begin{prog}{#1} #3 \end{prog}\vspace{2mm}} 
			\caption{#1 \label{#2}} \end{figure}}

\renewcommand{\phi}{\varphi}
\newcommand{\eps}{\epsilon}

\newcommand{\half}{\ensuremath{\frac{1}{2}}}

\newcommand{\poly}{\operatorname{poly}}
\newcommand{\dist}{\mbox{\sf dist}}
\newcommand{\diam}{\mbox{\sf diam}}
\newcommand{\reals}{{\mathbb R}}
\newcommand{\framework}{{\mathbb F}}

\newenvironment{properties}[2][0]
{
\begin{enumerate} \setcounter{enumi}{#1}}{\end{enumerate}}

\newcommand{\vol}{\operatorname{Vol}}

\newcommand{\shortestpathquery}{\mbox{\sf{shortest-path-query}}}

\newcommand{\pquery}{\mbox{\sf{path-query}}\xspace}
\newcommand{\exspquery}{\mbox{\sf{expander-short-path-query}}\xspace}
\newcommand{\expquery}{\mbox{\sf{expander-short-path-query}}\xspace}

\newcommand{\spquery}{\mbox{\sf{short-path-query}}\xspace}
\newcommand{\shortestpath}{\mbox{\sf{SSSP-query}}\xspace}
\newcommand{\algbasic}{\ensuremath{\mathsf{AlgBasic}}\xspace}
\newcommand{\proccut}{\ensuremath{\mathsf{ProcCut}}\xspace}
\newcommand{\initnc}{\ensuremath{\mathsf{InitNC}}\xspace}
\newcommand{\computenc}{\ensuremath{\mathsf{InitNC}}\xspace}

\newcommand{\algmaintainNC}{\ensuremath{\mathsf{AlgMaintainNC}}\xspace}

\newcommand{\maintaincluster}{\ensuremath{\mathsf{MaintainCluster}}\xspace}
\newcommand{\maintaingoodcluster}{\ensuremath{\mathsf{MaintainGoodCluster}}\xspace}
\newcommand{\algmaintaingoodcluster}{\ensuremath{\mathsf{AlgMaintainGoodCluster}}\xspace}
\newcommand{\algmaintaincluster}{\ensuremath{\mathsf{AlgMaintainCluster}}\xspace}
\newcommand{\recdynnc}{\ensuremath{\mathsf{RecDynNC}}\xspace}
\newcommand{\recdynNC}{\ensuremath{\mathsf{RecDynNC}}\xspace}
\newcommand{\algslow}{\ensuremath{\mathsf{AlgSlow}}\xspace}
\newcommand{\algpseudo}{\ensuremath{\mathsf{AlgPseudocut\&Expander}}\xspace}
\newcommand{\algmaintainexp}{\ensuremath{\mathsf{AlgMaintainExpander}}\xspace}

\newcommand{\dquery}{\mbox{\sf{dist-query}}\xspace}

\newcommand{\distquery}{\mbox{\sf{dist-query}}\xspace}

\newcommand{\CONNSF}{\mbox{\sf{CONN-SF}}}

\newcommand{\EST}{\mbox{\sf{ES-Tree}}\xspace}
\newcommand{\initexp}{\ensuremath{\mathsf{InitializeExpander}}\xspace}

\newcommand{\ESTs}{\mbox{\sf{ES-Trees}}\xspace}

\newcommand{\conn}{\mbox{\sf{conn}}}

\newcommand{\hW}{\hat W}
\newcommand{\hN}{\hat N}
\newcommand{\LP}{\mbox{LP}}
\begin{document}

\begin{titlepage}
	
	\title{Decremental All-Pairs Shortest Paths in Deterministic Near-Linear Time\footnote{A preliminary version appeared in STOC 2021}}
	\author{Julia Chuzhoy\thanks{Toyota Technological Institute at Chicago. Email: {\tt cjulia@ttic.edu}. Supported in part by NSF grants CCF-1616584 and CCF-2006464.}}
	\maketitle
\pagenumbering{gobble}
	
	\thispagestyle{empty}

\begin{abstract}
	We study the decremental All-Pairs Shortest Paths (APSP) problem in undirected edge-weighted graphs. The input to the problem is an undirected $n$-vertex $m$-edge graph $G$ with non-negative lengths on edges, that undergoes an online sequence of edge deletions. The goal is to support approximate shortest-paths queries: given a pair $x,y$ of vertices of $G$, return a path $P$ connecting $x$ to $y$, whose length is within factor $\alpha$ of the length of the shortest $x$-$y$ path, in time $\tilde O(|E(P)|)$, where $\alpha$ is the approximation factor of the algorithm. APSP is one of the most basic and extensively studied dynamic graph problems.
	A long line of work culminated in the algorithm of [Chechik, FOCS 2018] with near optimal guarantees: for any constant $0<\eps\leq 1$ and parameter $k\geq 1$, the algorithm achieves approximation factor $(2+\eps)k-1$, and  total update time $O(mn^{1/k+o(1)}\log (nL))$, where $L$ is the ratio of longest to shortest edge lengths. Unfortunately, as much of prior work, the algorithm is randomized and needs to assume an \emph{oblivious adversary}; that is, the input edge-deletion sequence is fixed in advance and may  not depend on the algorithm's behavior. 
	
	In many real-world scenarios, and in applications of APSP to static graph problems, it is crucial that the algorithm works against an \emph{adaptive adversary}, where the edge deletion sequence  may  depend on the algorithm's past behavior arbitrarily; ideally, such an algorithm should be \emph{deterministic}. Unfortunately, unlike the oblivious-adversary setting, its adaptive-adversary counterpart is still poorly understood. For unweighted graphs, the algorithm of [Henzinger, Krinninger and Nanongkai, FOCS '13, SICOMP '16] achieves a $(1+\eps)$-approximation with total update time $\tilde O(mn/\eps)$; the best current total update time guarantee of $n^{2.5+O(\eps)}$ is achieved by the recent deterministic algorithm of [Chuzhoy, Saranurak, SODA'21], with 
	$2^{O(1/\eps)}$-multiplicative and $2^{O(\log^{3/4}n/\eps)}$-additive approximation. To the best of our knowledge, for arbitrary non-negative edge weights, the fastest current adaptive-update algorithm  has total update time $O(n^{3}\log L/\eps)$, achieving a $(1+\eps)$-approximation. 
	Even if we are willing to settle for any $o(n)$-approximation factor, no currently known algorithm has a better than $\Theta(n^3)$ total update time in weighted graphs and better than $\Theta(n^{2.5})$ total update time in unweighted graphs.
	Several conditional lower bounds suggest that no algorithm  with a sufficiently small approximation factor can achieve an $o(n^3)$ total update time. 
	
	Our main result is a deterministic algorithm for decremental APSP in undirected edge-weighted graphs, that, for any $\Omega(1/\log\log m)\leq \eps< 1$, achieves approximation factor $(\log m)^{2^{O(1/\eps)}}$, with total update time $O\left (m^{1+O(\eps)}\cdot (\log m)^{O(1/\eps^2)}\cdot \log L\right )$. In particular, we obtain a $(\poly\log m)$-approximation in time $\Otilde(m^{1+\eps})$ for any constant $\eps$, and, for any slowly growing function $f(m)$, we obtain $(\log m)^{f(m)}$-approximation in time $m^{1+o(1)}$.
	We also provide an algorithm with similar guarantees for decremental Sparse Neighborhood Covers.
	
\end{abstract}

\newpage

\tableofcontents{}
\end{titlepage}

\pagenumbering{arabic}

\input{intro}

\input{prelims}

\input{recdynNCproblem.tex}

\input{cut-procedure.tex}

\input{alg-slow-simple.tex}
\input{pseudocuts.tex}

\input{good-clusters}
\input{contracted-graph2}

\input{final-proof}

\input{params}
\input{application}

\input{ack}

\newpage
\appendix

\input{prelims-appx}

\input{ES-tree-appx.tex}
\input{rec-comp.tex}

\input{proofs-basic-alg}

\newpage
\bibliographystyle{alpha}
\bibliography{APSP-p3}



\end{document}

%% file: intro.tex
\section{Introduction}
\label{sec: intro}
We study the decremental All-Pairs Shortest-Paths    (\APSP)  problem  in weighted undirected graphs. 
In this problem, we are given as input an undirected graph $G$ with lengths $\ell(e)\geq 1$ on its edges, that undergoes an online sequence of edge deletions. The goal is to support (approximate) shortest-path queries $\shortestpathquery(x,y)$: given a pair $x,y$ of vertices of $G$, return a path connecting $x$ to $y$, whose length is within factor $\alpha$ of the length of the shortest $x$-$y$ path in $G$, where $\alpha$ is the \emph{approximation factor} of the algorithm. We also consider approximate distance queries, $\distquery(x,y)$: given a pair $x,y$ of vertices of $G$, return an estimate $\dist'(x,y)$ on the distance $\dist_G(x,y)$ between $x$ and $y$ in graph $G$, such that $\dist_G(x,y)\leq \dist'(x,y)\leq \alpha\cdot \dist_G(x,y)$. \APSP is one of the most basic and extensively studied problems in dynamic algorithms, and in graph algorithms in general. Algorithms for this problem often serve as building blocks in designing algorithms for other graph problems, in both the classical static and the dynamic settings. Throughout, we denote by $m$ and $n$ the number of edges and the number of vertices in the initial graph $G$, respectively, and by $L$ the ratio of largest to smallest edge length. 
In addition to the approximation factor of the algorithm, two other central measures of its performance are: \emph{query time} -- the time it takes to process a single query; and \emph{total update time} -- the total time that the algorithm takes, over the course of the entire update sequence, to maintain its data structures. Ideally, we would like the total update time of the algorithm to be close to linear in $m$, and the query time for $\shortestpathquery$ to be bounded by $\tilde O(|E(P)|)$, where $P$ is the path that the algorithm returns. 

A straightforward algorithm for the decremental \APSP problem is the following: every time a query $\shortestpathquery(x,y)$ arrives, compute the shortest $x$-$y$ path in $G$ from scratch. This algorithm solves the problem exactly, but it has query time $\Theta(m)$. Another approach is to rely on \emph{spanners}. A spanner of a dynamic graph $G$ is another dynamic graph $H\subseteq G$, with $V(H)=V(G)$, such that the distances between the vertices of $G$ are approximately preserved in $H$; ideally a spanner $H$ should be very sparse. For example, a work of
\cite{BaswanaKS12} provides a randomized algorithm that maintains a spanner of a fully dynamic $n$-vertex graph $G$ (that may undergo both edge deletions and edge insertions), that, for any parameter $k$, achieves approximation factor $(2k-1)$, has expected amortized update time $O(k^2\log^2n)$ per update operation, and expected spanner size $O(kn^{1+1/k}\log n)$.
A recent work of
\cite{new-spanner} provides a randomized algorithm for maintaining a spanner of a fully dynamic $n$-vertex graph $G$ with approximation factor $O(\poly\log n)$ and total update time $\Otilde(m^*)$, where $m^*$ is the total number of edges ever present in $G$; the number of edges in the spanner $H$ is always bounded by $O(n\poly\log n)$. One significant advantage of this algorithm over the algorithm of \cite{BaswanaKS12} is that, unlike the algorithm of \cite{BaswanaKS12}, it can withstand an adaptive adversary; we provide additional discussion of oblivious versus adaptive adversary below.   An algorithm for the \APSP problem can naturally build on such constructions of spanners: given a query $\shortestpathquery(x,y)$ or $\distquery(x,y)$, we simply compute the shortest $x$-$y$ path in the spanner $H$. For example, the algorithm for graph spanners of \cite{new-spanner} implies a randomized $\poly\log n$-approximation algorithm for \APSP that has $O(m\poly\log n)$ total update time.
A recent work of \cite{newest-spanner} provides additional spanner-based algorithms for \APSP. 
Unfortunately, it seems inevitable that this straightforward spanner-based approach to \APSP  must have query time $\Omega(n)$ for both $\shortestpathquery$ and \distquery. 

In this paper, our focus is on developing algorithms for the \APSP problem, whose query time is $\tilde O(|E(P)|)$ for $\shortestpathquery$, where $P$ is the path that the query returns, and $O(\poly\log(mL))$ for $\distquery$. 
There are several reasons to strive for these faster query times. First, we typically want responses to the queries to be computed as fast as possible, and the above query times are close to the fastest possible. Second, obtaining $\tilde O(|E(P)|)$ query time for $\shortestpathquery$ is often crucial to obtaining fast algorithms for classical (static) graph problems that use algorithms for \APSP as a subroutine. We provide an example of such an application to (static) \MMF / \MM in uncapacitated graphs in \Cref{sec: application}.

We distinguish between dynamic algorithms that work against an \emph{oblivious adversary}, where the input sequence of edge deletions is fixed in advance and may not depend on the algorithm's past behavior, and algorithms that work against an \emph{adaptive adversary}, where the input update sequence may depend on the algorithm's past responses and inner states arbitrarily. We refer to the former as \emph{oblivious-update} and to the latter as \emph{adaptive-update} algorithms. We note that any deterministic algorithm for the \APSP problem is an adaptive-update algorithm by definition.

The classical data structure of Even and Shiloach \cite{EvenS,Dinitz,HenzingerKing}, that we refer to as \EST throughout the paper, implies an exact deterministic
algorithm for decremental unweighted $\APSP$ with $O(mn^{2})$ total update time, and the desired $O(|E(P)|)$ query time for $\shortestpathquery$, where $P$ is the returned path. 
Short of obtaining an exact algorithm for \APSP, the best possible approximation factor one may hope for is $(1+\eps)$, for any $\eps$. A long line of work 
\cite{BaswanaHS07,rodittyZ2,henzinger16,bernstein16} is dedicated to this direction. The fastest algorithms in this line of work, due to Henzinger, Krinninger, and Nanongkai \cite{henzinger16}, and due to Bernstein \cite{bernstein16} achieve total update time $\tilde{O}(mn/\epsilon)$; the former algorithm is deterministic but only works in unweighted undirected graphs, while the latter algorithm works in directed weighted graphs, with an overhead of $\log L$ in the total update time, but can only handle an oblivious adversary.
%
Unfortunately, known conditional lower bounds show that these algorithms are likely close to the best possible. Specifically, Dor, Halperin and Zwick \cite{DorHZ00}, and Roddity and Zwick \cite{RodittyZ11} showed that, assuming the Boolean Matrix Multiplication (BMM) conjecture\footnote{The conjecture states that
	there is no ``combinatorial'' algorithm for multiplying two Boolean matrices of size $n\times n$ in time $n^{3-\delta}$ for any constant $\delta>0$.}, for any $\alpha,\beta\geq 1$ with $2\alpha+\beta<4$, no combinatorial algorithm for \APSP achieves a multiplicative $\alpha$ and additive $\beta$ approximation, with total update time
$O(n^{3-\delta})$ and query time $O(n^{1-\delta})$, for any constant $0<\delta<1$. Henzinger et al. \cite{HenzingerKNS15} generalized this result to show the same lower bounds for all algorithms and not just combinatorial ones, assuming the Online Boolean Matrix-Vector Multiplication (OMV) conjecture\footnote{The conjecture assumes that there is no $n^{3-\delta}$-time algorithm, for any constant $0<\delta<1$, for the OMV problem, in which the input is a Bollean $(n\times n)$ matrix, with $n$ Boolean dimension-$n$ vectors $v_1,\ldots,v_n$ arriving online. The algorithm needs to output $Mv_i$ immediately after $v_i$ arrives.}. 
The work of Vassilevska Williams and Williams \cite{williams2018subcubic}, combined with the work of  Roddity and Zwick \cite{RodittyZ11}, implies that obtaining such an algorithm would lead to subcubic-time algorithms for a number of important static problems on graphs and matrices.

Due to these negative results, much work on the \APSP problem inevitably focused on higher approximation factors. 
In this regime, the oblivious-update setting is now reasonably well understood. A long line of work \cite{BernsteinR11,henzinger16,abraham2014fully,HenzingerKN14_focs} recently culminated with a randomized algorithm of Chechik~\cite{chechik}, that, for any integer $k\ge1$ and parameter $0<\eps<1$, obtains a
$((2+\epsilon)k-1)$-approximation, with total update time $O(mn^{1/k+o(1)}\log L)$, when the input graph is  weighted and undirected. This result is near-optimal, as all its parameters almost match
the best static algorithm of \cite{TZ}. We note that this result was recently slightly improved by \cite{lkacki2020near}, who obtain total update time $O(mn^{1/k}\log L)$, and improve query time for $\dquery$.

In contrast, progress in the adaptive-update setting has been much slower. 
Until recently, the fastest adaptive-update algorithm for {\bf unweighted} graphs, due to Henzinger, Krinninger, and Nanongkai
\cite{henzinger16}, only achieved an $\tilde{O}(mn/\epsilon)$ total update
time (for approximation factor $(1+\eps)$); the algorithm was recently significantly simplified by Gutenberg and Wulff-Nilsen  \cite{GutenbergW20}.
A recent work of \cite{APSP-old} provided a deterministic algorithm for {unweighted} undirected graphs, that, for any parameter $1\le k\le o(\log^{1/8}n)$, in response to query $\shortestpathquery(x,y)$, returns a path of length at most $3\cdot 2^{k}\cdot\dist_{G}(x,y)+2^{(O(k\log^{3/4}n)}$, with query time $O(|E(P)|\cdot n^{o(1)})$ for $\shortestpathquery$, and total update time $n^{2.5+2/k+o(1)}$. To the best of our knowledge, the fastest current adaptive-update algorithm for {\bf weighted} graphs has total update time $O(n^{3}\log L/\eps)$ and approximation factor $(1+\eps)$ (see \cite{reliable-hubs}).

Interestingly, even if we allow an $o(n)$-approximation factor,
no adaptive-update algorithms with better than $\Theta(n^{3})$ total
update time and better than $\Theta(n)$ query time for $\shortestpathquery$ and $\distquery$ are currently known for weighted undirected graphs, 
and no adaptive-update algorithms with better than $\Theta(n^{2.5})$ total update time  and better than $\Theta(n)$ query time are currently known for unweighted undirected graphs. Moreover, even for the seemingly simpler Single-Source Shortest Path problem (\SSSP), where all queries must be between a pre-specified source vertex $s$ and another arbitrary vertex of $G$, no algorithms achieving a better than $\Theta(n^2)$ total update time, and better than $\Theta(n)$ query time for $\shortestpathquery$ are known.
To summarize, ideally we would like an algorithm for decremental \APSP in weighted undirected graphs that achieves the following properties:

\begin{itemize}
	\item it can withstand an adaptive adversary (and is ideally deterministic);
	\item it has query time $\tilde O(|E(P)|)$ for $\shortestpathquery$, where $P$ is the returned path, and query time $\tilde O(1)$ for \distquery;
	\item it has near-linear in $m$ total update time; and
	\item it has a reasonably low approximation factor (ideally, polylogarithmic or constant).
\end{itemize}

Our main result comes close to achieving all these properties. Specifically, we provide a \emph{deterministic} algorithm for \APSP in weighted undirected graphs. For any precision parameter $\Omega(1/\log\log m)<\eps<1$, the algorithm achieves approximation factor 
$(\log m)^{2^{O(1/\eps)}}$, with total update time: $$O\left (m^{1+O(\eps)}\cdot (\log m)^{O(1/\eps^2)}\cdot \log L\right ).$$ The query time for processing $\distquery$ is $O(\log m \log\log L)$, and the query time for $\shortestpathquery$ is $O(|E(P)|)+O(\log m \log\log L)$, where $P$ is the returned path.
In particular, by letting $\eps$ be a small enough constant, we obtain a $O(\poly\log m)$-approximation with total update time time $O(m^{1+\delta})$, for any constant $0<\delta<1$, and by letting $1/\eps$ be a slowly growing function of $m$ (for example, $1/\eps=O(\log(\log^*m))$), we obtain an approximation factor $(\log m)^{O(\log^* m)}$, and total update time $O(m^{1+o(1)})$.

In fact we design an algorithm for a more general problem: \emph{dynamic Sparse Neighborhood Cover}. 
Given a graph $G$ with lengths on edges, a vertex $v\in V(G)$, and a distance parameter $D$, we denote by $B_G(v,D)$ the \emph{ball of radius $D$ around $v$}, that is, the set of all vertices $u$ with $\dist_G(v,u)\leq D$.
Suppose we are given a static graph $G$ with non-negative edge lengths, a distance parameter $D$ (that we call \emph{target distance threshold}), and a desired approximation factor $\alpha$. A $(D,\alpha\cdot D)$-\emph{neighborhood cover} for $G$ is a collection $\cset$ of vertex-induced subgraphs of $G$ (that we call \emph{clusters}), such that, for every vertex $v\in V(G)$, there is some cluster $C\in \cset$ with $B_G(v,D)\subseteq V(C)$. Additionally, we require that for every cluster $C\in \cset$, for every pair $x,y\in V(C)$ of its vertices, $\dist_G(x,y)\leq \alpha\cdot D$; if this property holds, then we say that $\cset$ is a \emph{weak} $(D,\alpha\cdot D)$-neighborhood cover of $G$. If, additionally, the diameter of every cluster $C\in \cset$ is bounded by $\alpha\cdot D$, then we say that  $\cset$ is a \emph{strong} $(D,\alpha\cdot D)$-neighborhood cover of $G$. Ideally, it is also desirable that the neighborhood cover is \emph{sparse}, that is, every edge (or every vertex) of $G$ only lies in a small number of clusters of $\cset$. For this static setting of the problem, the work of \cite{neighborhood-cover2,neighborhood-cover1} provides a deterministic algorithm that produces a strong $(D,O(D\log n))$-neighborhood cover of graph $G$, where every edge lies in at most $O(\log n)$ clusters, with running time $\Otilde(|E(G)|+|V(G)|)$.

In this paper we consider a \emph{decremental} version of the problem, in which the input graph $G$ undergoes an online sequence of edge deletions. We are required to maintain a weak $(D,\alpha\cdot D)$-neighborhood cover $\cset$ of the graph $G$, and we require that the clusters in $\cset$ may only be updated in a specific fashion: once an initial neighborhood cover $\cset$ of $G$  is computed, we are only allowed to delete edges or vertices from clusters that lie in $\cset$, or to add a new cluster $C$ to $\cset$, which must be a subgraph of an existing cluster of $\cset$. Additionally, we require that the algorithm supports  queries $\spquery(C,v,v')$: given two  vertices $v,v'\in V$, and a cluster $C\in \cset$ with $v,v'\in C$, return a path $P$ in the current graph $G$, of length at most $\alpha\cdot D$ connecting $v$ to $v'$ in $G$, in time $O(|E(P)|)$. The algorithm must also maintain, for every vertex $v\in V(G)$, a cluster $C=\coveringcluster(v)$ in $\cset$, with $B_G(v,D)\subseteq V(C)$. Lastly, we require that the neighborhood cover is \emph{sparse}, namely, for every vertex $v$ of $G$, the total number of clusters of $\cset$ to which $v$ may ever belong over the course of the algorithm is small. 
It is not hard to verify that an algorithm for the decremental Sparse Neighborhood Cover problem that we just defined immediately implies an algorithm for decremental \APSP with the same approximation factor, and the same total update time (to within $O(\log L)$-factor). We provide a deterministic algorithm for the dynamic Sparse Neighborhood Cover problem with approximation factor $\alpha=O\left ((\log m)^{2^{O(1/\eps)}}\right )$, and total update time $O\left (m^{1+O(\eps)}\cdot (\log m)^{O(1/\eps^2)}\right )$. Our algorithm ensures that, for every vertex $v\in V(G)$, the total number of clusters of $ \cset$ that $v$ ever belongs to is bounded by $m^{O(1/\log\log m)}$. 
We note that algorithms for static Sparse Neighborhood Covers have found many applications in the area of graph algorithms, and so we believe that our algorithm for dynamic Sparse Neighborhood Cover is interesting in its own right. 
A Sparse Neighborhood Cover for a dynamic graph $G$ naturally provides an emulator for $G$. If graph $G$ is decremental, then, while the edges may sometimes be inserted into the emulator (when a new cluster is added to the neighborhood cover $\cset$), due to the restrictions that we impose on the types of allowed updates to the clusters of $\cset$, such edge insertions are limited to very specific types, and so they are relatively easy to deal with. 
This allows us to compose emulators given by the neighborhood covers recursively. 
We note that the idea of using clustering of a dynamic graph $G$ in order to construct an emulator was used before (see e.g. the constructions of \cite{ForsterG19,chechikLowStretch,ForsterGH20} of dynamic low-stretch spanning trees). In several of these works, a family of clusters of a dynamic graph $G$ is constructed and maintained, and the restrictions on the allowed updates to the cluster family are similar to the ones that we impose; it is also observed in several of these works that with such restrictions one can naturally compose the resulting emulators recursively -- an approach that we follow here as well.
However, neither of these algorithms provide neighborhood covers, and in fact the clusters that are maintained for each distance scale are disjoint (something that cannot be achieved by neighborhood covers). Additionally, all the above-mentioned algorithms are randomized and assume an oblivious adversary. On the other hand, the algorithms of \cite{henzinger16,GutenbergW20} implicitly provide a deterministic algorithm for maintaining a neighborhood cover of a dynamic graph. However, these algorithms have a number of drawbacks: first, the running time for maintaining the neighborhood cover is too prohibitive (the total update time is $O(mn)$). Second, the neighborhood cover maintained is not necessarily sparse; in fact a vertex may lie in a very large number of resulting clusters. Lastly, clusters that join the neighborhood cover as the algorithm progresses may be arbitrary. The restriction that, for every cluster $C$ added to the neighborhood cover $\cset$, there must be a cluster $C'$ containing $C$ that already belongs to $\cset$, seems crucial in order to allow an easy recursive composition of emulators obtained from the neighborhood covers, and the requirement that the neighborhood cover is sparse is essential for bounding the sizes of the graphs that arise as the result of such recursive compositions.


We provide an application of our algorithm for the \APSP problem: a deterministic algorithm for \MMF and \MM in unit-capacity graphs. In both problems, the input is an undirected $n$-vertex $m$-edge graph $G$, and a collection $\mset=\set{(s_1,t_1),\ldots,(s_k,t_k)}$ of pairs of its vertices,  called \emph{demand pairs}. In the \MMF problem, the goal is to send maximum amount of flow between the demand pairs, such that the total amount of flow traversing each edge is at most $1$. We denote by $\optmcf$ the value of the optimal solution to this problem. 
In the \MM problem, given a graph $G$ and a collection $\mset$ of demand pairs as before, the goal is to select a minimum-cardinality subset $E'\subseteq E(G)$ of edges, such that, for all $1\leq i\leq k$, vertices $s_i$ and $t_i$ lie in different connected components of $G\setminus E'$. We use the standard primal-dual technique-based algorithm of \cite{GK98, Fleischer00}, that can equivalently be viewed as an application of the multiplicative weight update paradigm \cite{AroraHK12}, which essentially reduces the Multicommodity Flow problem to decremental \APSP; this reduction was first discovered by \cite{Madry10_stoc}.
Plugging in our algorithm for \APSP, we obtain a deterministic algorithm for \MMF, that, for any $0<\eps<1$,  achieves approximation factor $O\left ((\log m)^{2^{O(1/\eps)}}\right )$, and has running time $\tilde O\left (m^{1+O(\eps)}(\log m)^{2^{O(1/\eps)}}+k/\eps\right )$. The algorithm also provides an integral solution to the \MMF problem with congestion $O(\log n)$, and a fractional solution to the standard LP-relaxation for \MM. Using the standard ball-growing technique of \cite{LR,GVY}, we then obtain an algorithm for Minimum Multicut, with the same asymptotic running time, and similar approximation factor. 
The fastest previous approximation algorithms for \MMF, achieving $(1+\eps)$-approximation, have running times $O(k^{O(1)}\cdot m^{4/3}/\eps^{O(1)})$ \cite{KelnerMP12} and $\Otilde(mn/\eps^2)$ \cite{Madry10_stoc}; we are not aware of any algorithms that achieve a faster running time with possibly worse approximation factors, and we are not aware of any fast algorithms for the \MM problem. The best polynomial-time algorithm for \MM, due to \cite{LR,GVY}, achieves an $O(\log n)$-approximation.

Before we discuss our results and techniques in more detail, we provide some additional background on related work.

\subsection{Other Related Work}

\paragraph{\APSP on Expanders.}
A very interesting special case of the \APSP problem is \APSP on expanders. In this problem, we are given an initial graph $G$ that is a $\phi$-expander. Graph $G$ undergoes a sequence of edge deletions and isolated vertex deletions, that arrive in batches. We are guaranteed that, after each such batch of updates, the resulting graph $G$ remains an $\Omega(\phi)$-expander. As in the general \APSP problem, the goal is to support approximate  $\shortestpathquery$ in graph $G$. This problem is especially interesting for several reasons. First, it seems to be a relatively simple special case of the \APSP problem, and, if our goal is to obtain better algorithms for general \APSP, solving the problem in expander graphs is a natural starting step. Second, this problem arises in various algorithms for \emph{static} cut and flow problems, and seems to be intimately connected to efficient implementations of the Cut-Matching game of \cite{KRV}, which is a central tool in the design of fast algorithms for cut and flow problems (see, e.g. \cite{detbalanced}). Third, expander graphs are increasingly becoming a central tool for designing algorithms for various dynamic graph problems, and obtaining good algorithms for \APSP on expanders will likely become a powerful tool in the toolkit of algorithmic techniques in this area. A recent work of \cite{APSP-old}, building on \cite{detbalanced}, 
implies a deterministic algorithm for \APSP in expanders with approximation factor $O\left( \Delta^2(\log n)^{O(1/\eps^2)}/\phi\right )$, query time $O(|E(P)|)$ for \shortestpathquery, where $P$ is the returned path, and total update time $O\left (n^{1+O(\eps)}\Delta^7(\log n)^{O(1/\eps^2)}/\phi^5\right )$; here, $\Delta$ is the maximum vertex degree of $G$, $\phi$ is its expansion, and $\eps$ is a given precision parameter\footnote{The work of \cite{APSP-old} only explicitly provides such an algorithm for a specific setting of the parameter $\eps$, but it is easy to see that the same algorithm works for the whole range of values of $\eps$. We prove this  in \Cref{thm: expander APSP}  for completeness.}. In fact, algorithms in this paper also use this algorithm for \APSP in expanders as a subroutine.

\paragraph{Single-Source Shortest Paths.}
Single-Source Shortest Paths (\SSSP) is a special case of \APSP, where all queries must be between a fixed source vertex $s$ and arbitrary other vertices in the graph $G$. This problem has also been studied extensively. Algorithms for decremental \SSSP are a well-established tool in the design of fast algorithms for various variants of maximum $s$-$t$ flow and minimum $s$-$t$ cut problems (see, e.g. \cite{Madry10_stoc,fast-vertex-sparsest,APSP-old}). 

In the oblivious-adversary setting, our understanding of the problem is almost complete: a sequence of works \cite{HenzingerKN14_soda, BernsteinR11,HenzingerKN14_focs} has led to a $(1+\eps)$-approximation algorithm, that achieves total update time $O(m^{1+o(1)}\log L)$, which is close to the best possible. 
The query time of the algorithm is also near optimal: query time for $\distquery$ is $\poly\log n$, and query time for $\shortestpathquery$ is
$\tilde{O}(|E(P)|)$, where $P$ is the returned path. Conditional lower bounds
of \cite{DorHZ00,RodittyZ11} (that are based on the Boolean Matrix Multiplication conjecture) and of \cite{HenzingerKNS15} (based on the Online Matrix-vector Multiplication  conjecture), 
show that no algorithm that solves the problem exactly can simultaneously achieve an $O(n^{1-\delta})$ query time, and $O(n^{3-\delta})$ total update time, for any constant $\delta>0$, in graphs with $m = \Theta(n^2)$. 
The work of Vassilevska Williams and Williams \cite{williams2018subcubic}, combined with the work of  Roddity and Zwick \cite{RodittyZ11}, implies that obtaining an exact algorithm with similar total update time and query time would lead to subcubic-time algorithms for a number of important static problems on graphs and matrices.
This shows that the above oblivious-update algorithm is likely close to the best possible.

For the adaptive-update setting, the progress has been slower. It is well known
that the classical \EST data structure of Even and Shiloach \cite{EvenS,Dinitz,HenzingerKing},
combined with the standard weight rounding technique (e.g.~\cite{Zwick98,bernstein16}) gives a
$(1+\eps)$-approximate deterministic algorithm for $\SSSP$ with $\tilde O(mn\log L)$
total update time and near-optimal query time. 
Recently, Bernstein and Chechik  \cite{BernsteinChechik,Bernstein, BernsteinChechikSparse}, provided algorithms with total update time $\tilde{O}(n^{2}\log L)$ and $\tilde{O}(n^{5/4}\sqrt{m}) \leq \tilde O(mn^{3/4})$, while Gutenberg and Wulff-Nielsen \cite{GutenbergW20} showed an algorithm with $O(m^{1+o(1)}\sqrt{n})$ total update time. 
Unfortunately, all these algorithms only support distance queries, and they cannot handle shortest-path queries. 
This problem was recently addressed by \cite{fast-vertex-sparsest,APSP-old}, leading to a deterministic algorithm with total update time $O(n^{2+o(1)}\log L/\eps^2)$, that achieves a $(1+\eps)$-approximation factor, and has query time $O(|E(P)|\cdot n^{o(1)}\log\log L)$ for $\shortestpathquery$.
Lastly, the work of~\cite{new-spanner} on dynamic spanners also provides a randomized adaptive-update $(1+\eps)$-approximation
algorithm with total update time $O(m\sqrt{n})$, and query time $\tilde O(n)$. As mentioned already, they also provide an algorithm for dynamic spanners, leading to a $\poly\log n$-approximation algorithm with total update time $O(m\poly\log n)$ for \APSP, and hence for \SSSP, with query time $\tilde O(n)$.
To the best of our knowledge, our result for the \APSP problem is also the first adaptive-adversary algorithm for \SSSP with near-linear total update time, that achieves an approximation that is below $\Theta(n)$, and query time $\otilde(|E(P)|)$ for $\shortestpathquery$. We now discuss our results and techniques in more detail.

\subsection{Our Results and Techniques}

Our main result is a deterministic algorithm for decremental \APSP, that is summarized in the following theorem.
\begin{theorem}\label{thm: main}
	There is a deterministic algorithm, that, given an $m$-edge graph $G$ with length $\ell(e)\geq 1$ on its edges, that undergoes an online sequence of edge deletions, together with a parameter $c/\log\log m<\eps<1$ for some large enough constant $c$, supports approximate $\shortestpathquery$ queries and $\distquery$ queries with approximation factor $O\left ((\log m)^{2^{O(1/\eps)}}\right )$.  The query time for processing $\distquery$ is $O(\log m \log\log L)$, and the query time for processing $\shortestpathquery$ is $O(|E(P)|)+O(\log m \log\log L)$, where $P$ is the returned path, and $L$ is the ratio of longest to shortest edge length. The total update time of the algorithm is bounded by: $$O\left (m^{1+O(\eps)}\cdot (\log m)^{O(1/\eps^2)}\cdot \log L\right ).$$
\end{theorem}

Our proof exploits the decremental Sparse Neighborhood Cover problem, for which we provide the following algorithm:

\begin{theorem}\label{thm: NC}
	There is a deterministic algorithm, that, given an $m$-edge graph $G$ with integral lengths $\ell(e)\geq 1$ on its edges, that undergoes an online sequence of edge deletions, together with parameters $c/\log\log m<\eps<1$ for some large enough constant $c$, and $D\geq 1$, maintains a weak $(D,\alpha\cdot D)$-neighborhood cover $\cset$ of $G$, for $\alpha=O\left ((\log m)^{2^{O(1/\eps)}}\right )$, and  supports  queries $\spquery(C,v,v')$: given a cluster $C\in \cset$, and two  vertices $v,v'\in V(C)$, return a path $P$ connecting $v$ to $v'$ in $G$, of length at most $\alpha\cdot D$, in time $O(|E(P)|)$.
	Additionally, for every vertex $v\in V(G)$, the algorithm maintains a cluster $C=\coveringcluster(v)$ in $\cset$, with $B_G(v,D)\subseteq V(C)$.
	The algorithm starts with $\cset=\set{G}$, and the only allowed changes to the clusters in $\cset$ are: (i) delete an edge from a cluster $C\in \cset$; (ii) delete an isolated vertex from a cluster $C\in \cset$; and (iii) add a new cluster $C'$ to $\cset$, where $C'\subseteq C$ for some cluster $C\in \cset$. The algorithm has total update time $O\left (m^{1+O(\eps)}\cdot (\log m)^{O(1/\eps^2)}\right )$ and ensures that, for every vertex $v\in V(G)$, the total number of clusters $C\in \cset$ to which $v$ ever belongs over the course of the algorithm is at most  $m^{O(1/\log\log m)}$. 
\end{theorem}

We remark that the above theorem requires that we initially set $\cset=\set{G}$. Clearly, this initial cluster set $\cset$ may not be a valid neighborhood cover of $G$. Therefore, before the algorithm processes any updates of graph $G$, it may update this initial cluster set $\cset$, via changes of the types that are allowed by the theorem, until it becomes a valid neighborhood cover. We also note that we allow graphs to have parallel edges, so $m$ may be much larger than $|V(G)|$.

Lastly, we provide an efficient algorithm for the \MM and \MMF  
problems in unit-capacity graphs.

\begin{theorem}\label{thm: MMF and MM}
	There is a deterministic algorithm, that, given an $n$-vertex $m$-edge graph $G$, a collection $\mset=\set{(s_1,t_1),\ldots,(s_k,t_k)}$ of pairs of its vertices, called \emph{demand pairs}, and a precision parameter $c/\log\log m<\eps<1$ for some large enough constant $c$, computes, in time $\tilde O\left (m^{1+O(\eps)}(\log m)^{2^{O(1/\eps)}}+k/\eps\right )$, a solution to the \MMF instance $(G,\mset)$, of value at least \newline  
	$\Omega\left (\optmcf/(\log m)^{2^{O(1/\eps)}}\right )$, and a solution to the \MM instance $(G,\mset)$, of cost at most  $O\left ((\log m)^{2^{O(1/\eps)}}\cdot \optmm\right )$, where $\optmcf$ and $\optmm$ are optimal solution values to instance $(G,\mset)$ of \MMF and \MM, respectively.
\end{theorem}

The proof of \Cref{thm: MMF and MM} follows immediately from the proof of \Cref{thm: NC}  via standard techniques; see \Cref{sec: application} for more details. 
It is also immediate to obtain the proof of \Cref{thm: main} from \Cref{thm: NC} using the standard approach of considering each distance scale separately; see \Cref{subsubsec: first thm proof} for a formal proof. We now focus on describing our algorithm for the \NC problem from \Cref{thm: NC}, introducing our new ideas and techniques one by one.

\paragraph{Recursive Dynamic Neighborhood Cover.}
As mentioned already, one advantage of considering the \NC problem is that its solution naturally provides an emulator for the input graph $G$, which in turn can be used in order to compose algorithms for \NC recursively.
In fact, we initially prove a weaker version of \Cref{thm: NC}, by providing an algorithm (that we denote here for brevity by $\alg'$), that achieves a similar approximation factor, but a slower running time of: $$O\left (m^{1+O(\eps)}\cdot \poly(D)\cdot (\log m)^{O(1/\eps^2)}\right )$$ (on the positive side, the algorithm maintains a {\bf strong} neighborhood cover of the graph $G$). Recall that we call the parameter $D$ the \emph{target distance threshold} for the \NC problem instance.
We use the recursive composability of \NC in order to obtain the desired running time, as follows\footnote{A similar approach of recursive composition of emulators was used in numerous algorithms for \APSP; see, e.g. \cite{chechik}.}. Using standard rescaling techniques, we can assume that $1\leq D\leq \Theta(m)$. For all $1\leq i\leq O(1/\eps)$, let $D_i=m^{\eps i}$. We obtain an algorithm for the Sparse \NC problem for each target distance threshold $D_i$ recursively. For the base of the recursion, when $i=1$, we simply run Algorithm $\alg'$, to obtain the desired running time of $O\left (m^{1+O(\eps)}\cdot (\log m)^{O(1/\eps^2)}\right )$. Assume now that we have obtained an algorithm for target distance threshold $D_i$, that maintains a neighborhood cover $\cset_i$ of graph $G$. In order to obtain an algorithm for target distance threshold $D_{i+1}$, we construct a new graph $H$, by starting with $H=G$, deleting all edges of length greater than $D_{i+1}$, and rounding the lengths of all remaining edges up to the next integral multiple of $D_i$. Additionally, for every cluster $C\in \cset_i$, we add a vertex $u(C)$ (called a supernode), that connects, with an edge of length $D_i$, to every vertex $v\in V(C)\cap V(G)$. It is not hard to show that this new graph $H$ approximately preserves all distances between the vertices of $G$, that are in the range $(D_i,D_{i+1}]$. Since the length of every edge in $H$ is an integral multiple of $D_i$, scaling all edge lengths down by factor $D_i$ does not change the problem. It is then sufficient to solve the \NC problem in the resulting dynamic graph $H$, with target distance threshold $D_{i+1}/D_i=m^{\eps}$, which can again be done via Algorithm $\alg'$, with total update time $O\left (m^{1+O(\eps)}\cdot (\log m)^{O(1/\eps^2)}\right )$. The final algorithm for \Cref{thm: NC} is then obtained by recursively composing Algorithm $\alg'$ with itself $O(1/\eps)$ times. 

In order to be able to compose algorithms for the \NC problem using the above approach, we define the problem slightly differently, and we call the resulting variation of the problem Recursive Dynamic Neighborhood Cover, or \recdynnc. We assume that the input is a bipartite graph $H=(V,U,E)$, with non-negative edge lengths. Intuitively, the vertices in set $V$, that we refer to as \emph{regular vertices}, correspond to vertices of the original graph $G$, while the vertices in set $U$, that we call \emph{supernodes}, represent some neighborhood cover $\cset$ of the graph $G$ that is possibly maintained recursively: $U=\set{u(C)\mid C\in \cset}$. (In order to obtain the initial graph $H$, we subdivide every edge of $G$ by a new regular vertex; we view the original vertices of $G$ as supernodes; and for every vertex $v\in V(G)$, we add a new regular vertex $v'$ that connects to $v$ with a length-$1$ edge.) In addition to supporting standard edge-deletion and isolated vertex-deletion updates, we require that the algorithm for the \recdynnc problem supports a new update operation, that we call \emph{supernode splitting}\footnote{We note that a similar approach to handling cluster-splitting in an emulator that is based on clustering was used before in numerous works, including, e.g., \cite{BernsteinChechik,Bernstein,fast-vertex-sparsest,chechikLowStretch}.}. In this operation, we are given a supernode $u\in V(H)$, and a subset $E'$ of edges that are incident to $u$ in graph $H$. The update creates a new supernode $u'$ in graph $H$, and, for every edge $e=(u,v)\in E'$, adds a new edge $(u',v)$ of length $\ell(e)$ to $H$. 
The purpose of this update operation is to mimic the addition of a new cluster $C$ to $\cset$, where $C\subseteq C'$ for some existing cluster $C'\in \cset$. The supernode-splitting operation is applied to supernode $u(C')$, with edge set $E'$ containing all edges $(v,u(C'))$ with $v\in V(C)$, and the operation creates a new supernode $u(C)$.
Supernode-splitting operation, however, may insert some new edges into the graph $H$.
This creates several difficulties, especially in bounding total update times in terms of number of edges. We get around this problem as follows. Recall that the supernodes in set $U$ generally correspond to clusters in some dynamic neighborhood cover $\cset$, that we maintain recursively. We ensure that this neighborhood cover is sparse, that is, every regular vertex may only belong to a small number of such clusters (typically, at most $m^{1/O(\log\log m)}$). This in turn ensures that, in graph $H$, for every regular vertex $v\in V(H)$, the total number of edges incident to $v$ that ever belong to $H$ is also bounded by $m^{1/O(\log\log m)}$. We refer to this bound as the \emph{dynamic degree bound}, and denote it by $\mu$. Therefore, if we denote by $N(H)$ the number of regular vertices that belong to the initial graph $H$, then the total number of edges that ever belong to $H$ is bounded by $N(H)\cdot \mu$. This allows us to use the number of regular vertices of $H$ as a proxy to bounding the number of edges in $H$.

To summarize, the definition of the \recdynnc problem is almost identical to that of the Sparse \NC problem. The main difference is that the input graph now has a specific structure (that is, it is a bipartite graph), and, in addition to edge-deletions, we also need to support isolated vertex deletions and supernode-splitting updates. Additional minor difference is that we only require that the covering properties of the neighborhood cover hold for the regular vertices of $H$ (and not necessarily the supernodes), and we only bound the number of clusters ever containing a vertex for regular vertices (and not supernodes). These are minor technical details that are immaterial to this high-level overview. 


\paragraph{Procedure $\proccut$ and reduction to the \maintaincluster problem.}
One of the main building blocks of our algorithm is Procedure \proccut. Suppose our goal is to design an algorithm for the \recdynnc problem on input graph $H$, with target distance threshold $D$, and let $\cset$ be the neighborhood cover that we maintain. 
We denote by $N$ the number of regular vertices in the initial graph $H$, and, for each subgraph $H'\subseteq H$, we denote by $N(H')$ the number of regular vertices in $H'$.
Given a cluster $C\in \cset$, and two vertices $x,y\in C$, such that $\dist_C(x,y)> \Omega(D\poly\log N)$, procedure $\proccut$ produces two vertex-induced subgraphs $C',C''\subseteq C$, such that $N(C')\leq N(C'')$, $\diam(C')\leq O(D\poly\log N)$, and each of $C',C''$ contains exactly one of the two vertices $x,y$. Moreover, it guarantees that, for every vertex $v\in V(C)$, either $B_C(v,D)\subseteq C'$, or $B_C(v,D)\subseteq C''$ holds. We then add $C'$ to $\cset$, and update $C$ by deleting edges and vertices from it, until $C=C''$ holds. This procedure is exploited by our algorithm in two ways: first, we compute an initial strong $(D,D\cdot \poly\log N)$-neighborhood cover $\cset$ of the input graph $H$, before it undergoes any updates, by repeatedly invoking this procedure. Later, as the algorithm progresses, and update operations are applied to $H$, the diameters of some clusters $C\in \cset$ may grow. Whenever we identify such a situation, we use Procedure $\proccut$ in order to cut the cluster $C$ into smaller subclusters. We note that, if $C'$ and $C''$ are the outcome of applying Procedure \proccut to cluster $C$, then we cannot guarantee that the two clusters are disjoint, so they may share edges and vertices. Therefore, a vertex of $H$ may belong to a number of clusters in $\cset$. The main challenge in designing Procedure \proccut is to ensure that every vertex of $H$ only belongs to a small number of clusters (at most $N^{O(1/\log\log N)}$) over the course of the entire algorithm. The procedure uses a carefully designed modification of the ball-growing technique  of \cite{LR} that allows us to ensure this property. 
We note that several previous works used the ball-growing technique in order to compute and maintain a clustering of a graph. For example, \cite{chechikLowStretch} employ this technique in order to maintain clustering at every distance scale. However, the clusters that they maintain at each distance scale are disjoint, and so they can use the standard ball-growing procedure of \cite{LR} in order to ensure that relatively few edges have endpoints in different clusters. In contrast, in order to maintain a neighborhood cover, we need to allow clusters at each distance scale to overlap. While one can easily adapt the standard ball-growing procedure of \cite{LR} to still ensure that the total number of edges in the resulting clusters is sufficiently small, this would only ensure that every vertex belongs to relatively few clusters {\bf on average}. It is the strict requirement that {\bf every} vertex may only ever belong to few clusters in the neighborhood cover that makes the design of Procedure \proccut challenging. We are not aware of any other work that adapted the ball-growing technique to this type of requirement, except for the algorithm of \cite{neighborhood-cover2,neighborhood-cover1}, who did so in the static setting. It is unclear though how to adapt their techniques to the dynamic setting.

We also use Procedure \proccut to reduce the \recdynnc problem to a new problem, that we call \maintaincluster. In this problem, we are given some cluster $C$ that was just added to the neighborhood cover $\cset$. The goal is to support  queries $\spquery(C,v,v')$: given a pair $v,v'\in V(C)$ of vertices of $C$, return a path $P$ connecting $v$ to $v'$ in $C$, of length at most $\alpha\cdot D$, in time $O(|E(P)|)$. However, the algorithm may, at any time, raise a flag $F_C$, to indicate that the diameter of $C$ has become too large. When flag $F_C$ is raised, the algorithm must provide two vertices $x,y\in C$, with $\dist_C{(x,y)}>\Omega(D\poly\log N)$. The algorithm then obtains a sequence of update operations (that we call a \emph{flag-lowering sequence}), at the end of which either $x$ or $y$ are deleted from $C$, and flag $F_C$ is lowered. Queries $\spquery$ may only be asked when the flag $F_C$ is down. Once flag $F_C$ is lowered, the algorithm may raise it again immediately, as long as it supplies a new pair $x',y'\in V(C)$ of vertices with $\dist_C{(x',y')}>\Omega(D\poly\log N)$. Intuitively, once flag $F_C$ is raised, we will simply run Procedure \proccut on cluster $C$, with the vertices $x,y$ supplied by the algorithm, and obtain two new clusters $C'$, $C''$. Assume that $C'$ contains fewer regular vertices than $C''$. We then add $C'$ to $\cset$, and delete edges and vertices from $C$ until $C=C''$ holds, thus creating a flag-lowering update sequence for it. In order to obtain an algorithm for the \recdynnc problem, it is then enough to obtain an algorithm for the \maintaincluster problem. We focus on this problem in the remainder of this exposition.

\paragraph{Pseudocuts, expanders, and their embeddings.}

The next central tool that we introduce is balanced pseudocuts. Consider a cluster $C$, for which we would like to solve the \maintaincluster problem, as $C$ undergoes a sequence of online updates, with  target distance threshold $D$. 
For a given balance parameter $\rho$, a standard balanced multicut for $C$ can be defined as a set $E'\subseteq E(C)$ of edges, such that every connected component of $C\setminus E'$ contains at most $N(C)/\rho$ regular vertices. We weaken this notion of balanced multicut, and use instead \emph{balanced pseudocuts}. Let $D'=\Theta(D\poly\log N)$. A $(D',\rho)$-pseudocut in cluster $C$ is a collection $E'$ of its edges, such that, in graph $C\setminus E'$, for every vertex $v\in V(C)$, the ball $B_{C\setminus E'}(v,D')$ contains at most $N(C)/\rho$ regular vertices. In particular, once all edges of $E'$ are deleted from $C$, if we compute a strong $(D,D')$-neighborhood cover $\cset'$ of $C$, then we are guaranteed that for all $C'\in \cset'$, $N(C')\leq N(C)/\rho$. We note that standard balanced multicuts also achieve this useful property. An advantage of using pseudocuts is that we can design a near-linear time algorithm that computes a $(D',\rho)$-pseudocut $E'$ in graph $C$, for $\rho=N^{\eps}$, and additionally it computes an expander $X$, whose vertex set is $\set{v_e\mid e\in E''}$, where $E''\subseteq E'$ is a large subset of the edges of $E'$, together with an embedding of $X$ into $C$, via short embedding paths, that causes a low edge-congestion  (see \Cref{thm: compute pseudocut and expander} for details).
This allows us to build on known expander-based techniques in order to design an efficient algorithm for the \maintaincluster problem. Consider the following algorithm, that consists of a number of phases. In every phase, we start by computing a $(D',\rho)$-pseudocut $E'$ of $C$, the corresponding expander $X$, and its embedding into $C$. Let $E''\subseteq E'$ be the set of edges $e$, whose corresponding vertex $v_e$ lies in the expander $X$, so $V(X)=\set{v_e\mid e\in E''}$. We then use two data structures. The first data structure is an \EST $\tau$, whose root $s$ is a new vertex, that connects to each endpoint of every edge in $E''$, and has depth $O(D\poly\log N)$. This data structure allows us to ensure that every vertex of $C$ is close enough to some edge of $E''$, and to identify when this is no longer the case, so that flag $F_C$ is raised. Additionally, we use known algorithms for \APSP on expanders, together with the algorithm of  \cite{expander-pruning} for expander pruning, in order to maintain the expander $X$ (under update operations performed on the cluster $C$), and its embedding into $C$. This allows us to ensure that all edges in $E''$ remain sufficiently close to each other. These two data structures are sufficient in order to support the $\spquery(C,v,v')$ queries. If the initial pseudocut $E'$ was sufficiently large, then these data structures can be maintained over a long enough sequence of update operations to cluster $C$. Once a large enough number of edges are deleted from $C$, expander $X$ can no longer be maintained, and we recompute the whole data structure from scratch. Therefore, as long as the pseudocut $E'$ that our algorithm computes is sufficiently large (for example, its cardinality is at least $(N(C))^{1-\eps}$), we can support the $\spquery(C,v,v')$ queries as needed, with a very efficient algorithm. 

It now remains to deal with the situation where the size of the pseudocut $E'$ is small. One simple way to handle it is to maintain $2|E'|$ \EST data structures, each of which is rooted at an  endpoint of a distinct edge of $E'$, and has depth threshold $\Theta(D\poly\log N)$.
As long as the root vertex of an \EST $\tau$ remains in the current cluster $C$, we say that the tree $\tau$ \emph{survives}. As long as at least one of the \ESTs rooted at the endpoints of the edges in $E'$ survives, we can support the $\spquery(C,v,v')$ queries using any such tree. We can also use such a tree in order to detect when the diameter of the cluster becomes too large, and, when this happens, to identify a pair $x,y$ of vertices of $C$ with $\dist_C(x,y)$ sufficiently large. Once every \EST that we maintain is destroyed, we are guaranteed that all edges of $E'$ are deleted from $C$. We can then iteratively apply Procedure \proccut in order to further decompose $C$ into a collection of low-diameter clusters (that is, we compute a collection $\cset'$ of subgraphs of $C$, such that $\cset'$ is a $(D,D')$-neighborhood cover for $C$). Since $E'$ was a $(D',\rho)$-pseudocut for the original cluster $C$, we are then guaranteed that every cluster in $\cset'$ is significantly smaller than $C$, and contains at most $N(C)/\rho$ regular vertices. We can then initialize the algorithm for solving the \maintaincluster problem on each cluster of $\cset'$.
This approach already gives non-trivial guarantees (though in order to optimize it, we should choose a different  threshold for the cardinality of $E'$: if $|E'|>\sqrt{N(C)}$, we should use the expander-based approach, and otherwise we should maintain the \EST's). 
Our rough estimate is that such an algorithm would result in total update time $O\left (m^{1.5+O(\eps)}\cdot (\log m)^{O(1/\eps^2)}\right )$, but it is still much higher than our desired update time.

In order to achieve our desired near-linear total update time, 
we exploit again the recursive composability properties of the \recdynnc problem. 
Specifically, consider the situation where the pseudocut $E'$ that we have computed is small, that is, $|E'|<(N(C))^{1-\eps}$, and consider the graph $H'=C\setminus E'$. For all $1\leq i\leq \ceil{\log D}$, we solve the \recdynnc problem in graph $H'$ with target distance threshold $D_i=2^i$ recursively. Fix some index $1\leq i\leq \ceil{\log D}$, and let $\cset_i$ be the initial strong $(D_i,D_i\cdot \poly\log N)$-neighborhood cover that this algorithm computes. The properties of the balanced pseudocut ensure that each cluster $C'\in \cset_i$ is significantly smaller that $C$: namely, $N(C')\leq N(C)/\rho\leq (N(C))^{1-\eps}$. Therefore, we can solve the \maintaincluster problem on each such cluster recursively, and we also do so for every cluster that is later added to $\cset_i$. Let $\tilde\cset=\bigcup_i\cset_i$ be the dynamic collection of clusters that we maintain.

We use the set $\tilde\cset$ of clusters in order to construct a \emph{contracted graph} $\hat H$. The vertex set of $\hat H$ consists of the set $S$ of regular vertices -- all regular vertices that serve as endpoints of the edges of $E'$ (the edges of the pseudo-cut); and the set $U'=\set{u(C')\mid C'\in \tilde\cset}$ of supernodes. For every edge $e=(u,v)\in E'$, where $v\in S$ is a regular vertex, we add an edge connecting $v$ to every supernode $u(C')$, such that cluster $C'$ contains either $v$ or $u$. The length of the edge is $D_i$, where $i$ is the index for which $C'\in \cset_i$ holds.
It is not hard to show that the distances between the vertices of $S$ are approximately preserved in graph $\hat H$. As cluster $C$ undergoes a sequence of update operations, the neighborhood covers $\cset_i$ evolve, which in turn leads to changes in the contracted graph $\hat H$. However, we ensure that all changes to the neighborhood covers $\cset_i$ are only of the types allowed by
\Cref{thm: NC}, namely: (i) delete an edge from a cluster of $\cset_i$; (ii) delete an isolated vertex from a cluster  of $\cset_i$; or (iii) add a new cluster $C''$ to $\cset_i$, where $C''\subseteq C'$ for some cluster $C'\in \cset_i$. We are then guaranteed that all resulting changes to graph $\hat H$ can be implemented via allowed update operations: namely edge deletions, isolated vertex deletions, and supernode splitting.

We construct two data structures. First, an \EST $\tau$, in the graph obtained from $C$ by adding a new source vertex $s^*$, that connects to every vertex in $S$ with a length-$1$ edge. The depth of the tree is $O(D\poly\log N)$. This data structure allows us to ensure that every vertex of $C$ is sufficiently close to some vertex of $S$, and, when this is no longer the case, to raise the flag $F_C$, and to supply two vertices $x,y\in V(C)$ that are sufficiently far from each other.

The second data structure is obtained by applying the algorithm for the \maintaincluster problem recursively to the contracted graph $\hat H$. This data structure allows us to ensure that all vertices of $S$ are sufficiently close to each other, and, when this is no longer the case, it supplies a pair of vertices $s,s'\in S$, that are sufficiently far from each other in $\hat H$, and hence in $C$. Since we only use this algorithm in the scenario where $|E'|\leq (N(C))^{1-\eps}$, we are guaranteed that $|S|\leq (N(C))^{1-\eps}$, so graph $\hat H$ is significantly smaller than $C$.

To summarize, in order to solve the \maintaincluster problem in graph $C$, we use an expander-based approach, as long as the size of the pseudocut $E'$ that our algorithm computes is above $(N(C))^{1-\eps}$. Once this is no longer the case, we recursively solve the problem on clusters that are added to the neighborhood covers $\cset_i$ of graph $H=C\setminus E'$, for $1\leq i\leq \ceil{\log D}$. 
This allows us to maintain the neighborhood covers $\set{\cset_i}$, which, in turn, allow us to maintain the contracted graph $\hat H$. We then solve the \maintaincluster problem recursively on the contracted graph $\hat H$.
Once all edges of $E'$ are deleted from $C$, we start the whole algorithm from scratch. Since we ensure that the diameter of $C$ is bounded by $D'$, from the definition of a balanced pseudocut, we are guaranteed that $N(C)$ has decreased by at least a factor $N^{\eps}$.

\paragraph{Directions for future improvements.}
A major remaining open question is whether we can obtain an algorithm for decremental \APSP with a significantly better approximation factor, while preserving the near-linear total update time and the near-optimal query time. 
While it seems plausible that the new tools presented in this paper may lead to an improved $(\log m)^{\poly(1/\eps)}$-approximation algorithm with similar running time guarantees, 
improving the approximation factor beyond the $(\log m)^{\poly(1/\eps)}$ barrier seems quite challenging. A necessary first step toward such an improvement is to obtain better approximation algorithms for the decremental \APSP problem on expanders. 
We believe that this is a very interesting problem in its own right, and it is likely that better algorithms for this problem will lead to better deterministic algorithms for basic cut and flow problems, including Minimum Balanced Cut and Sparsest Cut, via the techniques of \cite{detbalanced}. This, however, is not the only barrier to obtaining an  approximation factor below $(\log m)^{\poly(1/\eps)}$ for decremental \APSP in near-linear time. In order to bring the running time of the algorithm for the \recdynnc problem down from $O\left (m^{1+O(\eps)}\cdot \poly(D)\cdot (\log m)^{O(1/\eps^2)}\right )$  to the desired running time of $O\left (m^{1+O(\eps)}\cdot (\log m)^{O(1/\eps^2)}\right )$, we recursively compose instances of \recdynnc with each other. This leads to recursion depth $O(1/\eps)$, and unfortunately the approximation factor accumulates with each recursive level. If the running time of our basic algorithm for \recdynnc (see \Cref{thm: main dynamic NC algorithm}) can be improved to depend linearly instead of cubically on $D$, it seems conceivable that one could use the approach of \cite{BernsteinChechik,Bernstein}, together with Layered Core Decomposition of \cite{APSP-old} in order to avoid this recursion (though it is likely that, in the running time of the resulting algorithm, term $m^{1+O(\eps)}$ will be replaced with $n^{2+O(\eps)}$). 
Lastly, our algorithm for the \maintaincluster problem needs to call to itself recursively on the contracted graph $\hat H$, which again leads to a recursion of depth $O(1/\eps)$, with the approximation factor accumulating at each recursive level. One possible direction for reducing the number of the recursive levels is designing an algorithm for computing a pseudocut $E'$, its corresponding expander $X$, and an embedding of $X$ into the cluster $C$ with a better balance parameter $\rho$  (see \Cref{thm: compute pseudocut and expander}). We believe that obtaining an analogue of \Cref{thm: compute pseudocut and expander} with stronger parameters and faster running time is a problem of independent interest. 


\subsection{Organization}
Most of this paper focuses on the proof of \Cref{thm: NC}. We start with preliminaries in \Cref{sec: prelims}, and then define the Recursive Dynamic Neighborhood Cover problem (\recdynNC), and state our main result for it in \Cref{sec: valid inputs}. We also show in \Cref{sec: valid inputs} that the proofs of \Cref{thm: main} and \Cref{thm: NC} follow from this result. After that, we gradually introduce our new technical tools. In \Cref{sec: proccut} we describe and analyze Procedure \proccut, and use it in order to reduce the \recdynnc problem to a new problem that we define in the same section, called \maintaincluster. In the same section, we show a slow and simple algorithm for the \maintaincluster problem, that we will use as a recursion base. In \Cref{sec: pseudocuts} we define balanced pseudocuts, and provide our main algorithm for them, that allows us to compute a pseudocut in a given cluster $C$, together with the corresponding expander $X$, and its embedding into $C$. We also provide an algorithm for \APSP on expanders, that is implicit in \cite{APSP-old}. Lastly, we show that these new tools already lead to a somewhat faster algorithm for the \maintaincluster problem. In \Cref{sec: pivot decomposition} we define good clusters; intuitively, these are clusters on which we can solve the \maintaincluster problem using the tools that we have developed so far, and we provide an algorithm that does exactly that. We then define a contracted graph and analyze its properties in \Cref{sec: conracted graph}, and we complete the algorithm for the \maintaincluster problem, by combining all tools that we have introduced, in \Cref{sec: final proof}. This completes the proof of \Cref{thm: main} and \Cref{thm: NC}. In \Cref{sec: params} we summarize all main parameters used in the proof.
Lastly, in \Cref{sec: application} we provide our algorithm for \MMF and \MM, proving \Cref{thm: MMF and MM}.

%% file: prelims.tex
\section{Preliminaries}\label{sec: prelims}

All logarithms in this paper are to the base of $2$. 
All graphs in this paper are undirected. Graphs may have parallel edges, except for simple graphs, that cannot have them.
Throughout the paper, we use a $\tilde O(\cdot)$ notation to hide multiplicative factors that are polynomial in $\log m$ and $\log n$, where $m$ and $n$ are the number of edges and vertices, respectively, in the initial input graph.

We follow standard graph-theoretic notation.   Given a graph $G=(V,E)$ and two disjoint subsets $A,B$ of its vertices, we denote by $E_G(A,B)$ the set of all edges with one endpoint in $A$ and another in $B$, and by $E_G(A)$ the set of all edges with both endpoints in $A$. We also denote by $\delta_G(A)$ the set of all edges with exactly one endpoint in $A$. For a vertex $v\in V(G)$, we denote by $\delta_G(v)$ the set of all edges incident to $v$ in $G$, and by $\deg_G(v)$ the degree of $v$ in $G$.
We may omit the subscript $G$ when clear from context. Given a subset $S\subseteq V$ of vertices of $G$, we denote by $G[S]$ the subgraph of $G$ induced by $S$.

Given a graph $G$ and a weight function $w: V(G)\rightarrow \reals$ on its vertices, for a subset $V'\subseteq V(G)$ of its vertices, we denote by $W(V')=\sum_{v\in V'}w(v')$ the total weight of all vertices in $V'$. Abusing the notation, for a subgraph $C\subseteq G$, we denote the weight of the subgraph $W(C)=\sum_{v\in V(C)}w(v)$.

Given a graph $G$ with lengths $\ell(e)\geq 0$ on edges $e\in E(G)$, for a pair of vertices $u,v\in V(G)$, we denote by $\dist_G(u,v)$ the \emph{distance} between $u$ and $v$ in $G$, that is, the length of the shortest path between $u$ and $v$ with respect to the edge lengths $\ell(e)$.  For a pair $S,T$ of subsets of vertices of $G$, we define the distance between $S$ and $T$ to be $\dist_G(S,T)=\min_{s\in S,t\in T}\dist_G(s,t)$.
For a vertex $v\in V(G)$, and a vertex subset $S\subseteq V(G)$, we also define the distance between $v$ and $S$ as $\dist_G(v,S)=\min_{u\in S}\dist_G(v,u)$.
 The \emph{diameter} of the graph $G$, denoted by $\diam(G)$, is the maximum distance between any pair of vertices in $G$. For a path $P$ in $G$, we denote its length by $\ell_G(P)=\sum_{e\in E(P)}\ell(e)$.
 For a vertex $v\in V(G)$ and a distance parameter $D\geq 0$, we denote by $B_G(v,D)=\set{u\in V(G)\mid \dist_G(u,v)\leq D}$ the \emph{ball of radius $D$ around $v$}.
Similarly, for a vertex subset $S\subseteq V(G)$, we let the ball of radius $D$ around $S$ be $B_G(S,D)=\set{u\in V(G)\mid \dist_G(u,S)\leq D}$.
 We will sometimes omit the subscript $G$ when clear from context.

\paragraph{Neighborhood Covers.}
Neighborhood Cover is a central notion that we use throughout the paper. We use both a strong and a weak notion of neighborhood covers, that are defined as follows.

\begin{definition}[Neighborhood Cover]
	Let $G$ be a graph with lengths $\ell(e)>0$ on edges $e\in E(G)$, let $S\subseteq V(G)$ be a subset of its vertices, and let $D\leq D'$ be two distance parameters. A \emph{weak $(D,D')$-neighborhood cover} for vertex set $S$ in $G$ is a collection $\cset=\set{C_1,\ldots,C_r}$ of vertex-induced subgraphs of $G$ called \emph{clusters}, such that:
	
	\begin{itemize}
		\item for every vertex $v\in S$, there is some index $1\leq i\leq r$ with $B_G(v,D)\subseteq V(C_i)$; and
		\item for all $1\leq i\leq r$, for every pair $s,s'\in S\cap V(C_i)$ of vertices, $\dist_{G}(s,s')\leq D'$.
	\end{itemize} 
	A set $\cset$ of subgraphs of $G$ is a \emph{strong $(D,D')$-neighborhood cover} for vertex set $S$ if it is a weak $(D,D')$-neighborhood cover for $S$, and, additionally, for every cluster $C\in \cset$, for every pair $s,s'\in S\cap V(C)$ of vertices, $\dist_{C}(s,s')\leq D'$.
	
	If the vertex set $S$ is not specified, then we assume that $S=V(G)$.
\end{definition}

\paragraph{Dynamic Graphs.}
Throughout, we consider a graph $G$ that undergoes an online sequence $\Sigma=(\sigma_1,\sigma_2,\ldots)$ of update operations. For now it may be convenient to think of the update operations being edge deletions, though we will consider additional update operations later. After each update operation (e.g. edge deletion), our algorithm will perform some updates to the data structure. 
We refer to different ``times'' during the algorithm's execution. The algorithm starts at time $0$. For each $t\geq 0$, we refer to ``time $t$ in the algorithm's execution'' as the time immediately after all updates to the data structures maintained by the algorithm following the $t$th update $\sigma_t\in \Sigma$ are completed. When we say that some property holds at every time in the algorithm's execution, we mean that the property holds at all times $t$ of the algorithm's execution, but it may not hold, for example, during the procedure that updates the data structures maintained by the algorithm, following some input update operation $\sigma_t\in \Sigma$.
For $t\geq 0$, we denote by $G^t$ the graph $G$ at time $t$; that is, $G^0$ is the original graph, and for $t\geq 0$, $G^t$ is the graph obtained from $G$ after the first $t$ update operations $\sigma_1,\ldots,\sigma_t$.

\paragraph{Even-Shiloach Trees~\cite{EvenS,Dinitz,HenzingerKing}.}
Suppose we are given a graph $G=(V,E)$ with integral lengths $\ell(e)\geq 1$ on its edges $e\in E$, a source $s$, and a distance bound $D\geq 1$. Even-Shiloach Tree (\EST) algorithm maintains, for every vertex $v$ with $\dist_G(s,v)\leq D$, the distance $\dist_G(s,v)$, under the deletion of edges from $G$. Moreover, it maintains a shortest-path tree $\tau$ rooted at vertex $s$, that includes all vertices $v$ with $\dist_G(s,v)\leq D$. We denote the corresponding data structure by $\EST(G,s,D)$. 
The total running time of the algorithm, including the initialization and all edge deletions, is $O(m\cdot D\log n)$, where $m$ is the initial number of edges in $G$ and $n=|V|$. Throughout this paper, we refer to the corresponding data structure as \emph{basic \EST}. We later define a slight generalization of this data structure that allows us to handle some additional update operations. This more general data structure will be referred to as \emph{generalized \EST}.

\paragraph{Cuts, Sparsity and Expanders.}
Given a graph $G$, a \emph{cut} in $G$ is a bipartition $(A,B)$ of the set $V(G)$ of its vertices, with $A,B\neq \emptyset$. The \emph{sparsity} of the cut $(A,B)$ is $\phi_G(A,B)=\frac{|E_G(A,B)|}{\min \set{|A|,|B|}}$. We denote by $\Phi(G)$ the smallest sparsity of any cut in $G$, and we refer to $\Phi(G)$ as the \emph{expansion} of $G$.

\begin{definition}[Expander]
	We say that a graph $G$ is a $\phi$-expander iff $\Phi(G)\geq \phi$.
\end{definition}

We will repeatedly use the following standard observation, whose proof is included in Section \ref{subsec: short paths in expanders} of Appendix for completeness.
\begin{observation}\label{obs: short paths in exp}
	Let $G$ be an $n$-vertex $\phi$-expander with maximum vertex degree at most $\Delta$. Then for any pair $u,v$ of vertices of $G$, there is a path connecting $u$ to $v$ that contains at most $\frac{8\Delta \log n}{\phi}$ edges.
\end{observation}

\paragraph{Expander Pruning.}

We use an algorithm for expander pruning by \cite{expander-pruning}. We slightly rephrase it so it is defined in terms of graph expansion, instead of conductance that was used in the original paper. For completeness, we include the proof in Section \ref{subsec: expander pruning} of Appendix.

\begin{theorem}[Adaptation of Theorem 1.3 in~\cite{expander-pruning}]\label{thm: expander pruning}
	There is a deterministic algorithm, that, given an access to the adjacency list of a graph $G=(V,E)$ that is a $\phi$-expander, for some parameter $0<\phi<1$, such that the maximum vertex degree in $G$ is at most $\Delta$, and a sequence $\Sigma=(e_1,e_2,\ldots,e_k)$ of $ {k\leq \phi |E|/(10\Delta)}$ online edge deletions from $G$, maintains a vertex set $S\subseteq V$ with the following properties. Let $G^i$ denote the graph $G\setminus\set{e_1,\ldots,e_i}$; let $S_0=\emptyset$ be the set $S$ at the beginning of the algorithm, and for all $0<i\leq k$, let $S_i$ be the set $S$ after the deletion of the edges of $e_1,\ldots,e_i$ from graph $G$. Then, for all $1\leq i\leq k$:
	
	\begin{itemize}
		\item $S_{i-1}\subseteq S_i$;
		\item $ |S_i|\leq 8i\Delta/\phi$;
		\item $|E(S_i,V\setminus S_i)|\leq 4i$; and
		\item graph $G^i[V\setminus S_i]$ is a $\phi/(6\Delta)$-expander.
	\end{itemize}
	
	The total running time of the algorithm is $\Otilde(k\Delta^2/\phi^2)$.
\end{theorem}

%% file: recdynNCproblem.tex

\section{Valid Input Structure, Valid Update Operations, Generalized Even-Shiloach Trees, and the Dynamic Recursive Neighborhood Cover Problem}
\label{sec: valid inputs}

Throughout this paper, we will work with inputs that have a specific structure. The structure is designed in a way that will allow us to naturally compose different instances recursively, by exploiting the notion of neighborhood covers. In order to avoid repeatedly defining such inputs, we provide a definition here, and then refer to it throughout the paper. In this section we also define the types of update operations that we allow for such inputs, and extend the basic Even-Shiloach Tree data structure to support such updates. Lastly, we formally define the Recursive Dynamic Neighborhood Cover problem (\recdynnc) and state our main result for this problem. We then show that this result immediately implies the proofs of Theorems \ref{thm: main} and \ref{thm: NC}.

\subsection{Valid Input Structure and Valid Update Operations}

We start by defining a valid input structure; the definition is used throughout the paper and is intended as a shorthand for the types of inputs most our subroutines use. 

\begin{definition}[Valid Input Structure]
	A valid input structure consists of a bipartite graph $H=(V,U,E)$, a distance threshold $D>0$ and integral lengths $1\leq \ell(e)\leq D$ for edges $e\in E$. 
	The vertices in set $V$ are called \emph{regular vertices} and the vertices in set $U$ are called \emph{supernodes}.
	We denote a valid input structure by $\iset=\left(H=(V,U,E),\set{\ell(e)}_{e\in E(H)}, D\right )$. If the distance threshold $D$ is not explicitly defined, then we set it to $\infty$.
\end{definition}

Intuitively, supernodes in set $U$ correspond to clusters in a Neighborhood Cover $\cset$ of the vertices in $V$ with some (smaller) distance threshold, that is computed and maintained recursively. %
Given a valid input structure $\iset=\left(H,\set{\ell(e)}_{e\in E(H)}, D\right )$, we will allow the following types of update operations, that we refer to as \emph{valid update operations}:

\begin{itemize}

\item {\bf Edge Deletion.} Given an edge $e\in E(H)$, delete $e$ from $H$.

\item {\bf Isolated Vertex Deletion.} Given a vertex $x\in V(H)$ that is an isolated vertex, delete $x$ from $H$; and



\item {\bf Supernode Splitting.} The input of this update operation is a supernode $u\in U$ and a non-empty subset $E'\subseteq \delta_H(u)$ of edges incident to $u$.
The update operation creates a new supernode $u'$, and, for every edge $e=(u,v)\in E'$, it adds a new edge $e'=(u',v)$ of length $\ell(e)$ to the graph $H$.  We will sometimes refer to $e'$ as a \emph{copy of edge $e$}.
\end{itemize}

For brevity of notation, we will refer to edge-deletion, isolated vertex deletion, and supernode-splitting operations as \emph{valid update operations}.
Notice that the update operations may not create new regular vertices, so vertices may be deleted from the vertex set $V$, but never added to it. A supernode splitting operation, however, adds a new supernode to the graph $H$, and also inserts edges into $H$.
Unfortunately, this means that the number of edges in $H$ may grow as the result of the update operations, which makes it challenging to analyze the running times of various algorithms that we run on subgraphs $C$ of $H$ in terms of $|E(C)|$. 
In order to overcome this difficulty, we use the notion of the \emph{dynamic degree bound}.

\begin{definition}[Dynamic Degree Bound]
We say that a valid input structure \newline  $\iset=\left(H=(V,U,E),\set{\ell(e)}_{e\in E(H)}, D\right )$, undergoing a sequence $\Sigma$ of valid update operations has dynamic degree bound $\mu$ iff for every regular vertex $v\in V$, the total number of edges incident to $v$ that are ever present in $H$ over the course of the update sequence $\Sigma$ is at most $\mu$. 
\end{definition}

We will usually denote by $N^0(H)$ the  number of regular vertices in the initial graph $H$. If $(\iset,\Sigma)$ have dynamic degree bound $\mu$, then we are guaranteed that the number of edges that are ever present in $H$ over the course of the update sequence $\Sigma$ is bounded by $N^0(H)\cdot \mu$.

In general, we will always ensure that the dynamic degree bound $\mu$ is quite low. It may be convenient to think of it as $m^{o(1)}$, where $m$ is the initial number of edges in the input graph $G$ for the \APSP problem. Intuitively, every supernode of graph $H$ represents some cluster $C$ in a  $(\hat D,\hat D')$-neighborhood cover $\cset$ of $G$, for some parameters $\hat D,\hat D'\ll D$. Typically, each regular vertex of $H$ represents some actual vertex of graph $G$,  and an edge $(v,u)$ is present in $H$ iff vertex $v$ belongs to the cluster $C$ that vertex $u$ represents. Intuitively, we will ensure that the neighborhood cover $\cset$ of $G$ is constructed and maintained in such a way that the total number of clusters of $\cset$ to which a given regular vertex $v$ ever belongs over the course of the algorithm is very small. This will ensure that the dynamic degree bound for graph $H$ is small as well. 

Note that we can assume without loss of generality that every vertex in the original graph $H^0$ has at least one edge incident to it, as otherwise it is an isolated vertex, and will remain so as long as it lies in $H$. Moreover, from the definition of a supernode-splitting operation, it may not be applied to an isolated vertex (as we require that the edge set $E'$ is non-empty). Therefore, any isolated vertex of $H^0$ can be ignored. We will therefore assume from now on that every supernode in the original graph $H^0$ has degree at least $1$. (This assumption is only used for convenience, so that we can bound the total number of vertices in $H^0$ by $O(|E(H^0)|)$.)

We use the following simple observation to show that the distances in the graph $H$ may not decrease as the result of a valid update operation.

\begin{observation}\label{obs: no dist increase}
	Consider the graph $H$ at any time during the execution of the sequence $\Sigma$ of valid update operations, and let $x,x'$ be any two vertices of $H$. Let $H'$ be the graph obtained after a single valid update operation on $H$. Then $\dist_{H'}(x,x')\geq \dist_{H}(x,x')$.
\end{observation}

\begin{proof}
	If the valid update operation is an edge deletion, or an isolated vertex deletion, then the statement is clearly true. Assume now that the valid update operation is supernode-splitting, applied to a supernode $u$, with the corresponding edge set $E'$. It is easy to verify that $\dist_{H'}(x,x')=\dist_H(x,x')$, since every path $P$ connecting $x$ to $x'$ in $H'$ can be transformed into a path $P'$ of the same length connecting $x$ to $x'$ in $H$, by replacing the vertex $u'$ with the vertex $u$ on it (if $u'\in V(P)$; otherwise $P\subseteq H$ holds). 
\end{proof}

Throughout the paper, it may sometimes be convenient for us to work with vertex weights. In such cases, we will always let the weight of every supernode be $0$, and the weight of every regular vertex be $1$. For a subgraph $C\subseteq H$, we can then denote by $W(C)$ the total weight of all vertices in $C$. From the above discussion, $|E(C)|\leq W(C)\cdot \mu$ always holds. Some of the technical tools that we develop work even when vertex weights are arbitrary. Since we believe that some of these tools are interesting in their own right, and may be useful for future work, we state them in the most general way, with arbitrary non-negative vertex weights. However, when applying these tools to our problem, we will always set vertex weights as described above.

\subsection{Generalized Even-Shiloach Trees}
\label{subsec: ES-trees}


In this subsection we show that the basic $\EST$ data structure of \cite{EvenS,Dinitz,HenzingerKing} can be extended to the setting of a valid input structure that undergoes valid update operations. We do so in the following theorem, whose proof is standard and is deferred to Section \ref{sec: proof of ES tree thm} of Appendix. We note that a similar data structure was used, either explicity or implicitly, in numerous previous papers.


\begin{theorem}\label{thm: ES-tree}
	There is a deterministic algorithm, that we refer to as \emph{generalized \EST}, that, given a valid input structure $\iset=\left(H,\set{\ell(e)}_{e\in E(H)},D\right )$, where graph $H$ undergoes a sequence $\Sigma$ of valid update operations with dynamic degree bound $\mu$, and given additionally a source vertex $s\in V(H)$, and a distance threshold $D^*>0$, such that the length of every edge in $H$ is bounded by $D^*$, supports  $\shortestpath$  queries: given a vertex $x\in V(H)$, either correctly establish, in time $O(1)$, that $\dist_H(s,x)>D^*$, or return a shortest $s$-$x$ path $P$ in $H$, in time $O(|E(P)|)$. The total update time of the algorithm is $\otilde(N^0\cdot \mu \cdot D^*)$, where $N^0$ is the number of regular vertices in the initial graph $H$.
\end{theorem}


\subsection{The Recursive Dynamic Neighborhood Cover (\recdynNC) Problem}
In this subsection we define the Recursive Dynamic Neighborhood Cover problem, and then state our main technical result that provides a deterministic algorithm for this problem. We then provide an algorithm with faster update time by exploiting the recursive composability properties of this problem. 
This faster algorithm, in turn, will immediately provide an algorithm for the \APSP problem, proving \Cref{thm: main}, and an algorithm for the \NC problem, proving \Cref{thm: NC}.

\paragraph{Problem Definition.}
The input to the Recursive Dynamic Neighborhood Cover (\recdynNC) problem is a valid input structure $\iset=\left(H=(V,U,E),\set{\ell(e)}_{e\in E},D \right )$, where graph $H$ undergoes a sequence $\Sigma$ of valid update operations with some given dynamic degree bound $\mu$. Additionally, we are given a desired approximation factor $\alpha$. We assume that we are also given some arbitrary fixed ordering $\oset$ of the vertices of $H$, and that any new vertex that is inserted into $H$ as the result of supernode-splitting updates appears at the end of the current ordering.
The goal is to maintain the following data structures:

\begin{itemize}
	\item   a collection $\uset$ of subsets of vertices of graph $H$, together with a collection $\cset=\set{H[S]\mid S\in\uset}$ of clusters in $H$, such that $\cset$ is a weak $(D,\alpha \cdot D)$ neighborhood cover for the set $V$ of regular vertices in graph $H$.
		For every set $S\in \uset$, the vertices of $S$ must be maintained in a list, sorted according to the ordering $\oset$;
	\item for every regular vertex $v\in V$, a cluster  $C=\clustercover(v)$, with $B_H(v,D)\subseteq V(C)$;
	\item for every vertex $x\in V(H)$, a list $\clusterlist(x)\subseteq \cset$ of all clusters containing $x$, and for every edge $e\in E(H)$, a list $\clusterlist(e)\subseteq \cset$ of all clusters containing $e$.
\end{itemize} 

The set $\uset$ of vertex subsets must be maintained as follows. Initially, $\uset=\set{V(H^0)}$, where $H^0$ is the initial input graph $H$. After that, the only allowed changes to vertex sets in $\uset$ are:

\begin{itemize}
	\item $\delvertex(S,x)$: given a vertex set $S\in \uset$, and a vertex $x\in S$, delete $x$ from $S$; 
	\item $\addsupernode(S,u)$: if $u$ is a supernode that is lying in $S$ that just underwent supernode splitting update, add the newly created supernode $u'$ to $S$; and
	\item $\csplit(S,S')$: given a vertex set $S\in \uset$, and a subset $S'\subseteq S$ of its vertices, add $S'$ to $\uset$.
\end{itemize}

We refer to the above operations as \emph{allowed changes to $\uset$}.
In other words, if we consider the sequence of changes that clusters in $\cset$ undergo over the course of the algorithm, the corresponding sequence of changes in vertex sets in $\set{U(C)\mid C\in \cset}$ must obey the above rules.


In addition to maintaining the above data structures, an algorithm for the \recdynnc problem needs to support queries $\spquery(C,v,v')$: given two {\bf regular} vertices $v,v'\in V$, and a cluster $C\in \cset$ with $v,v'\in C$, return a path $P$ in the current graph $H$, of length at most $\alpha\cdot D$ connecting $v$ to $v'$ in $H$, in time $O(|E(P)|)$. 
This completes the definition of the \recdynNC problem.



\subsection*{Statement of Main Technical Result}

Our main technical result is a deterministic algorithm for the \recdynNC problem, that is summarized in the following theorem.

\begin{theorem}\label{thm: main dynamic NC algorithm}
	There is a deterministic algorithm for the \recdynNC problem, that,  on input $\iset=\left(H=(V,U,E),\set{\ell(e)}_{e\in E},D \right )$ undergoing a sequence of valid update operations with dynamic degree bound $\mu$, and a parameter $c/\log \log W<\eps<1$, for some large enough constant $c$, where $W$ is the number of regular vertices in $H$ at the beginning of the algorithm, achieves approximation factor $\alpha=(\log(W\mu))^{2^{O(1/\eps)}}$, and  has total update time $O\left (W^{1+O(\eps)}\cdot \mu^{2+O(\eps)}\cdot D^3\cdot (\log (W\mu))^{O(1/\eps^2)}\right )$. Moreover, the algorithm ensures that for every regular vertex $v\in V$, the total number of clusters in the neighborhood cover $\cset$ that the algorithm maintains, to which vertex $v$ ever belongs over the course of the algorithm is bounded by $W^{O(1/\log\log W)}$. It also ensures that the neighborhood cover $\cset$ that it maintains is a strong $(D,\alpha\cdot D)$ neighborhood cover for the set $V$ of regular vertices of $H$.
\end{theorem}

The following sections of this paper focus on the proof of Theorem \ref{thm: main dynamic NC algorithm}. Notice however that the running time of the algorithm from \Cref{thm: main dynamic NC algorithm} is somewhat slow, and in particular it has a high dependence on the distance threshold $D$, that we would like to avoid. 
We now show an algorithm that exploits the recursive composability of instances of the \recdynNC problem in order to do so.

\subsection{A Faster Algorithm  for \recdynnc and Proofs of Theorems \ref{thm: main} and \ref{thm: NC}}

We prove the following theorem, that provides a faster algorithm for the \recdynNC problem, and follows from \Cref{thm: main dynamic NC algorithm}.
\begin{theorem}\label{thm: main final dynamic NC algorithm}
	There is a deterministic algorithm for the \recdynNC problem, that,  given a valid input structure $\iset=\left(H=(V,U,E),\set{\ell(e)}_{e\in E},D \right )$ undergoing a sequence of edge-deletion and isolated vertex-deletion operations, with dynamic degree bound $2$, and  a parameter $c/\log \log W<\eps<1$, for some large enough constant $c$, where $W$ is the number of regular vertices in $H$ at the beginning of the algorithm, achieves approximation factor $\alpha=(\log W)^{2^{O(1/\eps)}}$,  with total update time $O\left (W^{1+O(\eps)}\cdot (\log W)^{O(1/\eps^2)}\right )$. Moreover, the algorithm ensures that for every regular vertex $v\in V$, the total number of clusters in the weak neighborhood cover $\cset$ that the algorithm maintains, to which vertex $v$ ever belongs over the course of the algorithm, is bounded by $W^{O(1/\log\log W)}$.
\end{theorem}

We note that, since the dynamic degree bound is $2$, the degree of every regular vertex in the initial graph $H$ is at most $2$. Notice that supernode-splitting operations are not allowed in the above theorem, and so the degree of every regular vertex remains at most $2$ over the course of the algorithm. Therefore, the number of edges in $H$ is bounded by $2W$ throughout the algorithm.

We start by showing that Theorems \ref{thm: main} and \ref{thm: NC} immediately follow from \Cref{thm: main final dynamic NC algorithm}.

\subsubsection{Completing the proof of Theorem \ref{thm: NC}}
\label{subsubsection: completing the proof of main thm}

We assume that we are given an $m$-edge graph $G$ with integral length $\ell(e)\geq 1$ on its edges, that undergoes an online sequence of edge deletions, together with parameters $c'/\log \log m<\eps<1$ for some large enough constant $c'$, and $D\geq 1$.
Note that we can assume that $G$ is a connected graph, as otherwise we can run the algorithm on each of its connected components separately, so $|V(G)|\leq m$ holds. 
We construct a bipartite graph $H=(V,U,E)$ as follows. We start with the graph $G$, and we let $U=V(G)$ be the set of the supernodes of $H$. We then subdivide every edge $e\in E(G)$ with a new regular vertex $v_e$, and we set the lengths of both new edges to be $\ell(e)$. The set $V$ of regular vertices of $H$ is the union of two subsets: a set $\set{v_e\mid e\in E(G)}$ of vertices corresponding to edges of $G$, and another subset $S=\set{x'\mid x\in V(G)}$ of vertices corresponding to vertices of $G$. Every vertex $x'\in S$ connects to the corresponding vertex $x\in V(G)$ with a length-$1$ edge.  Once we delete all edges of length greater than $3D$, we obtain a valid input structure
$\iset=\left(H=(V,U,E),\set{\ell(e)}_{e\in E}, 3D\right )$. Given an online sequence $\Sigma$ of edge deletions for graph $G$, we can produce a corresponding online sequence $\Sigma'$ of edge deletions and isolated vertex deletions for graph $H$, as follows: whenever an edge $e\in E(G)$ is deleted from $G$, we delete its two corresponding edges (that are incident to $v_e$) from graph 
$H$, and we then delete vertex $v_e$ that becomes an isolated vertex. We have therefore obtained an instance $\iset$ of the \recdynNC problem, on valid input structure $\iset$ that undergoes a sequence of edge-deletion and isolated vertex-deletion operations. Since the degree of every regular vertex in $H$ is at most $2$, it is easy to see that $H$ has dynamic degree bound $2$. 
We let $W=|V|\leq 2m$ be the number of regular vertices in $H$. By letting $c'$ be a large enough constant, we can assume that $c/\log \log W<\eps<1$, where $c$ is the constant from \Cref{thm: main final dynamic NC algorithm}. We run the algorithm for the \recdynnc problem from \Cref{thm: main final dynamic NC algorithm} on input $\iset$ undergoing the sequence $\Sigma'$ of update operations. Let $\cset$ be the neighborhood cover that the algorithm maintains. We then define a neighborhood cover $\cset'$ for graph $G$ as follows. For every cluster $C\in \cset$, there is a cluster $C'\in \cset'$, which is a subgraph of $G$ induced by vertex set $\set{x\in V(G) \mid x'\in V(C)}$. Recall that cluster set $\cset$ 
	is initially defined to be $\cset=\set{H}$, so initially, $\cset'=\set{G}$ holds. After that, the only changes to vertex sets in $\uset=\set{U(C)\mid C\in \cset}$ are the allowed changes, that include:
	
	\begin{itemize}
		\item $\delvertex(R,x)$: given a vertex set $R\in \uset$, and a vertex $x\in R$, delete $x$ from $R$; if $x=y'$ for some vertex $y\in V(G)$, then we also delete $y$ from the cluster $C\in \cset'$ representing $H[R]$;
		 
		\item $\addsupernode(R,u)$: since we do not allow supernode splitting operations, no such updates will be performed; and
		\item $\csplit(\tilde R,R)$: given a vertex set $R\in \uset$, and a subset $\tilde R\subseteq R$ of its vertices, add $\tilde R$ to $\uset$. In this case, we create a new cluster in $\cset'$ that is a subgraph of $G$, induced by the set $\set{x\in V(G)\mid x'\in R'}$ of vertices.
	\end{itemize}
	
Additionally, when an edge $e$ is deleted from $G$, we need to delete it from every cluster of $\cset$ that contains it.	
The time that is needed to make all  these updates to cluster set $\cset'$ is subsumed by the time required to maintain cluster set $\cset$.

Consider now some vertex $x\in V(G)$. Recall that the algorithm for the \recdynnc problem maintains a cluster $C\in \cset$, with $B_{H}(x',3D)\subseteq V(C)$. It is easy to verify that $B_{H}(x',3D)$ contains, for every vertex $y\in B_G(x,D)$, the corresponding vertex $y'\in V$. We then set $\coveringcluster(x)=C'$, where $C'\in \cset$ is the cluster corresponding to $C$.

Lastly, suppose we are given a query 
	$\spquery(C',x,y)$, where $C'\in \cset'$, and $x,y\in V(C)$. We then run query $\spquery(C,x',y')$ in the data structure maintained by the algorithm for the \recdynnc problem, obtaining a path $P$ in graph $H$ that connects $x'$ to $y'$ and has length at most $3\alpha D$, for $\alpha=(\log m)^{2^{O(1/\eps)}}$. By deleting the first and the last vertex on $P$, and by suppressing some vertices, we obtain  a path $P'$ in graph $G$, connecting $x$ to $y$, whose length is at most $D\cdot (\log m)^{2^{O(1/\eps)}}$. The running time of this algorithm is $O(|E(P)|)$. The fact that we can support queries $\spquery(C',x,y)$, together with the above discussion, proves that the cluster set $\cset'$ that we maintain is a weak $(D,3\alpha D)$-neighborhood cover for graph $G$.
	Moreover, from \Cref{thm: main final dynamic NC algorithm}, the total update time of the algorithm is $O\left (m^{1+O(\eps)}\cdot (\log m)^{O(1/\eps^2)}\right )$, and we are also guaranteed that for every vertex $v\in V(G)$, the total number of clusters in $\cset'$ to which vertex $v$ ever belonged is at most $m^{O(1/\log\log m)}$.

\subsubsection{Completing the proof of \Cref{thm: main}}
\label{subsubsec: first thm proof}
By using standard scaling and rounding of edge lengths, at the cost of losing a factor $2$ in the approximation ratio, we can assume  that all edge lengths are integers between $1$ and $L$. 
We define $\lambda=\ceil{\log{L}}$ distance thresholds $D_1,\ldots,D_{\lambda}$, where for $1\leq i\leq \lambda$, $D_i=2^i$. 

For all $1\leq i\leq \lambda$, we apply the algorithm from \Cref{thm: NC} to graph $G$, for distance threshold $D_i$, the input parameter $\eps$, and approximation factor $\alpha=O\left ((\log m)^{2^{O(1/\eps)}}\right )$. We denote the neighborhood cover that the algorithm maintains by $\cset_i$ and the corresponding data structure by $\dset_i$. Note that the total update time of this algorithm, over all distance scales $D_i$ is $O\left (m^{1+O(\eps)}\cdot (\log m)^{O(1/\eps^2)}\cdot \log L\right )$, as required.
We now show algorithms for processing $\distquery$ and $\shortestpathquery$ queries.

Given a query $\distquery(x,y)$, for a pair $x,y\in V(G)$ of vertices, we perform a binary search on an integer $1\leq i\leq \lambda$, such that, if $C=\coveringcluster(x)$ is a cluster in $\cset_i$ that covers $x$, then $y\in V(C)$, but, if $C'=\coveringcluster(x)$ is the cluster of $\cset_{i-1}$ covering $x$, then $y\not \in C'$. In order to do so, we consider the current guess $i'$ on the integer $i$, find the cluster $C=\coveringcluster(x)$ in data structure $\dset_{i'}$ in time $O(1)$, and check whether $y\in V(C)$ in time $O(\log m)$. If this is the case, then in our next guess we increase the index $i'$; otherwise we decrease it. Once the desired index $i$ is found, we return the value $2D_i$. We claim that $D_i/2\leq \dist_G(x,y)\leq \alpha D_i$. Indeed, from the definition of weak $(D,\alpha D)$-neighborhood cover, since $x$ and $y$ lie in the same cluster of $\cset_i$, $\dist_G(x,y)\leq \alpha D_i$. Moreover, since $y$ does not lie in the cluster $C'=\coveringcluster(x)$ in $\cset_{i-1}$, we get that $y\not\in B_G(x,D_{i-1})$, and so $\dist_G(x,y)>D_{i-1}= D_i/2$. Since we perform a binary search over $\lambda=\ceil{\log L}$ values $i$, the running time of the algorithm for responding to the query is $O(\log m\log\log L)$.

Consider now a query $\shortestpathquery(x,y)$ for a pair $x,y\in V(G)$ of vertices. We start by computing an index $i$ exactly as in $\distquery(x,y)$, so that $D_i/2\leq \dist_G(x,y)\leq \alpha D_i$, together with the cluster  $C=\coveringcluster(x)$ in $\cset_i$, so that $y\in C$ holds. This can be done in time $O(\log m\log\log L)$ as before. We then run query $\spquery(C,x,y)$ in data structure $\dset_i$, to compute a path $P$ in graph $G$, connecting $x$ to $y$, of length at most $\alpha\cdot D_i\leq 2\alpha \dist_G(x,y)\leq (\log m)^{2^{O(1/\eps)}}\cdot \dist_G(x,y)$, in time $O(|E(P)|)$.

Therefore, in order to prove Theorems \ref{thm: main} and \ref{thm: NC}, it is enough to prove \Cref{thm: main final dynamic NC algorithm}.

\subsubsection{Proof of \Cref{thm: main final dynamic NC algorithm}}

In this subsection we prove \Cref{thm: main final dynamic NC algorithm} using the algorithm from \Cref{thm: main dynamic NC algorithm}.
We can assume that graph $H$  has no isolated vertices, as all such vertices can be ignored (e.g. each such vertex can be placed in a separate cluster).
	For simplicity, we denote by $n$ and $m$ the number of vertices and edges in the initial graph $H$. Notice that $m\leq 2W$, and $n\leq 3W$ must hold.

	We start by showing that, at the cost of losing a factor $2$ in the approximation ratio, we can assume that $D=2n$, and that all edge lengths are integers between $1$ and $D$. 
	We show this using standard arguments. Recall that, since $\iset$ is a valid input structure, all edges in $H$ have lengths at most $D$. Next, we set the length of each  edge $e$ to be $\ell'(e)=\ceil{n\ell(e)/D}$.
	
	For every pair $x,y$ of vertices, let $\dist'(x,y)$ denote the distance
	between $x$ and $y$ with respect to the new edge length values.
	Notice that for every pair $x,y$ of vertices, $\frac{n}{D}\cdot \dist(x,y)\leq\dist'(x,y)\leq\frac{n}{D}\cdot \dist(x,y)+n$,
	since the shortest $x$-$y$ path contains at most $n$ edges. 
	
Therefore, if $\dist(x,y)\leq D$, then $\dist'(x,y)\leq 2n$. Moreover, if $P$ is an $x$-$y$ path with $\ell'(P)\leq \alpha n$, then $\ell(P)\leq \alpha D$ must hold.
It is now enough to solve the problem on graph $H$ with the new edge weights $\ell'(e)$ for $e\in E(H)$, and distance bound $D'=2n$. 
Therefore, we assume from now on that $D=2n$, and that all edge lengths are integers between $1$ and $D$.

The main idea of the proof is to apply the algorithm from \Cref{thm: main dynamic NC algorithm} recursively, for smaller and smaller distance bounds. Specifically, we prove the following lemma by induction.
	
\begin{lemma}\label{lem: inductive dynamic NC algorithm}
	There is a universal constant $\tilde c$, and a deterministic algorithm for the \recdynNC problem, that,  given a valid input structure $\iset=\left(H=(V,U,E),\set{\ell(e)}_{e\in E(H)},D \right )$ undergoing a sequence of edge-deletion and isolated vertex-deletion operations, with dynamic degree bound $2$, and
  a parameter $0<\eps<1$,
such that $D \leq 6W^{\eps i}$ holds for some integer $i$, where $W$ is the number of regular vertices in $H$ at the beginning of the algorithm, achieves approximation factor $\alpha_i=(\log W)^{\tilde c i\cdot 2^{\tilde c/\eps}}$, and  has total updtate time at most $\left (\tilde c^i\cdot W^{1+\tilde c\eps}\cdot (\log W)^{\tilde c/\eps^2}\right )$. Moreover, the algorithm ensures that for every regular vertex $v\in V$, the total number of clusters in the weak neighborhood cover $\cset$ that the algorithm maintains, to which vertex $v$ ever belongs over the course of the algorithm, is bounded by $W^{O(1/\log\log W)}$.
\end{lemma}

Notice that the above lemma completes the proof of \Cref{thm: main final dynamic NC algorithm}, by setting $i=\ceil{1/\eps}$. The approximation factor achieved by the algorithm is at most $(\log W)^{2\tilde c\cdot 2^{\tilde c/\eps}/\eps}=(\log W)^{2^{O(1/\eps)}}$, and its running time is  $O\left (W^{1+O(\eps)}\cdot (\log W)^{O(1/\eps^2)}\right )$.
The proof of the lemma is somewhat technical and is deferred to Section \ref{sec:recursive composition} of Appendix. We provide here a short intuitive overview of the proof. The proof is by induction on $i$, where the base case, with $i=1$, follows immediately by invoking the algorithm from \Cref{thm: main dynamic NC algorithm}. In order to prove the lemma for some integer $i> 1$, we select a threshold $D'\approx 6W^{\eps(i-1)}$; we say that an edge $e\in E$ is long, if its length is greater than $D'$, and it is short otherwise. Let $H'$ be the graph obtained from $H$ after all long edges are deleted from it. We can then use the induction hypothesis in order to maintain a data structure $\dset(H')$, solving the \recdynnc problem on graph $H'$, with distance bound $\Theta(D')$, and approximation factor $\alpha_{i-1}=(\log W)^{\tilde c (i-1)\cdot 2^{\tilde c/\eps}}$. Let $\cset'$ be the collection of clusters maintained by the algorithm. Next, we define another graph $\hat H$, as follows. Initially, we start with $\hat H=H$, except that we increase all edge lengths to become integral multiples of $D'$; this does not distort the lengths of long edges by much, but it may increase lengths of short edges significantly. Additionally, for every cluster $C\in \cset'$, we add a supernode $u(C)$ to graph $\hat H$, that connects to every regular vertex $v\in V(C)$ with an edge of length $D'$. Let $\hat H^0$ be the initial graph $\hat H$, before any updates to graph $H$. We show that we can use the data structure $\dset(H')$, and the corresponding dynamic neighborhood cover $\cset'$, in order to produce a sequence of valid update operations for graph $\hat H$, so that, if the resulting sequence of update operations is applied to the initial graph $\hat H^0$, then the resulting dynamic graph is precisely $\hat H$. We also show that distances are approximately preserved in $\hat H$: that is, if $v,v'$ are two regular vertices in graph $H$ that lie at distance at most $D$ in $H$, then the distance between $v$ and $v'$ in $\hat H$ is also bounded by $O(D)$. Additionally, we show an algorithm that, given a path $P$ in graph $\hat H$, connecting any pair $v,v'$ of its regular vertices, produces a path $P'$ in graph $H$, connecting the same pair of vertices, such that the length of $P'$ is close to the length of $P$. Lastly, we scale all edge lengths in graph $\hat H$ down by factor $D'$, and use the algorithm from \Cref{thm: main dynamic NC algorithm} in order to maintain a neighborhood cover $\hat \cset'$ in the resulting graph $\hat H'$, with distance threshold $\hat D'=D/D'\leq W^{O(\eps)}$. We exploit this data structure (that we denote by $\dset(\hat H')$), and the fact that $V(H)\subseteq V(\hat H')$, in order to maintain neighborhood cover $\cset$ for graph $H$, as follows. Consider the collection of vertex sets $\hat \uset'=\set{V(C)\mid C\in \hat\cset'}$ that is maintained by data structure $\dset(\hat H')$.  For every vertex set $S\in \hat \uset'$, we define a set $S'=S\cap V(H)$. We then let $\cset=\set{H[S']\mid S\in \hat \uset'}$ be the neighborhood cover for graph $H$. In other words, we maintain the same set of clusters as $\hat \cset'$, except that we ignore vertices that do not lie in the graph $H$. This completes the high-level description of the proof of \Cref{lem: inductive dynamic NC algorithm}. 
A formal proof appears in  Section \ref{sec:recursive composition} of Appendix.

The remainder of this paper focuses on the proof of \Cref{thm: main dynamic NC algorithm}. In subsequent sections, we gradually introduce several new technical tools, which are used in order to define faster and faster algorithms for the \recdynnc problem, building up to the proof of  \Cref{thm: main dynamic NC algorithm}.

%% file: cut-procedure.tex
\section{First Set of Tools: Procedure \proccut, Initial Neighborhood-Cover Decomposition, and Reduction to \maintaincluster Problem}
\label{sec: proccut}

In this section we introduce one of the central tools that we use in order to maintain neighborhood cover --
Procedure \proccut. Intuitively,
suppose we are given a valid input structure $\iset=\left(H,\set{\ell(e)}_{e\in E(H)},D\right )$ that undergoes a sequence of valid update operations with dynamic degree bound $\mu$, and our goal is to solve the \recdynnc problem on it, for some approximation factor $\alpha$.
We set vertex weights in the usual way: the weight of every regular vertex is $1$, and the weight of each supernode is $0$. We assume that we are given some parameter $W\geq W^0(C)$, where $W^0(C)$ is the weight of all vertices in the initial graph $C$. 
We will maintain a neighborhood cover $\cset$ of $H$, starting from $\cset=\set{H}$.
Given some subgraph $C\subseteq H$ (that may lie in the neighborhood cover $\cset$), whose diameter is greater than $D'=\Theta(D\log^4W)$, the procedure produces two new clusters $C',C''\subseteq C$. It outputs  cluster $C'$ explicitly, by listing all its vertices and edges, while cluster $C''$ is output implicitly, by listing all vertices and edges of $C\setminus C''$. Cluster $C'$ is then added to the neighborhood cover $\cset$, and cluster $C\in \cset$ is updated so that $C=C''$ holds. We say that cluster $C'$ was \emph{split off of $C$}. In order to ensure that $\cset$ remains a valid neighborhood cover, we cannot guarantee that the two clusters $C'$ and $C''$ are disjoint, but we will ensure that they have a small overlap. At the beginning of our algorithm for the \recdynNC problem, we will compute an initial neighborhood cover $\cset$ via Algorithm \computenc, that is described later in this section, which iteratively employs Procedure \proccut on the input graph $H$. As the algorithm progresses and edges are deleted from $H$, the diameters of some clusters $C\in \cset$ may become large. Intuitively, whenever that happens, we would like to further decompose the cluster $C$ by employing Procedure \proccut again. In practice, we will define another problem, \maintaincluster, whose goal is to maintain the cluster $C\in \cset$ and to support queries $\spquery(C,v,v')$: given two {\bf regular} vertices $v,v'\in V(C)$, return a path $P$ in $C$, of length at most $\alpha\cdot D$ connecting $v$ to $v'$, in time $O(|E(P)|)$. The algorithm for the \maintaincluster problem on cluster $C$ may, at any time, raise a flag $F_C$ to indicate that it can no longer support $\spquery$ queries because the diameter of $C$ became too large. When flag $F_C$ is raised, the algorithm for the \maintaincluster problem needs to supply two vertices $x,y\in V(C)$ with $\dist_C(x,y)>\Omega(D\log^4 W)$. We will then employ Procedure \proccut on graph $C$, possibly multiple times, in order to update the cluster $C$, so that its diameter falls back below $\alpha D$. As part of this process, new clusters will be added to the neighborhood cover $\cset$.
As mentioned already, the clusters in $\cset$ may not be disjoint. However, we will ensure that every regular vertex lies in only a small number of such clusters. It is this requirement that makes the description and the analysis of Procedure \proccut somewhat challenging.



The remainder of this section is organized as follows. We start with some basic definitions and notation in \Cref{subsec: def and notation}. 
We describe Procedure \proccut and provide its analysis in \Cref{subsec: proccut definition}. Then we describe and analyze Procedure \initnc for computing an initial Neighborhood Cover, in \Cref{subsec: initial NC}. In \Cref{subsec: algorithmic framework for NC}, we provide a general algorithmic framework for solving the \recdynNC problem that exploits both Procedure \initnc and Procedure \proccut. The framework assumes that we are given an oracle that flags clusters $C$ whose diameter become too large, and provides witnesses for this, in the form of a pair $x,y$ of vertices of $C$ whose distance in $C$ is large. In \Cref{subsec: maintain cluster problem} we define a new problem, called \maintaincluster. An algorithm for this problem will be used in order to implement the oracle in the above-mentioned algorithmic framework, and will also be responsible for supporting $\spquery$ queries within the clusters. We also state our main result for the \maintaincluster problem -- namely, a fast dynamic algorithm for it, and complete the proof of
\Cref{thm: main dynamic NC algorithm} using it in \Cref{subsec: maintain cluster problem}. The remainder of the paper then focuses on desiginig a fast algorithm for the \maintaincluster problem.
As a first step towards this goal, in \Cref{subsec: alg slow and simple} we provide a slow and simple algorithm for the \maintaincluster problem, called \algslow, that we will use as a recursion base.

\subsection{Basic Definition, Parameters and Notation}
\label{subsec: def and notation}

Throughout this section, we assume that we are given a valid input structure $\iset=\left(H,\set{\ell(e)}_{e\in E(H)},D\right )$, with $H=(V,U,E)$, that undergoes a sequence $\Sigma=(\sigma_1,\sigma_2,\ldots)$ of valid update operations with dynamic degree bound $\mu$. 
We define vertex weights in the usual way: the weight of every regular vertex is $1$, and the weight of every supernode is $0$. We denote by $W$ the total weight of all vertices at the beginning of the algorithm. 
 We assume that we are given a target approximation factor $\alpha\geq \Omega(\log^4 W)$. Our goal is to solve the \recdynnc problem on $\iset$ with approximation factor $\alpha$. Therefore, we will maintain a collection $\cset$ of vertex-induced subgraphs of $H$, that we call clusters. At the beginning of the algorithm, we set $\cset=\set{H^0}$, where $H^0$ is the initial graph $H$ (but we will may make some additional updates to this initial set of clusters before processing the first update $\sigma_1\in\Sigma$). 

We use the following parameters: $r=2\ceil{\frac{\log W}{\log\log W}}$, and $\gamma=2^r=W^{O(1/\log\log W)}$. As mentioned above, we aim to construct and maintain a strong $(D,\alpha D)$-neighborhood cover $\cset$. We will ensure that for every regular vertex $v\in V(H)$, the number of clusters in $\cset$ to which $v$ ever belongs  is at most $\gamma$. 

Given a regular vertex $v\in V(H)$, we denote by $n_v$ the number of clusters in the current cluster set $\cset$ containing $v$.
In order to bound the total number of clusters to which a regular vertex belongs, we use the notion of \emph{vertex budgets} that we define next.

\begin{definition}[Vertex Budgets]
	Let $x\in V(H)$ be a vertex, and let $C\in \cset$ be a cluster that contains $x$. Then the \emph{budget of $x$ for cluster $C$} is $\beta_C(x)=\left (1+\frac{\log W(C)}{\log^2W}\right )\cdot w(x)$. The total budget of $x$ is $\beta(x)=\sum_{\stackrel{C\in \cset:}{x\in V(C)}}\beta_C(x)$. For a subset $X\subseteq V(H)$ of vertices, we denote by $\beta(X)=\sum_{x\in X}\beta(x)$.
\end{definition}

Note that for a supernode $u\in U$, $\beta(u)=0$, since $w(u)=0$, and for a regular vertex $v$, for every cluster $C$ containing $v$, $w(v)\leq \beta_C(v)\leq (1+1/\log W)w(v)$, and overall, $w(v)n_v\leq \beta(v)\leq (1+1/\log W)w(v)n_v$.

Assume now that we are given some collection $\cset$ of clusters of $H$, and let $C\in \cset$ be any such cluster. As graph $H$ undergoes valid update operations, cluster $C$ is updated accordingly. Specifically, consider any update opertion $\sigma_t\in \Sigma$. If $\sigma_t$ is a deletion of an edge $e$, and $e\in E(C)$, then we delete $e$ from $C$ as well. If $\sigma_t$ is a deletion of an isolated vertex $x$, and $x\in V(C)$, then we delete $x$ from $C$ as well. 
Assume now that $\sigma_t$ is a supernode splitting operation, applied to a supernode $u$, and a set $E'\subseteq \delta_H(u)$ of its incident edges. If $u\not\in C$, then we do not need to update the cluster $C$. Otherwise, we let $E'_C=E'\cap E(C)$. We then update the cluster $C$ by performing a supernode-splitting operation in it, for vertex $u$, with edge set $E'_C$. 
Note that such update operations may not introduce new regular vertices into $C$. We then immediately obtain the following observation.

\begin{observation}\label{obs: no change in budges due to update operations}
	If an edge-deletion or a supernode-splitting update operation is performed in graph $H$, then for every vertex $x\in V(H)$, budget $\beta(x)$ remains unchanged, and for every cluster $C$ containing $x$, $\beta_C(x)$ remains unchanged. Additionally, for every cluster $C$, $\sum_{x\in V(C)}\beta_C(x)$ remains unchanged, and the total budget of all vertices $\sum_{x\in V(H)}\beta(x)$ remains unchanged. If an isolated vertex deletion update operation is performed in graph $H$, with the deleted vertex $y$, then $\sum_{x\in V(H)}\beta(x)$ decreases by at least $\beta(y)$; for every cluster $C'\in \cset$ with $y\not \in C'$, the budget $\sum_{x\in V(C)}\beta_C(x)$ does not change, and for every cluster $C'\in \cset$ with $y\in C'$, $\sum_{x\in V(C)}\beta_C(x)$  decreases by at least $\beta_C(y)$.
\end{observation}

Recall that $r=2\ceil{\frac{\log W}{\log\log W}}$.
The regular vertices in $V$ are partitioned into classes $S_0,\ldots,S_r$, as follows.
\begin{definition}[Vertex Classes]
	For a regular vertex $v\in V$ and an integer $0\leq j<r$, we say that $v$ belongs to class $S_j$, iff $2^j\leq n_v<2^{j+1}$. If $n_v\geq 2^r$, then we say that $v$ belongs to class $S_r$. For all $1\leq j\leq r$, we denote $S_{\geq j}=S_j\cup\cdots\cup S_r$.
\end{definition}
	


\subsection{Procedure \proccut}
\label{subsec: proccut definition}


The input to Procedure \proccut is a cluster $C\in \cset$, some vertex $x\in V(C)$, and a distance threshold $D$ (which is identical to the input distance parameter $D$). The procedure runs Dijkstra's algorithm (or weighted BFS) from vertex $x$ in graph $C$, up to a certain depth $D^*$, that will be determined later. 

Recall that Dijkstra's algorithm maintains a set $S$ of ``discovered'' vertices, where at the beginning $S=\set{x}$. Throughout the algorithm, for every vertex $y\in S$, we maintain the distance $\dist_C(x,y)$, and a neighbor vertex $a_y$ of $y$ that does not lie in $S$, and minimizes the length of the edge $(y,a_y)$. In every iteration, we select a vertex $y\in S$, for which $\dist_C(x,y)+\ell(y,a_y)$ is minimized, and add vertex $a_y$ to $S$. We are then guaranteed that $\dist_C(x,a_y)=\dist_C(x,y)+\ell(y,a_y)$. 
Assume that we are given, for every vertex $y\in V(C)$, a list $\lambda(y)$ of its neighbors $a$, sorted according to the length $\ell(a,y)$ of the corresponding edge, from smallest to largest. Then Dijkstra's algorithm can be implemented so that, if $S_i$ is the set $S$ after the $i$th iteration, then the total running time of the algorithm up to and including iteration $i$ is $O(E(S_i))$. In order to do so, we maintain, for every vertex $y\in V(G)$, a pointer $p_y$ to the vertex $a_y$ on the list $\lambda(y)$. We also maintain a heap of vertices in set $\set{a_y\mid y\in S}$, whose key is $\dist_G(x,y)+\ell(y,a_y)$. In every iteration, we select a vertex $a=a_y$ from the top of the heap, add it to $S$, and then advance the pointer $p_y$ until the first vertex that does not lie in $S$ is encountered (if vertex $a$ that was added to $S$ serves as vertex $a_y$ for several vertices $y\in S$, we advance the pointer $p_y$ for each such vertex $y$). We also initialize pointer $p_a$. 

For all $i\geq 0$, we denote by $L_i$ the set of all vertices of $C$ that lie at distance $2(i-1)D+1$ to $2iD$ from $x$ in $C$. In other words:

\[L_i=B_C(x,2iD)\setminus B_C(x,2(i-1)D).   \]

We refer to the vertices of $L_i$ as \emph{layer $i$ of the BFS}. We denote by $W_i=\sum_{y\in L_i}w(y)$ the total weight of all vertices in $L_i$. 
Recall that we have denoted by $W$ the total weight of all vertices of $H$ at the beginning of the algorithm. The following definition is crucial for the description of Procedure \proccut.

\begin{definition}[Eligible Layer]
For an integer $i>1$, we say that layer $L_i$ of the BFS is \emph{eligible} iff the following conditions hold:
	
	\begin{properties}{C}
		\item $\sum_{i'\leq i}\tW_{i'}\leq \tW(C)/2$; \label{cut condition 1: half the weight}
		\item $\tW_i\leq  \left(\sum_{i'<i}\tW_{i'}\right )/(64\log^2W)$; and  \label{cut condition 2: small weight}
\item for all $1\leq j\leq r$: 
	$\tW(L_i\cap S_{\geq j}) \leq \left(\sum_{i'<i}\tW( L_{i'}\cap S_{\geq j})\right )/\left (64\log^2W\right )$.  \label{cut condition 3: levels}
	\end{properties}

\end{definition}

We need the following claim, whose proof uses standard arguments and appears in Section \ref{subsec: proof of proccut eligible layer} of Appendix.
\begin{claim}\label{claim: eligible layer}
	If the total weight of all vertices in $B_C (x, 256D\log^4 W )$ is at most $\tW(C)/2$, then there is some eligible layer $L_i$ with $1<i<  128\log^4W$.
\end{claim}

We are now ready to describe the algorithm for \proccut.
The input to the procedure is a cluster $C\in \cset$, a vertex $x\in V(C)$, and a distance threshold $D>0$. We assume that we are given, for every vertex $y\in V(C)$, a list $\lambda(y)$ of its neighbors $a$ in $C$, sorted according to the length $\ell(a,y)$ of the corresponding edge, from smallest to largest.
 The procedure runs Dijkstra's algorithm from the input vertex $x$ in graph $C$, until it encounters the first layer $L_i$, such that either $\sum_{i'\leq i}\tW_{i'}> \tW(C)/2$, or $L_i$ is an eligible layer. In the former case, the algorithm outputs FAIL. In the latter case, it computes two new clusters: cluster $C'$ is the subgraph of $H$ induced by vertex set $L_1\cup\cdots\cup L_{i}$. Cluster $C''$ is obtained from cluster $C$ by deleting all vertices in $L_1\cup \cdots L_{i-1}$ from it. 
The algorithm outputs cluster $C'$ explicitly, by listing all its vertics and edges, and it outputs cluster $C''$ implicitly, by listing the edges and the vertices of $C\setminus C''$.
Notice that vertices of $L_i$, and edges that connect them, belong to both $C'$ and $C''$.
The following observation immediately follows from Claim \ref{claim: eligible layer} and the definition of a layer.

\begin{observation}\label{obs: diam of cut off cluster}
	If Procedure $\proccut(C,x,D)$ does not return FAIL, then it produces two clusters $C',C''$, with $W(C')\leq W(C)/2$, and $\diam(C')\leq 512D\log^4W$. Moreover, $V(C')\subseteq B_C(x,256D\log^4W)$.
\end{observation}

Next, we establish that the two new clusters $C',C''$ contain a $D$-neighborhood every vertex of $C$ in the following claim, whose proof appears in Section \ref{subsec: proof of maintain neighborhood cover} of Appendix.

\begin{claim}\label{claim: cut gives valid neighborhood cover}
Assume that Procedure $\proccut(C,x,D)$ produces two clusters $C'$ and $C''$. Then for every vertex $y\in V(C)$, either $B_C(y,D)\subseteq V(C')$ or $B_C(y,D)\subseteq V(C'')$ must hold.
	\end{claim} 

We view Procedure \proccut as creating a new cluster $C'$ that is split off of the cluster $C$; recall that $\tW(C')\leq \tW(C)/2$ must hold. We then update the cluster $C$, by setting $C= C''$.
Notice that the running time of the procedure is $O(|E(C')|)\leq O(W(C')\mu)$. 

Lastly, we prove the following lemma, that is central to bounding the number of clusters that a regular vertex may belong to.
The proof appears in \Cref{subsec: proving bound on budgets} of Appendix.

\begin{lemma}\label{lem: change in budget after proccut}
	Assume that Procedure $\proccut(C,x,D)$, when applied to a cluster $C\in \cset$, produced two clusters $C'$, $C''$, and assume that cluster set $\cset$ was updated accordingly. Denote by $\beta$ and $\beta'$ the total budget of all vertices in $V(H)$ at the beginning and at the end of the procedure, respectively. For all $1\leq j\leq r$, let $\beta_{\geq j}$ and $\beta'_{\geq j}$ denote the total budget of all vertices in $S_{\geq j}$ at the beginning and at the end of the procedure, respectively, and let $\tW'_j$ be the total weight of new regular vertices that join class $S_j$ at the end of the procedure. Then $\beta'\leq \beta$, and moreover, for all $1\leq j\leq r$, $\beta_{\geq j}'\leq \beta_{\geq j}+2^{j}(1+1/\log W) \tW'_j$.
\end{lemma}

\paragraph{Modified \proccut.}
So far we have described Procedure \proccut, whose input is a cluster $C$, a vertex $x$, and a distance parameter $D$. Recall that, if the total weight of all vertices in $B_C (x, 256D\log^4 W)$ is greater than $\tW(C)/2$, the procedure may return FAIL. We will sometimes use the procedure slightly differently. The input to this modified procedure, that we denote by $\proccut'$, is a cluster $C$, a distance bound $D$, and two vertices $x,y\in V(C)$ with $\dist_C(x,y)\geq 1024 D\log^4W$. As before, we assume that we are given, for every vertex $a\in V(C)$, a list $\lambda(a)$ of its neigbhors $b$, sorted by the length of the corresponding edge $(a,b)$.
We run the procedure \proccut simultaneously from vertex $x$ and from vertex $y$ in graph $C$. In other words, we run Dijkstra's algorithm from both vertices simultaneously, so that the number of edges that the two algorithms discover at each time step (or, equivalently, the time they invest), remain within a constant factor from each other.
Once an eligible layer is reached by either of the procedures, we terminate both algorithms (since $B_C (x, 256D\log^4 W )\cap B_C (y, 256D\log^4 W )=\emptyset$, this is guaranteed to happen, from \Cref{claim: eligible layer}, with a layer $L_i$ whose index $1<i<  128\log^4W$). Assume w.l.o.g. that $\proccut(C,x)$ has reached an eligible layer $L_i$ before $\proccut(C,y)$. Let $C',C''$ be the two resulting clusters, so that $x\in C'$, $V(C')=L_1\cup\cdots\cup L_i$, while $y\in C''$, with $V(C'')=V(C)\setminus (L_1\cup\cdots\cup L_{i-1})$, and $W(C')\leq W(C)/2$. 
We then let $C''$ be the cluster obtained from $C$ by deleting the vertices of $L_1\cup\cdots\cup L_{i-1}$. The outcome of the procedure $\proccut'(C,x,y,D)$ is the pair $C',C''$ of clusters, where, as before, cluster $C'$ is returned explicitly, by listing all of its vertices and edges, and cluster $C''$ is returned implicitly, by listing all vertices and edges of $C\setminus C''$.
We summarize the properties of Procedure $\proccut'$ in the next claim.

\begin{claim}\label{claim: proccut modified}
	Procedure $\proccut'(C,x,y,D)$, given a pair $x,y\in V(C)$ of vertices with $\dist_C(x,y)\geq 1024 D\log^4W$ produces two subgraphs $C',C''$ of $C$, with $x\in C'\setminus C''$ and $y\in C''\setminus C'$, or the other way around. Moreover, $\diam(C')\leq 512D\log^4W$, $W(C')\leq W(C)/2$, and for every vertex $z\in V(C)$, either $B_C(z,D)\subseteq V(C')$ or $B_C(z,D)\subseteq V(C'')$. The running time of the algorithm is $O(|E(C')|)$, provided it is given, for every vertex  $a\in V(C)$, a list $\lambda(a)$ of its neigbhors $b$ in $C$, sorted by the length of the corresponding edge $(a,b)$.
\end{claim}

We note that the budget analysis from \Cref{lem: change in budget after proccut} continues to hold for $\proccut'$ as well, as we can equivalently view $\proccut'$ as running $\proccut$ from one of the vertices, $x$ or $y$.

\subsection{Computing the Initial Neighborhood Cover -- Procedure \initnc}
\label{subsec: initial NC}


We now describe an algorithm, that we refer to as \initnc, for computing an initial neighborhood cover. We assume that we are given a valid input structure  $\iset=\left(H=(V,U,E),\set{\ell(e)}_{e\in E}, D\right )$, where the degree of every regular vertex is at most $\mu$, and we denote $D'=1024D\log^4W$. Our algorithm will compute a strong $(D,D')$-neighborhood cover $\cset$ for $V(H)$ in $H$.

The algorithm maintains a cluster $C^*\subseteq H$, where at the beginning, $C^*=H$. For every vertex $y\in V(C^*)$, it maintains a list $\lambda(y)$ of neighbors $a$ of $y$ in $C^*$, sorted by the length of the corresponding edge $(a,y)$.It also maintains a collection $\cset$ of clusters, that is initialized into $\emptyset$. We will ensure that every cluster in $\cset$ has diameter at most $D'$. Lastly, the algorithm maintains a set $U^*\subseteq V(C^*)$ of potential centers, that is initially set to $V(H)$.  We initialize the lists $\lambda(y)$ for every vertex $y\in V(H)$ at the beginning. The algorithm performs iterations, as long as $U^*\neq \emptyset$.

In order to perform an iteration, we select an arbitrary vertex $x\in U^*$, and execute $\proccut(C^*,x)$. Assume first that procedure does not return FAIL, and instead computes two clusters $C',C''$ (recall that $C''$ is only returned implicitly, by listing the vertices and edges of $C\setminus C''$). In this case, we say that the iteration is \emph{good}. Note that we are guaranteed, from \Cref{obs: diam of cut off cluster}, that the diameter of $C'$ is bounded by $D'/2$. We then add $C'$ to $\cset$, set $C^*=C''$ (by deleting the edges and vertices of $C\setminus C''$ from $C^*$), and delete from $U^*$ all vertices that no longer belong to $C^*$. We also update the lists $\lambda(y)$ of vertices $y\in V(C^*)$ by deleting the neighbors of $y$ that no longer lie in $C^*$, finishing the current iteration.

Assume now that Procedure \proccut returned FAIL. In this case, we say that the iteration is \emph{bad}. We are then guaranteed, from \Cref{claim: eligible layer}, that the total weight of vertices in $B_{C^*}(x,256D\log^4W)$ is at least $W(C^*)/2$.
We initialize an ES-tree data structure $\tau$, rooted at vertex $x$, up to depth  $\hat D=512D\log^4W$. Throughout the remainder of the algorithm, we denote by $B'=B_{C^*}(x,256D\log^4W)$, and we denote by  $B''=B_{C^*}(x,512D\log^4W)$ -- the set of all vertices that belong to $\tau$. Note that both vertex sets can be maintained by the algorithm that maintains the tree $\tau$.

In the remainder of the algorithm, we will always let $U^*$ contain all vertices of $V(C^*)\setminus B''$. Therefore, initially, we delete from $U^*$ all vertices of $B''$. Whenever some vertex $a$ leaves the set $B''$, we will add it back to set $U^*$. We then denote $x^*=x$, and continue to the next iteration. From this time onward, the vertices in $B'$ will never be deleted from $C^*$, since for every vertex $y\in U^*$, $B'\cap B_{C^*}(y,256D\log^4W)=\emptyset$ currently holds, and therefore will continue to hold as vertices and edges are deleted from $C^*$. Notice that, as vertices and edges are deleted from $C^*$, over the course of the \initnc algorithm, $W(C^*)$ may decrease, but $W(B')$ will not decrease, and so $W(B')\geq W(C^*)/2$ will continue to hold for the remainder of the algorithm.
Therefore, after a single bad iteration, for the remainder of the algorithm, for every vertex $y\in U^*$, $W(B_{C^*}(y,256  D\log^4W))\leq W(C^*)/2$ must hold, since $W(B')\geq W(C^*)/2$, and $B'\cap B_{C^*}(y,256D\log^4W)=\emptyset$. It follows that at most one iteration of the algorithm may be bad. 

The algorithm terminates when $U^*=\emptyset$ holds. If $C^*\neq \emptyset$ at the time of the termination, then $V(C^*)\subseteq B''$ must hold. We then add $C^*$ as the last cluster to $\cset$. We call $C^*$ a \emph{distinguished cluster} of $\cset$. Note that for every cluster $C\in \cset\setminus \set{C^*}$, $W(C)\leq W(H)/2$ holds. It may be sometimes convenient to think of the algorithm $\initnc$ as iteratively splitting the clusters of $\cset\setminus\set{C^*}$ off of cluster $C=H$, and so we may view $C^*$ as the updated version of the cluster $C$ at the end of the algorithm, instead of viewing it as a newly created cluster.

From \Cref{obs: diam of cut off cluster} and \Cref{claim: cut gives valid neighborhood cover}, it is immediate to verify that the collection $\cset$ of clusters that we obtain at the end of the algorithm is a strong $(D,D')$-neighborhood cover for $H$.  

The running time for a bad iteration is at most $|E(H)|$, while the running time of a good iteration that creates a cluster $C'$ is bounded by $|E(C')|$. Additionally, the running time required for maintaining the ES-tree $\tau$ is $O(|E(H)|\cdot D\cdot \poly\log |E(H)|)$.
Therefore, the total running time of the algorithm is bounded by $O(|E(H)|\cdot D\cdot \poly\log |E(H)|+\sum_{C\in \cset}|E(C)|)\leq O(W(H)\cdot D\cdot \poly\log (W(H)\mu)+\sum_{C\in \cset}W(C))\cdot \mu$. Since, for every vertex $x\in V(H)$ and cluster $C\in \cset$ containing $x$, $\beta_C(x)\geq w(x)$ holds, we get that  $\sum_{C\in \cset}W(C)\leq \sum_{x\in V(H)}\beta(x)$. Lastly, since, from \Cref{lem: change in budget after proccut}, the total budget of all vertices $\sum_{x\in V(H)}\beta(x)$ does not increase over the course of the algorithm, and since, at the beginning of the algorithm, $\sum_{x\in V(H)}\beta(x)\leq 2W(H)$ held, we get that the total running time is bounded by $O(W(H)\cdot\mu\cdot D\cdot\poly\log(W(H)\mu))\leq O(N^0(H)\cdot\mu \cdot D\cdot \poly\log(N^0(H)\mu) )$, where $N^0(H)$ is the number of regular vertices in $H$.
Therefore, we have proved the following lemma.

\begin{lemma}\label{lem: initialize neighborhood cover}
	Algorithm \initnc described above correctly computes a strong $(D,D')$-neighborhood cover $\cset$ for $H$, for $D'=1024D\log^4W$, where $W=N^0(H)$ is the number of regular vertices in $H$. The algorithm is deterministic, with running time $O(|E(H)|\cdot D\cdot \poly\log |E(H)|+\sum_{C\in \cset}|E(C)|)\leq O(N^0(H)\cdot D\cdot \mu\cdot \poly\log(N^0(H)\mu))$. 
\end{lemma}

Note that Procedure \initnc can be viewed as a series of applications of Procedure \proccut. We will sometimes use this view, which will, for example, allow us to use the results from \Cref{lem: change in budget after proccut} to bound the budgets of vertices.


\subsection{Algorithmic Framework for Maintaining Neighborhood Cover}
\label{subsec: algorithmic framework for NC}

We now describe a general algorithmic framework for maintaining a strong neighborhood cover in a given valid input structure that undergoes a sequence of valid update operations. This framework will be used in the proof of \Cref{thm: main dynamic NC algorithm} in several different ways. We denote the algorithm described in this subsection by \algmaintainNC.

We assume that we are given a valid input structure $\iset=\left(H,\set{\ell(e)}_{e\in E(H)},D\right )$, that undergoes a sequence of valid update operations, with dynamic degree bound $\mu$.
The algorithm maintains a collection $\cset$ of subgraphs of $H$ that we refer as clusters, with the guarantee that for every regular vertex $v\in V(H)$, there is a cluster $C\in \cset$ with $B_H(v,D)\subseteq V(C)$. 
The algorithm maintains, for every regular vertex $v\in V(H)$, a cluster $C=\coveringcluster(v)$, with $B_H(v,D)\subseteq V(C)$. It also maintains, for every vertex $x\in V(H)$, a list $\clusterlist(x)$ of all clusters in $\cset$ containing $x$, and for every edge $e\in E(H)$, a list $\clusterlist(e)$ of all clusters in $\cset$ containing $e$.
Lastly, for every cluster $C\in \cset$ and vertex $x\in V(C)$, it maintains a list $\lambda_C(x)$ of all neighbors $a$ of $x$ in $C$, sorted by the length of the corresponding edge $(a,x)$.

As before, we set the weight of every regular vertex to be $1$, and of every supernode to be $0$. We denote by $N^0(H)$ the number of regular vertices in $H$ at the beginning of the algorithm, and by $W=N^0(H)$ the weight of all vertices of $H$ at the beginning of the algorithm.

At the beginning of the algorithm, before any update operations from $\Sigma$ are processed, we apply Procedure \initnc to graph $H$, and we add to $\cset$ the resulting collection of clusters, that form a strong $(D,1024D\log^4W)$-neighborhood cover of $H$. We also initialize all data structures $\coveringcluster(v)$ for regular vertices $v\in V(H)$,  $\clusterlist(x)$ for vertices $x\in V(H)$, and $\clusterlist(e)$ for edges $e\in E(H)$. For every cluster $C\in \cset$ and vertex $x\in V(C)$, we initialize the list $\lambda_C(x)$ of neighbors of $x$ in $C$.

 As graph $H$ undergoes a sequence $\Sigma$ of valid update operations, we update every cluster $C\in \cset$ accordingly. Specifically, consider any update operation $\sigma_t\in \Sigma$. If $\sigma_t$ is a deletion of an edge $e$, then for every cluster $C\in \clusterlist(e)$, we delete $e$ from $C$ as well. If $\sigma_t$ is a deletion of an isolated vertex $x$, then for every cluster $C\in \clusterlist(x)$, we delete $x$ from $C$ as well. 
 Assume now that $\sigma_t$ is a supernode splitting operation, applied to a supernode $u$, and a set $E'\subseteq \delta_H(u)$ of its incident edges. For every cluster $C\in \clusterlist(u)$, we let $E'_C=E'\cap E(C)$. If $E'_C\neq \emptyset$, then we update the cluster $C$ by performing a supernode-splitting operation in it, for vertex $u$, with edge set $E'_C$. 
We then add $C$ to $\clusterlist(u')$ of the newly created supernode $u'$, and also initialize $\clusterlist(e)$ for every edge $e$ that was just added to $H$.
In order to implement these updates efficiently, we process the edges of $E'$ one-by-one. Let $e=(u,v)$ be an edge of $E'$ that is currently processed. We then consider every cluster $C\in \clusterlist(e)$ one-by-one. For each such cluster $C$, we mark $C$ as a cluster on which the supernode splitting operation needs to be executed (if it has not been marked yet), and add the edge $e$ to the set $E'_C$ (if set $E'_C$ is not yet initialized, we set $E'_C=\set{e}$). Once all edges in $E'$ are processed, we perform a supernode splitting operation on each marked cluster $C$, using the edge set $E'_C$. After each update operation, we update the lists $\lambda_C(x)$ of all relevant vertices $x\in V(C)$.

Note that the total processing time of the update operation $\sigma_t$ is asymptotically bounded by the total number of edges deleted from the clusters in $\cset$, or inserted into the clusters in $\cset$, and the total number of vertices deleted from the clusters in $\cset$, while processing $\sigma_t$.

Additionally, we assume that we are given an oracle, that, at any time, may raise a flag $F_C$ for a cluster $C\in \cset$. When the oracle raises flag $F_C$, it needs to provide a pair  $x,y\in V(C)$ of vertices of $C$ with $\dist_C(x,y)\geq 1024 D\log^4 W$. The algorithm then runs $\proccut'(C,x,y,D)$, obtaining two clusters $C',C''$. The procedure outputs cluster $C'$ explicitly, and it outputs cluster $C''$ implicitly, by listing the edges and the vertices of $C\setminus C''$. Assume w.l.o.g. that Procedure $\proccut'$ found an eligible layer $L_i$ when running Dijktra's algorithm from vertex $x$. Then for every regular vertex $v\in V(H)$ with $\coveringcluster(v)=C$, if $v\in B_C(x,2iD-D)$, we set $\coveringcluster(v)=C'$; otherwise, we are guaranteed that $B_H(v,D)\subseteq V(C'')$, and $\coveringcluster(v)$ does not need to be updated. Additionally, for every vertex $x\in V(C')$, and for every edge $e\in E(C')$, we add $C'$ to $\clusterlist(x)$ or to $\clusterlist(e)$, respectively. Similarly, for every edge $e$ that was deleted from $C$, and for every vertex $x$ that was deleted from $C$, we delete $C$ from $\clusterlist(x)$ or from $\clusterlist(e)$, respectively. 
For every vertex $x\in V(C')$, we initialize the neighbor list $\lambda_{C'}(x)$, and for every edge $e$ that was deleted from $C$, we update the neighbor lists $\lambda(x)$ of its endpoint(s) that remain in $C$. 
Note that all these updates can be made in time $O(|E(C')|)$.
We say that cluster $C'$ was \emph{split off of $C$}. Once the update is completed, flag $F_C$ is lowered. No update operations from $\Sigma$ are processed when flag $F_C$ is up for any cluster $C$.

This completes the description of Algorithm \algmaintainNC. Note that for each cluster $C\in \cset$, from the moment $C$ is added to $\cset$, it undergoes a sequence of updates that are valid update operations for the corresponding valid input structure $\iset_C=\left(C,\set{\ell(e)}_{e\in E(C)},D\right )$. Some of these update operations mirror the update operations from $\Sigma$ that graph $H$ undergoes, and some update operations (edge deletions and isolated vertex deletions) arise from splitting clusters off of $C$. Next, we analyze some properties of \algmaintainNC that will be useful for us later.



\paragraph{Bounding the running time.}
Recall that we denote by $W=W(H^0)$ the total weight of all vertices of $H$ at the beginning of the algorithm.
For every cluster $C\in \cset$, we denote by $W^0(C)$ the total weight of all vertices of $C$ at the time when $C$ is added to $\cset$. We also denote by $\beta^0(C)$ the sum of the budgets $\beta_C(v)$ of all vertices $v\in V(C)$ at the moment when cluster $C$ is created. From the definition of budgets, $W^0(C)\leq \beta^0(C)$ for every cluster $C$.
Recall that, from the time that $C$ is added to $\cset$, $W(C)$ may only decrease. Notice that, once $C$ is added to $\cset$, it remains there until the end of the algorithm, though it is possible that all edges and all vertices are deleted from  $C$, and it becomes empty. 
We denote by $\cset'$ the set of all clusters that ever belonged to $\cset$ (or equivalently, it is the set of all clusters lying in $\cset$ at the end of the algorithm.)
We start with the following claim.

\begin{claim}\label{claim: bound total weight}
	$\sum_{C\in \cset'}W^0(C)\leq O(W\log W)$.
\end{claim}
\begin{proof}
	From the above discussion, for every cluster $C\in \cset'$, 	$W^0(C)\leq \beta^0(C)$. Moreover, it is easy to see that, if we denote by $\beta^0$ the sum of the budgets of all vertices $x\in V(H)$ at the beginning of the algorithm, then $\beta^0\leq 2W$. Therefore, it is enough to prove that $\sum_{C\in \cset'}\beta^0(C)\leq O(\beta^0\log W)$.
	
	We construct a partitioning tree $\tau$, whose vertex set is $\set{v(C)\mid C\in \cset'}$. The root of the tree is the vertex $v(H)$, corresponding to the original graph $H$. Consider now some vertex $v(C)$ of the tree, where $C$ is some cluster. Let $C_1,C_2,\ldots,C_q$ be all clusters that were split off of $C$ over the course of the algorithm. Recall that, from Condition \ref{cut condition 1: half the weight}, we are guaranteed that for all $1\leq i \leq q$, $W(C_i)\leq W(C)/2$ holds when $C_i$ is split off of $C$, and so $W^0(C_i)\leq W^0(C)/2$. We add edges connecting each of the vertices $v(C_1),\ldots,v(C_q)$ to $v(C)$, and these vertices become children of the vertex $v(C)$ in the tree. We need the following observation. 
	
	\begin{observation}\label{obs: weights do not grow}
		Let $v(C)$ be a vertex in the tree $\tau$, and let $v(C_1),\ldots,v(C_q)$ be its child vertices. Then $\sum_{i=1}^q\beta^0(C_i)\leq \beta^0(C)$.
	\end{observation}
	
	Assume first that the observation is correct. Since, for every child vertex $v(C_i)$ of vertex $v(C)$, $W^0(C_i)\leq W(C)/2\leq W^0(C)/2$ holds,  the depth of the tree $\tau$ is bounded by $\ceil{\log W}$. Moreover, if we denote, for $1\leq i\leq \ceil{\log W}$, by $\cset_i\subseteq \cset'$ the set of all clusters $C$ such that the distance from $v(C)$ to the root of $\tau$ is exactly $i$, then, from \Cref{obs: weights do not grow}, $\sum_{C\in \cset_i}\beta^0(C)\leq \beta^0(H)$. Therefore, $\sum_{C\in \cset'}W^0(C)\leq \sum_{C\in \cset'}\beta^0(C)\leq O(\beta^0(H)\log W)\leq O(W\log W)$. In order to complete the proof of \Cref{claim: bound total weight}, it is now enough to prove \Cref{obs: weights do not grow}.

	\begin{proofof}{\Cref{obs: weights do not grow}}
		We assume that the clusters $C_1,C_2,\ldots,C_q$ where split off of $C$ in this order. Consider the iteration of the algorithm when cluster $C_i$ was split off of cluster $C$, by applying Procedure \proccut to cluster $C$. Let $C''$ denote the cluster $C$ at the end of this procedure, and let $C$ denote  the same cluster at the beginning of this procedure. From \Cref{lem: change in budget after proccut}, the total budget of all vertices in $H$ did not increase as the result of applying \proccut to cluster $C$. If we denote $\beta$ the sum of budgets of all vertices of $H$ before the procedure is applied to cluster $C$, and by $\beta'$ the sum of budgets of all vertices after the procedure is applied, then:
		
		\[\beta'-\beta=\sum_{v\in C_i}\beta_{C_i}(v)+\sum_{v\in C''}\beta_{C''}(v)-\sum_{v\in C}\beta_C(v).\]
		
		Since $\beta'\leq \beta$, we get that $\sum_{v\in C_i}\beta_{C_i}(v)+\sum_{v\in C''}\beta_{C''}(v)\leq \sum_{v\in C}\beta_C(v)$. By applying this argument iteratively to all clusters $C_1,\ldots,C_q$, and recalling that, from \Cref{obs: no change in budges due to update operations}, update operation cannot increase the total budget of vertices in a cluster, we get that $\sum_{i=1}^q\beta^0(C_i)\leq \beta^0(C)$.
	\end{proofof}
\end{proof}

We obtain the following immediate corollary of \Cref{claim: bound total weight}, bounding the running time of Algorithm \algmaintainNC.

\begin{corollary}\label{cor: running time of algmaintainNC}
	The running time of \algmaintainNC is $O(W\cdot \mu \cdot D\cdot \poly\log (W\mu))$.
\end{corollary}
\begin{proof}
Recall that, from \Cref{lem: initialize neighborhood cover}, the running time of Algorithm \initnc is $O\left (W  \mu D\poly\log(W\mu) \right )$. 
For every cluster $C\in \cset$, let $m(C)$ be the total number of edges that ever belonged to cluster $C$ over the course of the algorithm. 
Clearly, $m(C)\leq W^0(C)\cdot \mu$. The time spent on creating the cluster $C$ for the first time (by splitting it off of some other cluster $C^*$, including the time needed to update $C^*$) is bounded by $O(m(C))$. Subsequently, the time needed to process all update operations on $C$ is bounded by the total number of edges that are either deleted from $C$ or inserted into $C$, and the total number of vertices deleted from $C$. As the number of supernodes that are ever present in $C$ is bounded by $m(C)$, we get that the total update time spent on processing all update operations is bounded by $O(m(C))$. The total running time of the algorithm, excluding the time needed to run \initnc, is then bounded by: $\sum_{C\in \cset'}O(m(C))\leq \sum_{C\in \cset'}O(W^0(C)\cdot \mu)\leq O(W\mu\log W)$.
\end{proof}



\paragraph{Bounding number of copies of each vertex.}
Next, we bound the number of clusters in $\cset'$ that may contain a regular vertex $v$ of $H$, in the following theorem.

\begin{theorem}\label{thm: bound number of copies}
	For every regular vertex $v\in V(H)$, the total number of clusters $C\in \cset'$, such that $v$ lied in $C$ at any time during the algorithm's execution is at most $W^{O(1/\log\log W)}$.
\end{theorem}

\begin{proof}
	For the sake of the proof, it is convenient to disregard update operations when isolated vertices are deleted from $H$; we will simply assume that such vertices remain in graph $H$ and in every cluster $C$ to which they belonged at the time of deletion; we will not add such vertices to any new cluster. This is done in order to avoid the reduction in $W(H)$ (and in the budget of all vertices) following isolated vertex deletion.

Algorithm \algmaintainNC can be equivalently described as follows. We start with $\cset=\set{H}$, and then perform iterations. In ever iteration, we apply Procedure \proccut to a cluster $C\in \cset$, with a vertex $x\in V(C)$ and distance bound $D$, obtaining two clusters $C',C''$. We then add $C'$ to $\cset$ and replace $C$ with $C''$ in $\cset$.
(We also perform edge-deletion and supernode-splitting operations, but these operations do not affect vertex weights or budgets so they are immaterial to this discussion).	
	
	Consider now some application of the \proccut operation to some cluster $C\in \cset$, in which a new cluster $C'$ is split off of $C$; we denote by $C''$ the cluster $C$ at the end of this operation. Recall that $W^0(C')\leq W(C)/2\leq W^0(C)/2$ must hold. If a regular vertex $v\in V(C)$ lies in $C'$ at the end of this operation, but it does not lie in $C''$, then we say that this copy of $v$ was \emph{moved from cluster $C$ to cluster $C'$}. If vertex $v$ lies in both $C'$ and $C''$, then we say that \emph{a new copy of $v$ was created} -- the copy that lies in $C'$. Since, whenever vertex $v$ is moved from cluster $C$ to cluster $C'$, $W^0(C')\leq W^0(C)/2$ must hold, a copy of $v$ may only be moved from one cluster to another at most $O(\log W)$ times. Therefore, it is now enough to bound the total number of copies of a vertex that the algorithm may ever create. Since we do not delete isolated vertices, once a new copy of a vertex $v$ is created, it is never deleted, and may only move from cluster to cluster. In particular, if a regular vertex $v$ is added to set $S_{\geq j}$, then it remains in set $S_{\geq j}$ until the end of the algorithm. Recall that vertex $v$ belongs to class $S_r$ iff at least $2^r$ copies of vertex $v$ exist in $\cset$. We show in the next claim that at the end of the algorithm, $S_r=\emptyset$. It then follows that for any vertex $v$ that every belonged to the graph $H$, the total number of copies of $v$ that may be created over the course of the algorithm is bounded by $2^r$. Since a copy of $v$ may be moved at most $O(\log W)$ times, the total number of clusters in $\cset'$ to which $v$ ever belonged over the course of the algorithm is bounded by $O(2^r\log W)\leq W^{O(1/\log\log W)}$, since $r=2\ceil{\log W/\log\log W}$. The following claim will then finish the proof of \Cref{thm: bound number of copies}.

	\begin{claim}\label{claim: last class empty}
		Throughout the algorithm, $S_r=\emptyset$ holds.
	\end{claim}
	\begin{proof}
		We prove by induction that, for all $1\leq j\leq r$, $\tW(S_{\geq j})\leq W/\log^j W$ holds over the course of the entire algorithm. Assume first that this is indeed the case. Then we get that $\tW(S_r)\leq W/\log^rW<1$, since $r=2\ceil{\log W/\log\log W}$.
		Since the weight of every regular vertex is $1$, it follows that $S_r=\emptyset$. Therefore, it is now enough to prove the following claim.
		
		\begin{claim}
			For all $1\leq j\leq r$, $\tW(S_{\geq j})\leq W/\log^j W$ holds over the course of the entire algorithm.
		\end{claim}
		\begin{proof}
			Consider some index $1\leq j\leq r$, and some regular vertex $v\in V(H)$. Vertex $v$ may  only be added to $S_{\geq j}$,  if $v\in S_{j-1}$, and Procedure \proccut creates an additional copy of vertex $v$, so that the number of copies of $v$ becomes $2^j$. 
			From the moment a vertex is added to $S_{\geq j}$, it remains there until the end of the algorithm. Therefore, it is enough to prove that, at the end of the algorithm,  $\tW(S_{\geq j})\leq W/\log^j W$ holds.

			
			
			
			We prove this by induction on $j$. We start with the base case, where $j=1$. Notice that, from the definition of vertex budgets, for every regular vertex $v\in V\setminus S_{\geq 1}$, $w(v)\leq \beta(v)\leq w(v)(1+1/\log W)$ holds throughout the algorithm. At the beginning algorithm, the total budget of all regular vertices, $\beta=\sum_{v\in V}\beta(v)\leq \sum_{v\in V}w(v)(1+1/\log W)=W(1+1/\log W)$ holds (recall that supernodes all have weight $0$ and budget $0$). Since we assume that isolated vertices are not deleted from $H$, throughout the algorithm, $\sum_{v\in V}w(v)=W$ holds. Assume for contradiction that $\sum_{v\in S_{\geq 1}}w(v)> W/\log W$ holds at the end of the algorithm. 
			Observe that, at the end of the algorithm, a regular vertex $v\in V\setminus S_{\geq 1}$ has budget $\beta(v)\geq w(v)$, while a vertex $v\in S_{\geq 1}$ has budget at least $2w(v)$. Therefore, the total budget of all vertices at the end of the algorithm is at least:
			
			\[\sum_{v\in V}w(v)+\sum_{v\in S_{\geq 1}}w(v)>W(1+1/\log W).  \]
			
			This is a contradiction since the total budget of all vertices may not increase over the course of the algorithm.

			Assume now that for some $1\leq j<r$, $\sum_{v\in S_{\geq j}}w(v)\leq W/\log^j W$ holds at the end of the algorithm. We will now prove that this inequality holds for $j+1$. 
			
			Note that every vertex $v$ that belongs to $S_{\geq j}$ at the end of the algorithm has budget at least $2^j w(v)$, while a vertex $v$ that belongs to $S_{\geq j+1}$ has budget at least $2^{j+1}w(v)$ at the end of the algorithm. Assume for contradiction that at the end of the algorithm, $\sum_{v\in S_{\geq j+1}}w(v)> W/\log^{j+1} W$ holds. Then, at the end of the algorithm:
			
			\begin{equation}\label{eq: beta bounds}
			 \sum_{v\in S_{\geq j}}\beta(v)\geq 2^j\sum_{v\in S_{\geq j}}w(v)+2^{j}\sum_{v\in S_{\geq j+1}}w(v)> 2^j\sum_{v\in S_{\geq j}}w(v)+2^{j}W/\log^{j+1}W.
			 \end{equation}

Consider now changes to the value $\beta_{\geq j}=\sum_{v\in S_{\geq j}}\beta(v)$ over the course of the algorithm. When the algorithm starts, $\beta_{\geq j}=0$. From \Cref{lem: change in budget after proccut},  the value $\beta_{\geq j}$ may only increase when a vertex $v$ is added to $S_{\geq j}$. In such a case, the increase in $\beta_{\geq j}$ is bounded by $2^{j}(1+1/\log W)w(v)$. Update operations in the input sequence $\Sigma$ and other updates due to application of \proccut may not increase $\beta_{\geq j}$. Therefore, at the end of the algorithm:
			
			\[\begin{split}
			\beta_{\geq j}&\leq 2^{j}(1+1/\log W)\cdot \sum_{v\in S_{\geq j}}w(v)\\
			&\leq 2^j\sum_{v\in S_{\geq j}}w(v) + 2^j \sum_{v\in S_{\geq j}}w(v)/\log W\\
			&\leq 2^j\sum_{v\in S_{\geq j}}w(v) + 2^j W/\log^{j+1}W,
			\end{split} \]
			
			from the induction hypothesis, contradicting Inequality \ref{eq: beta bounds}.
		\end{proof}
	\end{proof}
\end{proof}

\subsection{\maintaincluster Problem and a Reduction from \recdynnc to \maintaincluster}
\label{subsec: maintain cluster problem}

In this subsection, we define a new problem, that we call \maintaincluster. Intuitively, this problem will be used in order to simulate the oracle needed for \algmaintainNC. Additionally, an algorithm for this problem will need to support $\spquery(C,v,v')$ queries as required in the \recdynNC problem. We state our main result for the \maintaincluster problem -- a fast deterministic algorithm for solving this problem, and show that this new algorithm, combined with Algorithm \algmaintainNC presented in the previous subsection, provides an algorithm for the \recdynnc problem, leading to the proof of \Cref{thm: main dynamic NC algorithm}.
We start with the definition of the \maintaincluster problem.

\subsubsection{\maintaincluster Problem}
Intuitively, the input to the \maintaincluster problem is a valid input structure $\iset=\left(C,\set{\ell(e)}_{e\in E(C)},D\right )$, where  $C$ is a connected subgraph of the original graph $H$.
Graph $C$ undergoes a sequence of valid update operations, with dynamic degree bound $\mu$. We assume that we are given as input a parameter $\hat W\geq N^0(C)\mu$, where $N^0(C)$ is the initial number of regular vertices in $C$. We are required to support queries 
$\spquery(C,v,v')$, in which, given a pair $v,v'\in V(C)$ of regular vertices of $C$, we need to return a path of length at most $\alpha D$ connecting them in $C$, where $\alpha$ is the desired approximation factor. However, the algorithm is allowed to raise a flag $F_C$ at any time. When flag $F_C$ is raised, the algorithm is required to supply a pair 
$v,v'$ of regular vertices of $C$, with $\dist_C(v,v')>1024D\log^4\hat W$.  The algorithm then receives, as part of its input update sequence $\Sigma$, a sequence $\Sigma'$ of valid update operations, at the end of which either $x$ or $y$ are deleted from $C$. We called sequence $\Sigma'$ a \emph{flag lowering sequence}.
Once the flag lowering sequence $\Sigma'$ is processed, flag $F_C$ is lowered. However, the algorithm may raise the flag $F_C$ again immediately, as long  as it provides a new pair $\hat v,\hat v'$ of vertices with $\dist_C(\hat v,\hat v')>1024D\log^4\hat W$. We emphasize that we view the resulting flag lowering sequences $\Sigma'$ as part of the input sequence $\Sigma$ of valid update operations, and so the restriction that the total number of edges incident to a regular vertex of $C$ over the course of the update sequence is bounded by $\mu$ continues to hold.
Queries $\spquery$ may only be asked when flag $F_C$ is down.
We also emphasize that the initial cluster $C$ that serves as input to the \maintaincluster problem may have an arbitrarily large diameter, and so the algorithm for the \maintaincluster problem  may repeatedly raise the flag $F_C$, until it is able to support $\spquery(C,v,v')$ (that intuitively means that the diameter of $C$ has fallen under $\alpha D$).
 We now provide a formal definition of the problem.

\begin{definition}[\maintaincluster problem] The input to the \maintaincluster problem is a valid input structure $\iset=\left(C,\set{\ell(e)}_{e\in E(C)},D\right )$, where $C$ is a connected graph. 
Graph $C$ undergoes an online sequence $\Sigma$ of valid update operations, and we are given its dynamic degree bound $\mu$. Additionally, we are given the desired approximation factor $\alpha$ and a parameter $\hat W\geq N^0(C)\cdot \mu$, where $N^0(C)$ is the number of regular vertices of $C$ at the beginning of the algorithm. 
The algorithm must  support queries $\spquery(C,v,v')$: given a pair $v,v'\in V(C)$ of regular vertices of $C$, return a path $P$ of length at most $\alpha D$ connecting them in $C$, in time $O(|E(P)|)$. The algorithm may, at any time, raise a flag $F_C$, at which time it must  supply a pair  $\hat v,\hat v'$ of regular vertices of $C$, with $\dist_C(\hat v,\hat v')>1024D\log^4\hat W$.
 Once flag $F_C$ is raised, the algorithm will obtain, as part of its input update sequence $\Sigma$, a sequence $\Sigma'$ of valid update operations called flag-lowering sequence, at the end of which 
 either $\hat v$ or $\hat v'$ are deleted from $C$. Flag $F_C$ is lowered after these updates are processed by the algorithm. Queries $\spquery$ may only be asked when flag $F_C$ is down.
\end{definition}

Our main result for the \maintaincluster problem is summarized in the following theorem.

\begin{theorem}\label{thm: main maintain cluster algorithm}
	There is a deterministic algorithm for the \maintaincluster problem, that,  on input $\iset=\left(C,\set{\ell(e)}_{e\in E(C)},D \right )$, that undergoes a sequence of valid update operations with dynamic degree bound $\mu$, and parameters $c/\log \log N^0(C)<\eps<1$ for some large enough constant $c$, and $\hat W\geq N^0(C)\mu$, where $N^0(C)$ is the number of regular vertices in $C$ at the beginning of the algorithm, achieves approximation factor $\alpha=(\log \hat W)^{2^{O(1/\eps)}}$, and  has total update time $O\left (N^0(C)\cdot \mu^2\cdot D^3\cdot \hat W^{O(\eps)}\cdot (\log \hat W)^{O(1/\eps^2)}\right )$. 
\end{theorem}

Next, we show that  \Cref{thm: main dynamic NC algorithm} follows from \Cref{thm: main maintain cluster algorithm}.

\subsubsection{Completing the Proof of \Cref{thm: main dynamic NC algorithm}}

Recall that we are given  a valid input structure $\iset=\left(H=(V,U,E),\set{\ell(e)}_{e\in E(H)},D \right )$, undergoing a sequence $\Sigma$ of valid update operations with dynamic degree bound $\mu$, and a desired approximation factor $\alpha$. As before, we set the weight of every regular vertex to $1$, the weight of every supernode to $0$, and we denote by $W=N^0(H)$ the weight of all vertices of $H$ at the beginning of the algorithm. We are also given a precision parameter $c/\log\log W<\eps<1$, and we use a new parameter $\hat W=W\cdot \mu$.

We employ Algorithm \algmaintainNC from \Cref{subsec: algorithmic framework for NC} on the input $\iset$, and the sequence $\Sigma$ of valid update operations. 
Recall that the algorithm maintains a set $\cset$ of clusters. Whenever Algorithm \algmaintainNC adds a new cluster $C$ to $\cset$, we initialize the algorithm for the \maintaincluster problem from \Cref{thm: main maintain cluster algorithm} on cluster $C$, with the same approximation factor $\alpha$, distance bound $D$, and parameter $\hat W$ as defined above.
We denote by $N^0(C)=W^0(C)$ the number of regular vertices of $C$ (or equivalently, the total weight of all vertices of $C$), when $C$ is added to $\cset$.
 This algorithm will serve two purposes: first, it will process queries $\spquery(C,v,v')$ for cluster $C$. Additionally, it will serve as an oracle that is used by Algorithm \algmaintainNC in order to update the clusters.

Recall that Algorithm \algmaintainNC starts with $\cset=\set{H}$. The only changes that Algorithm \algmaintainNC may perform to the clusters of $\cset$ are the following:

\begin{itemize}
	\item  delete an edge from $C$ (when the corresponding edge is deleted from $H$);
	\item delete an isolated vertex from $C$ (when the corresponding vertex is deleted from $H$);
	\item add a new supernode $u'$ together with adjacent edges when a supernode $u\in C$ undergoes supernode splitting update in $H$; and
	\item split a cluster $C'$ off of $C$.
\end{itemize} 

(We view the initial execution of the Algorithm \initnc on the original graph $H$ as a series of cluster splitting operations).
In the last operation, cluster $C'$ is a vertex-induced subgraph of $C$. The operation can be implemented by first adding $C'$ to $\cset$, and then deleting edges and vertices from $C$ as necessary.
Therefore, if we denote by $\uset=\set{V(C)\mid C\in \cset}$, then set $\uset$ only undergoes allowed changes. Algorithm  \algmaintainNC takes care of maintaining the data structures $\coveringcluster(v)$ for every regular vertex $v\in V(H)$, and data structures $\clusterlist(x)$ for vertices $x\in V(H)$ and $\clusterlist(e)$ for edges $e\in E(H)$.
In order to respond to  a query $\spquery(C,v,v')$, we simply run this query in the data structure for the \maintaincluster problem on cluster $C$, which must return  a path $P$ of length at most $\alpha D$ connecting $v$ to $v'$ in $C$, in time $O(|E(P)|)$. The fact that the data structures for the \maintaincluster problem on the clusters in $\cset$ can support such queries proves that at any time, $\cset$ is a strong $(D,\alpha\cdot D)$-neighborhood cover for the regular vertices in graph $H$.

Note that, from \Cref{thm: bound number of copies},
	for every regular vertex $v\in V(H)$, the total number of clusters in $\cset$ to which vertex $v$ ever belonged over the course of the algorithm is bounded by $W^{O(1/\log\log W)}$.

It now remains to bound the total running time of the algorithm. From \Cref{cor: running time of algmaintainNC},
	the running time of \algmaintainNC is $O(W\cdot D\cdot  \mu \cdot \poly\log (W\mu))$.

For every cluster $C\in \cset$, the running time of the algorithm for the \maintaincluster problem on $C$ is at most 
$O\left (W^0(C)\cdot \mu^2\cdot D^3\cdot \hat W^{O(\eps)}\cdot (\log \hat W)^{O(1/\eps^2)}\right )$.
 Lastly, if we denote by $\cset'$ the set of all clusters that were ever added to $\cset$, then, from 
\Cref{claim: bound total weight}, $\sum_{C\in \cset'}W^0(C)\leq O(W\log W)$. Therefore, the total running time of algorithms $\maintaincluster(C)$ for all clusters $C\in \cset$ is bounded by:

\[ 
\begin{split}
&O\left ( \left (\sum_{C\in \cset'}W^0(C)\right ) \cdot \mu^2\cdot D^3\cdot \hat W^{O(\eps)}\cdot (\log \hat W)^{O(1/\eps^2)}\right )\\
&\quad\quad\quad\quad\quad\leq   O\left ( (N^0(H))^{1+O(\eps)} \cdot \mu^{2+O(\eps)} \cdot D^3\cdot (\log (N^0(H)\mu))^{O(1/\eps^2)}\right ),
\end{split}
 \]

 as required.

The remainder of this paper focuses on the proof of \Cref{thm: main maintain cluster algorithm}. As we have shown, the proof of \Cref{thm: main dynamic NC algorithm}, and hence of \Cref{thm: main final dynamic NC algorithm}, and  Theorems \ref{thm: main} and \ref{thm: NC} follow from it.
We start with a slow and simple algorithm for the \maintaincluster problem, that will be used by our final recursive algorithm for the \maintaincluster problem as a recursion base.

%% file: alg-slow-simple.tex
\subsection{A Slow and Simple Algorithm for the \maintaincluster Problem}\label{subsec: alg slow and simple}

In this subsection, we present a straightforward algorithm, that we refer to as \algslow, for the \maintaincluster problem. Recall that in this problem, we are given a valid input structure $\iset=\left(C,\set{\ell(e)}_{e\in E(C)}, D\right )$, where $C$ is a connected graph that undergoes a sequence $\Sigma$ of valid update operations with dynamic degree bound $\mu$. 
We denote by $N^0(C)$ the number of regular vertices in $C$ at the beginning of the algorithm. We are also given a parameter $\hat W\geq N^0(C)\cdot \mu$. The goal is to  to support queries $\spquery(C,v,v')$: given a pair $v,v'\in V(C)$ of regular vertices of $C$, return a path $P$ of length at most $\alpha D$ connecting them in $C$, in time $O(|E(P)|)$, where $\alpha$ is the desired approximation factor. The algorithm may raise a flag $F_C$ at any time; it is then required to provide 
 supply a pair  $\hat v,\hat v'$ of regular vertices of $C$, with $\dist_C(\hat v,\hat v')>1024D\log^4\hat W$.
Once flag $F_C$ is raised, the algorithm will obtain, as part of its input update sequence, a sequence $\Sigma'$ of update operations, at the end of which 
either $\hat v$ or $\hat v'$ are deleted from $C$. Flag $F_C$ is lowered after these updates are processed by the algorithm. Queries $\spquery$ may only be asked when flag $F_C$ is down. Our  slow and simple algorithm for the \maintaincluster problem is summarized in the following theorem.

\begin{theorem}\label{thm: slow simple}
	There is a deterministic algorithm for the \maintaincluster problem, that we call \algslow, that achieves approximation factor $\alpha=1024\log^4\hat W$ and total update time: 
	$$O ((N^0(C))^2\mu D\poly\log \hat W).$$
\end{theorem}

\begin{proof}
	The algorithm simply maintains, for every regular vertex $v\in V(C)$, the generalized \EST data structure $\tau(v)$ rooted at $v$, with depth threshold $\alpha D$, using the algorithm from \Cref{thm: ES-tree}. For each such tree $\tau(v)$, it also maintains a list $L(v)$ of all regular vertices of $C$ that do not lie in the tree $\tau(v)$. 
	Lastly, for every vertex $x\in V(C)$, we maintain a pointer from $x$ to its location in each of the trees $\tau(v)$ that contain $x$, and similarly for every edge $e\in E(C)$, we maintain a pointer from $e$ to its location in each of the trees $\tau(v)$ containing $e$.
	It is immediate to generalize the algorithm from \Cref{thm: ES-tree} to maintain these additional data structures.

	Since the running time of the algorithm from \Cref{thm: ES-tree} is $\otilde(N^0(C)\cdot \mu \cdot \alpha D)=\otilde(N^0(C)\cdot \mu \cdot D\log^4\hat W)$, the total update time needed to maintain all these data structures is $O((N^0(C))^2\cdot \mu \cdot D\cdot \poly\log \hat W)$.

	Whenever, for any regular vertex $v$, $L(v)\neq \emptyset$ holds, we raise the flag $F_C$, with the corresponding pair of vertices being $v, v'$, where $v'$ is any vertex of $L(v)$. We are then guaranteed that $\dist_C(v,v')>1024D\log^4\hat W$.
	
	(We note that it is possible that the diameter of the initial cluster $C$ is greater than $\alpha D$. In this case, immediately after initializing our data structures, we will discover that for some regular vertices $v$, $L(v)\neq \emptyset$ holds. We will then raise the flag $F_C$ immediately, and will continue raising it until $L(v)=\emptyset$ holds for all regular vertices $v$, and so the diameter of $C$ falls below $\alpha D$. Only then do we start processing the updates from the input sequence and responding to queries).
	
	Recall that a query $\spquery(C,v,v')$ may only be asked when flag $F_C$ is down, so for every regular vertex $v''\in V(C)$, $L(v'')=\emptyset$. 
	In particular, $v'\in \tau(v)$ must hold. Given such a query, we retrace the path $P$ connecting $v'$ to $v$ in the tree $\tau(v)$, and then return this path. Clearly, the running time of this algorithm is $O(|E(P)|)$. Since the depth of the tree $\tau(v)$ is $\alpha D$, the length of path $P$ is at most $\alpha D$, as required.
\end{proof}


We will use \algslow as a subroutine in order to solve the \maintaincluster problem on clusters that are small, e.g. clusters $C$ with $W^0(C)<\hat W^{O(\eps)}$. 
Additionally, for some settings of the parameters, Algorithm \algslow immediately provides the desired running time.  We summarize these 
settings in the following observation that follows immediately from \Cref{thm: slow simple}.

\begin{observation}\label{obs: easy param settings for maintaincluster}
	Consider an input to the \maintaincluster  problem, consisting of a valid input structure $\iset=\left(C,\set{\ell(e)}_{e\in E(C)},D\right )$ that undergoes an online sequence $\Sigma$ of valid update operations with dynamic degree bound $\mu$, and a parameter  $\hat W\geq N^0(C)\cdot \mu$. 
	Then:
	
	\begin{itemize}
		\item if $D> \hat W$, then the running time of \algslow is $O (N^0(C)\mu D^2\poly\log \hat W)$;
		\item if $\hat W<2^{O(1/\eps^2)}$, then  the running time of \algslow is $O(N^0(C)\mu D\cdot 2^{O(1/\eps^2)}\poly\log \hat W)\leq O\left (N^0(C)\mu D(\log \hat W)^{O(1/\eps^2)}\right )$.
	\end{itemize}
\end{observation}

%% file: pseudocuts.tex
\section{Second Main Tool: Balanced Pseudocuts, Expanders, and Their Embeddings}\label{sec: pseudocuts}
Throughout this section, we assume that we are given a valid input structure $\iset=\left(C,\set{\ell(e)}_{e\in E(C)}, D\right )$, where $C$ is a connected graph, with weights $w(v)\geq 1$ on regular vertices $v\in V(C)$ and $w(u)=0$ on supernodes $u\in V(C)$. In this section we develop some tools that will be used in order to design a more efficient algorithm for the \maintaincluster problem. 

\subsection{Balanced Pseudocuts}

Intuitively, following a rather standard definition, for a given balance parameter $\rho$, a \emph{balanced multicut} in a vertex-weighted graph $C$ is a subset $E'\subseteq E(C)$ of edges of $C$, such that every connected component of $C\setminus E'$ has weight at most $W(C)/\rho$. We use instead a notion of \emph{balanced pseudocuts}, that can be viewed  as a relaxation of the notion of a balanced multicut.
The advantage of using this relaxed notion is that, as we later show, we can obtain an algorithm that computes a balanced pseudocut, and embeds an expander $X$, whose vertex set corresponds to a large subset of the edges of the pseudocut, into the graph $C$ with low congestion. We start by formally defining balanced pseudocuts (that we call pseudocuts for brevity). 

\begin{definition}[Pseudocut]
	Suppose we are given  a valid input structure  $\iset=\left(C,\set{\ell(e)}_{e\in E(C)},D\right )$ with weights $w(x)\geq 0$ for vertices $x\in V(C)$, together with another distance bound $\hat D>D$ and a parameter $\rho>1$. A subset $\hat E\subseteq E(C)$ of edges of $C$ is a $(\hat D,\rho)$-\emph{pseudocut}, iff for every vertex $x\in V(C)$, the total weight of all vertices in $B_{C\setminus \hat E}(x,\hat D)$ is at most $\frac{W(C)}{\rho}$. 
\end{definition}

One useful feature of a $(\hat D,\rho)$-pseudocut $\hat E$ is that, if we compute a Neighborhood Cover in graph $C\setminus \hat E$, with distance bound $\hat D/(1024\log^4W)$, using, for example, Procedure \initnc, then the weight of each resulting cluster will be significantly lower than $W(C)$. Standard balanced multicuts mentioned above provide this property as well. As already mentioned, an advantage of pseudocuts is that, as we show below, we can efficiently compute a pseudocut $\hat E$ in graph $C$, together with an expander $X$, whose vertex set $V(X)=\set{t_e\mid e\in \hat E^*}$ corresponds to a large subset $\hat E^*\subseteq \hat E$ of edges of the pseudocut, and an embedding of $X$ into $C$ with low congestion, such that every edge of $X$ is embedded into a short path. Intuitively, if the size of the pseudocut $\hat E$ is sufficiently large, then an algorithm for the \maintaincluster problem can proceed as follows: we maintain an \EST whose root is a new vertex $s$, that connects to every endpoint of every edge $e\in \hat E$ with $t_e\in V(X)$. This allows us to ensure that every vertex of $C$ is sufficiently close to some edge $e\in \hat E$ with $t_e\in V(X)$, and to easily detect when this property no longer holds, so flag $F_C$ can be raised. Additionally, we employ known algorithms for APSP in expander graphs, in order to ensure that the endpoints of all edges in $\set{e\in \hat E\mid t_e\in V(X)}$ remain close to each other. Lastly, as graph $C$ undergoes update operations, we can maintain the expander $X$ and its embedding into $C$ using the ``expander pruning'' algorithm from \Cref{thm: expander pruning}. This algorithm ensures that, if $|\hat E|$ is sufficiently large, then expander $X$ can be maintained over a long enough sequence of update operations of the input structure $\iset$. Once expander $X$ can no longer be maintained (that is, a long enough sequence of update operations has occurred since $X$ was created), we recompute the pseudocut and the corresponding expander, together with its embedding, from scratch. This is precisely the approach that we employ, if the size of the pseudocut that we compute is large enough. However, if $|\hat E|$ is small, then, using this approach, we may need to recompute the expander and the pseudocut too often, resulting in high running time. We discuss our approach for dealing with this situation later. We now proceed to define the main tools that we use.

\subsection{Expanders and Their Embeddings}

Recall that a graph $X$ is a $\phi$-expander, for a parameter $0<\phi<1$, iff for every partition $(A,B)$ of $V(X)$, $|E_X(A,B)|\geq \phi\min\set{|A|,|B|}$ holds.

\begin{definition}[Embeddings of Graphs]
Let $G$, $X$ be two graphs with $V(X)\subseteq V(G)$. An \emph{embedding} of $X$ into $C$ is a collection $\pset=\set{P(e)\mid e\in E(X)}$ of paths in graph $G$, such that, for every edge $e=(x,y)\in E(X)$, path $P(e)$ connects vertex $x$ to vertex $y$. The \emph{congestion} of the embedding is the maximum, over all edges $e'\in E(G)$, of the number of paths in $\pset$ containing $e'$. 
\end{definition}

Assume now that we are given a subset $\hat E\subseteq E(C)$ of edges of the graph $C$. We let $C_{|\hat E}$ be the graph obtained from $C$ by subdividing every edge $e=(u,v)\in \hat E$ with a vertex $t_e$. We set the lengths of the new edges $\ell(u,t_e)=\ell(v,t_e)=\ell(u,v)$.

\begin{definition}[Expander over an Edge Subset]
Given a subset $\hat E\subseteq E(C)$ of edges of the graph $C$, and another graph $X$, we say that $X$ is defined over the set $\hat E$ of edges if $V(X)=\set{t_e\mid e\in \hat E}$. 
\end{definition}

\input{compute-pseudocut-or-expander}

\input{APSP-in-expanders}

\input{maintain-expander}
\input{faster-alg}

%% file: compute-pseudocut-or-expander.tex
\subsection{Computing a Pseudocut, the Corresponding Expander, and its Embedding}

The following theorem is a central tool that we use to design efficient algorithms for the \maintaincluster problem.

\begin{theorem}\label{thm: compute pseudocut and expander}
	There is a deterministic algorithm, that we call \algpseudo, whose input is  a valid input structure $\iset=\left(C,\set{\ell(e)}_{e\in E(C)}, D\right )$,  where $C$ is a connected graph, with arbitrary weights $w(x)\geq 0$ for vertices $x\in V(C)$, and parameters  $0<\epsx<1$, $\hat W\geq \max\set{W(C),|E(C)|, 2^{\Omega(1/\epsx^2)}}$, and  $\hat D>D$,  such that, for every vertex $x\in V(C)$, $w(x)\leq \frac{W(C)}{4\hW^{\epsx}}$ holds. The algorithm 
	computes a $(\hat D,\rho)$-pseudocut $\hat E$ in $C$, with $\rho=\hat W^{\eps}$, and a $\phi^*$-expander $X$  defined over an edge set $\hat E^*\subseteq \hat E$, with $|\hat E^*|\geq \Omega(|\hat E|/\hW^{2\epsx})$, such that the maximum vertex degree in $X$ is $O(\log \hW)$, and $ \phi^*=1/(\log \hW)^{O(1/\epsx)} $. The algorithm also computes an embedding $\pset$ of $X$ into $C_{|\hat E^*}$ with congestion at most $\eta=\hat D\cdot \hat W^{\eps} (\log \hat W)^{O(1/\eps)}$, such that the length of every path in $\pset$ is at most $ O(\hat D\log^2\hat W)$. The running time of the algorithm is  $O\left (|E(C)|\cdot \hat D^2\cdot \hat W^{O(\epsx)}\cdot  (\log \hat W)^{O(1/\epsx^2)} \right )$.
\end{theorem}


We emphasize that the expander $X$ that we construct is unweighted. Therefore, when we talk about the length of a path $P$ in the expander $X$, we refer to the number of edges on $P$. 

\begin{proofof}{\Cref{thm: compute pseudocut and expander}}
	The following theorem is key to the proof of \Cref{thm: compute pseudocut and expander}.

\begin{theorem}\label{lem: smaller pseudocut or expander}
	There is a deterministic algorithm, whose input consists of:
	
	\begin{itemize} 
		\item  a valid input structure $\iset=\left(C,\set{\ell(e)}_{e\in E(C)}, D\right )$,  where $C$ is a connected graph;
		\item arbitrary weights $w(x)\geq 0$ for vertices $x\in V(C)$;
		\item parameters $0< \epsx< 1$, $\hat W\geq \max\set{W(C),|E(C)|,2^{\Omega(1/\epsx^2)}}$, and  $\hat D>D$; and
		\item a $(\hat D,\rho)$-pseudocut $\hat E$  for $C$ of cardinality $k$, with $\rho=\hat W^{\eps}$. 
	\end{itemize}

The algorithm computes one of the following:
	
	\begin{itemize}
		\item either a  $\phi^*$-expander $X$ defined over a subset $\hat E^*\subseteq \hat E$ of at least $\Omega(k/\hat W^{2\epsx})$ edges, with $\phi^*=1/(\log \hat W)^{O(1/\epsx)}$, of maximum vertex degree $O(\log \hat W)$, together with an embedding $\pset$ of $X$ into $C_{|\hat E^*}$ with congestion at most $\hat D\cdot \hat W^{\eps} (\log \hat W)^{O(1/\eps)}$, such that the length of every path in $\pset$ is at most $ O(\hat D\log^2\hat W)$; or
		
		\item another $(\hat D,\rho)$-pseudocut $\hat E'$ in $C$, with $|\hat E'|\leq |\hat E|\left (1-\frac{1}{\hat W^{\eps} (\log \hat W)^{O(1/\eps)} }\right )$.
	\end{itemize}  
	The running time of the algorithm is $O\left (|E(C)|\cdot  \hat D^2\cdot \hat W^{O(\eps)}(\log \hat W)^{O(1/\eps^2)}\right )$.
\end{theorem}

The proof of \Cref{thm: compute pseudocut and expander} easily follows from \Cref{lem: smaller pseudocut or expander}.
Throughout, we set $\rho=\hat W^{\eps}$. 
Recall that for every vertex $x\in V(C)$, $w(x)\leq \frac{W(C)}{4\hat W^{\epsx}}\leq \frac{W(C)}{4\rho}$ holds.
 We start with an initial $(\hat D,\rho)$-pseudocut $\hat E_0$, that is constructed as follows. We construct a partition $S_1,\ldots,S_r$ of $V(C)$, where $r=4\rho$, and for all $1\leq i\leq r$, $W(S_i)\leq W(C)/\rho$ via a simple greedy algorithm: start with $S_1=S_2=\cdots=S_r=\emptyset$, and then process the vertices of $C$ in the order of their weights from highest to lowest, breaking ties arbitrarily. When vertex $x$ is processed, we add it to the set $S_i$, whose current weight is the smallest (if several sets have the same weight, choose one arbitrarily). It is immediate to verify that, when the algorithm terminates, $\max_{1\leq i<j\leq r}|W(S_i)-W(S_j)|\leq 2\max_{x\in V(C)}w(x)\leq \frac{W(C)}{2\rho}$. Since $r=4\rho$, there must be some set $S_i$ with $W(S_i)\leq W(C)/(4\rho)$. We conclude that for every set $S_j$, $W(S_j)\leq W(C)/\rho$. We then let $\hat E_0$ contain all edges whose endpoints lie in different sets $S_1,\ldots,S_r$. It is immediate to verify that $\hat E_0$ is a $(\hat D,\rho)$-pseudocut (in fact it is a valid $(D^*,\rho)$-pseudocut for any parameter $D^*<\infty$).

We then perform iterations. The input to the $i$th iteration is the current $(\hat D,\rho)$-pseudocut $\hat E_{i-1}$, where the input to the first iteration is the pseudocut $\hat E_0$ that we have just computed. We apply the algorithm from \Cref{lem: smaller pseudocut or expander} to the same input $\iset$, and $(\hat D,\rho)$-pseudocut $\hat E_{i-1}$. If the outcome of the algorithm is another $(\hat D,\rho)$-pseudocut $\hat E'$ with $|\hat E'|\leq |\hat E_{i-1}|\left (1-\frac{1}{\hat W^{\eps} (\log \hat W)^{O(1/\eps)} }\right )$, then we set $\hat E_i=\hat E'$, and continue to the next iteration. Otherwise, the algorithm must return 
	a  $\phi^*$-expander $X$ defined over an edge set $\hat E^*\subseteq \hat E_{i-1}$, for $\phi^*=1/(\log \hat W)^{O(1/\epsx)}$, with $|\hat E^*|\geq \Omega(|\hat E_{i-1}|/\hat W^{2\epsx})$, such that the maximum vertex degree in $X$ is at most $O(\log \hat W)$. The algorithm also computes an embedding $\pset$ of $X$ into $C$ with congestion at most $\hat D\cdot \hat W^{\eps} (\log \hat W)^{O(1/\eps)}$, such that the length of every path in $\pset$ is at most $ O(\hat D\log^2\hat W)$. We then terminate the algorithm, and return $\hat E_{i-1}, X$ and $\pset$ as its outcome.
	
	Since we are guaranteed that, for all $i$,  $|\hat E_i|\leq |\hat E_{i-1}|  \left (1-\frac{1}{\hat W^{\eps} (\log \hat W)^{O(1/\eps)} }\right )$, and $|E(C)|\leq \hat W$, the number of iterations is bounded by $\hat W^{\eps} \cdot (\log \hat W)^{O(1/\eps)}$. Since the running time of each iteration is $O\left (|E(C)|\cdot  \hat D^2\cdot \hat W^{O(\eps)}(\log \hat W)^{O(1/\eps^2)}\right )$,
	the total running time of the algorithm is bounded by:
	
	\[O\left (|E(C)|\cdot  \hat D^2\cdot \hat W^{O(\eps)}\cdot (\log \hat W)^{O(1/\eps^2)}\right ).\]

	In order to complete the proof of \Cref{thm: compute pseudocut and expander}, it is now enough to prove \Cref{lem: smaller pseudocut or expander}. The proof is somewhat technical, and we provide it in Section \ref{subsec: pseudocut or expander} of Appendix.
	We provide here a high-level overview of the proof. A key part of the proof is an algorithm that either computes the desired expander $X$ and its embedding, or it computes a collection $E_1,\ldots,E_h$ of disjoint subsets of edges from $\hat E$, such that the cardinality of each subset is at least $k'\geq k/(\hat W^{O(\epsx)}(\log \hat W)^{O(1/\eps)})$ and $h\geq \rho$. The algorithm also computes a subset $E'$ of at most $k'/2$ edges of $C$, such that, in graph $C\setminus E'$, for every pair $E_i,E_j$ of distinct subsets of edges, the length of the shortest path connecting an edge of $E_i$ to an edge of $E_j$ is at least $10\hat D$. This is sufficient for us in order to define the new pseudocut, $\hat E'=(\hat E\setminus E_i)\cup E'$, where set $E_i$ is carefully chosen to ensure that $\hat E'$ is indeed a valid pseudocut. Intuitively, $E_i$ is chosen so that the total weight of all vertices in the ball of radius $2\hat D$ around $E_i$ in graph $C\setminus E'$ is the smallest.
	
	The algorithm mentioned above, that either computes the expander $X$ with its embedding, or the desired collection $E_1,\ldots,E_h$ of subsets of $\hat E$, proceeds by iteratively applying the cut-matching game of \cite{KRV,KhandekarKOV2007cut} to graph $C_{|\hat E}$. Using standard techniques, such an algorithm can be used in order to either compute an expander $X$ over a large enough subset $\hat E^*\subseteq \hat E$ of edges and to embed it into $C_{|\hat E}$, or to compute two large enough subsets $\hat E_1,\hat E_2$ of edges of $\hat E$, together with another relatively small subset $E'$ of edges of $C$, such that, in $C\setminus E'$, every path connecting an endpoint of an edge of $\hat E_1$ to an endpoint of an edge of $\hat E_2$ is sufficiently long. We start by applying this algorithm to the initial edge set $\hat E$ in graph $C_{|\hat E}$. If the algorithm returns the desired expander and its embedding, then we are done. Otherwise, it returns two subsets $\hat E_1,\hat E_2$ of edges of $\hat E$, and edge set $E'$. We then recursively apply the same algorithm to both $\hat E_1$ and $\hat E_2$, instead of $\hat E$. Again, if the outcome is an expander and its embedding into $C_{|\hat E}$ then we are done; otherwise, we obtain sufficiently large subsets $\hat E^1_1,\hat E_1^2$ of $\hat E_1$, and  $\hat E^1_2,\hat E_2^2$ of $\hat E_2$. We then continue applying the same algorithm to each of the four resulting subsets of $\hat E$. This process continues until we either compute an expander $X$ and its embedding into $C_{|\hat E}$, or the number of subsets of $\hat E$ reaches $h= \Omega(\rho)$. Throughout the algorithm, we maintain a subset $E'$ of edges of $C$, such that, in graph $C\setminus E'$, the distance between every pair of edge sets $\hat E_i,\hat E_j$ that we maintain is at least $\Omega(\hat D)$. Once the number of disjoint subsets of $\hat E$ that we maintain reaches $h$, if we did not terminate the algorithm with the desired expander $X$ and its embedding, we return the resulting subsets $\hat E_1,\ldots,\hat E_h$ of $\hat E$, together with the edge set $E'$ that the algorithm maintained.
	 The complete proof of \Cref{lem: smaller pseudocut or expander} appears in Section \ref{subsec: pseudocut or expander} of Appendix.
\end{proofof}

%% file: APSP-in-expanders.tex
\subsection{APSP in Expanders}

It is well known (see \Cref{obs: short paths in exp}), that, if an $n$-vertex graph $X$ is a $\phi$-expander, and maximum vertex degree in $X$ is bounded by $\Delta$, then for any pair $x,y$ of vertices in $X$, there is a path connecting them of length at most $O (\Delta\log n/\phi)$. We need an algorithm that supports short-path queries on an expander, as it undergoes edge deletions. In each such query, given a pair $x,y$ of vertices of $X$, we need to return a path connecting $x$ to $y$ in $X$ of length comparable to $O (\Delta\log n/\phi)$. As expander $X$ undergoes edge deletions, we will also prune some vertices out of it to ensure that it remains an expander. Such an algorithm was provided in \cite{APSP-old}, who used techniques similar to those in \cite{detbalanced}. However, we need a slightly different tradeoff between the approximation factor and the running time from that in \cite{APSP-old}, and we  summarize it in the following theorem. The proof is practically identical to that of \cite{APSP-old} and is deferred to Section \ref{subsec: expander APSP proof} of Appendix.

\begin{theorem} (Analogue of Theorem 3.13 in \cite{APSP-old})
	\label{thm: expander APSP} There is a deterministic algorithm that, given
	an $n$-vertex (unweighted) graph $X$ that is a $\phi$-expander with maximum vertex degree at  most $\Delta$, and a parameter $0<\eps<1$, such that $\phi\leq 1/2^{\Omega(1/\eps)}$, with graph $X$ undergoing a sequence at most $\phi|E(X)|/(20\Delta)$ edge deletions,
	maintains a vertex set $S\subseteq V(X)$, such that, for every $t>0$, after $t$ edges are deleted from $X$, $|S|\leq O(t\Delta/\phi)$ holds. 
	Vertex set $S$ is incremental, so vertices may join it but they may not leave it.
	The algorithm also supports queries $\exspquery$: given a pair of vertices $x,y\in V(X)\setminus S$, return an $x$-$y$  path $P$ in $X\setminus S$ of length at most $O\left( \Delta^2(\log n)^{O(1/\eps^2)}/\phi\right )$, with query time $O(|E(P)|)$. 
	The total update time of the algorithm is $O\left (\frac{n^{1+O(\eps)}\Delta^7(\log n)^{O(1/\eps^2)}}{\phi^5}\right ) $.
\end{theorem}

%% file: maintain-expander.tex
\subsection{Maintaining the Expander}

Suppose we are given some $\phi$-expander $X$ that is defined over some subset $\hat E^*$ of edges of $C$, with $|\hat E^*|=k$, and an embedding $\pset$ of $X$ into $C_{|\hat E^*}$, with congestion at most $\eta$. 
Next, we show that this expander can ``withstand'' a long sequence of valid update operations. In order to do so, we need to define a \emph{cost} of each valid update operation.

\begin{definition}[Update operation costs]
The costs of valid update operations are defined as follows. The cost of an edge-deletion is $1$, and the cost of every other valid update operation is $0$.
\end{definition}

The following observation is immediate.

\begin{observation}\label{obs: total cost of all updates}
	Let $\iset=\left(C,\set{\ell(e)}_{e\in E(C)},D\right )$ be a valid input structure, undergoing a sequence $\Sigma$ of valid update operations, with dynamic degree bound $\mu$. Then the total cost of all update operations in $\Sigma$ is at most $O(N^0(C)\mu)$, where $N^0(C)$ is the number of the regular vertices in the initial graph $C$.
\end{observation}

The observation follows from the fact that the total number of edges that are ever present in $C$ is bounded by $N^0(C)\mu$.
Lastly, we use the following theorem in order to maintain the expander $X$ computed by \Cref{thm: compute pseudocut and expander}.

\begin{theorem}\label{thm: maintain expander}
	There is a deterministic algorithm, that we call \algmaintainexp, that receives as input the following:
	
	\begin{itemize}
		\item a valid input structure $\iset=\left(C,\set{\ell(e)}_{e\in E(C)},D\right )$, where $C$ is a connected graph with maximum vertex degree at most $\mu$;
		
		\item parameters $0< \epsx< 1$, $\hat W\geq \max\set{N^0(C),\mu,|E(C)|,2^{\Omega(1/\epsx^2)}}$, and  $\hat D>D$, where $N^0(C)$ is the number of regular vertices in $C$ at the beginning of the algorithm;
		
		\item a $\phi^*$-expander $X$, for $\phi^*=1/(\log \hat W)^{O(1/\epsx)}$, defined over an edge set $\hat E^*\subseteq E(C)$ of cardinality $k$, such that the maximum vertex degree in $X$ is at most $ O(\log \hat W)$;
		
		\item an embedding $\pset$ of $X$ into $C_{|\hat E^*}$ with congestion $\eta=\hat D\cdot \hat W^{\eps} (\log \hat W)^{O(1/\eps)}$, where each path in $\pset$ has length at most $O(\hat D\log^2\hat W)$; and
		
		\item an online sequence $\Sigma$ of valid update operations of total cost at most $\frac{\phi^* k}{c\eta\log \hat W}$ for some large enough constant $c$, with dynamic degree bound $\mu$.
	\end{itemize} 
	 The algorithm maintains a non-empty subgraph $X'\subseteq X$, and supports  expander-short-path queries: given a pair of vertices $x,y\in V(X')$, return an $x$-$y$  path $P$ in graph $C_{|\hat E^*}$, of length at most $O\left(\hat D\cdot (\log \hat W)^{O(1/\epsx^2)}\right )$, with query time $O(|E(P)|)$. 
	Graph $X'$ is decremental, so vertices and edges may be deleted from it, but they may never be added to it.
	The total update time of the algorithm is bounded by:  $$O\left (N^0(C)\cdot \mu\cdot \hat D^2\cdot\hat W^{O(\epsx)}\cdot (\log \hat W)^{O(1/\epsx^2)}\right ) .$$
\end{theorem}

Note that, by substituting the parameters $\phi^*$ and $\eta$, Algorithm  \algmaintainexp  can maintain the graph $X'$ over the course of a sequence of valid update operations of total cost at most $\Omega \left ( \frac{k}{\hat D\cdot \hat W^{O(\epsx)\cdot }(\log \hat W)^{O(1/\epsx)}  }\right )$.

\begin{proof}
	Note that, since graph $C$ is connected, $k\leq |E(C)|\leq N^0(C)\mu$ must hold.
	We use the algorithm from \Cref{thm: expander APSP} in order to maintain the expander $X$ under edge deletions, with parameters $\eps$, $\phi=\phi^*$, and $\Delta=O(\log \hat W)$. 
	For every edge $e\in E(C)$, we maintain a list $L(e)$ of all edges $e'\in E(X)$ with $e\in P(e')$, where $P(e')\in \pset$ is the path embedding $e'$. (If edge $e$ lies in $\hat E^*$, then it is split into two edges by vertex $t_e$ in graph $C_{|\hat E^*}$; we then include in $L(e)$ all edges of $X$ whose embedding contains either one of the resulting edges.) As each such list $L(e)$ contains at most $\eta=\hat D\cdot \hat W^{\eps} (\log \hat W)^{O(1/\eps)})$ edges, initializing all lists $L(e)$ takes time $O(|E(X)|\cdot \eta)\leq O(|E(C)|\cdot \eta^2)\leq O(N^0(C)\mu  \hat D^2 \hat W^{O(\epsx)}\cdot (\log \hat W)^{O(1/\eps)})$.
	
	We handle update operations performed on cluster $C$ as follows. First, if an isolated vertex is deleted from $C$, or a supernode splitting operation is executed, then we do nothing. Assume now that an edge $e$ is deleted from $C$. Then for every edge $e'\in L(e)$, we delete $e'$ from the graph $X$, and update the data structure from \Cref{thm: expander APSP}  accordingly. Note that, if the cost of an update operation is $a$, then the number of edges deleted from $X$ as the result of this operation is at most $2\eta a$. Therefore, if we perform a sequence of valid update operations in graph $C$, whose total cost is at most $\phi^* k/(c\eta\log\hat W)$, this will result in a sequence of edge deletions to the data structure from \Cref{thm: expander APSP}, whose length is at most $2\phi^*k/(c\log \hat W) \leq O\left(\phi^*|E(X)|/(c\Delta)\right)$, where $\Delta$ is the maximum vertex degree in $X$. 
	By letting $c$ be a large enough constant, we can ensure that this is bounded by $\phi^*|E(X)|/(20\Delta)$.
	The graph $X'$ that we maintain is defined as follows. Consider some time $t$ in the algorithm, after $t$ updates to graph $C$ were executed. Let $M(t)$ denote the number of edge deletions from graph $X$ that were executed so far, and let $E_t$ be the corresponding set of edges. Let $S_t$ be the set $S$ that the algorithm from \Cref{thm: expander APSP} maintains, after the edges in $E_t$ are deleted from $X$. Then graph $X'$ at time $t$ is defined to be $(X\setminus E_t)\setminus S_t$. Observe that \Cref{thm: expander APSP} ensures that graph $X'$ non-empty, since for all $t$, $|S_t|\leq O(|E_t|\Delta /\phi^*)\leq O(\phi^*k\Delta/(c\phi^*\log \hat W)\leq O(k\Delta/(c\log \hat W))$; letting $c$ be a large enough constant, we can ensure that this number is below $k/2$. The total update time of the algorithm from \Cref{thm: expander APSP} is bounded by:
	$O\left (\frac{k^{1+O(\epsx)}\Delta^7(\log k)^{O(1/\epsx^2)}}{(\phi^*)^5}\right )\leq O\left ((N^0(C)\mu)^{1+O(\epsx)}(\log \hat W)^{O(1/\epsx^2)}\right )\leq O\left (N^0(C)\mu\hat W^{O(\epsx)}(\log \hat W)^{O(1/\epsx^2)}\right ) $.
	
	Consider now expander-short-path query, in which we are given a pair of vertices $x,y\in V(X')$, and we need to return an $x$-$y$  path $P$ in graph $C_{|\hat E^*}$, of length at most $O\left(\hat D \cdot (\log \hat W)^{O(1/\epsx^2)}\right )$, with query time $O(|E(P)|)$. We run query $\expquery(x,y)$ in graph $X'$, using the algorithm from \Cref{thm: expander APSP}, obtaining an $x$-$y$ path $Q$ in $X'$, of length at most 
	$O\left( \Delta^2(\log k)^{O(1/\epsx^2)}/(\phi^*)^5\right )\leq O\left((\log \hat W)^{O(1/\epsx^2)}\right ) $, in time $O(|E(Q)|)$. Next, we use, for each edge $e\in E(Q)$, the embedding path $P(e)\in \pset$, whose length is $O(\hat D\log^2\hat W)$, to replace the edge $e$ on the path $Q$. As the result, we obtain a path $Q'$ in graph $C_{|\hat E^*}$ connecting $x$ to $y$, of length $O\left(\hat D\cdot (\log \hat W)^{O(1/\epsx^2)}\right )$. The query time of the algorithm is  $O(|E(Q')|)$.
\end{proof}

%% file: faster-alg.tex
\subsection{A Faster Algorithm for the \maintaincluster Problem}
The tools that we have developed so far are already sufficient in order to obtain a faster algorithm for the \maintaincluster problem. Since we do not use this algorithm in our final result, we only provide its sketch here, starting with some intuition.

Intuitively, an algorithm for the \maintaincluster problem,  on a given input structure $\iset=\left(C,\set{\ell(e)}_{e\in E(C)}, D\right )$ undergoing a sequence of valid update operations with dynamic degree bound $\mu$ can proceed as follows. We set the weight of every regular vertex to be $1$, and of every supernode to $0$. We then apply Algorithm \algpseudo to compute a $(\hat D,\rho)$-pseudocut $\hat E$, for $\hat D=\Theta(D\log^4\hat W)$, and $\rho=\hat W^{\eps}$, where $\hat W=W\mu$, and $W$ is the given input parameter bounding the total number of regular vertices in the initial graph $C$. We also compute the corresponding expander $X$ defined over a large subset $\hat E^*\subseteq \hat E$ of edges, and its embedding $\pset$ into $C_{|\hat E^*}$. We then initialize Algorithm \algmaintainexp on the expander $X$, that will maintain a subgraph $X'\subseteq X$. Additionally, we construct an \EST in the graph $C_{|E^*}$, whose root vertex is a new vertex $s$, that connects with an edge to every vertex of $V(X')$. The depth threshold of the \EST is $2^{13} D\log^4 \hat W$. The \EST will allow us to ensure that every vertex of $C$ remains close to one of the vertices of $V(X')$; whenever this is not the case, this data structure allows us to detect it and to provide a violating pair $v,v'$ of regular vertices with $\dist_{C}(v,v')>1024 D\log^4 \hat W$. We are also guaranteed by \Cref{thm: maintain expander}, that every pair of vertices in the expander $X$ has a short path connecting them in $C_{|E^*}$. This ensures that, for every pair $e,e'\in \hat E^*$ of edges, there is a short path connecting their endpoints in $C$. These data structures can also easily support queries $\spquery$ for pairs of vertices in cluster $C$, that are required from an algorithm for the \maintaincluster problem. As observed already, this data structure can withstand a sequence of update operations to cluster $C$, whose cost is  $\Omega \left ( \frac{|\hat E^*|}{\hat D\cdot \hat W^{O(\epsx)}\cdot (\log \hat W)^{O(1/\epsx)}  }\right )$, and the total update time required to maintain the data structure is $O\left (N^0(C)\cdot \mu\cdot \hat D^2\cdot \hat W^{O(\epsx)}\cdot (\log \hat W)^{O(1/\epsx^2)}\right ) $.

Once a sequence of update operations of cost $\Theta \left ( \frac{|\hat E^*|}{\hat D\cdot \hat W^{O(\epsx)}\cdot (\log \hat W)^{O(1/\epsx)}  }\right )$ is performed in cluster $C$, we discard the current data structures, run the algorithm \algpseudo again, and recompute the data structures from scratch. We call the execution of the algorithm between a pair of successive calls to \algpseudo a single phase. If we are lucky, and every time that Algorithm \algpseudo is called, we obtain a pseudocut $\hat E$, whose cardinality is at least $|E(C)|/\hat W^{\epsx}$, then the number of phases will be bounded by $O\left (\hat W^{O(\epsx)}\mu D(\log \hat W)^{O(1/\epsx)}\right ) $. The total update time of this algorithm will then be at most $$O\left (N^0(C)\mu^2D^3\hat W^{O(\epsx)}(\log \hat W)^{O(1/\epsx^2)}\right ) .$$ However, we may not be so lucky, and it is possible that at some point Algorithm \algpseudo may return a pseudocut $\hat E$ of small cardinality.

One straightforward way to deal with this issue is to initialize an \EST in graph $C$, rooted at each regular vertex that serves as an endpoint of an edge in $\hat E$; the depth each tree is $2^{13} D\log^4\hat W$. Once the root of a tree is deleted from $C$, we say that the tree is destroyed. As long as not all such trees are destroyed, we can implement an algorithm for the \maintaincluster problem, similarly to \algslow. Once every tree is destroyed, we are guaranteed that all edges of $\hat E$ are deleted from $C$. From the properties of the pseudocut, and since we guarantee that the diameter of $C$ remains below $2^{13}D\log^4\hat W<\hat D$, this means that the total number of all regular vertices in $C$ has decreased by at least factor $\rho$. We can then restart the algorithm outlined above for the \maintaincluster problem in $C$ from scratch. This approach provides a faster algorithm than \algslow, since we no longer have an \EST data structure rooted at every single regular vertex of $C$, and since $|\hat E|$ is relatively small. In fact, in order to optimize the running time of the algorithm, we should use a threshold of $\sqrt{N^0(C)}$ for $|\hat E|$: as long as $|\hat E|\geq \sqrt{N^0(C)}$, we use the expander-based approach, and once $|\hat E|$ falls below this value, we use the \EST's; this will result in running time that is roughly $(N^0(C)\mu)^{1.5+O(\eps)}D^3(\log (W\mu))^{O(1/\eps)}$ for the \maintaincluster problem. One advantage of this approach is that it avoids recursion. 
We roughly estimate that, using this algorithm instead of \Cref{thm: main maintain cluster algorithm} would result in an algorithm for \APSP with total update time approximately $O\left (m^{1.5+O(\eps)}(\log m)^{O(1/\eps^2)}\right )$, and approximation factor $(\log m)^{O(1/\eps)}$.
This already gives a somewhat fast algorithm for sparse graphs, but this running time is still much higher than our desired bound.


%% file: good-clusters.tex
\section{Third Main Tool: Good Clusters}
\label{sec: pivot decomposition}

Throughout this section, we assume that we are given an input to the \maintaincluster problem, that consists of a valid input structure $\iset=\left(C,\set{\ell(e)}_{e\in E(C)},D \right )$, where $C$ is a connected graph undergoing a sequence $\Sigma$ of valid update operations with a dynamic degree bound $\mu$. We refer to $D$ as \emph{target distance threshold for $C$}. 
We are also given a parameter $\hat W\geq N^0(C)\cdot \mu$, where $N^0(C)$ is the number of regular vertices of $C$ at the beginning of the algorithm, and a precision parameter $0<\eps<1$. 
We also use a parameter $\hat D=2^{30}D\log^{10} \hat W$.

Throughout this section, we assume that $D\leq \hat W$, and $\hat W\geq 2^{\Omega(1/\eps^2)}$, since otherwise, from \Cref{obs: easy param settings for maintaincluster}, algorithm \algslow provides the desired bounds on the running time, with approximation factor $O(\log^4\hat W)$.
We define vertex weights as usual: the weight of each regular vertex is $1$, and the weight of each supernode is $0$. We denote by $W(C)$ the total weight of all vertices in $C$.
 We are then guaranteed that, throughout the execution of the 
update sequence $\Sigma$, $\hat W\geq \max\set{W(C)\mu,|E(C)|, 2^{\Omega(1/\epsx^2)}}$ holds.

Throughout the remainder of this section, we also use the following parameters from \Cref{thm: compute pseudocut and expander}: $\rho=\hat W^{\eps}$, $ \phi^*=1/(\log \hat W)^{O(1/\epsx)}$, and $\eta=
\hat D\cdot \hat W^{\eps} (\log \hat W)^{O(1/\eps)}$.
For convenience, all these parameters are summarized in  \Cref{sec: params}.

In this section, we will define a new problem, called \maintaingoodcluster, which is an easier version of the \maintaincluster problem. In this new problem, we will only need to maintain the cluster $C$ and support the $\spquery(C,v,v')$ queries as long as the cluster is ``good'' (we define what it means below). Once the cluster is no longer good, the algorithm can raise a flag $F^b_C$ to indicate that, and then the algorithm terminates. However, when flag $F^b_C$ is raised, the algorithm is required to provide a ``bad witness'' for the fact that $C$ is no longer a good cluster, which is a small $(\hat D,\rho)$-pseudocut of $C$. 

We start by defining good clusters and bad witnesses in \Cref{subsec: good clusters}. 
We then define the \maintaingoodcluster problem and provide an algorithm for it in \Cref{subsec: maintain good cluster problem}.
All main parameters that are used in this section are summarized in \Cref{sec: params}.

\subsection{Good Clusters and Witnesses}
\label{subsec: good clusters}

Assume that we are given a valid input structure $\iset=\left(C,\set{\ell(e)}_{e\in E(C)},D\right )$, parameters $\mu$, $0<\eps<1$, and $\hat W\geq N^0(C)\cdot \mu$, that we use in order to define parameters $\hat D$, $\rho$,  $\phi^*$ and $\eta$, as described above (see also \Cref{sec: params}).

\begin{definition}[Good Clusters and Witnesses]
We say that $C$ is a \emph{type-1 good cluster} iff $W(C)\leq \hat W^{10\eps}$. We say that $C$ is a \emph{type-2 good cluster} iff there is a collection $\hat E^*\subseteq E(C)$ of edges of $C$ of cardinality $k\geq \Omega(W(C)/\hat W^{3\epsx})$, and a $\phi^*$-expander $X$ defined over the edges of $\hat E^*$, with maximum vertex degree at most $O(\log \hat W)$, together with an embedding $\pset$ of $X$ into $C_{|\hat E^*}$ with congestion at most $\eta$, such that the length of every path in $\pset$ is at most $O(\hat D\log^2\hat W)$. We call $(\hat E^*,X,\pset)$ a \emph{good witness} for $C$.
\end{definition}

We will informally refer to a cluster that is type-1 or type-2 good as a good cluster. Note that a cluster may be a type-2 good cluster, but this does not guarantee that we can efficiently find a good witness for it; generally when we establish that a cluster is a type-2 good cluster, we will also provide a good witness for it.  

\begin{definition}[Bad Witness]
	At any point during the execution of the valid update sequence $\Sigma$ on $C$, we say that a set $\hat E$ of edges is a \emph{bad witness} for $C$ iff $|\hat E|<W(C)/\hat W^{\epsx}$, and $\hat E$ is a $(\hat D,\rho)$-pseudocut for $C$, where, as we defined before, $\hat D=2^{30}D\log^{10} \hat W$, and $\rho=\hat W^{\eps}$.
\end{definition}

(Notice that it is possible that a cluster is a good cluster and yet it has a bad witness).
We need the following observation that easily follows from \Cref{thm: compute pseudocut and expander}.

\begin{observation}\label{obs: good cluster or bad witness}
	There is a deterministic algorithm, that, 	given  a valid input structure \newline $\iset=\left(C,\set{\ell(e)}_{e\in E(C)}, D\right )$,  where $C$ is a connected graph, with weights $w(v)=1$ on its regular vertices and $w(u)=0$ on supernodes, together with parameters  $0<\epsx<1$ and $\hat W\geq \max\set{W(C),|E(C)|, 2^{\Omega(1/\epsx^2)}}$, either (i) correctly establishes that $C$ is a type-1 good cluster; or (ii) correctly establishes that $C$ is a type-2 good cluster and produces a good witness for it; or (iii) computes a bad witness $\hat E$ for $C$. The running time of the algorithm is $O\left (|E(C)|\cdot \hat D^2\cdot \hat W^{O(\epsx)}\cdot  (\log \hat W)^{O(1/\epsx^2)} \right )$.
\end{observation}
\begin{proof}
	Checking whether $C$ is a type-1  good cluster can be easily done in time $O(|E(C)|)$. 
	Assume now that $C$ is not a type-1 good cluster. Since the weight of every regular vertex is $1$, and $W(C)\geq \hat W^{10\eps}$, we get that for every vertex $x\in V(C)$, $w(x)\leq \frac{W(C)}{4\hW^{\epsx}}$. We apply  Algorithm \algpseudo from \Cref{thm: compute pseudocut and expander} to the cluster $C$. Consider the $(\hat D,\rho)$-pseudo-cut $\hat E$ that that the algorithm computes. If $|\hat E|<W(C)/\hat W^{\epsx}$, then $\hat E$ is a bad witness for $C$. Otherwise, the algorithm computes an edge set  $\hat E^*\subseteq \hat E$, with $|\hat E^*|\geq \Omega(|\hat E|/\hW^{2\epsx})\geq \Omega(W(C)/\hat W^{3\epsx})$, and a 
	$\phi^*$-expander $X$  defined over the edge set $\hat E^*$, such that the maximum vertex degree in $X$ is $O(\log \hat W)$. The algorithm also computes an embedding $\pset$ of $X$ into $C_{|\hat E^*}$ with congestion at most $\eta$, such that the length of every path in $\pset$ is at most $O(\hat D\log^2\hat W)$. Clearly, $(\hat E^*,X,\pset)$ is a good witness for $C$. Lastly, the running time of \algpseudo is $O\left (|E(C)|\cdot \hat D^2\cdot \hat W^{O(\epsx)}\cdot  (\log \hat W)^{O(1/\epsx^2)} \right )$.
\end{proof}

\input{maintain-good-cluster}

%% file: maintain-good-cluster.tex
\subsection{\maintaingoodcluster Problem and an Algorithm for It}\label{subsec: maintain good cluster problem}
In this subsection we define the \maintaingoodcluster problem, that is identical to the \maintaincluster problem, except that now the algorithm may raise flag $F_C^b$ at any point, to indicate that cluster $C$ is not a good cluster. When raising flag $F_C^b$, the algorithm must also supply a bad witness for $C$. The algorithm then terminates. 

\begin{definition}[\maintaingoodcluster Problem]
	In the \maintaingoodcluster problem, we are given a valid input structure $\iset=\left(C,\set{\ell(e)}_{e\in E(C)},D\right )$, together with weights $w(v)= 1$ on regular vertices $v\in V(C)$, and weights $w(u)=0$ on supernodes $u\in V(C)$, and parameters $\mu\geq 1$, $0<\eps<1$, and $\hat W\geq \max\set{N^0(C)\mu, 2^{\Omega(1/\epsx^2)}}$. Graph $C$ undergoes a sequence $\Sigma$ of valid update operations with dynamic degree bound $\mu$.
	The goal of the algorithm is to  support queries $\spquery(C,v,v')$: given a pair $v,v'\in V(C)$ of regular vertices of $C$, return a path $P$ of length at most
	 $O\left(D\cdot (\log \hat W)^{O(1/\eps^2)}\right )$ connecting them in $C$, in time $O(|E(P)|)$. Additionally, the algorithm may, at any time, raise a flag $F_C$. When the algorithm raises flag $F_C$, it must  supply a pair  $\hat v,\hat v'$ of regular vertices of $C$, with $\dist_C(\hat v,\hat v')>1024D\log^4\hat W$.
	 Once flag $F_C$ is raised, the algorithm will obtain, as part of its input update sequence, a sequence $\Sigma'$ of update operations called flag-lowering sequence, at the end of which 
	 either $\hat v$ or $\hat v'$ are deleted from $C$. Flag $F_C$ is lowered after these updates are processed by the algorithm. Lastly, the algorithm may raise, at any time, flag $F_C^b$. When flag $F_C^b$ is raised, the algorithm needs to supply a bad witness $\hat E$ for $C$. Once the flag $F_C^b$ is raised, the algorithm terminates.  Queries $\spquery$ may only be asked when flags $F_C,F_C^b$ are down.
\end{definition}

In the following theorem we provide an efficient algorithm for the \maintaingoodcluster problem.

\begin{theorem}\label{thm: alg for maintain good cluster}
	There is a deterministic algorithm, that we call \algmaintaingoodcluster for the \newline \maintaingoodcluster problem, whose total update time, on input cluster $C$ is:

$$O\left (N^0(C)\cdot \mu^2\cdot D^3\cdot\hat W^{O(\epsx)}\cdot (\log \hat W)^{O(1/\epsx^2)}\right ).$$	
	
%
\end{theorem}

\begin{proof}
	We define the parameters $\hat D$, $\rho$, $\eta$, $\phi^*$ exactly as before (see also \Cref{sec: params}).
	We partition the algorithm's execution into phases. We denote by $W_j(C)$ the weight of all vertices of $C$ at the beginning of the $j$th phase. At the beginning of the $j$th phase, we run the algorithm from \Cref{obs: good cluster or bad witness} in order to check whether $C$ is a good cluster. 
	Recall that the running time of this algorithm is:
$$O\left (|E(C)|\cdot \hat D^2\cdot \hat W^{O(\epsx)}\cdot  (\log \hat W)^{O(1/\epsx^2)} \right )\leq O\left (N^0(C)\cdot \mu \cdot D^2\cdot \hat W^{O(\epsx)}\cdot  (\log \hat W)^{O(1/\epsx^2)} \right ).$$

	
	We then consider three cases.
	
\paragraph{Case 1.}
	The first case happens if the algorithm computes a bad witness $\hat E$ for $C$. In this case, we say that the phase is a bad phase. We then raise the flag $F^b_C$, with the bad witness $\hat E$, and the algorithm terminates. 
	
\paragraph{Case 2.}
	The second case happens if the algorithm establishes that $C$ is a type-1 good cluster, that is, $W_j(C)\leq \hat W^{10\eps}$. In this case, the current phase becomes the last phase of the algorithm. From now on, we run Algorithm \algslow from \Cref{thm: slow simple} on cluster $C$. Recall that the algorithm solves the \maintaincluster problem for $C$, and, when flag $F_C$ is not raised, it supports queries $\spquery(C,v,v')$ queries, where, given a pair $v,v'\in V(C)$ of regular vertices of $C$, it returns a path $P$ of length at most $1024D\log^4\hat W$ connecting them  in $C$, in time $O(|E(P)|)$. When flag $F_C$ is raised, the algorithm produces a pair $\hat v,\hat v'$ of regular vertices of $C$, with $\dist_C(\hat v,\hat v')>1024D\log^4\hat W$  The total update time of this algorithm is:

	$$O (W_j(C))^2\mu D\poly\log \hat W)\leq O (N^0(C)\cdot \mu \cdot D\cdot \hat W^{O(\eps)}\cdot \poly\log \hat W).$$

	 If Case 2 happens, then we say that the phase is type-1 good. From the above discussion, at most one type-1 good phase may happen over the course of the algorithm.

\paragraph{Case 3.}
If Case 3 happens, then the algorithm from  \Cref{obs: good cluster or bad witness} establishes that $C$ is a type-2 good cluster, and produces a good witness $(\hat E^*,X,\pset)$ for $C$. Recall that $\hat E^*$ is a set of edges of $C$, whose cardinality is denoted by $k$, and we are guaranteed that $k\geq \Omega(W_j(C)/\hat W^{3\epsx})$. Additionally, $X$ is a $\phi^*$-expander, defined over the edges of $\hat E^*$, with maximum vertex degree at most $O(\log \hat W)$. Lastly, $\pset$ is an embedding of $X$ into $C_{|\hat E^*}$ with congestion at most $\eta$, such that the length of every path in $\pset$ is at most $O(\hat D\log^2\hat W)$.

If Case 3 happens, then we say that the current phase is a type-2 good phase. In this case, the phase continues until the total cost of update operations in the input sequence $\Sigma$ performed over the course of the phase reaches 
$\Theta \left (\frac{\phi^*k}{\eta \log \hat W}\right )=\Theta \left ( \frac{k}{\hat D\cdot \hat W^{O(\epsx)\cdot }(\log \hat W)^{O(1/\epsx)}  }\right ) = \Theta \left ( \frac{k}{D\cdot \hat W^{O(\epsx)\cdot }(\log \hat W)^{O(1/\epsx)}  }\right ) $.
 Since $k\geq \Omega(W_j(C)/\hat W^{3\epsx})$, the total cost of update operations over the course of a type-2 good phase is then at least $\Omega \left ( \frac{W_j(C)}{D\cdot \hat W^{O(\epsx)\cdot }(\log \hat W)^{O(1/\epsx)}  }\right ) $.

We need the following simple observation:

\begin{observation}\label{obs: number of good type-3 iterations}
	The total number of type-2 good phases over the course of the algorithm is at most $O\left( \mu\cdot  D\cdot \hat W^{O(\epsx)}\cdot (\log \hat W)^{O(1/\epsx)} \right )$.
\end{observation}

\begin{proof}
	Recall that a single type-2 good phase	is executed as long as the total cost of update operations performed over the course of the phase does not exceed $\Omega \left ( \frac{W(C)}{D\cdot \hat W^{O(\epsx)\cdot }(\log \hat W)^{O(1/\epsx)}  }\right )$, where $W(C)$ is the weight of all vertices of $C$ at the beginning of the phase. Recall also that, from \Cref{obs: total cost of all updates}, the total cost of all update operations performed on the cluster $C$ over the course of the algorithm is at most $O(W^0(C)\mu)$. 
	
	Consider some integer $i$, and let $t_i$ be the first time during the algorithm's execution, when $W(C)$ fell below $2^i$. Then the total cost of all operations performed on the cluster $C$ from time $t_i$ onward is at most $O(2^i\mu)$ (this is because we can consider the cluster $C$ at time $t$ as a ``fresh'' cluster that undergoes the same sequence of update operations, and apply \Cref{obs: total cost of all updates} to that cluster). Let $t_{i+1}$ be the first time when $W(C)$ falls below $2^{i+1}$, and let $T$ be the time period between $t_i$ and $t_{i+1}$. Let $W'\leq 2^i$ be the total weight of all vertices of $C$ at time $t_i$. Then the total number of type-2 good phases during the time period $T$ is bounded by: 
	
	\[O\left(\frac{W'\mu}{W'/( D\cdot \hat W^{O(\epsx)\cdot }(\log \hat W)^{O(1/\epsx)} ) }\right )\leq 
	O\left( \mu\cdot  D\cdot \hat W^{O(\epsx)\cdot }(\log \hat W)^{O(1/\epsx)} \right ).\]
	
	Since we only need to consider integers $i$ between $\log W^0(C)\leq \log \hat W$ and $1$, we get that the total number of type-2 good phases  over the course of the algorithm is at most $O\left( \mu\cdot  D\cdot \hat W^{O(\epsx)}\cdot (\log \hat W)^{O(1/\epsx)} \right )$.
\end{proof}

The algorithm for the type-2 good phase proceeds as follows. First, we use Algorithm \algmaintainexp from \Cref{thm: maintain expander}, on $C$, in order to maintain a non-empty subgraph $X'\subseteq X$ over the course of the phase. Recall that the update time of the algorithm is $O\left (W_j(C)\cdot \mu\cdot \hat D^2\cdot\hat W^{O(\epsx)}\cdot (\log \hat W)^{O(1/\epsx^2)}\right )\leq O\left (N^0(C)\cdot \mu\cdot D^2\cdot\hat W^{O(\epsx)}\cdot (\log \hat W)^{O(1/\epsx^2)}\right )$, and the algorithm supports expander-short-path queries: given a pair of vertices $x,y\in V(X')$, return an $x$-$y$  path $P$ in graph $C_{|\hat E^*}$, of length at most $O\left(\hat D\cdot (\log \hat W)^{O(1/\epsx^2)}\right )\leq O\left(D\cdot (\log \hat W)^{O(1/\eps^2)}\right )$, with query time $O(|E(P)|)$.

Additionally, we maintain a graph $C'$, that is obtained from the graph $C_{|\hat E^*}$, by adding a source vertex $s$, that connects, with an edge of length $1$, to every vertex $t_e$ lying in the current expander $X'$. 
We initialize  a generalized \EST data structure $\tau$ from \Cref{thm: ES-tree} in graph $C'$, whose root is vertex $s$, with depth threshold $\hat D+1$. We will also maintain a list $L$ of all regular vertices of $C$ that do not lie in the tree $\tau$. 
Lastly, for every vertex $x\in V(C)\setminus L$, we maintain a pointer from $x$ to its location in the tree $\tau$.
As before, it is immediate to generalize the algorithm from \Cref{thm: ES-tree} to maintain these additional data structures. 
The total update time of the \EST data structure is at most $O\left (W_j(C)\cdot \mu\cdot D \cdot \poly\log \hat W\right )\leq O\left (N^0(C)\cdot \mu\cdot D\cdot  \poly\log \hat W\right )$. 

Whenever the list $L$ becomes non-empty, we raise the flag $F_C$, with the corresponding pair of vertices being $v,v'$, where $v'$ is any vertex in list $L$, and $v$ is any regular vertex that serves as an endpoint of an edge $e\in \hat E$ with $t_e\in V(X')$. Notice that, since $v'\in L$, we are guaranteed that $\dist_{C_{|\hat E^*}}(v',t_e)\geq \hat D$. Since the length of the edge $e$ is at most $D$, while $\hat D=2^{30}D\log^{10} \hat W$, we are guaranteed that $\dist_C(v,v')\geq 1024D\log^4\hat W$, as required.

Recall that a query $\spquery(a,b)$ may only be asked when flag $F_C$ is not raised, so $L=\emptyset$. Given such a query, where $a,b\in V(C)$ are regular vertices of $C$, we use the tree $\tau$ to compute two paths: path $P'$ connecting $a$ to some vertex $x'\in X'$, and path $P''$ connecting $b$ to some vertex $x''\in X'$; recall that the lengths of both paths are at most  $\hat D=2^{30}D\log^{10} \hat W$, and they can be computed in time $O(|E(P')|+O(|E(P'')|)$. Next, we run query expander-short-path-query between $x'$ and $x''$ in the data structure maintained by \algmaintainexp, to compute a path $Q$ in graph $C_{|\hat E^*}$, that connects $x'$ to $x''$, of length at most $O\left(D\cdot (\log \hat W)^{O(1/\eps^2)}\right )$, with query time $O(|E(Q)|)$. The final path $P$ connecting $a$ to $b$ in $C$ is obtained by concatenating $P',Q$ and $P''$, and suppressing the vertices of $\set{t_e\mid e\in \hat E^*}$ as needed. We are then guaranteed that the length of $P$ is at most $ O\left(D\cdot (\log \hat W)^{O(1/\eps^2)}\right )$, and, from the above discussion, query time is $O(|E(P)|)$.
 
Altogether, the total update time of the algorithm over the course of a single type-2 good phase is at most $O\left (N^0(C)\cdot \mu\cdot D^2\cdot\hat W^{O(\epsx)}\cdot (\log \hat W)^{O(1/\epsx^2)}\right )$. Since, from \Cref{obs: number of good type-3 iterations}, there are at most $O\left( \mu\cdot  D\cdot \hat W^{O(\epsx)}\cdot (\log \hat W)^{O(1/\epsx)} \right )$ type-2 good phases in the algorithm, the total update time spent on all type-2 good phases is bounded by:

$$O\left (N^0(C)\cdot \mu^2\cdot D^3\cdot\hat W^{O(\epsx)}\cdot (\log \hat W)^{O(1/\epsx^2)}\right ).$$

\paragraph{Final Running Time Analysis.}
To summarize, the algorithm may have at most one type-1 good phase, and its running time over the course of such a phase is at most 
$O (N^0(C)\cdot \mu \cdot D\cdot \hat W^{O(\eps)}\cdot \poly\log \hat W)$. The total update time spent on all type-2 good phases is bounded by: $$O\left (N^0(C)\cdot \mu^2\cdot D^3\cdot\hat W^{O(\epsx)}\cdot (\log \hat W)^{O(1/\epsx^2)}\right ),$$ 

and the number of type-2 good phases is at most  $O\left( \mu\cdot  D\cdot \hat W^{O(\epsx)}\cdot (\log \hat W)^{O(1/\epsx)} \right )$.
Additionally, there is at most one bad phase over the course of the algorithm, and, at the beginning of each phase, the algorithm spends
at most $O\left (N^0(C)\cdot \mu \cdot D^2\cdot \hat W^{O(\epsx)}\cdot  (\log \hat W)^{O(1/\epsx^2)} \right )$ in order to run the algorithm from  \Cref{obs: good cluster or bad witness}.

Therefore, the total running time of the algorithm is bounded by:

$$O\left (N^0(C)\cdot \mu^2\cdot D^3\cdot\hat W^{O(\epsx)}\cdot (\log \hat W)^{O(1/\epsx^2)}\right ).$$

\end{proof}

%% file: contracted-graph2.tex
\section{Fourth Main Tool: Contracted Graph}

\label{sec: conracted graph}

In this section we develop the fourth and the last main tool that we use in our final algorithm for the \maintaincluster problem -- contracted graph. Throughout this section, we assume that we are given an input to the \maintaincluster problem, that consists of 
a valid input structure $\iset=\left(C,\set{\ell(e)}_{e\in E(C)},D\right )$, where $C$ is a connected graph. 
Graph $C$ undergoes a sequence $\Sigma$ of valid update operations with a dynamic degree bound $\mu$. Recall that we are guaranteed that for every regular vertex $v\in V(C)$, the total number of edges incident to $v$ that are ever present in $C$ is bounded by $\mu$.
Additionally, we are given a precision parameter $\eps$, and a parameter $\hat W\geq N^0(C)\cdot \mu$, where $N^0(C)$ is the number of regular vertices of $C$ at the beginning of the algorithm. 
Throughout this section, we assume that $\hat W\geq 2^{\Omega(1/\epsx^2)}$, so $\hat W\geq \max\set{N^0(C)\mu, 2^{\Omega(1/\epsx^2)}}$ holds.
Recall that our goal is to  support queries $\spquery(C,v,v')$: given a pair $v,v'\in V(C)$ of regular vertices, return a path $P$ of length at most $\alpha D$, connecting them in $C$, in time $O(|E(P)|)$, where $\alpha$ is the approximation factor that the algorithm achieves. 
The algorithm may, at any time, raise a flag $F_C$, at which time it must  supply a pair  $\hat v,\hat v'$ of regular vertices of $C$, with $\dist_C(\hat v,\hat v')>1024D\log^4\hat W$.
Once flag $F_C$ is raised, the algorithm will obtain, as part of its input update sequence, a sequence $\Sigma'$ of update operations called flag-lowering sequence, at the end of which 
either $\hat v$ or $\hat v'$ are deleted from $C$.
We use parameters $\hat D=2^{30}D\log^{10} \hat W$, $\rho=\hat W^{\eps}$ from \Cref{sec: pivot decomposition}, and we define a new parameter $\lambda=\ceil{\log(2048D\log^4\hat W)}$.

We start with some intuition. Our algorithm for the \maintaincluster problem starts by running Algorithm \algmaintaingoodcluster from \Cref{thm: alg for maintain good cluster} on the input $\iset$, with the same parameters $\hat W$ and $\eps$. 
As long as the algorithm does not raise the flag $F^b_{C}$, it can support queries $\spquery(C,v,v')$, and it provides all required guarantees for the \maintaincluster problem, with approximation factor  $(\log \hat W)^{O(1/\eps^2)}$. Once the algorithm raises the flag $F^b_{C}$, it terminates, and we need to continue maintaining the cluster $C$ by other means. Recall that, when the flag $F^b_{C}$ is raised, the algorithm produces a bad witness $\hat E$ for $C$, where $\hat E$ is a  $(\hat D,\rho)$-pseudocut of cardinality at most $N^0(C)/\hat W^{\epsx}$.
We denote by $S$ the set of all regular vertices that serve as endpoints of the edges in $\hat E$; clearly, $|S|\leq N^0(C)/\hat W^{\eps}$.
At this point, we can initialize an \EST $\tau$ in the graph that is obtained from $C$, by adding a source vertex $s$ to it, that connects to every vertex of $S$ with a length-$1$ edge, and has depth $\Theta(\hat D)$. This tree will allow us to ensure that, as long as some edge of $\hat E$ remains in $C$, every vertex of $C$ is sufficiently close to at least one vertex in $S$. Once all edges of $\hat E$ are deleted from $C$, we will start the whole algorithm from scratch. 
When the algorithm starts from scratch, we will ensure that the diameter of $C$ is below $\hat D$, by repeatedly raising the flag $F^C$ as needed. From the definition of a pseudocut, we are then guaranteed that the total number of regular vertices in $C$ has decreased by at least factor $\rho$, so we can afford this in terms of the running time. As long as $\hat E\neq \emptyset$, the \EST $\tau$ ensures that every vertex of $C$ is close to some vertex of $S$. But we also need an additional data structure to ensure that vertices of $S$ remain close to each other. The goal of this section is to design such a data structure.

We now provide an intuitive description of the data structure. Consider the graph $\tilde H=C\setminus \hat E$. Assume that, for all $1\leq i\leq \lambda$, we maintain a strong $(2^i,\alpha\cdot 2^i)$-neighborhood cover $\cset_i$ for graph $\tilde H$, for some approximation factor $\alpha\geq  2048 \log^4\hat W$. We will exploit the fact that $\hat E$ is a $(\hat D,\rho)$-pseudocut in order to ensure that, for all $1\leq i\leq \lambda$, for every cluster $C'\in \cset_i$, $N(C')\leq N^0(C)/\rho$ always holds (here, $N(C')$ is the number of regular vertices in $C'$). Additionally, we will maintain a \emph{contracted graph} $\hat H$. The vertex set $\hat H$ consists of a set $S$ of regular vertices -- the regular vertices that serve as endpoints of the edges in $\hat E$, and of a set $\set{u(C')\mid C'\in \bigcup_i\cset_i}$ of supernodes, representing the clusters in the neighborhood covers $\cset_1,\ldots,\cset_{\lambda}$. For every regular vertex $s\in S$ and super-node $u(C')$, we add an edge $e=(s,u(C'))$ to the graph if either $s\in V(C')$, or there is an edge $e'=(s,u')\in \hat E$, with $u'\in C'$. The length of the edge $e$ is $2^i+\ell(e')$, where $i$ is the index for which $C'\in \cset_i$. One can show that, in the contracted graph $\hat H$, all distances between the vertices of $S$ are approximately equal to those in the original cluster $C$. Moreover, since $|S|\leq N^0(C)/\hat W^{\eps}$, the graph is, in a sense, significantly ``smaller'' than $C$. We exploit this fact in order to solve the \maintaincluster problem recursively in graph $\hat H$. This will allow us to ensure that the vertices of $S$ remain close to each other in $C$, and, once this is no longer the case, to raise the flag $F_C$, and to provide a pair of violating vertices as required. As the neighborhood covers $\cset_1,\ldots,\cset_{\lambda}$ evolve, the resulting changes to their clusters will lead to changes to the contracted graph $\hat H$. However, all such changes can be implemented via valid update operations of graph $\hat H$. 

This leaves open the question of how to maintain the neighborhood covers $\cset_i$; equivalently, we need to solve the \recdynnc problem in graph $\tilde H=C\setminus \hat E$, for every distance threshold in $\set{2^i\mid 1\leq i\leq \lambda}$. As we have shown already, in order to obtain an algorithm for the \recdynnc problem, it is sufficient to obtain an algoritm for the \maintaincluster problem. Since all clusters in the neighborhood covers $\cset_i$ for which we will need to solve the \maintaincluster problem will have a significantly smaller number of regular vertices than $N^0(C)$, we can maintain the neighborhood covers $\cset_i$ by applying the algorithm for the \maintaincluster problem recursively to these much smaller clusters, that arise over the course of the algorithm for the \recdynnc problem.

This is precisely the approach that we use, except for a small caveat. Since we solve the \maintaincluster problem in the contracted graph $\hat H$ recursively, we need to ensure that it has a small dynamic degree bound. Consider some edge $e=(u,v)\in \hat E$, where $v$ is a regular vertex, and $u$ is a supernode. In the construction of graph $\hat H$ that we have described above, for every cluster $C\in \bigcup_i\cset_i$ that contains $u$, we will add the edge $(v,u(C))$ to $\hat H$. The analysis from \Cref{subsec: algorithmic framework for NC} only guarantees that every {\bf regular} vertex of $\tilde H$ belongs to a small number of clusters in $\bigcup_i\cset_i$, but we cannot guarantee it for supernodes (intuitively, this is because each supernode has weight $0$; as supernodes may be added to the graph via supernode splitting, letting the weight of a supernode be non-zero would result in total weight of vertices in $H$ increasing as well as decreasing, which would make the analysis much more difficult). We overcome this difficulty by defining the graph $\tilde H$ slightly differently: we start, as before, with $\tilde H=C\setminus \hat E$. But then we insert, for every edge $e=(u,v)\in \hat E$, where $v$ is a regular vertex, a new \emph{fake regular vertex} $v^F(e)$, and a new \emph{fake edge} $e^F=(u,v^F(e))$. We view vertex $v^F(e)$ as a \emph{fake copy} of the vertex $v$, and we view $e^F$ as a fake copy of the edge $e$. Notice that it is possible that, for a regular vertex $v\in S$, we have introduced several fake copies -- one copy for every edge of $\hat E$ incident to $v$. In the contracted graph $\hat H$, we connect a vertex $v\in S$ to a supernode $u(C')$ whenever cluster $C'\in \bigcup_i\cset_i$ contains either $v$ or any of its fake copies. Since we are guaranteed that the total number of edges incident to $v$ in graph $C$ always remains below $\mu$, the number of fake copies of $v$ in $\tilde H$ is also bounded by $\mu$. As the algorithm from  \Cref{subsec: algorithmic framework for NC} ensures that every regular vertex of $\tilde H$ belongs to a small number of clusters in each set $\cset_i$ for $1\leq i\leq \lambda$, the dynamic degree bound for graph $\hat H$ will also be small (but somewhat higher than $\mu$).

The remainder of this section is organized as follows. We define the modified graph $\tilde H$ in \Cref{subsec: modified}. We then define the corresponding contracted graph in \Cref{subsec: contracted}. Lastly, we prove that the distances  are approximately preserved in the contracted graph in \Cref{subsec: distance preservation}.

\input{modified-graph}
\input{contracted-graph-def}
\input{distance-preservation-contracted}

%% file: modified-graph.tex
\subsection{Modified Graph $\tilde H$}
\label{subsec: modified}

Throughout, we assume that we are given  
a valid input structure $\iset=\left(C,\set{\ell(e)}_{e\in E(C)},D\right )$, where $C$ is a connected graph, that undergoes an online sequence $\Sigma$ of valid update operations with dynamic degree bound $\mu$. We also assume that we are given parameters $0<\eps<1$ and $\hat W\geq N^0(C)\mu$ as before. 
We use the pareameters $\hat D=2^{30}D\log^{10} \hat W$, $\lambda=\ceil{\log(2048D\log^4\hat W)}$, and $\rho=\hat W^{\eps}$, that are defined exactly as before.

Lastly, we assume that, at the beginning of the algorithm, we are given a bad witness $\hat E$ for cluster $C$, so $\hat E$ is a $(\hat D,\beta)$-pseudocut, with $|\hat E|<N^0(C)/\hat W^{\epsx}$. 
We denote this initial edge set $\hat E$ by $\hat E^0$, and we let $S$ be the set of all regular vertices that serve as endpoints of edges in $\hat E^0$, so $|S|<N^0(C)/\hat W^{\epsx}$. As cluster $C$ undergoes a sequence of valid update operations, edge set $\hat E$ evolves as follows. First, if an edge $e\in \hat E$ is deleted from $C$ as part of the input update operation sequence, we delete $e$ from $\hat E$. Additionally, if a supernode splitting operation is performed on a supernode $u\in V(C)$, with the corresponding edge set $E'$, then for every edge $e=(u,v)\in \hat E\cap E'$, we add the new copy $e'=(u',v)$ of the edge $e$ that was created by the supernode splitting operation to set $\hat E$. 
Therefore, throughout the algorithm, an edge may leave $\hat E$ only if it is deleted from $C$, and an edge may join $\hat E$ only if it is a newly created copy of an edge that already lies in $\hat E$, where the copy is created via the supernode splitting operation. 
Set $S$ of vertices always contains all regular vertices that serve as endpoints of the edges in the current set $\hat E$. It is immediate to verify that vertices may leave $S$, but they may never join $S$.

We now define the modified graph $\tilde H$, which is a dynamic graph. 
Initially, before any updates are applied to graph $C$, we define the initial graph $\tilde H^0$ as follows. We start with $\tilde H^0=C\setminus \hat E^0$. Next, we process every edge $e=(v,u)\in \hat E^0$, where $v\in S$ is a regular vertex, one-by-one. When edge $e=(v,u)$ is processed, we insert a new regular vertex $v^F(e)$ into $\tilde H^0$, that we call a \emph{fake regular vertex}, and a \emph{fake copy of vertex $v$}. Additionally, we insert a new edge $e^F=(v^F(e),u)$ into $\tilde H^0$, whose length is $\ell(e)$. We say that $v^F(e)$ is a \emph{fake edge}, and that it is a \emph{fake replacement edge} for edge $e$. This completes the definition of the initial modified graph $\tilde H^0$. Observe that, for every edge $e\in \hat E$, there is a unique fake replacement edge for $e$, and for every regular vertex $s\in S$, we may have created a number of fake copies of $s$. In the remainder of the algorithm, we never create any new fake vertices, but we may create new fake edges, as new edges join set $\hat E$. We will ensure that the following invariant holds throughout the algorithm:

\begin{properties}{I}
\item  For every edge $e=(u,v)\in \hat E$, there is a unique fake replacement edge $e^F$ in $\tilde H$, whose one endpoint is $u$ and the other endpoint is a fake copy of $v$; edge $e$ itself does not lie in $\tilde H$. \label{inv: fake edges}
\end{properties}

For each vertex $s\in S$, let $R(s)$ be the set of all fake copies of vertex $s$. Since no new fake vertices are ever created after the initialization step, vertices may leave set $R(s)$, but they may never join it. Since the degree of every vertex in $C$ is at most $\mu$, $|R(s)|\leq \mu$ must hold for every vertex $s\in S$.

Next, we describe the updates to the graph $\tilde H$, as cluster $C$ undergoes the sequence $\Sigma$ of update operations. We denote by $\tilde H^t$ the graph $\tilde H$ just before the $(t+1)$th update operation $\sigma_{t+1}\in \Sigma$ is executed in graph $C$. The updated graph $\tilde H^{t+1}$ is defined as follows. If $\sigma_{t+1}$ is the deletion of an isolated vertex $x$, then we also delete $x$ from $\tilde H^t$ (where it must be an isolated vertex as well). If $\sigma_{t+1}$ is the deletion of an edge $e$, then we proceed as follows. First, if $e\not\in \hat E$, then $e\in E(\tilde H)$ must hold. We delete $e$ from $\tilde H^{t}$ as well. If $e\in \hat E$, then we let $e^F$ be the unique fake replacement edge of $e$. We delete $e^F$ from $\hat E$ and from $\tilde H^t$. If the endpoint $v^F$ of $e^F$ that is a fake regular vertex becomes an isolated vertex in $\tilde H$ as the result of this edge deletion, then we delete $v^F$ from $\tilde H^t$ as well. If the endpoint $v\in S$ of $e$ that is a regular vertex has no other edges of $\hat E$ that are incident to it, we delete $v$ from $S$. Note that we have defined a sequence of valid update operations in graph $\tilde H$ corresponding to the update $\sigma_{t+1}$, whose length is $O(1)$ (in addition to possible updates to edge set $\hat E$ and vertex set $S$).

Lastly, assume that $\sigma_{t+1}$ is a supernode-splitting operation, for some supernode $u$, with edge set $E'$. In this case, we define a new edge set $E''$, as follows. We process each edge $e\in E'$ one-by-one. If $e\not\in \hat E$, then we add $e$ to $E''$. Otherwise, we add to $E''$ the unique fake replacement edge $e^F$ of $e$. We then perform the supernode-splitting operation in graph $\tilde H^t$, with supernode $u$ and edge set $E''$. We also update edge set $\hat E$, as described above. 
Consider an edge $e=(u,v)\in E'\cap \hat E$. Recall that, following the supernode splitting operation in graph $C$, we have created a copy $e'=(u',v)$ of $e$ in graph $C$, and edge $e'$ was added to the set $\hat E$. Let $e^F=(u,v^F)$ be the fake replacement edge of edge $e$ in $\tilde H$. Then the supernode splitting operation of $u$ in $\tilde H$ has created a new edge $(e')^F=(u',v^F)$, corresponding to the edge $e^F$. We view edge $(e')^F$ as the fake replacement edge of the edge $e'$.
This completes the definition of the graph $\tilde H^{t+1}$. Notice that, if Invariant \ref{inv: fake edges} held for $\tilde H^t$, then it also holds for $\tilde H^{t+1}$.
The following claim is now immediate.

\begin{claim}\label{claim: get update sequence for transformed graph} There is a deterministic algorithm that, given an online valid update sequence $\Sigma=(\sigma_1,\sigma_2,\ldots)$ for $C$ produces, at each time $t>0$, a sequence $\tilde \Sigma_t$ of valid update operations for graph $\tilde H$, such that $\tilde \Sigma_t$ contains a constant number of valid update operations, and the length of the description of each operation, as well as the time required to compute $\tilde \Sigma_t$, is asymptotically bounded by the length of the description of $\sigma_t$.
Moreover, for all $t\geq 0$, the graph that is obtained from $\tilde H^0$ by applying the sequence $\tilde \Sigma_1\circ\tilde \Sigma_2\circ\cdots\circ\tilde \Sigma^t$ to it is precisely $\tilde H^{t+1}$. Additionally, if the dynamic graph $C$ has dynamic degree bound $\mu$, then the dynamic graph $\tilde H$ has dynamic degree bound $\mu$ as well. The algorithm also maintains the dynamic set $\hat E$ of edges, and the dynamic set $S$ of vertices.
\end{claim}

.

%% file: contracted-graph-def.tex
\subsection{Contracted Graph $\hat H$}
\label{subsec: contracted}

A contracted graph $\hat H$ is a dynamic graph that is associated with the modified graph $\tilde H$, and the set $\hat E$ of edges of $C$.  
We consider some time point $t\geq 0$, and the corresponding modified graph $\tilde H=\tilde H^t$, and edge set $\hat E=\hat E^t$. We assume that we are given, for all $1\leq i\leq \lambda$, a collection $\cset_i$ of subgraphs of $\tilde H$, such that $\cset_i$ is a strong $(2^i,\alpha\cdot 2^i)$ neighborhood cover for the set of regular vertices in graph $\tilde H$, where $\alpha\geq 2048\log^4\hat W$ is some given approximation factor.
We denote $\cset=\bigcup_{1\leq i\leq \lambda}\cset_i$.

The contracted graph $\hat H^t$ at time $t$ is defined as follows. Its vertex set consists of two subsets: set $S$ of regular vertices (as before, this set contains all regular vertices of $C$ that serve as endpoints of edges in $\hat E^t$), and set $\set{u(C')\mid C'\in \cset}$ of supernodes. The edge set is defined as follows. For every regular vertex $s\in S$, and supernode $u(C')$, we add an edge $(s,u(C'))$ iff cluster $C'$ contains at least one vertex in $R(s)\cup \set{s}$ (recall that $R(s)$ is the set of all fake copies of $s$). The length of the edge is $2^i$, where $i$ is the index for which $C'\in \cset_i$. 

Intuitively, we will maintain the neighborhood covers $\cset_i$, for $1\leq i\leq \lambda$ using the algorithmic framework from \Cref{subsec: algorithmic framework for NC}. Once the collection $\cset_i$ of clusters is initialized using algorithm \initnc, it will only undergo allowed updates, which include: \delvertex (delete a vertex from a cluster in $\cset_i$); \addsupernode (add a supernode to a cluster in $\cset_i$ to reflect a supernode-splitting operation); and \csplit (create a new cluster, that is a sub-graph of an existing cluster in $\cset_i$). Once graph $\hat H$ is initialized, all subsequent changes to $\hat H$, due to the changes in the collections $\set{\cset_i}_{i=1}^{\lambda}$ of clusters can be realized via standard update operations. Therefore, we can view the graph $\hat H$ as an instance of the \maintaincluster problem that undergoes a sequence of valid update operations.

The framework from \Cref{subsec: algorithmic framework for NC} guarantees (see \Cref{thm: bound number of copies}), that for all $1\leq i\leq \lambda$, every regular vertex of $\tilde H$ lies in at most $W^{O(1/\log\log W)}$ clusters in $\cset_i$ over the course of the algorithm, where $W$ is the total number of regular vertices in $\tilde H^0$; it is easy to verify that $W\leq N^0(C)+|\hat E|\leq 2N^0(C)$. Moreover, as observed already, for every vertex $s\in S$, $|R(s)|\leq \mu$. Therefore, for every vertex $s\in S$, there may be at most $\mu\cdot W^{O(1/\log\log W)}$ clusters in $\bigcup_i\cset_i$ that ever contained a vertex of $R(s)\cup \set{s}$, and so the total number of edges that are ever incident to $s$ in $\hat H$ is bounded by  $\mu\cdot W^{O(1/\log\log W)}$. We can then use $\mu'=\mu\cdot W^{O(1/\log\log W)}\leq \mu\cdot (N^0(C))^{O(1/\log\log N^0(C))}$ as the dynamic degree bound for the graph $\hat H$. Lastly, observe that $N^0(\hat H)\leq |\hat E^0|\leq N^0(C)/\hat W^{\eps}$. We now turn to prove that the distances between the vertices of $S$ are approximately preserved in $\hat H$.

%% file: distance-preservation-contracted.tex
\subsection{Distance Preservation in the Contracted Graph}
\label{subsec: distance preservation}

\begin{claim}\label{claim: distances preserved in contracted}
	At every time $t\geq 0$, for every pair $s,s'\in S$ of vertices, if $\dist_{C}(s,s')\leq 2048 D\log^4\hat W$, then $\dist_{\hat H}(s,s')\leq 4\dist_C(s,s')$.
\end{claim}

\begin{proof}
	Consider any time point $t\geq 0$, and let $\hat H$ be the contracted graph at time $t$.
	Let $s,s'\in S$ be a pair of vertices with $\dist_{C}(s,s')\leq 2048 D\log^4\hat W$,
	 Let $P$ be the shortest $s$-$s'$ path in graph $C$, and let $\ell^*\leq 2048 D\log^4\hat W$ be its length. Let $e_1,\ldots,e_q$ be the edges of $\hat E$ lying on $P$, indexed in the order of their appearance. We denote $s_0=s$, $s_{q+1}=s'$, and, for $1\leq j\leq q$, we let $s_j\in S$ be the endpoint of $e_j$ that is a regular vertex. Let $P_1,\ldots,P_{q+1}$ be the sub-paths of $P$ obtained by deleting the edges $e_1,\ldots,e_q$ from $P$, indexed in the order of their appearance on $P$. Notice that a path $P_j$ may consist of a single vertex $\hat s\in S$. Consider now the graph obtained from $P$ by first deleting the edges $e_1,\ldots,e_q$ from it, and then adding, for all $1\leq j\leq q$, the unique fake replacement edge $e_j^F$ of edge $e_j$ to the graph. Then the resulting graph is a subgraph of $\tilde H$, and it is a collection of $q+1$ paths $P_1',\ldots,P_{q+1}'$, where for all $1\leq j\leq q$, path $P_j'$ is obtained from $P_j$ by possibly appending one fake edge at the beginning, and possibly appending one fake edge at the end of it. 
	 (We note that it is possible that $s_j$ is an endpoint of both $e_{j}$ and $e_{j+1}$, and that the fake copies $e^F_{j},e^F_{j+1}$ share an endpoint $s_j^F$; in this case, we view $P'_{j-1}$ as terminating at vertex $s_j^F$, path $P'_j$ as consisting of only the vertex $s_j$, and path $P'_{j+1}$ as originating at $s_j^F$.)
	 Notice that for all $1\leq j\leq q$, either the last endpoint of $P_j'$ is $s_j$ and the first endpoint of $P_{j+1}'$ is a fake copy of $s_j$, or vice versa. Notice also that $P_1',\ldots,P_{q+1}'$ are paths in graph $\tilde H$, and that $\sum_j\ell_{\tilde H}(P_j')=\ell^*$. For all $1\leq j\leq q+1$, we denote $\ell_j=\ell_{\tilde H}(P'_j)$.
	
	Consider now some index $1\leq j\leq q+1$. Let $i_j$ be the smallest integer, so that $\ell_j\leq 2^{i_j}$. Since $\ell_j\leq \ell^*\leq 2048D\log^4\hat W$, and $\lambda=\ceil{\log(2048D\log^4\hat W)}$, $i_j\leq \lambda$ holds. Consider now cluster set $\cset_{i_j}$. Since this cluster set is a strong $(2^{i_j},\alpha\cdot 2^{i_j})$-neighborhood cover for graph $\tilde H$, there is some cluster $C_j\in \cset_{i_j}$ that contains the two endpoints of path $P_j$. Therefore, there is an edge $(s_{j-1},u(C_j))$ and an edge $(u(C_j),s_j)$ in graph $\hat H$. We let $\hat P_j$ denote the path in graph $\hat H$, that is the union of these two edges. As the length of each of the two edges is $2^{i_j}$, the length of path $\hat P_j$ is $2\cdot 2^{i_j}\leq 4\ell_j$.
	By concatenating the paths $\hat P_1,\ldots,\hat P_q$, we obtain a path in graph $\hat H$ connecting $s$ to $s'$, of length at most $4\sum_{j=1}^{q+1}\ell_j\leq 4\dist_C(s,s')$.
\end{proof}

It is easy to verify that the opposite direction is also true: $\dist_{\hat H}(s,s')\leq O(\alpha\cdot \dist_C(s,s'))$ for any pair $s,s'$ of vertices in $S$. We do not use this claim directly, so we do not prove it formally.

%% file: final-proof.tex
\section{Putting Everything Together: Proof of \Cref{thm: main maintain cluster algorithm}}
\label{sec: final proof}
In this section we combine all tools that we have developed so far, in order to obtain a proof of \Cref{thm: main maintain cluster algorithm}. In order to do so, we prove the following theorem by induction:

\begin{theorem}\label{thm: induction for maintain cluster algorithm}
	There is a universal constant $c>1$, and deterministic algorithm for the \maintaincluster problem, that,  given a valid input structure $\iset=\left(C,\set{\ell(e)}_{e\in E(C)},D \right )$, where $C$ is a connected graph, undergoing a sequence of valid update operations with dynamic degree bound $\mu\geq 1$, and parameters $c/\log \log ((N^0(C))^{2/z})<\eps<1$, and $\hat W\geq N^0(C)\mu$, where $N^0(C)$ is the number of regular vertices in $C$ at the beginning of the algorithm, such that $N^0(C)\leq \hat W^{\eps z/2}$, for some integer $1\leq z\leq \ceil{2/\eps}$, achieves approximation factor $\alpha_z=(\log \hat W)^{2c\cdot 2^{2z}/\eps^2}$ and  has total update time at most: 
	$$4^{cz}\cdot N^0(C)\cdot \mu^2\cdot D^3\cdot\hat W^{c\eps}\cdot  (\log \hat W)^{c/\epsx^2+cz}$$. 
\end{theorem}

\Cref{thm: main maintain cluster algorithm} follows immediately from \Cref{thm: induction for maintain cluster algorithm}, by setting $z=\ceil{2/\eps}$. The resulting algorithm achieves approximation factor at most $(\log \hat W)^{2c\cdot 2^{6/\eps}/\eps^2}=(\log \hat W)^{2^{O(1/\eps)}}$, and it has total update time at most:
$N^0(C)\cdot \mu^2\cdot D^3\cdot\hat W^{O(\eps)}\cdot  (\log \hat W)^{O(1/\eps^2)}$.

The remainder of this section is dedicated to the proof of \Cref{thm: induction for maintain cluster algorithm}.
Throughout the proof, we assume that $c$ is a large enough constant.

The proof is by induction on $z$. The base case is when $z\leq 10$, and so $N^0(C)\leq \hat W^{5\eps}$. In this case, we use Algorithm \algslow from \Cref{thm: slow simple}, that achieves approximation factor $O((\log \hat W)^4)\leq (\log \hat W)^{c}$ (assuming that $c$ is a large enough constant), and has total update time $O((N^0(C))^2\mu D\poly\log \hat W)\leq O(N^0(C)\mu D\hat W^{5\eps}\poly\log \hat W)$.

From now on we assume that we are given some integer $z>10$, and that the theorem holds for all integers smaller than $z$. 
From \Cref{obs: easy param settings for maintaincluster},  if $D> \hat W$, or $\hat W<2^{O(1/\eps^2)}$, then  the running time of \algslow is at most: 

$$O(N^0(C)\mu D^2\cdot 2^{O(1/\eps^2)}\poly\log \hat W)\leq O\left (N^0(C)\mu D^2(\log \hat W)^{O(1/\eps^2)}\right )<4^c\cdot N^0(C)\cdot \mu^2\cdot D^3\cdot\hat W^{c\eps}\cdot  (\log \hat W)^{c/\epsx^2},$$

 and it achieves approximation factor $O((\log \hat W)^4)\leq (\log \hat W)^{c}$ .  
Therefore, we assume from now on that $D<\hat W$, and that $\hat W\geq 2^{\Omega(1/\eps^2)}$. In particular, $\hat W\geq \max\set{N^0(C)\mu,|E(C)|, 2^{\Omega(1/\epsx)}}$ holds throughout the algorithm.
The algorithm consists of four phases, that we now describe.


\subsection{Phase 1}
In this phase, 	we run the algorithm \algmaintaingoodcluster from \Cref{thm: alg for maintain good cluster} for the \maintaingoodcluster problem on $\iset,\Sigma$, with the same parameters $\eps$, $\hat W$. Recall that, as long as the flag $F^b_C$ is not raised, the algorithm essentially solves the \maintaincluster problem on cluster $C$, with approximation factor  ${(\log \hat W)^{O(1/\eps^2)}}\leq \alpha_z$. The running time of this algorithm is bounded by:

 $${O\left (N^0(C)\cdot \mu^2\cdot D^3\cdot\hat W^{O(\epsx)}\cdot (\log \hat W)^{O(1/\epsx^2)}\right )}
\leq \frac{c}{4}\cdot N^0(C)\cdot \mu^2\cdot D^3\cdot\hat W^{c\eps}\cdot  (\log \hat W)^{c/\epsx^2+cz}.$$ Once the algorithm raises flag $F^b_C$, the first phase terminates.
Recall that, once flag $F^b_C$ is raised Algorithm \maintaingoodcluster supplies a bad witness $\hat E$ for $C$, where $\hat E$ is $(\hat D,\rho)$-pseudocut $\hat E$, with $|\hat E|<N^0(C)/\hat W^{\epsx}$; as before, $\hat D=2^{30}D\log^{10} \hat W$, and $\rho=\hat W^{\eps}$.

\subsection{Phase 2}

Our algorithm for Phase 2 follows the outline and the tools developed in \Cref{sec: conracted graph}, and consists of four parts (or four different algorithms) that we run in parallel. The first algorithm, that we refer to as $\alg_1$, constructs and maintains the modified graph $\tilde H$, and maintains the edge set $\hat E$, using the cluster $C$, the initial bad witness $\hat E^0$ for $C$, and the update sequence $\Sigma$ for $C$. This algorithm starts by constructing an initial graph $\tilde H$, and then produces an online update sequence $\tilde \Sigma$ for it, based on the update sequence $\Sigma$ for $C$. The second algorithm, that we call $\alg_2$, maintains the neighborhood covers $\cset_1,\ldots,\cset_{\lambda}$ for graph $\tilde H$, and it also maintains the contracted graph $\hat H$. This algorithm, after initializing the contracted graph $\hat H$, uses the update sequence $\tilde \Sigma$ for $\tilde H$ and the neighborhood covers $\cset_1,\ldots,\cset_{\lambda}$ of graph $\tilde H$, in order to produce an online update sequence $\hat \Sigma$ for graph $\hat H$, that allows us to maintain the graph $\hat H$. The third algorithm, that we refer to as $\alg_3$, simply runs the algorithm for the $\maintaincluster$ problem from the induction hypothesis on graph $\hat H$. Lastly, the fourth algorithm, called $\alg_4$, maintains an \EST in graph $C$, that is rooted at the vertices of $S$. We also provide an algorithm for supporting queries $\spquery(C,v,v')$ in graph $C$, using the data structures that are maintained these algorithms. We now describe each of these algorithms in turn.

\subsubsection{Maintaining the Modified Graph -- Algorithm $\alg_1$}

We initialize the set $S$ of regular vertices of $C$ to contain all regular vertices that are endpoints of the edges of $\hat E^0$. We also initialize the modified graph $\tilde H$. 
We then use the algorithm from \Cref{claim: get update sequence for transformed graph}, that, given the online update sequence $\Sigma=(\sigma_1,\sigma_2,\ldots)$ for $C$, produces, at each time $t>0$, a sequence $\tilde \Sigma_t$ of valid update operations for graph $\tilde H$, such that $\tilde \Sigma_t$ contains a constant number of valid update operations, and the length of the description of each operation, as well as the time required to compute $\tilde \Sigma_t$, is asymptotically bounded by the length of the description of $\sigma_t$. As shown in \Cref{claim: get update sequence for transformed graph}, this algorithm ensures that the graph obtained at time $t$ is precisely $\tilde H^t$. The algorithm also maintains the dynamic set $\hat E$ of edges and the set $S$ of vertices. Lastly, recall that the dynamic degree bound for the dynamic graph $\tilde H$ is $\mu$, and that the initial number of regular vertices in $\tilde H$, $N^0(H)\leq 2N^0(C)$.
We denote this algorithm by $\alg_1$. The running time of the algorithm is bounded asymptotically by the total length of the description of the update sequence $\Sigma$, which is in turn bounded $O(|E^*|)$, where $E^*$ is the total number of edges that are ever present in graph $C$, so in particular, $|E^*|\leq O(N^0(C)\mu)$. Therefore, the total update time of algorithm $\alg_1$ is $O(N^0(C)\mu)$. 

\subsubsection{Maintaining the Neighborhood Covers of $\tilde H$ and the Contracted Graph -- Algorithm $\alg_2$}

As before, we let $\lambda=\ceil{\log(2048D\log^4\hat W)}$. We now fix an index $1\leq i\leq \lambda$, and describe an algorithm for maintaining a strong $(2^i,\alpha_{z-1}\cdot 2^i)$-neighborhood cover $\cset_i$ of the regular vertices in the modified graph $\tilde H$, as it undergoes a sequence $\tilde \Sigma$ of valid update operations. The algorithm follows the framework from \Cref{subsec: algorithmic framework for NC}.  
For every vertex $x\in V(\tilde H)$, we maintain a list  $\clusterlist_i(x)$ of all clusters $C'\in \cset_i$ containing $x$, and similarly, for every edge $e\in E(\tilde H)$, we maintain a list $\clusterlist_i(e)$ of all clusters $C'\in \cset_i$ containing $e$.
For every cluster $C'\in \cset_i$, we denote by $N^0(C')$ the number of regular vertices that $C'$ contained when it was first added to $\cset_i$. 
We now describe the algorithm for maintaining the neighborhood cover $\cset_i$ in more detail.

\paragraph{Algorithm for Maintaining $\cset_i$.}
For convenience, we denote $W=N^0(\tilde H)$; recall that $W\leq 2N^0(C)$. We also denote $D_i=2^i$.
At the beginning of the algorithm, before any update operations from $\tilde \Sigma$ are processed, we apply Procedure \initnc from \Cref{subsec: initial NC} to graph $\tilde H$, with distance parameter $2^i$, and we add to $\cset_i$ the resulting collection of clusters, that, from \Cref{lem: initialize neighborhood cover}, form a strong $(D_i,1024D_i\log^4W)$-neighborhood cover of the set of regular vertices of $\tilde H$. The running time of the algorithm, from  \Cref{lem: initialize neighborhood cover}, is $O(N^0(\tilde H)\cdot D_i\cdot \mu\cdot \poly\log(N^0(\tilde H)\mu)$.
We need the following observation regarding this initial set $\cset_i$ of clusters.

\begin{observation}\label{obs: smaller weight}
	For every cluster $C'\in \cset_i$, $N^0(C')\leq 2N^0(C)/\hat W^{\eps}\leq \hat W^{\eps(z-1)/2}$.
\end{observation}
\begin{proof}
	Let $C'\in \cset_i$ be any cluster. Recall that, for every fake edge $e^F=(u,v^F)$, an endpoint $v^F$ of $e^F$, that is a fake regular vertex, has degree $1$ in the initial graph $\tilde H^0$.  Moreover, the total number of fake regular vertices in $\tilde H^0$ is at most $|\hat E|\leq N^0(C)/\hat W^{\epsx}$. Consider the graph $C''$, that is obtained from $C'$ after all fake edges are deleted from it. Then $C''$ is a subgraph of $C\setminus \hat E$, and it has diameter at most $2^{i+10}\log^4W\leq 2^{11}\cdot 2^\lambda\log^4\hat W<\hat D$, since $\lambda=\ceil{\log(2048D\log^4\hat W)}$, and $\hat D=2^{30}D\log^{10} \hat W$. 
	
	Since $\hat E$ is a $(\hat D,\rho)$-pseudocut in $C$, the number of regular vertices in $C''$ is at most $N^0(C)/\rho=N^0(C)/\hat W^{\eps}$. Therefore, altogether, $N^0(C')\leq 2N^0(C)/\hat W^{\eps}$. Since we have assumed that $N^0(C)\leq \hat W^{\eps z/2}$, we get that $N^0(C')\leq \hat W^{\eps (z-1)/2}$.
\end{proof}

As the algorithm progresses, new clusters may be added to $\cset_i$, but each such new cluster $C''$ will always be a subgraph of a cluster that is currently in $\cset_i$. Since no new regular vertices are ever added to $\tilde H$, this will guarantee that for every cluster $C'$ that is ever added to $\cset_i$, $N^0(C')\leq \hat W^{\eps (z-1)/2}$ always holds. Whenever a cluster $C'$ is added to $\cset_i$, we initialize the algorithm for the \maintaincluster problem from the induction hypothesis on it, using the same parameters $\eps$ and $\hat W$ as before (but note that the distance parameter, $2^i$, may be as large as $2^{\lambda}>D$). The algorithm then achieves approximation factor $\alpha_{z-1}$, and running time 
at most:

$$4^{c(z-1)}\cdot N^0(C')\cdot \mu^2\cdot D_i^3\cdot\hat W^{c\eps}\cdot  (\log \hat W)^{c/\epsx^2+c(z-1)}.$$


As graph $\tilde H$ undergoes the sequence $\tilde \Sigma$ of valid update operations, we update every cluster $C'\in \cset_i$ accordingly, in the same way as in the algorithm from \Cref{subsec: algorithmic framework for NC}. Consider any update operation $\tilde \sigma_t\in \Sigma$. If $\tilde \sigma_t$ is a deletion of an edge $e$, then for every cluster $C'\in \clusterlist_i(e)$, we delete $e$ from $C'$ as well. If $\tilde \sigma_t$ is a deletion of an isolated vertex $x$, then for every cluster $C'\in \clusterlist_i(x)$, we delete $x$ from $C$ as well. 
 Assume now that $\tilde \sigma_t$ is a supernode splitting operation, applied to a supernode $u$, and a set $E'\subseteq \delta_{\tilde H}(u)$ of its incident edges. For every cluster $C'\in \clusterlist(u)$, we let $E'_{C'}=E'\cap E(C')$. We then update the cluster $C'$ by performing a supernode-splitting operation in it, for vertex $u$, with edge set $E'_{C'}$. 
We then add $C'$ to $\clusterlist(u')$ of the newly created supernode $u'$, and also initialize $\clusterlist(e)$ for every edge $e$ that was just added to $\tilde H$.

When an algorithm $\maintaincluster(C')$ raises the flag $F_C$, for any cluster $C'\in \cset_i$, it needs to provide a pair  $v,v'\in V(C')$ of regular vertices of $C'$ with $\dist_{C'}(v,v')>1024\cdot D_i\log^4\hat W\geq 1024 \cdot D_i\cdot \log^4(2W)$. 
We then run Procedure $\proccut'(C',v,v',D_i)$, obtaining a cluster $C''\subseteq C'$, and another cluster $C'''$, which is given implicitly, by listing all edges and vertices of $C'\setminus C'''$. 
We add cluster $C''$ to $\cset_i$, initializing the $\maintaincluster(C'')$ data structure on it,  and update cluster $C'$, by first deleting all edges of $C'\setminus C'''$ from it, and then deleting resulting isolated vertices of $C'\setminus C'''$, so that at the end of this update, $C'=C'''$ holds. The procedure guarantees that either $v$ or $v'$ is deleted from $C'$ at the end of this update sequence.
We also update the data structures $\clusterlist(x),\clusterlist(e)$ for all vertices $x\in C''$ and $x\in C'\setminus C'''$, and for all edges $e\in E(C'')$ and $e\in E(C'\setminus C''')$, exactly like in the algorithm from \Cref{subsec: algorithmic framework for NC}. It is easy to verify that, throughout the algorithm, $\cset_i$ is a strong $(2^i,\alpha_{z-1}2^i)$ neighborhood cover for the set of regular vertices in $\tilde H$. Indeed, the algorithm from \Cref{subsec: algorithmic framework for NC} (and, in particular, Procedure \proccut) guarantee that for every regular vertex $v\in V(\tilde H)$, there is some cluster $C'\in \cset_i$ with $B(v,2^i)\subseteq V(C')$. This property holds immediately after the application of Procedure \initnc. 
Since valid update operations cannot decrease distances between regular vertices  (from \Cref{obs: no dist increase}), this property continues to hold after each update operation (if it held before the update operation).
From the analysis of Procedure \proccut, this property also continues to hold after each application of the procedure (if it held before the application of the procedure). Lastly, the fact that the algorithm for the \maintaincluster problem on each cluster $C'\in \cset_i$ supports queries $\spquery(C',v,v')$ with approximation factor $\alpha_{z-1}$ shows that $\cset_i$ is indeed a strong $(2^i,\alpha_{z-1}\cdot 2^i)$ neighborhood cover for the set of regular vertices in $\tilde H$.

We denote the algorithm that we have described above, that receives as input the modified graph $\tilde H$ and the online update sequence $\tilde \Sigma$ for it, and maintains the neighborhood covers  $\cset_1,\ldots,\cset_{\lambda}$ of $\tilde H$, by $\alg_2$.
If we exclude the update time for the $\maintaincluster$ algorithm on clusters $C'\in \cset_i$, then, from \Cref{cor: running time of algmaintainNC},  the total running time of $\alg_2$ is $O(W D_i \mu  \poly\log(W\mu))\leq O(N^0(C)D_i\mu \poly\log \hat W)$. Additionally, from \Cref{claim: bound total weight}, 	$\sum_{C'\in \cset_i}N^0(C')\leq O(W\log W)\leq O(N^0(C)\log \hat W)$. Therefore, the total running time of all algorithms $\maintaincluster(C')$ for all clusters $C'$ that ever belonged to $\cset_i$ is bounded by:

$$4^{c(z-1)}\cdot N^0(C)\cdot \mu^2\cdot D_i^3\cdot\hat W^{c\eps}\cdot  (\log \hat W)^{c/\epsx^2+c(z-1)+1}.$$

Summing up over all $1\leq i\leq \lambda$, and recalling that $2^{\lambda}\leq 2^{12}D\log^4\hat W$, we get that the total update time of algorithm $\alg_2$ is at most:

$$4^{c(z-1)}\cdot N^0(C)\cdot \mu^2\cdot D^3\cdot\hat W^{c\eps}\cdot  (\log \hat W)^{c/\epsx^2+c(z-1)+14}\leq 4^{c(z-1)}\cdot N^0(C)\cdot \mu^2\cdot D^3\cdot\hat W^{c\eps}\cdot  (\log \hat W)^{c/\epsx^2+cz}.$$

\paragraph{Maintaining the Contracted Graph.}
In addition to the above data structures, we maintain the contracted graph $\hat H$, as follows. At the beginning of the algorithm, once, for all $1\leq i\leq \lambda$, cluster set $\cset_i$ is initialized using the $\initnc$ procedure, we construct an initial graph $\hat H$. Vertex set of $\hat H$ consists of the set $S$ or regular vertices -- all regular vertices that serve as endpoints of the edges in $\hat E^0$, and the set $\set{u(C')\mid C'\in \bigcup_i\cset_i}$ of supernodes. There is an edge $(s,u(C'))$ between a regular vertex $s\in S$ and a supernode $u(C')$, where $C'\in \cset_i$ for $1\leq i\leq \lambda$ iff $C'$ contains $s$ or one of its fake copies. The length of the edge is $2^i$. 

As the algorithm progresses, and clusters in set $\cset=\bigcup_i\cset_i$ undergo changes, we update the contracted graph $\hat H$ by producing a sequence $\hat \Sigma$ of valid update operations for it, as follows. Recall that all changes to clusters in $\cset$ are allowed changes, and fall into one of the following three categories:

\begin{itemize}
	\item $\delvertex(C',x)$: delete vertex $x$ from cluster $C'$. If $x\in S$, or $x$ is a fake copy of some vertex $s\in S$, we may need to delete the edge $(s,u(C'))$ from $\hat H$. We check whether $s$ or any of its fake copies still lie in $C'$, using the $\clusterlist(x)$ data structure for vertices $x\in \set{s}\cup R(s)$, and, if this is not the case, we delete the edge $(s,u(C'))$ from $\hat H$. Since there are at most $\mu$ fake copies of $s$, this check can be performed in $O(\mu)$ time.
	
	\item $\addsupernode$: add a new supernode to an existing cluster $C'\in \cset$; no updates to $\hat H$ are necessary.
	
	\item $\csplit$: given a cluster $C'\in \cset$, and a subgraph $C''\subseteq C'$, add $C''$ to $\cset$. In this case, we perform a supernode split in graph $\hat H$, for supernode $u(C')$. The edge set $E'$ contains, for every vertex $s\in S$ such that $(\set{s}\cup R(S))\cap V(C'')\neq \emptyset$, the edge $(s,u(C'))$. This update can be executed in time $O(|V(C'')|\cdot \mu)$.
\end{itemize}

Additionally, when a vertex $s$ is deleted from the vertex set $S$ (which may only happen when $s$ no longer serves as endpoint of any edge in $\hat E$), we delete all edges incident to $s$ from $\hat H$, and then delete vertex $s$ from $\hat H$. Notice that, since $s$ no longer has any edge incident to it in $\hat E$, all fake copies of vertex $s$ have been deleted from the graph $\tilde H$ already.

It is immediate to verify that the graph $\hat H$ that we maintain via the above valid update operations is identical to the contracted graph defined in \Cref{subsec: contracted}. Recall that $N^0(\hat H)=|S|\leq |\hat E|\leq N^0(C)/\hat W^{\eps}\leq \hat W^{\eps(z-1)/2}$, and, as we established in \Cref{subsec: contracted}, graph $\hat H$ has dynamic degree bound  $\mu'=\mu\cdot W^{O(1/\log\log W)}\leq {\mu \cdot (N^0(C))^{O(1/\log\log N^0(C))}}$ (as before, $W=N^0(\tilde H)\leq 2N^0(C)$).
Let $D'=8 D$ be the distance threshold for the graph $\hat H$.
It is immediate to extend Algorithm $\alg_2$ so that, at every time $t>0$, given an online update sequence $\tilde \Sigma_t$ (that is produced by algorithm $\alg_1$ and corresponds to the update $\sigma_t$ in the input sequence $\Sigma$ for cluster $C$), produces a sequence $\hat \Sigma_t$ of valid update operations for graph $\hat H$, so that, at any time $t>0$, the graph that is obtained from the initial graph $\hat H$ by applying the update sequence $\hat \Sigma_1\circ\cdots\circ \hat \Sigma_t$ to it is precisely $\hat H^t$ -- that is, the contracted graph $\hat H$ at time $t$. This can be done without increasing the asymptotic running time of $\alg_2$.

\subsubsection{Running $\maintaincluster$ on the Contracted Graph -- Algorithm $\alg_3$}

We have now defined a valid input structure $\hat \iset=(\hat H,\set{\ell(e)}_{e\in E(\hat H)},D')$, where $D'=8D$, that undergoes an online sequence $\hat \Sigma$ of valid update operations, with dynamic degree bound $\mu'=\mu\cdot (N^0(C))^{O(1/\log\log N^0(C))}\leq \mu (N^0(C))^{\eps/2}$ (from the assumption that $\eps\geq \Omega(1/\log\log N^0(C))$ in the statement of \Cref{thm: main maintain cluster algorithm}, and the fact that $N^0(\hat H)\leq  N^0(C)/\hat W^{\eps}$).  From the above discussion, $N^0(\hat H)\mu'\leq N^0(C)\mu$, and so
we still get that $\hat W\geq N^0(\hat H)\cdot \mu'$. Additionally, $N^0(\hat H)(\mu')^2\leq N^0(C)\mu^2$ holds.

We apply the algorithm for the \maintaincluster problem from the induction hypothesis to $\hat \iset$, with the same parameters $\hat W$ and $\eps$; we refer to this algorithm as $\alg_3$. Recall that this algorithm achieves approximation factor $\alpha_{z-1}$, and has running time:

$$4^{c(z-1)}\cdot N^0(\hat H)\cdot (\mu')^2\cdot (8D)^3\cdot\hat W^{c\eps}\cdot  (\log \hat W)^{c/\epsx^2+c(z-1)}\leq 4^{c(z-1)}\cdot N^0(C)\cdot \mu^2\cdot D^3\cdot\hat W^{c\eps}\cdot  (\log \hat W)^{c/\epsx^2+cz}.$$


Whenever Algorithm $\alg_3$ raises a flag $F_{\hat H}$, it provides a pair $s,s'\in V(\hat H)$ of regular vertices, with $\dist_{\hat H}(s,s')>1024D'\log^4\hat W=2^{13} D\log^4\hat W$.
We  the raise flag $F_C$, with the vertices $s$ and $s'$. Note that from \Cref{claim: distances preserved in contracted}, $\dist_C(s,s')>1024D\log^4\hat W$ must hold.
We then obtain, as part of the input update sequence $\Sigma$ for cluster $C$, a sequence $\Sigma'$ (the flag lowering sequence) of valid update operations for cluster $C$, at the end of which either $s$ or $s'$ are deleted from $C$. This sequence of update operations is processed by Algorithm $\alg_1$, that produces a sequence $\tilde \Sigma'$ of update operations for the modified graph $\tilde H$. Update sequence $\tilde \Sigma'$ is in turn processed by Algorithm $\alg_2$, that produces a sequence $\hat \Sigma'$ of update operations  for the contracted graph $\hat H$. Lastly, Algorithm $\alg_3$ processes the update sequence $\hat \Sigma'$ on graph $\hat H$. Since either $s$ or $s'$ are deleted from $C$ at the end of sequence $\Sigma'$, we are guaranteed that either $s$ or $s'$ are deleted from $\hat H$ 
at the end of sequence $\hat \Sigma'$. This ensures that, once flag $F_{\hat H}$ is raised, algorithm $\alg_3$ receives, as part of its update sequence, a flag-lowering sequence $\hat \Sigma'$ of valid update operations, at the end of which at least one of $s,s'$ is deleted from $\hat H$, as required.

\subsubsection{The ES-Tree: Algorithm $\alg_4$}

Our last ingerdient is algorithm $\alg_4$, that maintains a generalized \EST $\tau$ from \Cref{thm: ES-tree}, in the graph that is obtained from cluster $C$, by adding a new source vertex $s^*$, that connects to every vertex in $S$ with an edge of length $1$. The tree $\tau$ has depth $2\hat D$. As cluster $C$ undergoes valid update operations, we perform the same valid update operations in the data structure $\tau$. When a vertex $s$ is removed from $S$, we delete the edge $(s,s^*)$ from the tree. Once vertex $s^*$ becomes an isolated vertex in the tree $\tau$, Phase 3 terminates. Observe that the latter can only happen when all vertices are deleted from $S$, which in turn can only happen when $\hat E=\emptyset$ holds.

We maintain a list $L$ of all regular vertices of $C$ that do not lie in the tree $\tau$. Whenever $L\neq \emptyset$, we raise the flag $F_C$, and supply a pair $(s,v)$ of vertices, where $s$ is any vertex in $S$, and $v$ is any vertex in $L$. We are then guaranteed that $\dist_C(s,v)\geq \hat D>1024D\log^4\hat W$.

From \Cref{thm: ES-tree}, the total update time of Algorithm $\alg_4$ is $\otilde(N^0(C)\cdot \mu \cdot \hat D)\leq N^0(C)\cdot \mu\cdot D \cdot (\log \hat W)^c$.

\subsection*{Responding to Queries}
Suppose we are given a query $\spquery(C,v,v')$, where $v$ and $v'$ are regular vertices in graph $C$. 
Recall that flag $F_C$ must be down when the query is received, and so the list $L$ maintained by $\alg_4$ is empty, and flag $F_{\hat H}$ is down as well.
We compute a path $P$ connecting $v$ to $v'$ in $C$, as follows. First, we use the \EST $\tau$ to compute a path $Q$ in graph $C$, connecting $v$ to some vertex $s\in S$, and path $Q'$ in graph $C$, connecting $v'$ to some vertex $s'\in S$. The time required to compute both paths is proportional the number of edges on them. If $s=s'$, then we obtain a path connecting $v$ to $v'$ in $C$ by concatenating the paths $Q$ and $Q'$. As both paths have length at most $2\hat D\leq O(D\log^{10}\hat W)$, we obtain a $v$-$v'$ path of desired length. We assume from now on that $s\neq s'$.

Next, we run query $\spquery(\hat H,s,s')$ in the data structure maintained by Algorithm $\alg_3$. Recall that the distance threshold for the \maintaincluster problem on graph $\hat H$ is $D'=8 D$, and that the algorithm achieves approximation factor $\alpha_{z-1}$. Therefore, in response to this query, we obtain a path $\hat P$ connecting $s$ to $s'$ in graph $\hat H$, whose length is  $\hat \ell\leq \alpha_{z-1}D'=8\alpha_{z-1}D$. 

We transform path $\hat P$ into a path connecting vertex $s$ to vertex $s'$ in graph $C$, as follows. Denote $\hat P=(s_0=s,u_1,s_1,\ldots,u_{r-1},s_r=s')$, where $s_0,\ldots,s_r$ are vertices of $S$, and $u_1,\ldots,u_{r-1}$ are supernodes. Fix an index $1\leq j<r$, and consider the supernode $u_j$; assume that $u_j=u(C_j)$, where $C_j\in \cset_{i_j}$, for some $1\leq i_j\leq \lambda$.

We are then guaranteed that either $s_{j-1}$, or one of its fake copies lie in $C_j$; in the former case, we set $s'_{j-1}=s_{j-1}$, and in the latter case, we let $s'_{j-1}$ be a fake copy of $s_{j-1}$ lying in $C_j$. Similarly, we are guaranteed that either $s_j$ or one of its fake copies lie in $C_j$; in the former case, we set $s'_{j}=s_{j}$, and in the latter case, we let $s'_{j}$ be a fake copy of $s_{j}$. Next, we use query $\spquery(C_j,s'_{j-1},s'_j)$ to the algorithm for solving the \maintaincluster problem in graph $C_j$ that is part of Algorithm $\alg_2$. As the result, we obtain a path $P'_j$, connecting $s'_{j-1}$ to $s'_j$ in $C_j$, of length at most $\alpha_{z-1}\cdot 2^{i_j}$. The time required to compute the path $P'_j$ is $O(|E(P'_j)|)$. For every fake regular vertex $v^F$ on path $P'_j$, we replace $v^F$ with the original regular vertex $v$ (for which $v^F$ is a fake copy). Similarly, every fake edge $e^F$ on $P'_j$ is replaced by its original edge $e\in E(C)$. We then obtain a path $P_j$ in graph $C_j$, connecting $s_{j-1}$ to $s_j$, of length at most $\alpha_{z-1}\cdot 2^{i_j}$. Recall that the lengths of the edges $(s_{j-1},u_j)$ and $(s_j,u_j)$ were $2^{i_j}$. By concatenating the paths $P_1,\ldots,P_{r-1}$, we obtain a path $P'$ in graph $C$, connecting $s$ to $s'$, whose length is at most $\alpha_{z-1}\ell_{\hat H}(\hat P)\leq 8\alpha^2_{z-1}D$. The time that the algorithm spent on computing the path $P'$ is $O(|E(P')|)$. Lastly, we let $P$ be the $v$-$v'$ path in graph $C$, obtained by concatenating paths $Q,P'$ and $Q'$. Then path $P$ has length 
at most:

\[\begin{split} 
8\alpha^2_{z-1}D+O(D\log^{10}\hat W)& \leq 10\alpha^2_{z-1}D
\\ &\leq 10((\log \hat W)^{2c\cdot 2^{2(z-1)}/\eps^2})^2\cdot D\\
&\leq (\log \hat W)^{2c\cdot 2^{2+2(z-1)}/\eps^2}\cdot D\\
&\leq (\log \hat W)^{2c\cdot 2^{2z}/\eps^2}\cdot D\\
&=\alpha_z\cdot D.
\end{split} \]

The total running time for responding to query $\spquery(C,v,v')$ is $O(|E(P)|)$.

\subsubsection{Running Time of Phase 2}

Recall that the running time of Algorithm $\alg_1$ is $O(N^0(C)\mu)$, the running time of Algorithm $\alg_2$ is bounded by  $4^{c(z-1)}\cdot N^0(C)\cdot \mu^2\cdot D^3\cdot\hat W^{c\eps}\cdot  (\log \hat W)^{c/\epsx^2+cz}$, the running time of Algorithm $\alg_3$ is bounded by $4^{c(z-1)}\cdot N^0(C)\cdot \mu^2\cdot D^3\cdot\hat W^{c\eps}\cdot  (\log \hat W)^{c/\epsx^2+cz}$, and the running time of Algorithm $\alg_4$ is bounded by $N^0(C)\cdot \mu\cdot D \cdot (\log \hat W)^c$. Therefore, the total running time of the algorithm for Phase 2 is at most:

\[ 3\cdot 4^{c(z-1)}\cdot N^0(C)\cdot \mu^2\cdot D^3\cdot\hat W^{c\eps}\cdot  (\log \hat W)^{c/\epsx^2+cz}\]

\subsection{Phase 3}
The third phase starts when $S=\emptyset$ holds, and terminates when the number of regular vertices in $C$ falls below $4N^0(C)/\rho$.
During this phase, the algorithm will repeatedly raise flag $F_C$, every time providing a pair $v,v'$ of regular vertices of $C$ with $\dist_C(v,v')>1024D\log^4\hat W$. Therefore, the only updates to cluster $C$ that occur over the course of the phase are due to flag-lowering sequences.

We denote by $C^*$ the cluster $C$ at the beginning of Phase 3. Note that, since Phase 2 has terminated, $\hat E=\emptyset$ now holds. We start with the following observation.

\begin{observation}\label{obs: few regular vertices in balls}
	For every regular vertex $v^*\in V(C^*)$, the number of regular vertices in $B_{C^*}(v^*,\hat D)$ is at most $N^0(C)/\rho$.
\end{observation}
\begin{proof}
	Consider the graph $C'=C\setminus \hat E$, as it evolves over the course of the second phase, and let $B=B_{C'}(v^*,\hat D)$. At the beginning of the second phase, from the definition of a bad witness, the initial ball $B$ contains at most $N^0(C)/\rho$ regular vertices. We claim that no new regular vertices join the set $B$ over the course of the second phase.
	
	In order to prove this, it is enough to consider every update operation $\sigma_t\in \Sigma$ one-by-one, and to show that no new regular vertex joins $B$ as the result of $\sigma_t$. Indeed, if $\sigma_t$ is the deletion of an edge, then this may only increase distances in graph $C'$, so no new vertices join the set $B$. Similarly, if $\sigma_t$ is the deletion of an isolated vertex of $C$, then no new vertices may join $B$. Assume now that $\sigma_t$ is a supernode splitting operation of a supernode $u$, with edge set $E'$. Recall that, as the result of this operation, we add a new supernode $u'$ to $C$. For every edge $e=(u,v)\in E'$, if $e\not\in \hat E$, then edge $e'=(u',v)$, of length $\ell(e')$ is added to $C$. if $e\in \hat E$, then edge $e'=(u',v)$ is added to both $C$ and $\hat E$, so it is not added to $C'$. Therefore, in graph $C'$, we insert a new vertex $u'$, and, for every edge $e=(u,v)\in E'\setminus \hat E$, we add edge $(u',v)$ of length $\ell(e)$ to $C'$. It is easy to see that the distance between any pair of regular vertices in $C'$ may not decrease as the result of this operation, so no new regular vertices may join set $B$.
\end{proof}

Recall that, from \Cref{obs: no dist increase}, valid update operations may not decrease distances between any pair of regular vertices. Therefore, throughout the remainder of the algorithm, as $C^*$ undergoes valid update operations, for every regular vertex $v^*\in V(C^*)$, the number of regular vertices in $B_{C^*}(v^*,\hat D)$ is always bounded by $N^0(C)/\rho$.

The algorithm in Phase $3$ consists of a number of iterations, where in each iteration  we will ensure that the number of regular vertices in $C$ reduces by a constant factor. The iterations continue as long as the number of regular vertices in $C$ is at least $4N^0(C)/\rho$. Each iteration consists of two steps. The first step is summarized in the following lemma.

\begin{lemma}\label{lem: step 1}
	There is a deterministic algorithm, that, given a graph $C'$ obtained from $C^*$ after a sequence of valid update operations, that contains at least $2N^0(C)/\rho$ regular vertices, computes a collection $\bset=\set{T_1,\ldots,T_q}$ of disjoint subsets of regular vertices of $C'$, such that:
	
	\begin{itemize}
		\item for all $1\leq i\leq q$, $|T_i|\leq N^0(C)/\rho$;
		\item for all $1\leq i<i'\leq q$, if $v\in T_i$ and $v'\in T_{i'}$, then $\dist_{C'}(v,v')> 1024D\log^4\hat W$; and
		\item $\sum_{i=1}^q|T_i|\geq N(C')/2$, where $N(C')$ is the number of regular vertices in $C'$. 
	\end{itemize}
The running time of the algorithm is $O(|E(C')|)=O(N^0(C)\mu)$.
\end{lemma}
\begin{proof}
	The proof uses the standard ball-growing technique. 
	Throughout, we denote $D^*=2048D\log^4\hat W$.
	We start with $\bset=\emptyset$, and $C''=C'$, and then perform iterations, where in the $i$th iteration we add vertex set $T_i$ to $\bset$, and we delete some vertices and edges from $C''$. We also maintain a set $T_0$ of \emph{discarded regular vertices}, that is initialized to $\emptyset$.
	We will ensure that the following invariants hold at the end of each iteration $i$:
	
	\begin{properties}{P}
		\item for all $1\leq i'\leq i$, $|T_{i'}|\leq N^0(C)/\rho$; \label{invariant: small balls}
		\item $|T_0|\leq \sum_{i'=1}^i|T_{i'}|$;  \label{invariant: few discarded regular vertices}
		\item for every pair $v,v'\in V(C'')\setminus T_0$ of regular vertices, if $\dist_{C''}(v,v')\geq D^*$, then $\dist_{C'}(v,v')\geq  D^*$.\label{invariant: distances preserved}
		\item for all $1\leq i'<i''\leq i$, $\dist_{C'}(T_{i'},T_{i''})\geq D^*$; and \label{invariant: balls far from each other}
		\item for all $1\leq i'\leq i$, 
		for all pairs $v\in T_{i'}$, $v'\in V(C'')\setminus T_0$ of regular vertices,
		$\dist_{C'}(v,v')\geq D^*$. \label{invariant: balls far from remaining guys}
	\end{properties}
	
	It is easy to verify that all invariants hold at the beginning of the algorithm.
	The iterations continue as long as there is at least one regular vertex in $V(C'')\setminus T_0$.
	 We now describe the execution of the $i$th iteration. We assume that all invariants hold at the beginning of the iteration.
	
	Let $v_i$ be an arbitrary regular vertex in $V(C'')\setminus T_0$. We run Dijkstra's algorithm (that we sometimes call weighted BFS) from vertex $v_i$ in graph $C''$. For $j\geq 1$, the $j$th layer $\Lambda_j$ of the BFS is the set of all vertices whose distance from $v_i$ is between $2(j-1)D^*$ and $2jD^*$, that is: $\Lambda_j=B_{C''}(v_i,2jD^*)\setminus B_{C''}(v_i,2(j-1)D^*)$. 
	We denote by $N_j$ the number of regular vertices in $\Lambda_j$ that do not lie in $T_0$, and we denote by $E_j$ the set of all edges with both endpoints in $\Lambda_1\cup\cdots\cup \Lambda_j$. For an index $j>1$, we say that layer $\Lambda_j$ is good iff $N_{j}\leq N_1+\cdots+N_{j-1}$, and $|E_{j}\setminus E_{j-1}|\leq |E_{j-1}|$. We need the following simple observation.
	
	\begin{observation}\label{obs: good layer}
		There is an index $1<j\leq 3\log \hat W$, such that layer $\Lambda_j$ is good.
	\end{observation}
\begin{proof}
	Assume otherwise, 
	 We say that layer $\Lambda_j$ is \emph{type-1 bad} if $N_j> N_1+\cdots+N_{j-1}$, and we say that it is \emph{type-2 bad} if  $|E_{j}\setminus E_{j-1}|> |E_{j-1}|$.
	Since we have assumed that there are more than $3\log \hat W$ bad layers, either at least $\ceil{\log \hat W}$ layers are type-1 bad, or at least $\ceil{2\log \hat W}$ layers are type-2 bad. Assume first that the former is true, and let $j_1<j_2<\cdots<j_{\ceil{\log \hat W}}$ be indices of  type-1 bad layers. Then for all $2\leq i\leq {\ceil{\log \hat W}}$, 
	$N_{j_i}>N_1+\cdots+N_{j_i-1}$ must hold. Therefore, $N_{j_{\ceil{\log \hat W}}}>N_1\cdot 2^{\ceil{\log \hat W}}\geq \hat W\geq N^0(C)$, which is impossible.
The proof for the second case, where at least $\ceil{2\log\hat W}$ layers are type-2 bad is almost identical and is omitted here.
\end{proof}

We run Dijkstra's algorithm starting from vertex $v_i$, until we reach the first index $j>1$, such that $\Lambda_j$ is a good layer. We then let $T_i$ contain all regular vertices of $(\Lambda_1\cup\cdots\cup\Lambda_{j-1})\setminus T_0$, and we let $T_0^i$ contain all regular vertices of $\Lambda_j\setminus T_0$.
From the definition of the good layer, $|T_0^i|\leq |T_i|$. We add the set $T_i$ to $\bset$, and we add the vertices of $T_0^j$ to $T_0$.
We also delete from $C''$ all vertices of $\Lambda_1\cup\ldots\cup \Lambda_{j-1}$ and all edges that are incident to them. This completes the description of the $i$th iteration. We now show that all invariants continue to hold.

First, from our discussion, Invariant \ref{invariant: few discarded regular vertices} continues to hold. Also, if Invariants \ref{invariant: balls far from each other} and
\ref{invariant: balls far from remaining guys} held at the beginning of the iteration, then Invariant \ref{invariant: balls far from each other} continues to hold.

We denote by $C''_0$ the cluter $C''$ at the beginning of the iteration, and by $C''_1$ the cluster $C''$ at the end of the iteration.
Next, we estabish that Invariant \ref{invariant: balls far from remaining guys} continues to hold. Consider any pair of vertices $v\in T_i$, $v'\in V(C''_1)\setminus T_0$.
Then $v\in B_{C''_0}(v_i,2(j-1)D^*)$, while $v'\not \in B_{C''_0}(v_i,2jD^*)$. Therefore, $\dist_{C''_0}(v,v')\geq D^*$. From Invariant \ref{invariant: distances preserved}, $\dist_{C'}(v,v')\geq D^*$, and so Invariant  \ref{invariant: balls far from remaining guys} continues to hold.

In order to establish Invariant \ref{invariant: distances preserved},  consider any pair of regular vertices $v,v'\in V(C''_1)\setminus T_0$, with $\dist_{C''_1}(v,v')\geq D^*$. We claim that $\dist_{C''_0}(v,v')\geq D^*$. Indeed, assume otherwise, and let $P$ be any $v$-$v'$ path in graph $C''_0$, whose length is less than $D^*$. Since $P\not\subseteq C''_0$, it must contain at least one vertex that was deleted in the current iteration, that is, a vertex $x\in \Lambda_1\cup\cdots\cup\Lambda_{j-1}$. But from the definition of set $T_0^i$, $v\not\in \Lambda_1\cup\cdots\cup \Lambda_{j}$. Therefore, $\dist_{C''_0}(x,v)\geq 2D^*$ must hold, contradicting the fact that the length of $P$ in $C''_0$ is less than $D^*$. We conclude that $\dist_{C''_0}(v,v')\geq D^*$, and, since Invariant \ref{invariant: distances preserved} held at the beginning of the iteration, we get that $\dist_{C'}(v,v')\geq D^*$, establishing Invariant \ref{invariant: distances preserved}

It now remains to establish Invariant \ref{invariant: small balls}, for which it is enough to show that $|T_i|\leq N^0(C)/\rho$. In order to do so, observe that $T_i\subseteq B_{C''_0}(v_i,8D^*\log \hat W)$. Since  $D^*=2048D\log^4\hat W$, while $\hat D=2^{30}D\log^{10} \hat W$, we get that $T_i\subseteq B_{C''_0}{(v_i,\hat D)}\subseteq B_{C'}(v_i,\hat D)$. From \Cref{obs: few regular vertices in balls}, and since valid update operations may not decrease distances between regular vertices, it then follows that $|T_i|\leq N^0(C)/\rho$. We conclude that all invariants continue to hold at the end of the current iteration.

Note that the running time of iteration $i$ is $O(|E_j|)$, and that all edges of $E_{j-1}$ are deleted from $C''$ at the end of the iteration. Since $|E_j|\leq O(|E_{j-1}|)$, the running time of iteration $i$ is asymptotically bounded by the number of edges deleted from $C''$ in this iteration. Therefore, the total running time of the algorithm is $O(|E(C')|)\leq O(N^0(C)\mu)$. The algorithm terminates once ever regular vertex of $C''$ lies in $T_0$. We then output all vertex sets that currently lie in $\bset$. From the invariants, it is easy to verify that these vertex sets have all required properties.
\end{proof}

We now proceed to describe Phase 3. 
The phase is executed as long as the number of regular vertices in graph $C$ remains at least $4N^0(C)/\rho$, and consists of a number of iterations. Each iteration is executed as follows.
We start by computing a collection $\bset$ of sets of regular vertices, using \Cref{lem: step 1}. Next,
as long as there are two distinct non-empty sets $T',T''\in \bset$ of vertices, we let $v\in T$, $v'\in T''$ be any pair of vertices. We raise flag $F_C$ with the pair $v,v'$ of vertices. 
From \Cref{lem: step 1}, and from the fact that valid update operations may not decrease distances between pairs of regular vertices, we are guaranteed that $\dist_C(v,v')> 1024D\log^4\hat W$. We then obtain a sequence of valid update operations (flag lowering sequence), after which $v$ or $v'$ must be deleted from $C$. For every regular vertex $v''$ that is deleted from $C$ as part of the update sequence, if $v''$ lies in any set $T'''\in \bset$, we delete $v''$ from set $T'''$. The iteration terminates once all but at most one set in $\bset$ become empty.
Recall that at the beginning of the iteration, $C$ contained at least $4N^0(C)/\rho$ regular vertices, and at least half of these vertices lied in the sets of $\bset$. Since the remaining set in $\bset$ may contain at most $N^0(C)/\rho$ regular vertices, at least a quarter of the regular vertices were deleted from $C$ over the course of the current iteration. Therefore, the total number of iterations in Phase 3 is bounded by $O(\log N^0(C))\leq O(\log \hat W)$. Each iteration consists of executing the algorithm from \Cref{lem: step 1}, whose running time is $O(N^0(C)\mu)$, and additional work that can be charged to the total number of edges and vertices deleted from $C$ over the course of the iteration. Therefore, the total running time of the algorithm for Phase 3 is $O(N^0(C)\mu\log \hat W)$.  Once Phase 3 terminates, we are guaranteed that the number of regular vertices in graph $C$ is bounded by $4N^0(C)/\rho\leq
 4\hat W^{\eps z/2}/\hat W^{\eps}\leq \hat W^{\eps(z-1)/2}$. Since flag $F_C$ is repeatedly raised over the course of Phase 3, no queries may be asked over the course of the phase.

\subsection{Phase 4}
Recall that, at the beginning of Phase 4, cluster $C$ contains at most $\hat W^{\eps(z-1)/2}$ regular vertices. We apply the algorithm for the \maintaincluster problem from the induction hypothesis to cluster $C$, until the end of the algorithm. Recall that the algorithm achieves approximation factor $\alpha_{z-1}\leq \alpha_z$, and has running time at most $4^{c(z-1)}\cdot N^0(C)\cdot \mu^2\cdot D^3\cdot\hat W^{c\eps}\cdot  (\log \hat W)^{c/\epsx^2+c(z-1)}$.

\subsection{Final Accounting of the Running Time}

To summarize, the running time of Phase 1 is at most: $$\frac{c}{4}\cdot N^0(C)\cdot \mu^2\cdot D^3\cdot\hat W^{c\eps}\cdot  (\log \hat W)^{c/\epsx^2+cz};$$ the running time of
Phase 2 is at most: $$3\cdot 4^{c(z-1)}\cdot N^0(C)\cdot \mu^2\cdot D^3\cdot\hat W^{c\eps}\cdot  (\log \hat W)^{c/\epsx^2+cz};$$ the running time of Phase 3 is at most: $$O(N^0(C)\mu\log \hat W);$$ and the running time of Phase 4 is at most: $$4^{c(z-1)}\cdot N^0(C)\cdot \mu^2\cdot D^3\cdot\hat W^{c\eps}\cdot  (\log \hat W)^{c/\epsx^2+c(z-1)}.$$

Altogether, the total update time of the algorithm is bounded by:

$$4^{cz}\cdot N^0(C)\cdot \mu^2\cdot D^3\cdot\hat W^{c\eps}\cdot  (\log \hat W)^{c/\epsx^2+cz}.$$

%% file: params.tex
\section{Parameters Used in Sections \ref{sec: pivot decomposition} -- \ref{sec: final proof}}
	\label{sec: params}
	
Input: valid input structure $\iset=\left(C,\set{\ell(e)}_{e\in E(C)},D \right )$, undergoing a sequence of valid update operations with dynamic degree bound $\mu\geq 1$, and parameters $0<\eps<1$, and $\hat W\geq N^0(C)\mu$, where $N^0(C)$ is the number of regular vertices in $C$ at the beginning of the algorithm.

New parameters:

\begin{itemize}
	\item $\hat D=2^{30}D\log^{10} \hat W$ -- distance parameter;
	\item $\rho=\hat W^{\eps}$ -- balance parameter for pseudocut;
	\item $\phi^*=1/(\log \hW)^{O(1/\epsx)}$ -- expansion of the expander $X$;
	\item $\eta=\hat D\cdot \hat W^{\eps} (\log \hat W)^{O(1/\eps)}$ -- bound on the congestion of the embedding $\pset$ of $X$.
	\item $\lambda=\ceil{\log(2048D\log^4\hat W)}$ -- number of distance scales for the neighborhood covers $\cset_1,\ldots,\cset_{\lambda}$.
\end{itemize}

For a cluster $C$, we denote by $W(C)$ the number of regular vertices in it.

\paragraph{Type-1 Good Cluser $C$:} $W(C)\leq \hat W^{10\epsx}$.


\paragraph{Good Witness for a type-2 Good Cluster:} $(\hat E^*,X,\pset)$, where $\hat E^*\subseteq E(C)$ has cardinality $\Omega(W(C)/\hat W^{3\epsx})$, $X$ is a $\phi^*$-expander defined over $\hat E^*$, with maximum vertex degree at most $O(\log \hat W)$, and $\pset$ is an embedding of $X$ into $C_{|\hat E^*}$ with congestion at most $\eta$ and path lengths at most $O(\hat D\log^2\hat W)$.

\paragraph{Bad Witness:}  $(\hat D,\rho)$-pseudocut $\hat E$ for $C$, with $|\hat E|<W(C)/\hat W^{\epsx}$.

%% file: application.tex
\section{Application: Fast Algorithm for \MMF and \MM}
\label{sec: application}

In this section, we provide an algorithm for (unweighted) \MMF and {\sf Minimum} {\sf Multicut}, proving \Cref{thm: MMF and MM}.
Recall that in both problems, the input is an undirected $n$-vertex $m$-edge graph $G=(V,E)$, and a collection $\mset=\set{(s_1,t_1),\ldots,(s_k,t_k)}$ of pairs of its vertices, called demand pairs.  
In the \MMF problem, the goal is to send maximum amount of flow between the demand pairs, such that the total amount of flow traversing any edge is at most $1$. We denote by $\optmcf$ the value of the optimal solution to this problem. 
In the \MM problem, the goal is to select a minimum-cardinality subset $E'\subseteq E(G)$ of edges, such that, for all $1\leq i\leq k$, vertices $s_i$ and $t_i$ lie in different connected components of $G\setminus E'$. We denote by $\optmm$ the value of the optimal solution to \MM. We use the standard primal-dual technique-based algorithm of \cite{GK98, Fleischer00} (see also \cite{Madry10_stoc}).

For all $1\leq i\leq k$, let $\pset_i$ be the set of all paths in $G$ connecting $s_i$ to $t_i$, and let $\pset=\bigcup_i\pset_i$. We assume that graph $G$ is connected (as otherwise we can solve both problems on each of its connected components separately), so in particular $\pset\neq \emptyset$.
Below is the standard LP-relaxation of the \MMF problem (denoted by $\LP_1$), and its dual (denoted by $\LP_2$), which is a relaxation of the \MM problem.

\begin{tabular}[t]{|l|l|}\hline &\\
	$\begin{array}{lll}
	\LP_1&&\\
	\text{Max}&\sum_i\sum_{P\in \pset_i}f(P)&\\
	\text{s.t.}&&\\
	&\sum_{i=1}^k\sum_{\stackrel{P\in\pset_i:}{e\in P}}f(P)\leq
	1&\forall e\in E\\
	&f(P)\geq 0&\forall 1\leq i\leq k, \forall P\in\pset_i\\
	\end{array}$
	&$
	\begin{array}{lll}
	{\LP_2}&&\\
	\text{Min}&\sum_{e\in E}x_e\\
	\text{s.t.}&&\\
	&\sum_{e\in P}x_e\geq 1&\forall 1\leq i\leq k,\forall  P\in\pset_i\\
	&x_e\geq 0&\forall e\in E\\
	&&\\
	\end{array}$\\ &\\ \hline
\end{tabular}

We now show an algorithm that approximately solves both $\LP_1$ and $\LP_2$.
Over the course of the algorithm, we maintain lengths $x_e$ on edges $e\in E$, where at the beginning, for every edge $e\in E$, we set $x_e=1/m$. As the algorithm progresses, we may increase the lengths of the edges. We also set $f(P)=0$ for every path $P\in \pset$. So far we have obtained a feasible solution to $\LP_1$ of value $0$, and a (possibly infeasible) solution to $\LP_2$, of value $1$. The remainder of the algorithm consists of a number of iterations.

Assume for now, that we are given an oracle $\oset$, that, in every iteration, either provides a simple path $P\in \pset$, whose length (with respect to current edge lengths $x_e$) is at most $1$, or certifies that every path $P\in \pset$ has length at least $1/\alpha$, for some approximation factor $\alpha\geq 1$.

The iterations continue as long as the oracle provides a simple path $P\in \pset$ of length at most $1$. The $j$th iteration is executed as follows. Let $P_j\in \pset$ be the path provided by the oracle. Then we set $f(P_j)=1$, and we double the length $x_e$ of every edge $e\in E(P_j)$. Notice that this increases the value of the primal solution by (additive) $1$, and it increases the value of the dual solution by at most (additive) $1$. Therefore, if we denote by $c_1$ the cost of the current solution to $\LP_1$, and by $c_2$ the cost of the current solution to $\LP_2$, then, throughout the algorithm, $c_1\geq c_2-1$ always holds. Since we have assumed that $|\pset|\neq \emptyset$, after the first iteration, $c_1\geq 1$, and so $c_2\leq 2c_1$ holds for the remainder of the algorithm.

The algorithm terminates when the oracle $\oset$ certifies that the length of every path in $\pset$ is at least $1/\alpha$. Note that, by setting $x'_e=x_e\cdot\alpha$, we obtain a feasible solution to $\LP_2$, of value at most $\alpha c_2\leq 2\alpha c_1$.

The flow values $\set{f(P)}_{P\in \pset}$ also provide a solution to $\LP_1$, but that solution may be infeasible, since some edges may carry more than one flow unit.
However, since we set, at the beginning, for every edge $e\in E(G)$, $x_e=1/m$, and since, whenever a path containing $e$ is added to $\pset$, we double the length of the edge $x_e$, it is easy to verify that the total flow that any edge $e\in E(G)$ carries is bounded by $\ceil{\log m}$. Let $f'$ be the multicommodity flow obtained by scaling the flow $f$ down by factor $1/\ceil{\log m}$. Then $f'$ is a feasible fractional solution to \MMF, of value $c_1'=c_1/\ceil{\log m}$. From the above discussion, $c_2\leq 2c_1 \leq 4c_1'\log m$.

Recall that, from LP-duality, $c_1'\leq \optmcf= \opt_{\LP_1}= \opt_{\LP_2}\leq  \alpha c_2$. Therefore, $\optmcf\leq \alpha c_2\leq O(\alpha \log m) c'_1$ holds.
We conclude that we have obtained a solution to the {\sf Maximum Multicommodity} {\sf Flow} problem, of value $\Omega(\optmcf/(\alpha\log m)$.

Additionally, $\optmm\geq \opt_{\LP_2}\geq c_1'\geq  \Omega(c_2/\log m)$. Therefore,
 we have obtained a fractional solution   $\set{x'_e}_{e\in E(G)}$ to $\LP_2$, of value $\alpha c_2\leq O(\alpha\log m) \optmm$. Our last step is to transform this fractional solution to the \MM instance $(G,\mset)$ into an integral one, using the standard ball-growing technique of \cite{LR,GVY}. The resulting deterministic algorithm (that is very similar in nature to our Procedure \proccut, which in fact was inspired by the algorithm of \cite{LR,GVY}), obtains an integral solution to the \MM problem instance $(G,\mset)$, in time $O(|E(G)|)$, of cost $O(\log m)\cdot c$, where $c\leq O(\alpha\log m)\optmm$ is the cost of the fractional solution to $\LP_2$. 

We conclude that the above algorithm provides an $O(\alpha\log m)$-approximate solution for the \newline \MMF problem, and an $O(\alpha\log^2 m)$-approximate solution for \MM, where $\alpha$ is the approximation factor of the oracle $\oset$.

\subsection{Implementing the Oracle}
We now show an algorithm to efficiently implement the oracle $\oset$. 
One difficulty in implementing it via the algorithm from \Cref{thm: NC}  in a straightforward way is that the algorithm from \Cref{thm: NC}, in response to a query $\spquery(C,s_i,t_i)$ may return an $s_i$-$t_i$ path $P$ that is non-simple, and moreover, if we let $P'$ be a simple path obtained from $P$ by removing all cycles, then it is possible that $|E(P)|\gg |E(P')|$. This is a problem because the algorithm spends time $O(|E(P)|)$ in order to process the query, but we will only double the lengths of the edges lying on the path $P'$. This in turn may result in a running time that is too high overall. Ideally, we would like to ensure that, if the algorithm from \Cref{thm: NC} returns an $s_i$-$t_i$ path $P$ that is non-simple, and $P'$ is the corresponding simple path, then $|E(P')|$ is close to $|E(P)|$. We overcome this difficulty as follows. Our algorithm consists of $O(1/\eps)$ phases. Denote $m'=\ceil{2m\log m}$. Let $\alpha^*=O\left ((\log m)^{2^{O(1/\eps)}}\right )$ be the approximation factor that the algorithm from \Cref{thm: NC} achieves on a graph with $m'$ edges. For $j\geq 0$, let $\alpha_j=(\alpha^*)^{2j}$, and let $L_j=m^{j\eps}$. We will ensure that the following invariant holds:

\begin{properties}{I}
	\item For all $j\geq 0$, at the beginning of Phase $(j+1)$, every path $P\in \pset$ whose length is at most $1/\alpha_j$ contains at least $L_j$ edges.\label{inv: short path few edges}
\end{properties}

Notice that the invariant clearly holds at the beginning of the first phase. We now describe the execution of the $(j+1)$th phase, for some $j\geq 0$.

\paragraph{Execution of Phase $(j+1)$, for $j\geq 0$.} 
We construct a graph $G_j$, whose vertex set is $V(G_j)=V(G)$. 
For every edge $e=(v,v')\in E(G)$, and for every integer $1\leq i\leq \ceil{\log m}$, we 
add an edge $e_i=(v,v')$ to $G_j$, of length $\frac{2^i}{m}+\frac{1}{2\alpha^*\cdot \alpha_j\cdot L_{j+1}}$. We call edge $e_i$ \emph{the $i$th copy of $e$}. 
Throughout the algorithm, whenever the length of edge $e$ in graph $G$ doubles, we delete from $G_j$ the lowest-length copy of the edge $e$. This ensures that, if the length of $e$ in $G$ is $2^i/m$, then every copy of $e$ in $G_j$ has length at least  $\frac{2^i}{m}+\frac{1}{2\alpha^*\cdot \alpha_j\cdot L_{j+1}}$.
We the initialize the data structure from \Cref{thm: NC} on this new graph $G_j$, with target distance threshold $D=1/(\alpha^*\cdot \alpha_j)$, and we denote by $\cset$ the weak $(D,\alpha^*\cdot D)$-neighborhood cover of $G_j$ that the algorithm maintains. (Recall that the definition of the \NC problem requires that the length of every edge is at least $1$. In order to achieve this, we need to scale all edge lengths so they become integral, and we need to do the same with the parameter $D$. As this does not change the problem in any way, we ignore this minor technicality).

We mark every demand pair $(s_i,t_i)\in \mset$ as \emph{unexplored}. As the algorithm progresses, we will mark some demand pairs as explored. For each such demand pair $(s_i,t_i)$, we will ensure that the distance, in the current graph $G_j$, between $s_i$ and $t_i$, is at least $1/(\alpha^*\cdot \alpha_j)$. We now describe a single iteration.

If every demand pair is marked as explored, then the phase terminates.
We are then guaranteed that every path in the current graph $G_j$, connecting any demand pair $(s_i,t_i)\in \mset$ has length at least $1/(\alpha^*\cdot \alpha_j)$ in $G_j$. We claim that in this case,  every path $P\in \pset$ whose length is at most $1/\alpha_{j+1}$ (in graph $G$), contains at least $L_{j+1}$ edges. Indeed, assume otherwise, and let $P\in \pset$ be a path connecting some demand pair $(s_i,t_i)\in \mset$, that has length $\ell\leq 1/\alpha_{j+1}$  in graph $G$, and contains fewer than $L_{j+1}$ edges. Let $P'$ be an $s_i$-$t_i$ path in graph $G_j$, obtained by taking, for every edge $e\in E(P)$, a copy that has smallest length. Then the length of path $P'$ in graph $G_j$ is bounded by:

\[\ell+\frac{1}{2\alpha^*\cdot \alpha_j}\leq \frac{1}{\alpha_{j+1}}+\frac{1}{2\alpha^*\cdot \alpha_j}\leq \frac{1}{\alpha_{j}\cdot (\alpha^*)^2}+\frac{1}{2\alpha^*\cdot \alpha_j}< \frac{1}{\alpha^*\cdot \alpha_j},\]

a contradiction to the fact that demand pair $s_i$-$t_i$ is marked as explored.
Therfore, when the phase terminates, Ivariant \ref{inv: short path few edges} holds.  

Assume now that not every demand pair in $\mset$ is marked as explored, and let $(s_i,t_i)\in \mset$ be any demand pair that is not marked as explored. Let $C=\coveringcluster(s_i)$ be the cluster of $\cset$ containing $B_{G_j}(s,D)$, that the algorithm   from \Cref{thm: NC} maintains. We start by checking, in time $O(1)$, whether $t_i\in C$. If this is not the case, then we are guaranteed that $\dist_G(s_i,t_i)>D=1/(\alpha^*\cdot \alpha_j)$. We then mark demand pair $(s_i,t_i)$ as explored, and continue to another unexplored demand pair. Otherwise, if $t_i\in C$, then we run query $\spquery(C,s_i,t_i)$ in the data structure maintained by the algorithm from \Cref{thm: NC}. The algorithm is then guaranteed to return a path connecting $s_i$ to $t_i$ in graph $G_j$, of length at most $1/\alpha_j$. 
We denote this path by $P$. From the way we set the lengths of the edges in graph $G_j$, we are guaranteed that $|E(P)|\leq 2\alpha^*\cdot L_{j+1}$.  Path $P$ immediately gives us the corresponding (possibly non-simple) path $P'$ in graph $G$, whose length is at most $1/\alpha_j$.  Let $P''$ be a simple path that is obtained from $P'$, after removing all cycles from it. Note that we can compute $P''$ in time $O(|E(P)|)$, and the query time $\spquery(C,s_i,t_i)$ also took time $O(|E(P)|)$. Then the length of path $P'$ is bounded by $1/\alpha_j$, and, from Invariant \ref{inv: short path few edges}, path $P''$ contains at least $L_j$ edges. We then return the path $P''$ and terminate the iteration.

The algorithm terminates after $t=\ceil{1/\eps}$ phases, at which time we are guaranteed, from Invariant \ref{inv: short path few edges}, that every path in $\pset$ has length at least $1/\alpha_t$, for $\alpha_t=(\alpha^*)^{O(1/\eps)}=O\left ((\log m)^{2^{O(1/\eps)}}\right )$. We denote $\alpha=\alpha_t$, the approximation factor of the oracle $\oset$. We now analyze the running time of a single phase. 

The time required to maintain the data structure from \Cref{thm: NC} is $O\left (m^{1+O(\eps)}\cdot (\log m)^{O(1/\eps^2)}\right )$. The time needed to process every query $\spquery(C,s_i,t_i)$ is $O(|E(P)|)$, where $P$ is the returned path. Recall that we have established that $P$ contains at most $2\alpha^*\cdot L_{j+1}$ edges, while its corresponding simple path $P''$ contains at least $L_j$ edges. Therefore, $|E(P)|\leq 2\alpha^*\cdot m^{\eps}|E(P'')|$. We charge every edge on path $P''$ for at most $2\alpha^*\cdot m^{\eps}$ edges on path $P$. Since we double the length of every edge on path $P''$ in graph $G$, and since  the length of every edge may only be doubled $O(\log m)$ times, 
an edge of $G$ may be charged at most $O(\log m)$ times over the course of a single phase. Therefore, the total time for processing all queries $\spquery(C,s_i,t_i)$ over the course of the phase, and also for computing the corresponding simple paths, is bounded by $O\left (m^{1+\eps}(\log m)^{2^{O(1/\eps)}}\right )$.
 Lastly, for every demand pair $(s_i,t_i)$, we may spend additional $O(1)$ time in the iteration in which the pair is marked as explored. Therefore, the total running time of 
 a single phase is bounded by $O\left (m^{1+O(\eps)}(\log m)^{2^{O(1/\eps)}}+k\right )$. Since the total number of phases is bounded by $O(1/\eps)$, the total running time of the algorithm implementing the oracle is $O\left (m^{1+O(\eps)}(\log m)^{2^{O(1/\eps)}}+k/\eps\right )$. The time required in order to implement the remainder of the algorithm (that is, updating the flow $f$ and the edge lengths) is subsumed by this running time. Therefore, the total running time of the algorithm is $O\left (m^{1+O(\eps)}(\log m)^{2^{O(1/\eps)}}+k/\eps\right )$.
Since this implementation of the oracle achieves approximation factor $\alpha=O\left ((\log m)^{2^{O(1/\eps)}}\right )$, the final approximation factor that we achieve for both \MMF and \MM is $O\left ((\log m)^{2^{O(1/\eps)}}\right )$.

%% file: ack.tex
\section{Acknowledgement}
\label{sec: ack}
We thank Merav Parter for introducing us to the notion of neighborhood covers and to the work of \cite{neighborhood-cover2,neighborhood-cover1}, that was the inspiration for this work.

%% file: prelims-appx.tex
\section{Proofs Omitted from \Cref{sec: prelims}}

\subsection{Short Paths in Expanders -- Proof of \Cref{obs: short paths in exp}}
\label{subsec: short paths in expanders}
The proof follows the standard ball-growing technique. For all $i\geq 1$, we denote $S_i=B_G(v,i)$, and $S'_i=B_G(u,i)$. 
Consider now some index $i$, and assume that $|S_i|\leq n/2$. Let $R_i=S_{i+1}\setminus S_i$ be the set of all vertices $v'$ whose distance from $v$ is exactly $i+1$. Then, since graph $G$ is a $\phi$-expander, $|E_G(S_i,R_i)|=|E_G(S_i,V(G)\setminus S_i)|\geq \phi |S_i|$. On the other hand, since every vertex in $R_i$ has degree at most $\Delta$, $|E_G(S_i,R_i)|\leq \Delta|R_i|$. Therefore, $|R_i|\geq \phi |S_i|/\Delta$ must hold. Since $S_{i+1}=S_i\cup R_i$, we get that $S_{i+1}\geq(1+\phi/\Delta)|S_i|$. We conclude that, for $q=\ceil{\Delta\log n/\phi}$, $|S_q|\geq n/2$ must hold. Using similar reasoning, $|S'_q|\geq n/2$ must hold, and so there must be a path connecting $u$ to $v$, whose length is at most $2q+1\leq 8\Delta\log n/\phi$.

\subsection{Expander Pruning -- Proof of \Cref{thm: expander pruning}}
\label{subsec: expander pruning}

For the sake of the proof, we need to define the notions of volumes of vertex sets, and graph conductance. Given a graph $G=(V,E)$, and a subset $S\subseteq V$ of its vertices, the \emph{volume} of $S$ in $G$ is $\vol_G(S)=\sum_{v\in S}\deg_G(v)$. Given a cut $(A,B)$ in graph $G$, its \emph{conductance} is $\psi_G(A,B)=\frac{|E_G(A,B)|}{\min\set{\vol_G(A),\vol_G(B)}}$. The \emph{conductance of a graph $G$}, denoted by $\Psi(G)$, is the minimum conductance of any cut in $G$. The following observation is immediate from the definitions of graph expansion and conductance.

\begin{observation}
	\label{prop:sparsity vs conductance}Let $G=(V,E)$ be a connected
	graph with maximum vertex degree $\Delta$, and let $(A,B)$ be any cut in $G$. Then: 
	
	$$\frac{\phi_{G}(A,B)}{\Delta} \le\psi_{G}(A,B)\le\phi_{G}(A,B).$$
	
	In particular:
		$$\frac{\Phi(G)}{\Delta}\le \Psi(G)\le \Phi(G).$$
	
\end{observation}

In order to prove \Cref{thm: expander pruning}, we use the following theorem from ~\cite{expander-pruning}

\begin{theorem}[Restatement of Theorem 1.3 in~\cite{expander-pruning}]\label{thm: expander pruning-inner}
	There is a deterministic algorithm, that, given an access to the adjacency list of a graph $G=(V,E)$ with $|E|=m$, a parameter $0<\psi\leq 1$, and a sequence $\Sigma=(e_1,e_2,\ldots,e_{k})$ of $ {k \leq \psi m/10}$ online edge deletions, maintains a vertex set $S\subseteq V$ with the following properties. Let $G_i$ be the graph $G$ after the edges $e_1,\ldots,e_i$ have been deleted from it; let $S_0=\emptyset$ be the set $S$ at the beginning of the algorithm, and for all $0<i\leq k$, let $S_i$ be the set $S$ after the deletion of $e_1,\ldots,e_i$. Then, for all $1\leq i\leq k$:
	
	\begin{itemize}
		\item $S_{i-1}\subseteq S_i$;
		\item $ {\vol_G(S_i)\leq 8i/\psi}$;
		\item $|E(S_i,V\setminus S_i)|\leq 4i$; and
		\item if  $\Psi(G)\geq \psi$, then the conductance of the graph $G_i[V\setminus S_i]$ is at least $\psi/6$.
	\end{itemize}
	
	The total running time of the algorithm is $O(k\log m/\psi^2)$.
\end{theorem}

We are now ready to prove \Cref{thm: expander pruning}. Let $G=(V,E)$ be the input graph, that is a $\phi$-expander, with maximum vertex degree at most $\Delta$. From the definition of expanders, $G$ must be a connected graph, and, from \Cref{prop:sparsity vs conductance}, if we denote by $\psi=\Psi(G)$, then $\psi\geq \phi/\Delta$. From the statement of   \Cref{thm: expander pruning}, the number of edges in the deletion sequence $\Sigma$ is $k\leq \frac{\phi|E|}{10\Delta}\leq \frac{\psi|E|}{10}$, as required. We can now apply the algorithm from \Cref{thm: expander pruning-inner} to graph $G$ and edge deletion sequence $\Sigma$. We are then guaranteed that for all $1\leq i\leq k$, $S_{i-1}\subseteq S_i$, as required. Moreover, we are guaranteed that for each resulting graph $G'_i=G_i[V\setminus S_i]$, the conductance of $G'_i$ is at least $\psi/6$, and so, from \Cref{prop:sparsity vs conductance}, we get that $\Phi(G'_i)\geq \Psi(G'_i)\geq \psi/6\geq \phi/(6\Delta)$. Therefore, graph $G'_i$ is a $\phi/(6\Delta)$-expander. We are also guaranteed that $|E(S_i,V\setminus S_i)|\leq 4i$. Lastly, since graph $G$ is connected, for every vertex set $S$, $|S|\leq \vol_G(S)$ must hold, and so for all $i$, 
$|S_i|\leq \vol_G(S_i)\leq 8i/\psi\leq 8i\Delta/\phi$. The running time of the algorithm is $O(k\log |E|)/\psi^2=\otilde(k\Delta^2/\phi^2)$.

%% file: ES-tree-appx.tex
\section{Generalized Even-Shiloach Trees -- Proof of \Cref{thm: ES-tree}}
\label{sec: proof of ES tree thm}

\paragraph{Data Structures.}
We maintain the graph $H$ as an adjacency list: for every vertex $v$, we maintain a linked list of its neighbors. 
Throughout the algorithm, we will maintain the following data structures:

\begin{itemize}
	\item A shortest-path tree $T$ rooted at vertex $s$, that contains all vertices $x\in V(H)$ with $\dist_H(s,x)\leq D^*$. For every such vertex $x$, its correct distance $\lambda (x)=\dist_H(s,x)$ from $s$ is stored together with $x$.
	
	\item For every vertex $x\in V(T)$, let $\pred(x)$ be the set of all vertices $y\in V(T)$ with $(x,y)\in E(H)$. We maintain a heap $\heap(x)$ in which all elements of $\pred(x)$ are stored, where the key associated with each element $y\in \pred(x)$ is $\lambda (y)+\ell(y,x)$.

	\item For every edge $e=(y,x)$ with $y,x\in V(T)$, we store, together with vertex $y$, a pointer to its corresponding element in the heap $\heap(x)$ and vice versa.
\end{itemize}

Throughout,  we denote by $m'$ the total number of edges that are ever present in $H$, so  $m'\leq N^0\cdot \mu$.

\paragraph{Initialization.}
We run Dijkstra's algorithm on graph $H$ time $\otilde(|E(H)|)\leq \otilde(m')$, up to distance threshold $D^*$, to construct the initial tree $T$. For every vertex $x\in V(T)$, we initialize $\lambda(x)=\dist_H(s,v)$, and for all $x\not\in V(T)$, we set $\lambda(x)=\infty$. We then process all vertices of $V(T)$ in the order of their distance from $s$, and insert each such vertex $x$ into all heaps $\heap(y)$ of all vertices $y\in V(H)$ that are neighbors of $x$ in $H$. Clearly, initialization takes time at most $\otilde(m')$.

We now assume that we are given a valid data structure, and describe algorithms for handling updates.

\paragraph{Edge Deletions.}
The procedure for processing edge deletions is completely standard. The description provided here is due to Chechik \cite{chechik}. It is somewhat different but equivalent to the standard description.
Suppose an edge $e=(x,y)$ is deleted from the graph $H$.  We start by updating the heaps $\heap(x)$, $\heap(y)$ with the deletion of this edge.  If $e\not\in E(T)$, then no further updates are necessary.

Therefore, we assume from now on that $e=(x,y)$ is an edge of the tree $T$, and we assume w.l.o.g. that $y$ is the parent of $x$ in $T$. Let $S$ contain all vertices of $T$ that lie in the subtree $T_x$ of $x$. We delete the edge $(x,y)$ from $T$, thereby disconnecting all vertices of $S$ from $T$.  The remainder of the algorithm consists of two phases. In the first phase, we identify the set $R\subseteq S$ of all vertices, whose distance from $s$ has increased, and connect all remaining vertices of $S$ to $T$. In the second phase, we attempt to reconnect vertices of $R$ to $T$.

In order to implement the first phase, we maintain a heap $Q$ of all vertices that need to be examined, where the key associated with each vertex $a$ in $Q$ is $\lambda(a)$.
We initialize $Q$ to maintain a single vertex -- the vertex $x$, and we also initialize $R=\emptyset$. Heap $Q$ will have the property that, if some vertex $a$ belongs to $Q$, and vertex $a'$ was the parent of $a$ in $T$, then $a'$ was added to $R$ (or, if $a=x$, then $a'\in T$). Moreover, if $a$ is a vertex of $Q$ with smallest key $\lambda(a)$, and  $a'$ is the element lying at the top of $\heap(a)$, then $a'\in T$ must hold. Both these invariants hold at the beginning.

The algorithm iterates, as long as $Q\neq \emptyset$. Let $a$ be a vertex of $Q$ with smallest key $\lambda(a)$. Let $b$ be the element lying at the top of $\heap(a)$. We check whether connecting $a$ to the tree $T$ via vertex $b$ will allow us to keep $\lambda(a)$ unchanged, or, equivalently, whether $\lambda(a)=\lambda(b)+\ell(a,b)$. If this is the case, then we connect $a$ to the tree $T$ via $b$, delete $a$ from $Q$, and proceed to the next iteration. Otherwise, we are guaranteed that $\lambda(a)$ must increase. We then add $a$ to $R$, add all children of $a$ in the original tree $T$ to the heap $Q$, and delete $a$ from $\heap(b')$ for all neighbors $b'$ of $a$. 

Note that the algorithm examines vertices of $S$ in the non-decreasing order of their label $\lambda(a)$ (except that, when some vertex $a$ is reconnected to the tree $T$, then its descendants will never be examined). A vertex $a$ that is examined is either connected to the tree, or it is added to $R$, and all its children are added to $Q$. This ensures that, throughout the algorithm, if $a$ is a vertex of $Q$ with smallest key $\lambda(a)$, and  $a'$ is the element lying at the top of $\heap(a)$, then $a'\in T$ must hold. Indeed, if $a'\not\in T$, then we should have examined $a'$ before, and, if it was added to $R$, then it would have been deleted from $\heap(a)$.

The first phase terminates once $Q=\emptyset$. It is not hard to see, using standard analysis, that its running time is bounded by $\Otilde(\sum_{x\in R}\deg_H(x))$.

In the second phase, we run Dijkstra's algorithm on the vertices in $R$, up to distance $D^*$, trying to reconnect them to the tree $T$. We also update the heaps of all vertices of $R$, and of their neighbors, accordingly. This step can also be implemented in time $\Otilde(\sum_{x\in R}\deg_H(x))$.

To summarize, in edge-update operations, whenever, for any vertex $x$, its distance from $s$ increases, we may need to pay $\otilde(\deg_H(x))$ in running time. It is then easy to see that the total update time due to edge deletions is bounded by $\otilde(m'D^*)$.

\paragraph{Deletion of Isolated Vertices.}
Deletion of isolated vertices is straightforward, and takes $O(1)$ time per vertex. Note that an isolated vertex may not belong to $T$ (unless that vertex is $s$ and $T$ only contains the vertex $s$), so apart from deleting the vertex from $H$, no further updates are necessary.

\paragraph{Supernode Splitting.}
Recall that in a supernode-splitting update, we are given a supernode $u\in U$, and a set $E'\subseteq \delta_H(e)$ of edges. 
We need to add a new supernode $u'$ to the graph, and, for every edge $e=(u,v')\in E'$, insert an edge $e'=(u',v')$ of length $\ell(e)$ into $H$. 

If $u\not \in T$, then we simply set $\lambda(u')=\lambda(u')=\infty$ and terminate the update algorithm.

Assume now that $u\in T$, and assume first that $u\neq s$. Let $v$ be the parent of the vertex $u$ in the tree $T$. We now proceed as follows:

\begin{enumerate}
	\item Add a new vertex $u'$ as the child of $v$ to the tree $T$ (it is convenient to think of it as a copy of $u$); set $\ell(u',v)=\ell(u,v)$, $\lambda(u')=\lambda(u)$, add $v$ to $\heap(u')$, and add $u'$ to $\heap(v)$.
	
	\item For every edge $e=(u,v')\in E'$ with $v'\neq v$, add an edge $e'=(u',v')$ of length $\ell(e)$ to the graph $H$, add $v'$ to $\heap(u')$, and add $u'$ to $\heap(v')$. Notice that the insertion of these edges does not decrease the distance of any vertex from $s$, since  $\lambda(u)=\lambda(u')$, and every regular vertex that serves as endpoint to a newly inserted edge is also a neighbor of $u$.
	
	
	\item If edge $(u,v)\not\in E'$, delete the edge $(v,u')$ from the graph $H$, and from the \EST data structure, using the edge-deletion update operation.
\end{enumerate}

If $u=s$, then the update procedure is almost identical, except that we initially insert $u'$ as a child of $u$, and set the length of the edge $(u,u')$ to be $1/2$. This ensures that, as edges corresponding to the edge set $E'$ are inserted into the graph, the distances from vertices of $H$ to $s$ do not decrease, and, since the lengths of all edges in $H$ are at least $1$, $T$ remains a valid shortest-path tree. 
In the last step, we delete the edge $(u,u')$ from our data structure using the edge-deletion update operation.

The processing time of this update procedure, excluding the calls to the edge-deletion updates in the \EST data structure, is $\otilde(|E'|)$. 
The possible insertion of edge $(u,u')$, or edge $(u,v)$, if this edge does not lie in $E'$, increases the total number of edges inserted into our data structure. However, since $|E'|\geq 1$ must hold, this new inserted edge can be charged to some edge of $E'$, increasing the total number of edges that are ever present in our data structure by at most factor $2$. 
From the above discussion, if we denote the length of the input update sequence $\Sigma$ by $k$, the total update time of the \EST data structure is bounded by $\otilde(m'D^*+k)\leq \otilde(N^0\mu D^*+k)$.

Lastly, observe that each update operation in $\Sigma$ either inserts at least one edge into $H$, or deletes at least one vertex or an edge from $H$. Therefore, $k\leq O(m')$, and the total running time of the algorithm is at most $\otilde(N^0\mu D^*)$.

\paragraph{Responding to Queries $\shortestpath$.} Recall that in  $\shortestpath$, we are given a vertex $x\in V(H)$, and our goal is to either correctly establish, in time $O(1)$, that $\dist_H(s,x)>D^*$, or to return a shortest $s$-$x$ path $P$, in time $O(|E(P)|)$. Given a query vertex $x$, we check whether $\lambda(x)=\infty$. If so, we report that $\dist_H(s,x)>D^*$. Otherwise, we retrace the unique path $P$ connecting $x$ to $s$ in the tree $T$, and return it, in time $O(|E(P)|)$.

%% file: rec-comp.tex
\section{Recursive Composition of \recdynnc Instances -- Proof of \Cref{lem: inductive dynamic NC algorithm}}
\label{sec:recursive composition}

	If $\eps\geq c/\log \log W$ for some constant $c$, then, since we can assume that $\tilde c$ is a large enough constant, the desired approximation factor $\alpha_i=(\log W)^{\tilde c i\cdot 2^{\tilde c/\eps}}>W^3$. In this case, we can simply use a known deterministic algorithm for fully dynamic minimum spanning forest of \cite{dynamic-connectivity}, that has total update time $O(W\log^4W)$. We simply let $\cset$ be the set of all connected components of the graph $H$. It is easy to verify that $\cset$ can be maintained using allowed changes only, in the same asymptotic total update time. Queries $\spquery(C,v,v')$ can be easily handled by returning the unique $v$-$v'$ path in the minimum spanning tree of $C$. Therefore, we assume from now on that $\eps< c/\log \log W$.

	The proof is by induction on $i$. The base case is when $i=1$, so $D\leq 6W^{\eps}$.
	We use the algorithm from \Cref{thm: main dynamic NC algorithm}, whose approximation factor is $(\log(W\mu))^{2^{O(1/\eps)}}\leq \alpha_1$ (if constant $\tilde c$ is large enough), and running time is  bounded by:
	
	$O\left (W^{1+O(\eps)}\cdot D^3\cdot (\log W)^{O(1/\eps^2)}\right )\leq \left(\tilde c \cdot W^{1+\tilde c\eps}\cdot (\log W)^{\tilde c/\eps^2}\right ),$
	as required.

	For the induction step, we consider some integer $1< i\leq \ceil{1/\eps}$, and assume that the lemma holds for integers below $i$.
	Consider the input graph $H$ and the distance threshold $D\leq  6W^{\eps i}$.
We use a parameter $D'=\floor{W^{\eps (i-1)}/4}$.
Clearly, $D'\geq W^{\eps(i-1)}/8$. 
We say that an edge $e$ of $H$ is \emph{long} if $\ell(e)>D'$, and we say that it is \emph{short} otherwise.

Let $H'$ be a dynamic graph that is  obtained from $H$ by deleting all long edges from it. We can generate a valid update sequence for graph $H'$ from the valid update input sequence $\Sigma$ for graph $H$ in a natural way, by ignoring updates concerning long edges. 
Therefore, we now obtained a valid input structure  $\iset'=\left(H',\set{\ell(e)}_{e\in E(H')}, 4D' \right )$, that undergoes a sequence $\Sigma'$ of edge-deletion and isolated vertex-deletion operations, with dynamic degree bound $2$, and the initial number of the regular vertices in $H'$ is at most $W$. Since $4D'\leq W^{\eps (i-1)}$, we can apply the induction hypothesis to input structure $\iset'$, sequence $\Sigma'$ of update operations, distance bound $4D'$, and approximation factor $\alpha_{i-1}=(\log W)^{\tilde c (i-1)\cdot 2^{\tilde c/\eps}}$.   The total update time of this algorithm is bounded by  $\left (\tilde c^{i-1}\cdot W^{1+\tilde c\eps}\cdot (\log W)^{\tilde c/\eps^2}\right )$. 
We are also guaranteed that for every regular vertex $v$ of $H'$, the total number of clusters in $\cset'$ that ever contain $v$ is at most $W^{O(1/\log\log W)}$.
We denote by $\dset(H')$ the corresponding data structure, by
 $\cset'$ the collection of clusters that the algorithm maintains, and by $\uset'=\set{V(C)\mid C\in \cset'}$ the corresponding vertex subsets.
We denote the algorithm that we have described so far, that maintains the data structure $\dset(H')$, by $\alg_1$.

	We use the neighborhood cover $\cset'$ in order to define another dynamic graph $\hat H$, as follows.
	We start with letting $\hat H=H$, and then round all edge lengths up to the next integral multiple of $D'$, denoting the resulting new length of each edge $e$ by $\hat \ell(e)$. 
	 Notice that, if $e$ is a short edge, then $\ell(e)\leq \hat\ell(e)= D'$, and if $e$ is a  long edge, then $\ell(e)\leq \hat \ell(e)\leq 2 \ell(e)$. Additionally, for every cluster $C\in \cset'$, we add a supernode $u(C)$ to graph $\hat H$, and connect $u(C)$ with an edge to every {\bf regular} vertex $v\in V(H')$ that lies in $C$; the length of this edge is $4D'$. This ensures that the length of every edge in $\hat H$ is an integral multiple of $D'$. For each time point $t$, we denote by $\hat H^t$ the graph $\hat H$ obtained immediately after the $t$th update in sequence $\Sigma$ to graph $H$ is processed.

	We now proceed as follows. First, we show an algorithm to construct an initial graph $\hat H^0$, and an algorithm to produce an online sequence of valid update operations $\hat \Sigma$ for this graph, so that, at every time $t$ (that is, after $t$th update in sequence $\Sigma$ for graph $H$ is processed by our algorithm), the resulting graph that we obtain is precisely $\hat H^t$. We will also show that, for every pair of regular vertices $v,v'\in V(H)$, the distance between them in $\hat H$ is close to that in $H$. We will use the algorithm from \Cref{thm: main dynamic NC algorithm} recursively on graph $\hat H$, to maintain a neighborhood cover $\hat \cset$ of regular vertices in graph $\hat H$. Lastly, we show that we can use the resulting dynamic neighborhood cover $\hat C$ for graph $\hat H$ in order to maintain the desired neighborhood cover for graph $H$.

\paragraph{Maintaining Graph $\hat H$.}
Recall that, from the definition of the \recdynnc problem, initially, the cluster set $\cset'$ consists of a single cluster $H'$, and its corresponding collection $\uset'$ of vertex subsets contains a single vertex set, $V(H')$. As the time progresses, vertex sets in $\uset'$ may only be updated via allowed change operations: $\delvertex$, $\addsupernode$, and $\csplit$.
We define the initial graph $\hat H^0$ as follows. We set $\hat H^0=H$, except that we set the edge lengths $\set{\hat \ell(e)}_{e\in E(\hat H^0)}$ as described above. Additionally, we add a single supernode $u(H')$, that connects to every regular vertex of $H'$ with an edge of length $4D'$. We use the following claim in order to produce input update sequence $\hat \Sigma$ for $\hat H$.

\begin{claim}\label{claim: maintain hat H}
	There is a deterministic algorithm that, given, at each time $t\geq 1$, the $t$-th update $\sigma_t$ in the input update sequence $\Sigma$ for $\iset$, produces a sequence $\hat \Sigma_t$ of valid update operations for graph $\hat H$, such that, for all $t\geq 0$, the graph obtained from $\hat H^0$ by applying the update sequence $(\hat \Sigma_1\circ \cdots\circ \hat \Sigma_t)$ to it is precisely $\hat H^t$. 
	The dynamic degree bound for the resulting dynamic graph $\hat H$ is at most $W^{O(1/\log\log W)}$.
	The total update time of the algorithm is bounded by  $O\left (\tilde c^{i-1}\cdot W^{1+\tilde c\eps}\cdot (\log W)^{\tilde c/\eps^2}\right )$.
\end{claim}

\begin{proof}
	We run Algorithm $\alg_1$, that solves the \recdynnc problem on graph input structure  $\iset'=\left(H',\set{\ell(e)}_{e\in E(H')}, 4D' \right )$, that undergoes a sequence $\Sigma'$ of edge-deletion and isolated vertex-deletion operations, with dynamic degree bound $2$, with the initial number of the regular vertices in $H'$ bounded by $W$. Recall that the algorithm 
	maintains the neighborhood cover $\cset'$, achieves approximation factor $\alpha_{i-1}=(\log W)^{\tilde c (i-1)\cdot 2^{\tilde c/\eps}}$, and has running time at most  $\left (\tilde c^{i-1}\cdot W^{1+\tilde c\eps}\cdot (\log W)^{\tilde c/\eps^2}\right )$. We are also guaranteed that for every regular vertex $v$ of $H'$, the total number of clusters in $\cset'$ that ever contain $v$ is at most $W^{O(1/\log\log W)}$. 
	
	Recall that initially, $\cset'=\set{H'}$. Consider now some time $t>0$, when an update $\sigma_t\in \Sigma$ for input structure $\iset$ arrives. We initialize $\hat \Sigma_t=\emptyset$, and we start by updating the data structure $\dset(H')$ with the update operation $\sigma_t$,  which may result in some changes to the vertex sets in $\uset'=\set{V(C)\mid C\in \cset'}$. We consider the resulting changes to vertex sets of $\uset'$ one by one. 
		If the change is $\addsupernode(S,u)$, for some vertex set $S\in \uset'$, then we ignore this update.
	If the change is $\delvertex(S,x)$, where a vertex $x\in S$ is deleted from vertex set $S\in \uset'$, and $x$ is a regular vertex, then we add to the sequence $\hat \Sigma_t$ an edge-deletion operation for the edge  $(x,u(C))$, where $C$ is the cluster of $\cset'$ with $V(C)=S$; if $x$ is a supernode then we ignore this change. 
	If the change is $\csplit(S,S')$, where $S\in \uset'$ and $S'\subseteq S$ is a new vertex set that is added to $\uset'$, then we add a supernode split operation to $\hat \Sigma_t$, defined as follows. Let $C$ be the cluster of $\cset'$ with $V(C)=S$. The supernode split operation is performed on  supernode $u(C)$, and the corresponding edge set $E'$ contains all edges $(u(C),v)$, where $v$ is a regular vertex of $H$ lying in $S'$. The new supernode that is added to the graph is $u(C')$, where $C'$ is the new cluster with $V(C')=S'$. Note that the time that is needed in order to compute the set $E'$ of edges is $|S'|$; since the algorithm for the \recdynnc problem on $\iset'$ needs to add vertex set $S'$ to $\uset'$,  this time is subsumed by the time that Algorithm $\alg_1$ takes in order to execute the $\csplit(S,S')$ operation.

Lastly, we consider the update operation $\sigma_t$ itself, and add additional updates to $\hat \Sigma_t$ as follows. If $\sigma_t$ is the deletion of an isolated vertex $x$ from graph $H$, then $x$ must also currently be an isolated vertex of $\hat H$ (as it was just deleted from all clusters containing it). We then add isolated vertex deletion operation for vertex $x$ to $\hat \Sigma_t$. If $\sigma_t$ is the deletion of an edge $e$, then add to $\hat \Sigma_t$ the deletion of $e$.
This completes the description of the algorithm for producing the sequence $\hat \Sigma_t$.

It is immediate to verify that  for all $t>0$, the graph obtained by applying update sequence $(\hat \Sigma_1\circ \cdots\circ \hat \Sigma_t)$ to graph $\hat H^0$ is precisely $\hat H^t$.
Next, we bound the dynamic vertex degree for the resulting dynamic graph $\hat H$. Recall that, from the statement of \Cref{lem: inductive dynamic NC algorithm}, for every regular vertex $v$ of graph $H'$, the total number of clusters in $\cset'$ that ever contain $v$ is bounded by $W^{O(1/\log\log W)}$. Since the inital degree of every regular vertex in $H'$ is at most $2$, and no supernode splitting operations are allowed in $H$, we get that the dynamic degree bound for $\hat H$ is $W^{O(1/\log\log W)}$.

The total update time of this algorithm is subsumed by the total update time of Algorithm $\alg_1$, and is bounded by   $O\left (\tilde c^{i-1}\cdot W^{1+\tilde c\eps}\cdot (\log W)^{\tilde c/\eps^2}\right )$.  
\end{proof}

\paragraph{Distance Preservation.}
We now show that distances between regular vertices in graph $\hat H$ are not much larger than the corresponding distances in graph $H$.

\begin{lemma}\label{lem: distance preservation}
	Throughout the algorithm, for every pair $v,v'\in V(H)$ of regular vertices, if $\dist_H(v,v')\leq D$, then  $\dist_{\hat H}(v,v')\leq 20\cdot \dist_H(v,v')+8D'$.
\end{lemma}

\begin{proof}
	Consider some pair $v,v'\in V(H)$ of regular vertices, with $\dist_H(v,v')\leq D$. Let $P$ be the shortest $v$-$v'$ path in graph $H$, whose length $\ell(P)=\dist_H(v,v')$ is denoted by $\ell^*$. We assume first that $P$ contains at least one long edge, so $\ell^*\geq D'$.
	
	Let $\set{e^1,\ldots,e^k}$ be the set of all long edges that appear on path $P$, indexed in the order of their appearance on $P$. 
	
	For every long edge $e^i$, we denote its endpoints by $v_i$ and $u_i$, where $u_i$ is a supernode. Let $\hat e^i$ be the unique edge on path $P$ that shares the endpoint $u_i$ with $e^i$. For all $1\leq i\leq k$, we delete both $e^i$ and $\hat e^i$ from path $P$. Let $\pset=\set{P_1,\ldots,P_{q}}$ denote the resulting set of subpaths of $P$ (we do not include subpaths consisting of a single vertex), that are indexed in the order of their appearance on path $P$. 
	Notice that, by deleting, for all $1\leq i\leq k$, both edge $e^i$ and $\hat e^i$,  we have ensured that both endpoints of each path in $\pset$ are regular vertices.

	We say that a path $P_i\in \pset$ is \emph{long} iff its length (in graph $H'$) is greater than $4D'$, and we say that it is short otherwise. 

	Consider first some short path $P_i\in \pset$, and let $\hat v,\hat v'$ be its endpoints. Recall that $\cset'$ is a set of clusters maintained as a solution to the \recdynnc on input structure $\iset'$, with graph $H'$ and distance bound $4D'$. Since $\dist_{H'}(v,v')\leq \ell_{H'}(P_i)\leq 4D'$, 
	there is a cluster $C\in \cset'$ containing both $v$ and $v'$. Therefore, there is a vertex $u(C)$ in graph $\hat H$, and edges $(\hat v, u(C)),(\hat v',u(C))$ of length $4D'$ each. We let $Q_i$ be the path obtained by concatenating these two edges. Then $Q_i$ is a path in graph $\hat H$, connecting the endpoints of $P_i$, of length at most $8D'$. 
	
	
	Next, we consider a long path $P_i\in \pset$, whose length must be greater than $4D'$. 
	Let $x_0$ denote the first endpoint of $P_i$, and let $y$ denote its last endpoint. 
	Since the length of every edge in graph $H'$ (and hence on path $P_i$) is at most $D'$, we can find a collection $x_1,x_2,\ldots,x_r$ of {\bf regular} vertices on path $P_i$, that appear on $P_i$ in this order, such that, if we denote $x_{r+1}=y$, and, for all $0\leq j\leq r$, we denote the subpath of $P_i$ from $x_j$ to $x_{j+1}$ by $R_j$, then for all $0\leq j<r$, the length of $R_j$ is at least $D'$ and at most $4D'$, and the length of $R_{r}$ is at most $4D'$.
	
	As before, since, for every regular vertex $v''\in V(H')$, some cluster $C\in \cset'$ must contain $B_{H'}(v'',4D)$, we get that for all $0\leq j\leq r$, there is a vertex $u(C_j)$ in graph $\hat H$, and edges $(x_j,u(C_j))$, $(u(C_j),x_{j+1})$ of length $4D'$ each. By concatenating all such edges, we obtain a path $Q_i$ in graph $\hat H$, connecting $x_0$ to $x_{r+1}$, of length at most $8(r+1)D'$. From our definition of vertices $x_1,\ldots,x_r$, the length of path $P_i$ in $H'$ is at least $rD'$. Since the length of path $P_i$ is greater than $4D'$, we are guaranteed that $r\geq 1$. Therefore, the length of path $Q_i$ in graph $\hat H$ is at most $16\ell_{H'}(P_i)$.
	
	Let $\qset=\set{Q_i\mid P_i\in \pset}$. From the above discussion, the total length of all paths in $\qset$ (in graph $\hat H$) is at most  $16\sum_i\ell_{H'}(P_i)+8|\pset|D'\leq 16\ell_{H}(P)\leq 16\ell^*$ (we have used the fact that $|\pset|\leq k+1$, where $k$ is the number of the long edges on path $P$, and each such long edge has length at least $D'$).

	Lastly, recall that for every long edge $e^i$, we have deleted $e^i$ and $\hat e^i$ from path $P$. The length of edge $e^i$ in graph $\hat H$ is at most $2\ell_{H'}(e^i)$, and the length of edge $\hat e^i$ in graph $\hat H$ is at most $\max\set{D',2\ell_{H'}(\hat e^i)}$. Since the length of $e^i$ in graph $H'$ is at least $D'$, the sum of lengths of both edges in graph $\hat H$ is bounded by $4\ell_{H'}(e^i)+4\ell_{H'}(\hat e^i)$.  Therefore, by concatenating all paths in $\qset$, and edges in $\set{e^i,\hat e^i}_{i=1}^k$, we obtain a path $Q^*$ in graph $\hat H$, connecting $v$ to $v'$, whose length is at most $20\ell^*$.

Assume now that path $P$ contains no long edge. Then we process path $P$ exactly like we processed paths in $\pset$: namely, if the length of $P$ in graph $H'$ is at most $4D'$, then we are guaranteed that there is some vertex $u(C)$, with edges $(v,u(C))$, $(v',u(C))$ of length $4D'$ each in graph $\hat H$, so $\dist_{\hat H}(v,v')\leq 8D'$. Otherwise, we treat $P$ like a long path in $\pset$, obtaining a path of length at most $16\ell_{H'}(P)$ connecting $v$ to $v'$ in graph $\hat H$.
\end{proof}

Next, we show that any path connecting a pair of regular vertices in graph $\hat H$ can be efficiently transformed into a path in graph $H$, connecting the same pair of vertices, without increasing the path length by much.

\begin{claim}\label{claim: path transforming}
	There is a deterministic algorithm, that we call \algtransformpath, that, given a path $P$ in graph $\hat H$, connecting a pair $v,v'$ of regular vertices, computes a path $P'$ in graph $H$, connecting the same pair of vertices, such that $\ell_{H}(P')\leq \alpha_{i-1}\cdot \hat \ell_{\hat H}(v,v')$. The running time of the algorithm is $O(|E(P')|)$.
\end{claim}

\begin{proof}
We process every supernode $u(C)$ on path $P$, with $C\in \cset'$ one by one. Consider any such supernode, and let $v_C,v'_C$ be the regular vertices of $\hat H$ appearing immediately before and immediately after $u(C)$ on path $P$.  We perform query $\spquery(C,v_C,v'_C)$ to the data structure $\dset(H')$, that maintains a solution to \recdynnc problem on graph $H'$, and obtain a path $P_C$, of length at most $\alpha_{i-1}\cdot D'$, connecting $v$ to $v'$ in $H$, in time $O(|E(P_C)|)$. Notice that, since the lengths of the edges $(v_C,u(C)),(v'_C,u(C))$ are $4D'$ each, after this step is completed for every supernode $u(C)$ on path $P$ with $C\in \cset''$, we obtain a path $P'$ in graph $H$, connecting $v$ to $v'$, whose length is at most $\alpha_{i-1}\cdot \hat \ell_{\hat H}(v,v')$. 
\end{proof}

\paragraph{Remainder of the Algorithm.}
Consider now the dynamic graph $\hat H$. 
Recall that the length of every edge in graph $\hat H$ is an integral multiple of  $D'$. 

 Let $\hat H'$ be a graph that is identical to $\hat H$, except that for every edge $e\in E(\hat H)$, we set its new length $\hat \ell'(e)=\hat \ell(e)/D'$ (recall that $\hat \ell(e)$ is the length of $e$ in graph $\hat H$). 
 

We set $\hat D=50D/D'$. Since $D\leq 6W^{\eps i}$, while $D'=\floor{W^{\eps (i-1)}/4}$, we get that $\hat D=\Theta(W^{\eps})$.

We have now defined a valid input structure $\hat \iset'=\left(\hat H',\set{\hat \ell'(e)}_{e\in E(\hat H')}, \hat D \right )$.
Using the algorithm from \Cref{claim: maintain hat H}, we obtain an online sequence $\hat \Sigma=(\hat \Sigma_1\circ \hat \Sigma_2\circ\cdots)$ of valid update operations for graph $\hat H'$, with dynamic degree bound  $\mu\leq W^{O(1/\log\log W)}$; recall that initially, the number of regular vertices in $\hat H$ is at most $W$.
We use the algorithm from \Cref{thm: main dynamic NC algorithm} in order to maintain a solution to the \recdynnc problem on graph $\hat H'$, with distance bound $\hat D$, and approximation factor:

 $$\hat \alpha=(\log(W\mu))^{2^{O(1/\eps)}}\leq (\log W)^{2^{O(1/\eps)}}\leq (\log W)^{2^{\tilde c/\eps}},$$
 
since $\mu \leq W^{O(1/\log\log W)}$, and $\tilde c$ is a large constant.

 Recall that the total update time of this algorithm is:

$$O\left (W^{1+O(\eps)}\cdot \mu^{2+O(\eps)}\cdot \hat D^3\cdot (\log (W\mu))^{O(1/\eps^2)}\right ) \leq O\left (W^{1+O(\eps)}\cdot (\log W)^{O(1/\eps^2)}\right ),$$

since $\hat D=O(W^{\eps})$, $\mu \leq W^{O(1/\log\log W)}$, and $\eps\geq c/\log \log W$.
We denote the corresponding data structure by $\dset(\hat H')$. The corresponding neighborhood cover is denoted by $\hat \cset$, and we denote by $\hat \uset=\set{V(C)\mid C\in \hat \cset}$.
We also denote the algorithm that we have just described, for maintaining the data structure $\dset(\hat H')$, by $\alg_2$.

We are now ready to complete the description of our algorithm for solving the \recdynnc problem on graph $H$. 
Recall that $V(H)\subseteq V(\hat H')$.
The neighborhood cover $\cset$ that we maintain is defined as follows: for every vertex set $S\in \hat \uset$, we let $S'=S\cap V(H)$. We then define $\uset=\set{S'\mid S\in \hat \cset'}$, and we let $\cset$ contain, for each vertex set $S'\in \uset$ the graph $H[S']$. Since we are guaranteed that the collection $\hat \uset$ of vertex subsets of $\hat H'$ only undergoes allowed changes, it is easy to see that so does $\uset$.

Consider any regular vertex $v\in V(H)$ at any point in the algorithm's execution, and vertex set $X=B_H(v,D)$. Recall that, from \Cref{lem: distance preservation}, for every regular vertex $v'\in X$, $\dist_{\hat H}(v,v')\leq 20\cdot \dist_H(v,v')+8D'$, and so $\dist_{\hat H'}(v,v')\leq (20\dist_H(v,v')+8D')/D'\leq \hat D$.
Similarly, if $u\in X$ is a supernode, then there is a regular vertex $v'\in X$, that is a neighbor of $u$, with $\dist_H(v,v')\leq \dist_H(v,u)-\ell_H(v',u)$. 
We then get that $\dist_{\hat H}(v,v')\leq 20\cdot \dist_H(v,v')+8D'$, and $\dist_{\hat H}(v,u)\leq 20\cdot \dist_H(v,v')+9D'+\ell_H(v',u)\leq 20D+9D'$. Therefore, $\dist_{\hat H'}(v,u)\leq (20D+9D')/D'\leq \hat D$.
We conclude that
$X=B_H(v,D)\subseteq B_{\hat H'}(v,\hat D)$. 
Let $\hat C=\coveringcluster(v)$ be the cluster of $\hat \cset$ containing $B_{\hat H'}(v,\hat D')$. Denote $S=V(\hat C)$, and let $S'=S\cap V(H)$ be the corresponding vertex set in $\uset$. Then $B_{H}(v,D)\subseteq S'$. Therefore, if we denote by $C$ the cluster of $\cset$ with $V(C)=S'$, then we can set $\coveringcluster(v)$ in $\cset$ to be $C$. 

Using the data structure $\dset(\hat H')$, it is then immediate to maintain, for every regular vertex $v\in V(H)$, a cluster  $C=\clustercover(v)$ in $\cset$, with $B_H(v,D)\subseteq V(C)$. We can also maintain, for every vertex $x\in V(H)$, a list $\clusterlist(x)\subseteq \cset$ of clusters containing $x$, and for every edge $e\in E(H)$, a list $\clusterlist(e)\subseteq \cset$ of clusters containing $e$, using similar data structures for graph $\hat H'$.

Recall that the algorithm from \Cref{thm: main dynamic NC algorithm}, that we used in order to maintain neighborhood cover $\hat \cset'$ in graph $\hat H'$ ensures that for every regular vertex $v\in V(\hat H')$, the total number of clusters in the neighborhood cover $\hat\cset$ that the algorithm maintains, to which vertex $v$ ever belonged is bounded by $W^{O(1/\log\log W)}$.
Therefore, for every regular vertex $v\in V(H)$, the total number of clusters in the cluster set $\cset$ to which $v$ may ever belong is also bounded by  $W^{O(1/\log\log W)}$.

Lastly, we show an algorithm for responding to queries $\spquery(C,v,v')$, where $C$ is a cluster in $\cset$, and $v,v'\in V(C)$
are regular vertices lying in $C$. Denote $S'=V(C)$, and let $\hat C\in \hat \cset$ be the corresponding cluster (with $S'\subseteq V(\hat C)$). We run query $\spquery(\hat C,v,v')$ in data structure $\dset(\hat H')$, obtaining a path $P$ in graph $\hat H'$, connecting $v$ to $v'$, of length at most $\hat D \cdot \hat \alpha\leq (50D\cdot (\log W)^{2^{\tilde c/\eps}})/D'$ (we have used the fact that $\hat D=50D/D'$ and $\hat \alpha=(\log W)^{2^{\tilde c/\eps}}$). Note that the length of path $P$ in graph $\hat H$ is at most $50D\cdot  (\log W)^{2^{\tilde c/\eps}}\leq D\cdot (\log W)^{2\cdot 2^{\tilde c/\eps}}$. The time required to process query $\spquery(\hat C,v,v')$ in data structure $\dset(\hat H')$ is $O(|E(P)|)$. Lasly, we apply Algorithm  \algtransformpath from \Cref{claim: path transforming} to path $P$ in graph $\hat H$, to obtain a path $P'$ in graph $H$, connecting $v$ to $v'$, whose length is bounded by:

\[
\begin{split}
\alpha_{i-1}\cdot \hat \ell_{\hat H}(v,v')&\leq \alpha_{i-1}\cdot D\cdot(\log W)^{2\cdot 2^{\tilde c/\eps}}\\
&\leq D\cdot (\log W)^{\tilde c (i-1)\cdot 2^{\tilde c/\eps}}\cdot (\log W)^{2\cdot 2^{\tilde c/\eps}}\\
&\leq D\cdot (\log W)^{\tilde c i\cdot 2^{\tilde c/\eps}}\\
&=D\cdot \alpha_i. 
\end{split}
\]

 The running time of the algorithm is $O(|E(P')|)$.


From the above discussion, and the fact that our algorithm supports $\spquery$ queries, it is immediate to verify that, throughout the algorithm, $\cset$ is a weak $(D,\alpha_i\cdot D)$-neighborhood cover of the regular vertices of $H$.

\paragraph{Total Update Time.}
We now bound the total update time of the algorithm. The update time is dominated by the update time of the algorithm from \Cref{claim: maintain hat H}, and the algorithm $\alg_2$. The former has update time $O\left (\tilde c^{i-1}\cdot W^{1+\tilde c\eps}\cdot (\log W)^{\tilde c/\eps^2}\right )$, while the latter has update time $O\left (W^{1+O(\eps)}\cdot (\log W)^{O(1/\eps^2)}\right )$. Since we can assume that $\tilde c$ is a sufficiently large constant, the total update time of the algorithm is bounded by $ \left(\tilde c^{i}\cdot W^{1+\tilde c\eps}\cdot (\log W)^{\tilde c/\eps^2}\right )$.

%% file: proofs-basic-alg.tex
\section{Proofs Omitted from \Cref{sec: proccut}}
\subsection{Proof of Claim \ref{claim: eligible layer}} \label{subsec: proof of proccut eligible layer}

	Assume otherwise. Then for all $1< i< 128\log^4 W$, layer $L_i$ is ineligible. Since the total weight of all vertices in $B_C\left (x,256D\log^4 W\right )$ is at most $W(C)/2$, Condition \ref{cut condition 1: half the weight} holds for all $1< i< \floor{128\log^4W}$. 
	For an index $1< i< \floor{128\log^4W}$, we say that layer $L_i$ is \emph{type-0 ineligible} iff Condition \ref{cut condition 2: small weight} is violated for it, and we say that it is \emph{type-$j$ ineligible}, for some $1\leq j\leq r$, if Condition \ref{cut condition 3: levels} for the index $j$ is violated for it. The following two observations, whose proofs are standard, bound the number of ineligible layers of each of these types.
	
	\begin{observation}\label{obs: few type 1 ineligible}
		The total number of type-0 ineligible layers is at most $128\log^3W$.
	\end{observation}
\begin{proof}
	Assume otherwise. Denote $y=\ceil{128\log^3W}-2$.  Let $1<i_1<i_2<\cdots<i_y$ be indices of $y$ type-0 ineligible layers. From the definition of a type-1 ineligible layer, for all $1\leq q\leq y$:

	\[W_{i_q}>\frac{\sum_{i'<i_q}W_{i'}}{64\log^2W}\geq \frac{\sum_{q'<q}W_{i_{q'}}}{64\log^2W}.\]
	
	Therefore, if we denote, for all $1\leq q\leq y$, $A_q=\sum_{q'<q}W_{i_{q'}}$, then for all $1<q\leq y$, $A_q\geq \left (1+\frac{1}{64\log^2W}\right )A_{q-1}$ must hold. Since the weight of every regular vertex is at least $1$, and each edge of $H$ has length at most $D$, $W_1\geq 1$ holds, and therefore:
	
	\[A_y\geq \left (1+\frac{1}{64\log ^2W }\right )^{y}>W,\]
	
	since $y=\ceil{128\log^3 W}-2$. This is a contradiction, since $W$ is the total weight of all vertices of $H$ at the beginning of the algorithm, and, as the algorithm progresses, this weight may only decrease.
%
%
%
%
%
%
\end{proof}

	\begin{observation}\label{obs: few type 2 ineligible}
	For all $1\leq j\leq r$, the number of type-$j$ ineligible layers is at most $128\log^3W$.
\end{observation}
\begin{proof}
	Assume that the observation is false for some $1\leq j\leq r$. Denote $z=\ceil{128\log^3W}-2$, and  let $1<i_1<i_2<\cdots<i_z$ be indices of $z$ type-$j$ ineligible layers. 
	From the definition of a type-$j$ ineligible layer, for all $1\leq q<z'$:

\[W(S_{\geq j}\cap L_{i_q}) > \frac{\sum_{i'<i_q}W( L_{i'}\cap S_{\geq j})}{64\log^2W }\geq \frac{\sum_{q'=1}^{q-1}W( L_{i_{q'}}\cap S_{\geq j})}{64\log^2W }.\]

	
	Therefore, if we denote, for all $1\leq q\leq y$, $A_q=\sum_{q'<q}W( L_{i_{q'}}\cap S_{\geq j})$, then for all $1<q\leq z$, $A_q\geq \left (1+\frac{1}{64\log^2W}\right )A_{q-1}$ must hold. Moreover, since $i_1$ is a type-$j$ ineligible layer, $W(S_{\geq j}\cap L_{i_1})\geq 1$ must hold. Therefore:
	
	\[A_{z}\geq \left (1+\frac{1}{64\log^2W}\right )^{z-1}>W,\]
	
	since 
	$z=\ceil{128\log^3W}-2$, a contradiction.
\end{proof}

We have shown that for all $0\leq j\leq r$, the total number of type-$j$ ineligible layers is bounded by $128\log^3 W$. Since $r<\log W-2$, we conclude that the total number of ineligible layers $L_i$, for $i>1$, is at most $128\log^4W-2$, a contradiction.

\subsection{Proof of \Cref{claim: cut gives valid neighborhood cover}}\label{subsec: proof of maintain neighborhood cover}
Consider some vertex $y\in V(C)$. Recall that $V(C')=L_1\cup\cdots\cup L_i=B_C(x,2Di)$. We now consider two cases. First, if $y\in B_C(x,2Di-D)$, then clearly $B_C(y,D)\subseteq C'$. Otherwise,  $y\not \in B_C(x,2Di-D)$, and so $B_C(y,D)\cap B_C(x,2Di-2D)=\emptyset$. Since cluster $C''$ contains all vertices of $C$ except those lying in $B_C(x,2Di-2D)=L_1\cup\cdots\cup L_{i-1}$, we get that $B_C(y,D)\subseteq V(C'')$.

\subsection{Proof of \Cref{lem: change in budget after proccut}}\label{subsec: proving bound on budgets}
We start by proving that $\beta'\leq \beta$.
We assume that \proccut chose a layer $L_i$ of $C$ as an eligible layer, and we denote by $C'$ and $C''$ the resulting two clusters. Recall that for every cluster $\hat C$ and for every vertex $y\in \hat C$, the budget of $y$ with respect to $\hat C$ is $\beta_{\hat C}(y)=\left (1+\frac{\log W(\hat C)}{\log^2W}\right )\cdot w(y)$. Therefore, it is easy to verify that for every vertex $y\in V(C'')\setminus L_i$, $\beta_{C''}(y)\leq \beta_C(y)$. We conclude that:

\begin{equation}
\sum_{y\in V(C'')\setminus L_i}\beta_{C''}(y)\leq \sum_{y\in V(C'')\setminus L_i}\beta_{C}(y).\label{eq: first bound}
\end{equation}

Consider now some vertex $y\in V(C')\setminus L_i$. Since, by Condition \ref{cut condition 1: half the weight}, $W(C')\leq W(C)/2$, we get that $\log( W(C'))\leq \log (W(C))-1$. Therefore:

\[\beta_{C'}(y)=  \left (1+\frac{\log (W(C'))}{\log^2W}\right )\cdot w(y) \leq \left (1+\frac{\log (W(C))}{\log^2W}\right )\cdot w(y)-\frac{w(y)}{\log^2W}. \]

We conclude that:

\begin{equation}
\sum_{y\in V(C')\setminus L_i}\beta_{C'}(y)\leq \sum_{y\in V(C')\setminus L_i}\beta_{C}(y)-\frac{W(V(C')\setminus L_i)}{\log^2W}.\label{eq: second bound}
\end{equation}

Consider now some vertex $y\in L_i$. Before $\proccut$ was executed, the copy of vertex $y$ in cluster $C$ contributed $\beta_C(y)$ to the total budget $\beta$. At the end of the procedure, we have created two copies of $y$, one of which lies in $C'$ and another in $C''$. Note that $\beta_{C''}(y)\leq \beta_C(y)$, while $\beta_{C'}(y)\leq 2w(y)$. Therefore, the contribution of these two copies of $y$ to the total budget $\beta'$ is bounded by:

\[ \beta_{C'}(y)+\beta_{C''}(y)\leq 2w(y)+\beta_{C}(y). \]

The total contribution of the vertices of $L_i$ to the new budget $\beta'$ is then bounded by:

\[ 2W(L_i)+\sum_{y\in L_i}\beta_C(y).\]

From Condition \ref{cut condition 2: small weight}, $W(L_i)\leq W(L_1\cup\cdots\cup L_{i-1})/(64\log^2W)=W(V(C')\setminus L_i)/(64\log^2W)$. Therefore, we get that:
\begin{equation}
\sum_{y\in L_i}(\beta_{C'}(y)+\beta_{C''}(y))\leq \frac{W(V(C')\setminus L_i)}{32\log^2W}+\sum_{y\in L_i}\beta_C(y).\label{eq: bound 3}  
\end{equation}

Lastly, by summing up bounds from Equations \ref{eq: first bound}, \ref{eq: second bound} and \ref{eq: bound 3}, we get that:

\[\sum_{y\in V(C')}\beta_{C'}(y)+\sum_{y\in V(C'')}\beta_{C''}(y)\leq  \sum_{y\in V(C)}\beta_C(y)- \frac{W(V(C')\setminus L_i)}{\log^2W}+\frac{W(V(C')\setminus L_i)}{32\log^2W}\leq  \sum_{y\in V(C)}\beta_C(y).\]

Note that the only values $\beta_{\hat C}(y)$ that change are for cluster $\hat C=C$ and vertices $y\in V(C)$. Therefore, we conclude that $\beta'\leq \beta$.

Next, we fix an index $1\leq j\leq r$, and we prove that $\beta'_{\geq j}\leq \beta_{\geq j}+2^{j}(1+1/\log W)W'_j$. In order to do so, we think of the procedure $\proccut$ as consisting of two steps. In the first step, we compute the two clusters $C'$ and $C''$ as before, but we do not update yet the classes of the vertices. That is, if, for some vertex $y$, a new copy of $y$ was introduced due to the procedure, and $y$ needs to move from class $S_a$ to class $S_{a+1}$, we do not move vertex $y$ yet. The total budget $\beta'_{\geq j}$ of the vertices lying in classes $S_{\geq j}$ at the end of this step is calculated with respect to the original classes. We will show that at the end of the first step, $\beta'_{\geq j}\leq \beta_{\geq j}$ holds. Then in the second step we appropriately update the class of each vertex. If $\hat S_j$ is the set of all vertices that are added to class $S_j$ during the second step, then each such vertex $y\in \hat S_j$ has $n_y=2^j$, and, for every cluster $\hat C$ containing $y$, $\beta_{\hat C}(y)=\left (1+\frac{\log W(\hat C)}{\log^2W}\right )\cdot w(y)\leq \left (1+\frac{1}{\log W}\right )\cdot w(y)$. Therefore, the increase in $\beta'_{\geq j}$ due to the second step is bounded by $2^j\cdot (1+1/\log W)\cdot W(\hat S_j)=2^j\cdot (1+1/\log W)W'_j$. In order to complete the proof of the lemma, it is enough to show that, at the end of the first step, $\beta_{\geq j}'\leq \beta_{\geq j'}$ holds. The proof is very similar to the proof of the first part of the lemma. Recall that the layer $L_i$ that was chosen by the procedure is an eligible layer. Therefore, 
$W(S_{\geq j}\cap L_i) \leq \left(\sum_{i'<i}W( L_{i'}\cap S_{\geq j})\right )/\left (64\log^2W\right )$ must hold.

As before, for every vertex $y\in V(C'')\setminus L_i$, $\beta_{C''}(y)\leq \beta_C(y)$ must hold, and so:

\begin{equation}
\sum_{y\in (V(C'')\cap S_{\geq j})\setminus L_i}\beta_{C''}(y)\leq \sum_{y\in (V(C'')\cap S_{\geq j}) \setminus L_i}\beta_{C}(y).\label{eq: first bound 2}
\end{equation}

Consider now some vertex $y\in V(C')\setminus L_i$. Since, by Condition \ref{cut condition 1: half the weight}, $W(C')\leq W(C)/2$, we get that $\log( W(C'))\leq \log( W(C))-1$. Therefore, as before:

\[\beta_{C'}(y)\leq  \left (1+\frac{\log W(C')}{\log^2W}\right )\cdot w(y) \leq \left (1+\frac{\log W(C)}{\log^2W}\right )\cdot w(y)-\frac{w(y)}{\log^2W}. \]

We conclude that:

\begin{equation}
\sum_{y\in (V(C')\cap S_{\geq j})\setminus L_i}\beta_{C'}(y)\leq \sum_{y\in (V(C')\cap S_{\geq j})\setminus L_i}\beta_{C}(y)-\frac{W(V(C'\cap S_{\geq j})\setminus L_i)}{\log^2W}.\label{eq: second bound 2}
\end{equation}

Consider now some vertex $y\in L_i\cap S_{\geq j}$. Before $\proccut$ was executed, the copy of vertex $y$ in cluster $C$ contributed $\beta_C(y)$ to the budget $\beta_{\geq j}$. At the end of the procedure, we have created two copies of $y$, one of which lies in $C'$ and another in $C''$. Exactly as before, the contribution of these two copies of $y$ to the total budget $\beta'$ is bounded by:

\[ \beta_{C'}(y)+\beta_{C''}(y)\leq 2w(y)+\beta_{C}(y). \]

Therefore, the total contribution of the vertices of $L_i\cap S_{\geq j}$ to the new budget $\beta'_{\geq j}$ is bounded by:

\[ 2W(L_i\cap S_{\geq j})+\sum_{y\in L_i\cap S_{\geq j}}\beta_C(y).\]

From Condition \ref{cut condition 3: levels},  
$W(L_i\cap S_{\geq j})\leq \left(\sum_{i'<i}W( L_{i'}\cap S_{\geq j})\right )/\left (64\log^2W\right )=W((C'\cap S_{\geq j})\setminus L_i)/\left ( 64\log^2W\right ) $.

Therefore, we get that:
\begin{equation}
\sum_{y\in L_i\cap S_{\geq j}}(\beta_{C'}(y)+\beta_{C''}(y))\leq W((C'\cap S_{\geq j})\setminus L_i)/\left (32\log^2W\right )+\sum_{y\in L_i\cap S_{\geq j}}\beta_C(y).\label{eq: bound 3 2}  
\end{equation}

Lastly, by summing up bounds from Equations \ref{eq: first bound 2}, \ref{eq: second bound 2} and \ref{eq: bound 3 2}, we get that:

\[
\begin{split}
\sum_{y\in V(C')\cap S_{\geq j}}\beta_{C'}(y)+&\sum_{y\in V(C'')\cap S_{\geq j}}\beta_{C''}(y)\\ & \leq  \sum_{y\in V(C)\cap S_{\geq j}}\beta_C(y)- \frac{W((V(C')\cap S_{\geq j})\setminus L_i)}{\log^2W}+\frac{W((V(C') \cap S_{\geq j})\setminus L_i)}{32\log^2W}\\ &\leq  \sum_{y\in V(C)\cap S_{\geq j}}\beta_C(y).
\end{split}
\]

We conclude that, at the end of the first step, $\beta'_{\geq j}\leq \beta_{\geq j}$ holds.

\input{pseudocut-or-expander}

\input{expander-APSP}

%% file: pseudocut-or-expander.tex
\section{Proofs Omitted from \Cref{sec: pseudocuts}}
\subsection{Proof of \Cref{lem: smaller pseudocut or expander}}
\label{subsec: pseudocut or expander}

We start with the following lemma, whose proof uses standard techniques and is deferred to the following subsection.
\begin{lemma}\label{lem: expander or cut}
	There is a constant $c\geq 1$ and a deterministic algorithm, that, given a graph $G$ with integral lengths $\ell(e)\geq 1$ on its edges $e\in E(G)$, a subset $S$ of its vertices, where $|S|=k$ is an even integer, and parameters $\eta, D'\geq 1$ and $\eps \geq \Omega(1/\log k)$, computes one of the following:
	
	\begin{itemize}
		\item either a graph $X$ with $V(X)\subseteq S$, and $|V(X)|\geq k/2$, such that the maximum vertex degree in $X$ is at most $O(\log k)$, and $X$ is a $\phi$-expander, for $\phi=1/(\log k )^{c/\eps}$, together with an embedding $\pset$ of $X$ into $G$, with congestion at most $\eta$, such that every path in $\pset$ has length at most $D'$; or
	
		\item a subset $E'$ of at most $\frac{ckD'\log^2k}{\eta}$ edges of $G$, and two disjoint subsets $S_1,S_2\subseteq S$ of vertices of cardinality at least $k/(\log k)^{c/\eps}$ each, such that, in graph $G\setminus E'$, $\dist(S_1,S_2)> D'$.
\end{itemize}
The running time of the algorithm is $\otilde\left (k^{1+\eps}(\log k)^{O(1/\eps^2)}\right )+\otilde \left (|E(G)|D'\eta\right )$.
\end{lemma}

We prove \Cref{lem: expander or cut} in \Cref{subsec: prove auxiliary lemma}, after we complete the proof of  \Cref{lem: smaller pseudocut or expander} using it.
We start with the following corollary of \Cref{lem: expander or cut}.

\begin{corollary}\label{cor: expander or terminal partition}
	There is a constant $c'\geq 1$, and a deterministic algorithm, that, given a graph $G$ with non-negative lengths $\ell(e)$ on its edges $e\in E(G)$, a subset $S$ of its vertices, where $|S|\leq k$ for some parameter $k$, and $|S|$ is an even integer, together with parameters $\eta, D''\geq 1$ and $\eps \geq \Omega(1/\log k)$, computes one of the following:
	
	\begin{itemize}
		\item either a graph $X$ with $V(X)\subseteq S$, and $|V(X)|\geq |S|/2$, such that the maximum vertex degree in $X$ is at most $O(\log k)$, and $X$ is a $\phi$-expander, for $\phi=1/(\log k )^{c'/\eps}$, together with an embedding $\pset$ of $X$ into $G$, with congestion at most $\eta$, such that every path in $\pset$ has length at most $64D''\log^2k$; or
		
		\item a subset $E'$ of at most $\frac{c' |S| D''\log^4k}{\eta}$ edges of $G$, and a partition $(S'_0,S'_1,S'_2)$ of $S$ into disjoint subsets, such that $|S'_1|,|S'_2|\geq |S|/(\log k)^{c'/\eps}$, $|S'_0|\leq \frac{\min\set {|S'_1|,|S'_2|}}{8\log k}$, and, in graph $G\setminus E'$, $\dist(S'_1,S'_2)> D''$.
	\end{itemize}
	The running time of the algorithm is $\otilde\left (k^{1+\eps}(\log k)^{O(1/\eps^2)}\right )+\otilde \left (|E(G)|D''\eta\right )$.
\end{corollary}

\begin{proof}
	We assume that $c'$ is a large enough constant, and in particular $c'>c$, where $c$ is the constant from the statement of \Cref{lem: expander or cut}. We set $D'= 64 D''\log^2k$, and apply the algorithm from \Cref{lem: expander or cut} to graph $G$, with the new parameter $D'$, and parameters $\eta,\eps$ that remain unchanged. We now consider two cases. In the first case, the outcome of the algorithm from \Cref{lem: expander or cut} is a graph $X$ with $V(X)\subseteq S$, and $|V(X)|\geq |S|/2$, such that the maximum vertex degree in $X$ is at most $O(\log |S|)\leq O(\log k)$, and $X$ is a $\phi$-expander, for $\phi=1/(\log |S| )^{c/\eps}\geq 1/(\log k)^{c'/\eps}$, together with an embedding $\pset$ of $X$ into $G$, with congestion at most $\eta$, such that every path in $\pset$ has length at most $D'= 64 D''\log^2k$. In this case, we terminate the algorithm, and return the graph $X$ and its embedding $\pset$ as the algorithm's outcome.
	
	From now on we assume that the second case happened, and the algorithm from \Cref{lem: expander or cut} returned a subset $E'$ of at most $\frac{c |S| D'\log^2k}{\eta}\leq \frac{c'|S|D''\log^4k}{\eta}$ edges of $G$, and two disjoint subsets $S_1,S_2\subseteq S$ of vertices of cardinality at least $|S|/(\log |S|)^{c/\eps}\geq |S|/(\log k)^{c'/\eps}$ each, such that, in graph $G\setminus E'$, $\dist(S_1,S_2)> D'$.
	
	In the remainder of the algorithm, we will compute the desired partition $(S'_0,S'_1,S'_2)$ by using standard ball-growing technique (somewhat similar to the one employed in Procedure \proccut from \Cref{subsec: proccut definition}).
	
	We let $L_0^1=S_1$.
	For $i\geq 1$, we define the $i$th layer for $S_1$ as $L_i^1=B_{G\setminus E'}(S_1,2iD'')\setminus B_{G\setminus E'}(S_1,2(i-1)D'')$. We denote by $K^1_i=|S\cap L^1_i|$, and by $J^1_i=K^1_0+\cdots+K^1_i$, where $K^1_0=|S_1|$.
	
	For $i\geq 1$, we say that layer $L_i^1$ is \emph{eligible} iff:
	
	\begin{itemize}
		\item $K^1_i\leq J^1_{i-1}/(8\log k)$; and
		\item $J^1_i\leq |S|/2$.
	\end{itemize}
	
	Similarly, we let $L_0^2=S_2$, and,	for $i\geq 1$, we define the $i$th layer for $S_2$ as $L_i^2=B_{G\setminus E'}(S_2,2iD'')\setminus B_{G\setminus E'}(S_2,2(i-1)D'')$. We denote by $K^2_i=|S\cap L^2_i|$, and by $J^2_i=K^2_0+\cdots+K^2_i$, where $K^2_0=|S_2|$. As before, for $i\geq 1$, we say that layer $L_i^2$ is \emph{eligible} iff:
	
	\begin{itemize}
		\item $K^2_i\leq J^2_{i-1}/(8\log k)$; and
		\item $J^2_i\leq |S|/2$.
	\end{itemize}

We need the following simple observation.

\begin{observation}\label{obs: eligible layer}
	There is an integer $1\leq i\leq 8\log^2k$, such that either $L_i^1$ or $L_i^2$ is an eligible layer.
\end{observation}
\begin{proof}
	Recall that $\dist_{G\setminus E'}(S_1,S_2)> D'$, and, since  $D'= 64 D''\log^2k$, the vertex sets  $B_{G\setminus E'}(S_1,32D''\log^2k)$ and $B_{G\setminus E'}(S_2,32D''\log^2k)$ are disjoint. Let $i'=\floor{16\log^2k}$. Then either $J^1_{i'}\leq |S|/2$ or $J^2_{i'}\leq |S|/2$ must hold. Assume without loss of generality that it is the former. We now prove that there must be some integer $1\leq i\leq i'$, such that layer $L_i^1$ is eligible.
	
	Assume otherwise. Note that for all $1\leq i\leq i'$, $J^1_i\leq |S|/2$ must hold. Since we have assumed that layer $L_i^1$ is ineligible, $K^1_i> J^1_{i-1}/(8\log k)$ must hold, and so $J^1_i\geq J^1_{i-1}(1+1/(8\log k))$. But then we get that: 
	$$J^1_{i'}\geq |S_1|\cdot \left (1+\frac 1 {8\log k}\right )^{i'}\geq |S_1|\cdot  \left (1+\frac 1 {8\log k}\right )^{\floor{16\log^2k}}\geq k,$$ a contradiction.
\end{proof}
	
We run two BFS searches in graph $G\setminus E'$ in parallel, one starting from $S_1$, and another starting from $S_2$, so that at any time point, both searches discover the same number of edges, until we encounter the first index $1\leq i\leq 8\log^2k$, such that either $L_i^1$ or $L_i^2$ is an eligible layer. Assume without loss of generality that $L_i^1$ is an eligible layer. We then let $S_0'$ contain all vertices of $S\cap L_i^1$, $S_1'$ contain all vertices of $S$ lying in $B_{G\setminus E'}(S_1,2(i-1)D'')$, and $S_2'$ contain all remaining vertices of $S$. Clearly, $|S_1'|\leq |S_2'|$, and, from the definition of an eligible layer, $|S_0'|\leq |S_1'|/(8\log k)$. Moreover, $S_1\subseteq S_1'$, so $|S_1'|\geq |S|/(\log k)^{c'/\eps}$, as required.
Since $S_1'\subseteq B_{G\setminus E'}(S_1,2(i-1)D'')$, while $S_2'\cap B_{G\setminus E'}(S_1,2iD'')=\emptyset$, we get that $\dist_{G\setminus E'}(S_1',S_2')> D''$. It now remains to analyze the running time of the algorithm. The running time for the algorithm from \Cref{lem: expander or cut}  is  $\otilde\left (|S|^{1+\eps}(\log k)^{O(1/\eps^2)}\right )+\otilde \left (|E(G)|D'\eta\right )= \otilde\left (|S|^{1+\eps}(\log k)^{O(1/\eps^2)}\right )+\otilde \left (|E(G)|D''\eta\right )$. The additional running time required to compute an eligible layer and the vertex sets $S_0',S_1',S_2'$ is bounded by $O(|E(G)|)$. therefore, the total running time of the algorithm is $\otilde\left (|S|^{1+\eps}(\log k)^{O(1/\eps^2)}\right )+\otilde \left (|E(G)|D''\eta\right )$.
\end{proof}

We need the following, somewhat stronger corollary, that allows us to either compute an expander $X$ and its embedding into $G$ as before, or to compute a number of subsets of the set $S$ that are far from each other in $G\setminus E'$, but now we additionally guarantee that the cardinality of each subset is significantly smaller than the cardinality of the initial set $S$; as before, we also compute a small subset $S_0$ of vertices of $S$ that we will discard.



\begin{corollary}\label{cor: expander or terminal partition2}
	There is a constant $c''\geq 1$, and a deterministic algorithm, that, given a graph $G$ with non-negative lengths $\ell(e)$ on its edges $e\in E(G)$, a subset $S$ of its vertices, where $|S|\leq k$ for some parameter $k$, together with parameters $\eta, D''\geq 1$ and $\eps \geq \Omega(1/\log k)$, such that $|S|\geq 40(\log k)^{c''/\eps}$ holds, computes one of the following:
	
	\begin{itemize}
		\item either a graph $X$ with $V(X)\subseteq S$, and $|V(X)|\geq |S|/4$, such that the maximum vertex degree in $X$ is at most $O(\log k)$, and $X$ is a $\phi$-expander, for $\phi=1/(\log k )^{c''/\eps}$, together with an embedding $\pset$ of $X$ into $G$, with congestion at most $\eta$, such that every path in $\pset$ has length at most $64D''\log^2k$; or
		
		\item a subset $E'$ of at most $\frac{|S| D''(\log k)^{c''/\eps}}{\eta}$ edges of $G$, and a partition $(S''_0,S''_1,S''_2,\ldots,S''_q)$ of $S$ into disjoint subsets, such that, for all $1\leq i\leq q$, $|S|/(\log k)^{c''/\eps}\leq |S''_i|\leq |S|/2$, $|S''_0|\leq |S|/(4\log k)$, and, in graph $G\setminus E'$, for all $1\leq i<i'\leq q$, $\dist(S''_i,S''_{i'})> D''$.
	\end{itemize}
	The running time of the algorithm is $\otilde\left (\left (k^{1+\eps}(\log k)^{O(1/\eps^2)}+|E(G)|D''\eta\right )\cdot (\log k)^{O(1/\eps)} \right )$.
\end{corollary}

\begin{proof}
	We assume that the constant $c''$ is large enough, so that $c''\geq c'$, where $c'$ is the constant from \Cref{cor: expander or terminal partition}. 
	The algorithm consists of at most $(\log k)^{c''/\eps}$ iterations. Over the course of the algorithm, we maintain a subset $S_0''$ of $S$, another subset $S^*$ of $S\setminus S_0''$, and a partition $\sset$ of $S\setminus (S_0''\cup S^*)$. We also maintain a collection $E'\subseteq E(G)$ of edges. We ensure that the following invariants hold throughout the algorithm:
	
	\begin{itemize}
		\item $|S_0''|\leq 1+\sum_{S'\in \sset}|S'|/(4\log k)$;
		\item $|S^*|> |S|/2$, and $|S^*|$ is even;
		\item For all $S',S''\in \sset$ with $S'\neq S''$, $\dist_{G\setminus E'}(S',S'')> D''$; 
		\item For all $S'\in \sset$, $\dist_{G\setminus E'}(S',S^*)> D''$; and
		\item For all $S'\in \sset$, $|S|/(\log k)^{c''/\eps}\leq |S'|\leq |S|/2$.
	\end{itemize}
	
	At the beginning of the algorithm, if $|S|$ is even, then we set $S_0''=\emptyset$, and $S^*=S$; otherwise, we let $S_0''$ contain an arbitrary vertex $s$ of $S$, and we set $S^*=S\setminus\set{s}$. We also set $\sset=\emptyset$, and $E'=\emptyset$. It is easy to verify that all invariants hold.
	
	We now describe an execution of an iteration. We apply the algorithm from \Cref{cor: expander or terminal partition} to graph $G\setminus E'$, vertex set $S^*$, and parameters $k,\eta,D''$ and $\eps$. We now consider two cases.
	
	In the first case, the algorithm returns  a graph $X$ with $V(X)\subseteq S^*$, and $|V(X)|\geq |S^*|/2\geq |S|/4$, such that the maximum vertex degree in $X$ is at most $O(\log k)$, and $X$ is a $\phi$-expander, for $\phi=1/(\log k )^{c'/\eps}\geq1/(\log k )^{c''/\eps} $, together with an embedding $\pset$ of $X$ into $G$, with congestion at most $\eta$, such that every path in $\pset$ has length at most $64D''\log^2k$. We then terminate the algorithm, and return the expander $X$ and its embedding.
	
	In the second case, the algorithm  from \Cref{cor: expander or terminal partition} returns a subset $E^*$ of at most $\frac{c' |S| D''\log^4k}{\eta}$ edges of $G$, and a partition $(S'_0,S'_1,S'_2)$ of $S^*$ into disjoint subsets. We assume without loss of generality that $|S'_1|\leq |S_2'|$, so in particular $|S_1'|\leq |S|/2$. Recall that the algorithm guarantees that $|S'_1|\geq |S^*|/(\log k)^{c'/\eps}\geq |S|/(\log k)^{c''/\eps}$, and that $|S'_0|\leq |S_1'|/(8\log k)$. Additionally, in graph $G\setminus (E'\cup E^*)$, $\dist(S'_1,S'_2)> D''$.
    If $|S_2'|$ is odd, then we let $s'\in S_2'$ be any vertex, and we move $s'$ from $S_2'$ to $S_0'$.	Notice that, since $|S^*|\geq |S|/2\geq 20(\log k)^{c''/\eps}$, while 
$|S'_1|\geq |S^*|/(\log k)^{c'/\eps}$, $|S_0'|\leq  |S_1'|/(4\log k)$ still holds.

We add the edges of $E^*$ to $E'$, the vertices of $S_0'$ to $S_0''$, and the set $S_1'$ to $\sset$. We also set $S^*=S_2'$. Since $|S_0''|\leq 1+\sum_{S'\in \sset}|S'|/(4\log k)$ held at the beginning of the current iteration, and  $|S_0'|\leq  |S_1'|/(4\log k)$, it is immediate to verify that $|S_0''|\leq 1+\sum_{S'\in \sset}|S'|/(4\log k)$ continues to hold.

It is also easy to verify that  for all $S',S''\in \sset$ with $S'\neq S''$, $\dist_{G\setminus E'}(S',S'')> D''$ continues to hold. Indeed, if $S',S''\neq S_1'$, then $\dist_{G\setminus E'}(S',S'')\geq D''$ held at the beginning of the current iterations, and, since we have only added edges to $E'$, the inequality continues to hold at the end of the iteration. Otherwise, if, for example, $S''=S_1'$, then, since $\dist_{G\setminus E'}(S',S^*)\geq D''$ held at the beginning of the current iteration, and $S''\subseteq S^*$, $\dist_{G\setminus E'}(S',S'')\geq D''$ holds at the end of the current iteration. Using similar reasonings, it is easy to verify that,  for all $S'\in \sset$, $\dist_{G\setminus E'}(S',S^*)\geq D''$ holds at the end of the current iteration.
Finally, since 	$ |S|/(\log k)^{c''/\eps}\leq |S_1'|\leq |S|/2$, the last invariant also continues to hold. If $|S^*|> |S|/2$ also continues to hold, then all invariants hold, and we continue to the next iteration. 

Assume now that $|S^*|\leq |S|/2$. In this case, we terminate the algorithm, and we return the current set $E'$ of edges, set $S''_0$ of vertices of $S$, and we let $S''_1,\ldots,S''_q$ be all subsets of $S$ lying in $\sset$, together with the set $S^*$. Note that we are guaranteed that $|S^*|\geq |S|/4$, and, from the invariants, it is easy to verify that we obtain a valid output (except for the bound on $|E'|$ that we establish below).
 Note that the cardinality of the set $S^*$ decreases by at least factor $(1-1/(\log k)^{c'/\eps})$ in each iteration, so the number of iterations in the algorithm is bounded by $O\left ((\log k)^{c'/\eps+1}\right )\leq (\log k)^{4c'/\eps}$. Since the cardinality of edge set $E'$ increases by at most $\frac{c' |S| D''\log^4k}{\eta}$ in each iteration, we get that, at the end of the algorithm, $|E'|\leq \frac{|S|D''(\log k)^{c''/\eps}}{\eta}$. Lastly, since the running time for each iteration is $\otilde\left (k^{1+\eps}(\log k)^{O(1/\eps^2)}\right )+\otilde \left (|E(G)|D''\eta\right )$, the total running time of the algorithm is bounded by $\otilde\left (\left (k^{1+\eps}(\log k)^{O(1/\eps^2)}+|E(G)|D''\eta\right )\cdot (\log k)^{O(1/\eps)} \right )$.
\end{proof}

Lastly, the following corollary of \Cref{cor: expander or terminal partition2} allows us to either compute an expander $X$ and embed it into $\pset$ as before, or to compute a large collection of subsets of the input terminal set $T$ that are far enough from each other.

\begin{corollary}\label{cor: expander or terminal partition many sets}
	There is a constant $c'''$, and a deterministic algorithm, that, given a graph $G$ with non-negative lengths $\ell(e)$ on its edges $e\in E(G)$, a subset $T$ of its vertices called terminals, with  $|T|= k$ for some parameter $k$, together with parameters $\eta, D'',h\geq 1$ and $\eps \geq \Omega(1/\log k)$, 
	such that $ h\leq k/ (100(\log k)^{c'''/\eps})$,
	computes one of the following:
	
	\begin{itemize}
		\item either a graph $X$ with $V(X)\subseteq T$, and $|V(X)|\geq k/(64h)$, such that the maximum vertex degree in $X$ is at most $O(\log k)$, and $X$ is a $\phi$-expander, for $\phi=1/(\log k )^{c'''/\eps}$, together with an embedding $\pset$ of $X$ into $G$, with congestion at most $\eta$, such that every path in $\pset$ has length at most $64D''\log^2k$; or
		
		\item a subset $E'$ of at most  $\frac{k D''(\log k)^{c'''/\eps}}{\eta}$ edges of $G$, and a collection $(T_1,\ldots,T_h)$ of disjoint subsets or $T$ of cardinality at least $k/(h(\log k)^{c'''/\eps})$, such that for all $1\leq i<i'\leq h$, in graph $G\setminus E'$, $\dist(T_i,T_{i'})> D''$.
	\end{itemize}
	The running time of the algorithm is  $\otilde\left (\left (k^{1+\eps}(\log k)^{O(1/\eps^2)}+|E(G)|D''\eta\right )\cdot h(\log k)^{O(1/\eps)} \right )$.
\end{corollary}

\begin{proof}
	The proof easily follows by iteratively applying the algorithm from \Cref{cor: expander or terminal partition2}.
	We assume that $c'''\gg c''$, where $c''$ is the constant from \Cref{cor: expander or terminal partition2}.
	The algorithm consists of $r\leq \ceil{\log h}+2$ phases. For all $1\leq j\leq r$, at the beginning of the $j$th phase, we are given a partition $\Pi^j=\set{T_0^j,T_1^j,\ldots,T_{q_j}^j}$ of the set $T$ of terminals, and a set $E^j\subseteq E(G)$ of edges, such that the following hold:
	
	\begin{itemize}
		\item $|T_0^j|\leq jk/(4\log k)$;
		\item $|E^j|\leq \frac{kj D''(\log k)^{c''/\eps}}{\eta}$;
		\item for all $1\leq i\leq q_j$, $\frac{k}{2^{j-1}\cdot (\log k)^{c''/\eps}}\leq |T_i^j|\leq k/2^{j-1}$; and
		\item for all $1\leq i<i'\leq q_j$, $\dist_{G\setminus E^j}(T_i^j,T_{i'}^j)>D''$.
	\end{itemize}

At the beginning of the algorithm, we set $T_0^1=\emptyset$, $T_1^1=T$, and $q_1=1$. We also set $E^1=\emptyset$. It is easy to verify that all invariants hold.

We now describe the $j$th phase. At the beginning of the phase, we set $T_0^{j+1}=T_0^j$, $E^{j+1}=E^j$, and $\Pi^{j+1}=\emptyset$. We then perform  $q_j$ iterations, where in the $i$th iteration we process terminal set $T_i^j\in \Pi^j$.

We now describe an iteration for processing the terminal set $T_i^j$. If $|T_i^j|\leq k/2^j$, then we simply add $T_i^j$ to $\Pi^{j+1}$, set $T_0^{i,j+1}=\emptyset$, $E_i^{j+1}=\emptyset$, and continue to the next iteration. Assume now that $|T_i^j|> k/2^j$. Notice that in this case, from our assumption that $h\leq k/ (100(\log k)^{c'''/\eps})$, we get that $|T_i^j|\geq k/(4h)\geq 40(\log k)^{c''/\eps}$ holds.
 We then apply the algorithm from \Cref{cor: expander or terminal partition2} to graph $G\setminus E^j$, vertex set $T^j_i$, and the same parameters $k,\eta,D'',\eps$ as before. We now consider two cases. In the first case, the outcome of the algorithm from \Cref{cor: expander or terminal partition2} is 
 a graph $X$ with $V(X)\subseteq T_i^j$, and $|V(X)|\geq |T_i^j|/4\geq k/2^{j+2}\geq k/(64h)$, such that the maximum vertex degree in $X$ is at most $O(\log k)$, and $X$ is a $\phi$-expander, for $\phi=1/(\log k )^{c''/\eps}\geq 1/(\log k )^{c'''/\eps}$, together with an embedding $\pset$ of $X$ into $G$, with congestion at most $\eta$, such that every path in $\pset$ has length at most $64D''\log^2k$. We then terminate the algorithm, and return the expander $X$ and its embedding $\pset$. 
 
Therefore, we assume from now on that the second case happens, and the algorithm from \Cref{cor: expander or terminal partition2} returns 
a subset $E_i^{j+1}$ of at most $\frac{|T_i^j| D''(\log k)^{c''/\eps}}{\eta}$ edges of $G\setminus E^j$, and a partition $(S''_0,S''_1,S''_2,\ldots,S''_r)$ of $T_i^j$ into disjoint subsets, such that, for all $1\leq a\leq z$, $|T_i^j|/(\log k)^{c''/\eps}\leq |S''_a|\leq |T_i^j|/2$, $|S''_0|\leq |T^j_i|/(4\log k)$, and, in graph $G\setminus (E^j\cup E_i^{j+1})$, for all $1\leq i<i'\leq z$, $\dist(S''_i,S''_{i'})> D''$.

We set $T_0^{i,j+1}=S''_0$; note that $|T_0^{i,j+1}|\leq |T^j_i|/(4\log k)$. Since $k/2^j\leq |T_i^j|\leq k/2^{j-1}$, we get that for all $1\leq a\leq z$, $k/(2^j(\log k)^{c''/\eps})\leq |S''_a|\leq k/2^{j}$. We then add all sets $S_1'',\ldots,S''_z$ to $\Pi^{j+1}$, and continue to the next iteration.

The phase terminates when all sets in $\Pi^j$ are processed. We set $E^{j+1}=E^j\cup (E_1^{j+1}\cup\cdots\cup E^{j+1}_{q_j})$. Since, for all $1\leq i\leq q_j$, $|E_i^{j+1}|\leq \frac{|T_i^j| D''(\log k)^{c''/\eps}}{\eta}$, we get that $|E_1^{j+1}\cup\cdots\cup E^{j+1}_{q_j}|\leq \frac{k D''(\log k)^{c''/\eps}}{\eta}$, and so altogether, $|E^{j+1}|\leq \frac{k(j+1) D''(\log k)^{c''/\eps}}{\eta}$. 
We also set $T_0^{j+1}=T_0^j\cup (T_0^{1,j+1}\cup\cdots\cup T_0^{q_j,j+1})$.
Since, for all $1\leq i\leq q_j$, $|T_0^{i,j+1}|\leq |T^j_i|/(4\log k)$, we get that  $|T_0^{1,j+1}\cup\cdots\cup T_0^{q_j,j+1}|\leq k/(4\log k)$, and $|T_0^{j+1}|\leq k(j+1)/(4\log k)$.
From the above discussion, for all $T'\in \Pi^{j+1}$, $\frac{k}{2^{j}\cdot (\log k)^{c''/\eps}}\leq |T'|\leq k/2^{j}$ holds, and it is immediate to verify that for all $T',T''\in \Pi^{j+1}$, $\dist_{G\setminus E^{j+1}}(T',T'')>D''$.

If the algorithm does not terminate with an expander $X$ and its embedding $\pset$, then it terminates after $r=\ceil{\log h}+2$ phases. Since we are guaranteed that, for all $T'\in \Pi^{r}$, $|T'|\leq k/2^{r}\leq k/(2h)$, while $|T_0^r|\leq rk/(4\log k)\leq k/2$  we get that $|\Pi^{r}|\geq h$. We let $T_1,\ldots,T_h$ be arbitrary $h$ vertex sets in $\Pi^r$. 
We let $E'=E^r$; from our invariants, $|E'|\leq \frac{k D''(\log k)^{c''/\eps+1}}{\eta}\leq \frac{k D''(\log k)^{c'''/\eps}}{\eta}$. Clearly, for all $1\leq i<i'\leq h$, in graph $G\setminus E'$, $\dist(T_i,T_{i'})> D''$. Our invariants also guarantee that, for all $1\leq i\leq h$, $|T_i|\geq \frac{k}{4h\cdot (\log k)^{c''/\eps}}\geq  \frac k {(h(\log k)^{c'''/\eps})}$.

It now remains to analyze the running time of the algorithm. The algorithm consists of $O(\log k)$ phases, and each phase consists of $O(h)$ iterations. From \Cref{cor: expander or terminal partition2}, the running time of each iteration is bounded 	by $\otilde\left (\left (k^{1+\eps}(\log k)^{O(1/\eps^2)}+|E(G)|D''\eta\right )\cdot (\log k)^{O(1/\eps)} \right )$. Therefore, the total running time of the algorithm is bounded by $\otilde\left (\left (k^{1+\eps}(\log k)^{O(1/\eps^2)}+|E(G)|D''\eta\right )\cdot h(\log k)^{O(1/\eps)} \right )$.
\end{proof}

We are now ready to complete the proof of \Cref{lem: smaller pseudocut or expander}. Recall that we are given as input  a valid input structure $\iset=\left(C,\set{\ell(e)}_{e\in E(C)}, D\right )$,  where $C$ is a connected graph, with arbitrary weights $w(x)\geq 0$ for vertices $x\in V(C)$, and parameters $0<\epsx<1$, $\hat D>D$, and $\hat W\geq \max\set{W(C),|E(C)|,2^{\Omega(1/\epsx)}}$, together with a $(\hat D,\rho)$-pseudocut $\hat E$  for $C$ of cardinality $k$, where $\rho=\hat W^{\eps}$.

We need to consider a special case, where $\epsx<O(1/\log k)$, or, equivalently, $k<2^{O(1/\epsx)}$. In this case, since we have assumed that $\hat W>2^{\Omega(1/\epsx)}$, $k<\hat W^{\epsx}$ holds. We then select an arbitrary edge $e\in \hat E$, and set $\hat E^*=\set{e}$. We let the expander $X$ contain a single vertex $t_e$ and no edges, and its embedding into $C_{|\hat E^*}$ is $\pset=\emptyset$. We then return $\hat E^*,X$, and $\pset$ and terminate the algorithm. Therefore, we assume from now on that $\epsx\geq \Omega(1/\log k)$.
We also assume that $k/ (100(\log k)^{c'''/\eps})\geq \ceil{\hat W^{\eps}}$, where $c'''$ is the constant from \Cref{cor: expander or terminal partition many sets}; otherwise, $k\leq \hat W^{2\eps}$ holds. In this case, we proceed as before: we select an arbitrary edge $e\in \hat E$, set $\hat E^*=\set{e}$, and return an expander $X$ containing a single vertex $t_e$ and no edges, and its embedding $\pset=\emptyset$  into $C_{|\hat E^*}$. Therefore, we assume from now on that $k/ (100(\log k)^{c'''/\eps}\geq \ceil{\hat W^{\eps}}$.

Let $G=C_{|\hat E}$ be the graph obtained from $C$ by subdividing every edge $e=(u,v)\in \hat E$ with a vertex $t_e$. We denote  the resulting edges by $e_1=(u,t_e)$, $e_2=(v,t_e)$, and we set the length of each of these edges to be $\ell(e)$. Let $T=\set{t_e\mid e\in \hat E}$ be the set of the newly created vertices, that we refer to as \emph{terminals}. Note that $k\leq \hat W$ must hold, since $C$ is a connected graph and $\hat W\geq |E(C)|$.

We use the following parameters: $D''=40\hat D$, $h=\ceil{\hat W^{\eps}}$, and $\eta=16D''\hat W^{\eps} (\log k)^{2c'''/\eps}\leq \hat D\cdot \hat W^{\eps} (\log \hat W)^{O(1/\eps)}$.
Note that, from our assumption, $h\leq k/ (100(\log k)^{c'''/\eps})$ holds. 
We apply the algorithm from \Cref{cor: expander or terminal partition many sets} to graph $G$, the set $T$ of terminals, parameters $\eta,D''$ and $h$ that we just defined, and the input parameter $\eps$. We now consider two cases, depending on the algorithm's outcome.

In the first, case, the algorithm from \Cref{cor: expander or terminal partition many sets} returns a graph $X$ with $V(X)\subseteq T$, and $|V(X)|\geq k/(64h)\geq \Omega(k/\hat W^{\eps})$, such that the maximum vertex degree in $X$ is at most $O(\log k)\leq O(\log \hat W)$, and $X$ is a $\phi^*$-expander, for $\phi^*=1/(\log k )^{O(1/\eps)}\geq 1/(\log \hat W )^{O(1/\eps)}$, together with an embedding $\pset$ of $X$ into $G$, with congestion at most $\eta\leq \hat D\cdot \hat W^{\eps} (\log \hat W)^{O(1/\eps)}$, such that every path in $\pset$ has length at most $64D''\log^2k\leq O(\hat D\log^2\hat W)$.
We let $\hat E^*\subseteq \hat E$ be the set of edges corresponding to vertices of $X$, that is, $\hat E^*=\set{e\mid t_e\in V(X)}$. Then $|\hat E^*|=|V(X)|\geq \Omega(k/\hat W^{\eps})$. Moreover, the set $\pset$ of paths that embeds $X$ into $G$ immediately defines a set $\pset'$ of paths in graph $C_{|\hat E^*}$, embedding $X$ into this graph with congestion at most $\eta\leq \hat D\cdot \hat W^{\eps} (\log \hat W)^{O(1/\eps)}$, such that the length of every path in $\pset$ is at most $ O(\hat D\log^2\hat W)$ (note that graph $G$ can be obtained from $C_{|\hat E^*}$ by subdividing each edge in $\hat E\setminus \hat E^*$, so the transformation from path set $\pset$ to $\pset'$ is straightforward). We then return the expander $X$, the edge set $\hat E^*$, and the embedding $\pset'$ of $X$ into $C_{|E^*}$ as the outcome of the algorithm.

From now on we focus on the second case, where the algorithm from \Cref{cor: expander or terminal partition many sets} returned a subset $E'$ of edges, of cardinality at most:

\[\frac{k D''(\log k)^{c'''/\eps}}{\eta}\leq \frac{k}{16\hat W^{\eps} (\log k)^{c'''/\eps}}\] 

(since $\eta=16D''\hat W^{\eps} (\log k)^{2c'''/\eps}$),
 together with a collection $(T_1,\ldots,T_h)$ of disjoint subsets or $T$, such that the cardinality of each subset $T_i$ is at least:
 
 \[\frac{k}{h(\log k)^{c'''/\eps}}\geq \frac{k}{2\hat W^{\eps} (\log k)^{c'''/\eps}}\geq 8|E'|.\] 

Recall that we are guaranteed that, for all $1\leq i<i'\leq h$, $\dist_{G\setminus E'}(T_i,T_{i'})> D''=40\hat D$.
%
In the remainder of the proof, we will exploit the terminal sets $T_1,\ldots, T_h$, and the edge set $E'$, in order to construct a $(\hat D,\rho)$-pseudocut $\hat E'$ in $C$, with $|\hat E'|\leq |\hat E|\left (1-\frac{1}{\hat W^{\eps} (\log \hat W)^{O(1/\eps)}}\right )$.

We define a set $E''$ of edges in graph $C$, as follows. For every edge $e\in E'$, if $e\in E(C)$, then we add $e$ to $E''$. Otherwise, $e$ was obtained by subdividing some edge $e'\in \hat E$. We then add $e'$ to $E''$. Clearly, $|E''|\leq |E'|$.
For all $1\leq i\leq h$, we denote $\hat E_i=\set{e\in \hat E\setminus E''\mid t_e\in T_i}$, and we let $Z_i$ the set of all vertices of $C$ that serve as endpoints of the edges in $\hat E_i$. 
Note that $|\hat E_i|\geq |T_i|-|E''|\geq |T_i|/2\geq 4|E''|$.
 We need the following easy observation.

\begin{observation}\label{obs: vertex sets far from each other}
	Let $C'=C\setminus E''$. Then for all $1\leq i<j\leq h$, $\dist_{C'}(Z_i,Z_j)> 4\hat D$.
\end{observation}
\begin{proof}
	Assume otherwise, and let $x\in Z_i,y\in Z_j$ be any pair of vertices with $\dist_{C'}(x,y)\leq 4\hat D$, and $i\neq j$. Let $P$ be a shortest $x$-$y$ path in graph $C'$. 
	Let $e\in \hat E_i$ be the edge with endpoint $x$, and define $e'\in \hat E_j$ similarly. We claim that there is a path $P'$ of length at most $10\hat D$ connecting $t_e$ to $t_{e'}$ in graph $G\setminus E'$. Since $D''=40\hat D$, this will contradict the fact that $\dist_{G\setminus E'}(T_i,T_{j})> D''$.
	
	In order to obtain the path $P$, we start with the path $P'$; for every edge $\hat e\in E(P)$, if $\hat e\not \in E(G)$, then we subdivide $\hat e$ to obtain a pair of edges that lie in $G$. Notice that this increases the length of the path by at most factor $2$. We then append the edges $(t_e,x)$ and $(y,t_{e'})$ to the beginning and the end of the resulting path, respectively. Since the length of every edge in $C$ is at most $D\leq \hat D$, it is easy to verify that the length of $P'$ is at most $10\hat D$. From our construction of edge set $E''$, and edge sets $\hat E_1,\ldots,\hat E_h$, it is also immediate to verify that no edge of $P'$ lies in $E'$.
\end{proof}

For all $1\leq i\leq h$, we let $B_i$ the set of all vertices of $C$ that can reach a vertex in $Z_i$ via a path of length at most $2\hat D$ in graph $C\setminus E''$. Equivalently: $B_i=B_{C\setminus E''}(Z_i,2\hat D)$. Note that, from \Cref{obs: vertex sets far from each other}, for all $1\leq i<j\leq h$, $B_i\cap B_j=\emptyset$. Therefore, there must be an index $1\leq i\leq h$, such that the total weight of all vertices in $B_i$ is bounded by $\frac{W(C)}{h}\leq \frac{W(C)}{\hat W^{\eps}}$. We fix this index $i$ from now on.

We are now ready to define the final set $\hat E'$ of edges: $\hat E'=(\hat E\setminus \hat E_i)\cup E''$. Recall that $|\hat E_i|\geq |T_i|/2\geq 4|E''|$, and $|\hat E_i| \geq  \frac{k}{4\hat W^{\eps} (\log k)^{c'''/\eps}}$. Therefore, $|\hat E'|\leq |\hat E|-|\hat E_i|+|E''|\leq k \left (1-\frac{1}{\hat W^{\eps} (\log k)^{O(1/\eps)} }\right )\leq k \left (1-\frac{1}{\hat W^{\eps} (\log \hat W)^{O(1/\eps)} }\right ) $. It now remains to show that $\hat E'$ is a valid $(\hat D,\rho)$-pseudocut, for $\rho=\hat W^{\eps}$. We do so in the next claim.

\begin{claim}\label{claim: get a pseudocut}
	Edge set $\hat E'$ is a $(\hat D,\rho)$-pseudocut for graph $C$.
\end{claim}
\begin{proof}
Consider any vertex $x\in V(C)$, and let $B=B_{C\setminus \hat E'}(x,\hat D)$. It is enough to show that $W(B)\leq \frac{W(C)}{\hat W^{\eps}}$. We consider two cases.

Assume first that $B\cap Z_i\neq \emptyset$. In this case, $B\subseteq B_{C\setminus \hat E'}(Z_i,2\hat D)\subseteq B_{C\setminus E''}(Z_i,2\hat D)=B_i$ (since $E''\subseteq \hat E'$). But then, from our choice of the index $i$, $W(B)\leq W(B_i)\leq \frac{W(C)}{\hat W^{\eps}}$.

Therefore, we assume that $B\cap Z_i=\emptyset$. Since $\hat E\setminus \hat E_i\subseteq \hat E'$, and no edge of $\hat E_i$ is incident to a vertex of $B$, if we consider the sub-graph $H$ of $C\setminus \hat E'$ that is induced by $B$, then $H$ is also a subgraph of $C\setminus \hat E$. In other words, $B\subseteq B_{C\setminus \hat E}(x,\hat D)$. But then, since $\hat E$ was a valid $(\hat D,\rho)$-pseudocut for $C$, $W(B)\leq \frac{W(C)}{\hat W^{\eps}}$ must hold.
\end{proof}

It now remains to analyze the running time of the algorithm. 
In addition to running the algorithm from \Cref{cor: expander or terminal partition many sets}, we need to compute the balls $B_1,\ldots,B_h$, and the final edge set $\hat E'$. Since the balls $B_1,\ldots,B_h$ are mutually disjoint, they can be computed in time $O(|E(C)|)$. Therefore, the running time is dominated by the running time of the algorithm from \Cref{cor: expander or terminal partition many sets}, which is bounded by:

\[ 
\otilde\left (h(\log k)^{O(1/\eps)}\cdot \left (k^{1+\eps}(\log k)^{O(1/\eps^2)}+|E(G)|D''\eta\right )  \right )
\leq O\left (|E(C)|\cdot  \hat D^2\cdot \hat W^{O(\eps)}(\log \hat W)^{O(1/\eps^2)}\right ).
\]

since $h\leq 2W^{\eps}$, $D''=\Theta(D)$, $\eta=\hat D\cdot \hat W^{\eps} (\log \hat W)^{O(1/\eps)}$, $k\leq |E(C)|\leq \hat W$, and $|E(G)|\leq O(|E(C)|)$.

In order to complete the proof of \Cref{lem: smaller pseudocut or expander}, it is now enough to prove \Cref{lem: expander or cut}, which we do next.

\subsubsection{Proof of \Cref{lem: expander or cut}}
\label{subsec: prove auxiliary lemma}


The proof follows standard techniques, namely the Cut-Matching game of Khandekar, Rao, and Vazirani \cite{KRV}, and more specifically, its variant that was studied in \cite{KhandekarKOV2007cut}, and then used in \cite{detbalanced}. At a high level, the Cut-Matching game is a game played between two players: a cut player and a matching player. The game starts with a graph $X$ containing an even number $n$ of vertices and no edges, and then proceeds in iterations. In every iteration $i$, a cut player computes some partition $(A_i,B_i)$ of $V(X)$, with $|A_i|\leq |B_i|$. The matching player then returns a matching $M_i$ between vertices of $A_i$ and vertices of $B_i$, of cardinality $|A_i|$. The edges of $M_i$ are added to the graph $X$, and the iteration terminates. Once graph $X$ becomes a $\phi$-expander, for some given parameter $\phi$, the game terminates. Intuitively, the goal of the cut player is to make the game as short as possible, while the matching player wants to make it longer. We will use the following theorem, that was proved in \cite{detbalanced} in order to implement the cut player.

\begin{theorem}[Theorem 4.1 in \cite{detbalanced}]\label{thm: cut player rec}
	There are universal constants $c_0$, $N_0$ and a deterministic algorithm, that, given an $n$-vertex graph $G=(V,E)$ and parameters $N,q$ with $N>N_0$ an integral power of $2$, and $q\ge 1$ an integer, such that $n\leq N^q$, and the maximum vertex degree in $G$ is at most $c'\log n$ for a constant $c'$, computes one of the following:
	
	\begin{itemize}
		\item either a partition $(A,B)$ of $V$ with $|A|,|B|\geq n/4$, $|A|\leq |B|$, and $|E_G(A,B)|\leq n/100$; or
		\item a subset $R\subseteq V$ of at least $n/2$ vertices of $G$, such that graph $G[R]$ is a $\phi'$-expander, for $\phi'=1/\left (q\log N\right )^{8q}$.
	\end{itemize}
	
	The running time of the algorithm is $O\left (N^{q+1}\cdot(q\log N)^{c_0q^2}\right )$.
\end{theorem}

We will use the following immediate corollary of the theorem.
\begin{corollary}\label{cor: cut player}
	There is a universal constant $c''$, and a deterministic algorithm, that, given an $n$-vertex graph $G=(V,E)$, and a parameter $0<\eps<1$, such that $\eps\geq c''/\log n$, and the maximum vertex degree in $G$ is at most $c'\log n$ for some constant $c'$, computes one of the following:

\begin{itemize}
	\item either a partition $(A,B)$ of $V(G)$ with $|A|,|B|\geq n/4$, $|A|\leq |B|$, and $|E_G(A,B)|\leq n/100$; or
	\item a subset $R\subseteq V$ of at least $n/2$ vertices of $G$, such that graph $G[R]$ is a $\phi''$-expander, for $\phi''=1/(\log n )^{O(1/\eps)}$.
\end{itemize}

The running time of the algorithm is $O\left (n^{1+O(\eps)}\cdot (\log n)^{O(1/\eps^2)}\right )$.
\end{corollary}

\begin{proof}
 We set $q=\ceil{1/\eps}$, so that $1/\eps\leq q\leq 2/\eps$. Next, we let $N$ be the smallest integral power of $2$, so that $N^q\geq n$. It is easy to verify that $N\leq 2n^{1/q}\leq 2n^{\eps}$. 
Also, $N\geq n^{1/q}\geq n^{\eps/2}$. We can then set the constant $c''$ to be large enough to ensure that $n^{\eps/2}> N_0$, where $N_0$ is the constant in the statement of \Cref{thm: cut player rec}. We then apply \Cref{thm: cut player rec} to the graph $G$, together with parameters $q$ and $N$. 
If the algorithm returns a partition $(A,B)$ of $V(G)$ with $|A|,|B|\geq n/4$ and $|E_G(A,B)|\leq n/100$, then we also return this partition as the outcome of our algorithm. Otherwise, the algorithm from \Cref{thm: cut player rec} computes a subset $R\subseteq V$ of at least $n/2$ vertices of $G$, such that graph $G[R]$ is a $\phi''$-expander, for $\phi''=1/\left (q\log N\right )^{8q}$. Since $q=O(1/\eps)$, and $N^q\leq 2n$, it is easy to verify that $\phi''\geq 1/(\log n)^{O(q)}\geq  1/(\log n)^{O(1/\eps)}$. We then return $R$ as the outcome of the algorithm.

Lastly, the running time of the algorithm is $O\left (N^{q+1}\cdot(q\log N)^{c_0q^2}\right )$. Since $N^q\leq 2 n$ and $q=\Theta(1/\eps)$, we get that the running time is bounded by $O\left (n^{1+O(\eps)}\cdot (\log n)^{O(1/\eps^2)}\right )$.
\end{proof}

Consider now the following implementation of the cut-matching game. We start with a graph $X$ that contains an even number $n$ of vertices and no edges, and play the game for at most $c'\log n$ iterations, where $c'$ is the constant from \Cref{cor: cut player}. In every iteration $i$, we apply the algorithm from \Cref{cor: cut player} to the graph $X$, with the parameter $\eps$. We now consider two cases. First, if the outcome of the algorithm from \Cref{cor: cut player} is a partition $(A,B)$ of $V(G)$ with $|A|,|B|\geq n/4$, $|A|\leq |B|$, and $|E_G(A,B)|\leq n/100$, then we let $(A_i,B_i)$ be any partition of $V(G)$ with $|A_i|=|B_i|$, such that $A\subseteq A_i$. We then compute an arbitrary perfect matching $M_i$ between vertices of $A_i$ and vertices of $B_i$. Lastly, we add the edges of $M_i$ to graph $X$, and continue to the next iteration. In the second case, the outcome of the algorithm from \Cref{cor: cut player} is a subset $R\subseteq V$ of at least $n/2$ vertices of $G$, such that graph $G[R]$ is a $\phi''$-expander, for $\phi''=1/(\log n )^{O(1/\eps)}$. In this case, we let $A_i=V\setminus R$, $B_i=R$, and we compute an arbitrary matching $M_i$ of cardinality $|A_i|$ between vertices of $A_i$ and vertices of $B_i$. We then add the edges of $M_i$ to graph $X$ and terminate the algorithm. 

The following theorem follows immediately from the results of \cite{KhandekarKOV2007cut}, and was also proved explicitly in \cite{detbalanced} (see Theorem 2.5 in  \cite{detbalanced}).

\begin{theorem}
	\label{thm: CMG} There is a constant $\hat c$, such that, after at most $\hat c \log n$ iterations, the above algorithm is guaranteed to terminate, and the resulting graph $X$ is a $\hat \phi$-expander, for $\hat \phi=1/(\log n )^{\hat c/\eps}$.
\end{theorem}

Notice that the above theorem holds regardless of the specific choices of the matchings $M_i$ in each iteration (though the partitions $(A_i,B_i)$ of the vertices of $V$ that the algorithm computes in each iteration depend on the choices of the matchings from previous iterations). We note that the constant $c'$ in the statement of \Cref{thm: cut player rec} from \cite{detbalanced} was chosen to be such that $\hat c\leq c'$ holds.

The following theorem will be used in order to implement the matching player. The proof uses techniques that are similar to those introduced in \cite{fast-vertex-sparsest}, and then used in \cite{detbalanced} and \cite{APSP-old}.

\begin{theorem}\label{thm: matching player}
	There is a deterministic algorithm, that, given a graph $G$ with integral lengths $\ell(e)\geq 1$ on its edges $e\in E(G)$, two disjoint subsets $A,B$ of its vertices with $|A|\leq |B|$ and $|A|=\kappa$, together with parameters $D',\eta>0$, and $z\geq 0$, computes one of the following:
	
	\begin{itemize}
		\item either a collection $\pset$ of at least $|A|-z$ paths in $G$, where each path connects a distinct vertex of $A$ to a distinct vertex of $B$; every path has length at most $D'$; and every edge of $G$ participates in at most $\eta$ paths; or
		
		\item a collection $E'$ of at most $\frac{|A|\cdot D'\log \kappa}{2\eta}$ edges of $G$, and two subsets $A^*\subseteq A$, $B^*\subseteq B$ of at least $z/2$ vertices each, such that, in graph $G\setminus E'$, $\dist(A^*,B^*)> D'$.
	\end{itemize}
	
	The running time of the algorithm is $\otilde(|E(G)|D'\eta)$.
\end{theorem}

We delay the proof of \Cref{thm: matching player} to \Cref{subsec: matching-player-pseudocut}, after we complete the proof of \Cref{lem: expander or cut} using it.

In the remainder of the proof, we let $\hat c$ be the constant from \Cref{thm: CMG},  so that the number of iterations in the Cut-Matching game played over a set of $k$ vertices is bounded by $\hat c\log k$, and the final graph $X$ obtained at the end of the game is a $\hat \phi$-expander, for $\hat \phi=1/(\log k )^{\hat c/\eps}$.
 The algorithm consists of two stages. In the first stage, we compute an expander $X$ over the vertex set $S$, and a relatively small subset $F\subseteq E(X)$ of its edges (that we call \emph{fake edges}), together with an embedding of $X\setminus F$ into $G$; if we fail to compute this expander, then we will produce the desired set $E'$ of at most $\frac{ckD'\log^2k}{\eta}$ edges, and two disjoint subsets $S_1,S_2\subseteq S$ of vertices of cardinality at least $k/(\log k)^{c/\eps}$ each, such that, in graph $G\setminus E'$, $\dist(S_1,S_2)> D'$. In the second stage, we use expander pruning in order to get rid of the fake edges and compute a final expander $X'\subseteq X\setminus F$. We now describe each stage in turn.

\paragraph{Stage 1: Embedding the Expander.}
We start with the graph $X$, whose vertex set is $S$, and edge set is empty. We also let $\pset=\emptyset$ be an initial embedding of the graph $X$.
We then perform at most $\hat c\log k$ iterations, with the $i$th iteration performed as follows. First, we apply the algorithm from \Cref{cor: cut player} to graph $X$. Assume first that the algorithm produces  a partition $(A_i,B_i)$ of $V(X)$ with $|A_i|,|B_i|\geq k/4$, $|A_i|\leq |B_i|$, and $|E_X(A_i,B_i)|\leq k/100$. We let $(A_i',B_i')$ be any partition of $V(X)$ with $|A_i'|=|B_i'|$ and $A_i\subseteq A_i'$. 

Next, we apply the algorithm from \Cref{thm: matching player} to graph $C$, the sets $A_i',B_i'$ of its vertices, with the same parameter $D'$, congestion parameter $\hat \eta=\eta/(\hat c\log k)$, and $z=\hat \phi k/(20\hat c^2\log^2 k)$. Assume first that the algorithm returns a collection $E'$ of at most 
$\frac{|A_i'|\cdot D'\log k}{2\hat \eta}\leq \frac{\hat c|A_i'|\cdot D'\log^2 k}{2\eta}=\frac{\hat ckD'\log^2k}{4\eta}$ edges of $G$, and two subsets $A^*\subseteq A_i$, $B^*\subseteq B_i$ of at least $z/2=\hat \phi k/(40\hat c^2\log^2 k)$ vertices each, such that, in graph $G\setminus E'$, $\dist(A^*,B^*)> D'$. We then say that the current iteration was unsuccessful, terminate the algorithm, and return the sets $S_1=A^*,S_2=B^*$ of vertices, together with the set $E'$ of edges. It is easy to verify that they have all required properties, since $\hat \phi k/(40\hat c^2\log^2 k)=k/(\log k)^{O(1/\eps)}$. 

Assume now that the algorithm returned  a collection $\pset_i$ of at least $|A_i'|-z$ paths in $G$, where each path connects a distinct vertex of $A_i'$ to a distinct vertex of $B_i'$, every path has length at most $D'$, and every edge of $G$ participates in at most $\hat \eta=\eta/(\hat c\log k)$ paths. We then say that the current iteration is successful. We let $M'_i$ be a partial matching between vertices of $A_i'$ and vertices of $B_i'$, where $(x,y)\in M_i'$ iff some path in $\pset_i$ connects $x$ to $y$. Let $A''_i\subseteq A'_i, B''_i\subseteq B'_i$ be the sets of vertices that do not participate in the matching $M_i'$. We let $F_i$ be an arbitrary perfect matching between vertices of $A''_i$ and vertices of $B''_i$; observe that $|F_i|=|A'_i\setminus A''_i|\leq z$ must hold. We call the vertex pairs in $F_i$ \emph{fake edges for graph $X$}. We then let $M_i=M_i'\cup F_i$, and add the edges of $M_i$ into $X$, terminating the current iteration.

Finally, assume that the algorithm from  \Cref{cor: cut player} returns a subset $R\subseteq V(X)$ of at least $k/2$ vertices of $X$, such that graph $X[R]$ is a $\phi$-expander. Let $i$ be the index of the current iteration. We set $A_i'=V(X)\setminus R$ and $B'_i=R$, and apply the algorithm from \Cref{thm: matching player} to graph $G$, vertex sets $A_i',B_i'$, distance bound $D'$, congestion bound $\hat \eta=\eta/(\hat c\log k)$, and parameter $z=\hat \phi k/(20\hat c^2\log^2 k)$ as before. We then continue exactly as before. If the algorithm returns a collection $E'$ of at most $\frac{\hat c|A_i'|\cdot D'\log^2 k}{2\eta}\leq \frac{\hat ckD'\log^2k}{4\eta}$ edges of $G$, and two subsets $A^*\subseteq A_i$, $B^*\subseteq B_i$ of at least $z/2=\hat \phi k/(40\hat c^2\log^2 k)$ vertices each, such that, in graph $G\setminus E'$, $\dist(A^*,B^*)> D'$, then, as before, we say that the current iteration was unsuccessful, terminate the algorithm, and return the sets $S_1=A^*,S_2=B^*$ of vertices, together with the set $E'$ of edges. Otherwise, the algorithm from  \Cref{thm: matching player} returns  a collection $\pset_i$ of at least $|A_i'|-z$ paths in $G$, where each path connects a distinct vertex of $A_i'$ to a distinct vertex of $B_i'$, every path has length at most $D'$, and every edge of $G$ participates in at most $\hat \eta=\eta/(\hat c\log k)$ paths. We then say that the current iteration is successful. We let $M'_i$ be a partial matching between vertices of $A_i'$ and vertices of $B_i'$, where $(x,y)\in M_i'$ iff some path in $\pset_i$ connects $x$ to $y$. Let $A''_i\subseteq A'_i, B''_i\subseteq B'_i$ be the sets of vertices that do not participate in the matching $M_i'$. We let $F_i$ be an arbitrary matching between vertices of $A''_i$ and vertices of $B''_i$ of cardinality $|A''_i|$; observe that $|F_i|=|A'_i\setminus A''_i|\leq z$ must hold. We call the vertex pairs in $F_i$ \emph{fake edges for graph $X$}. We then let $M_i=M_i'\cup F_i$, and add the edges of $M_i$ into $X$, terminating the current the algorithm. 
 
 If any iteration of the algorithm was unsuccessful, then we have obtained the desired sets $S_1,S_2$ of vertices of $S$, and the edge set $E'$ as required. We assume now that every iteration of the algorithm was successful.  Then, from \Cref{thm: CMG}, the total number of iterations in the algorithm is bounded by $\hat c\log k$, and the final graph $X$ is a $\hat \phi$-expander, for $\hat \phi=1/(\log k )^{\hat c/\eps}$.
Let $F=\bigcup_iF_i$ be the set of all fake edges of $X$. Then $|F|\leq \hat cz\log k$. Let $\pset=\bigcup_i\pset_i$. Then the paths in $\pset$ provide an embedding of the graph $X\setminus F$ into $G$, where the length of every path in $\pset$ is at most $D'$. Since the number of iterations is at most $\hat c\log k$, and each set $\pset_i$ of paths causes edge-congestion at most $\eta/(\hat c\log k)$, the total edge-congestion caused by the paths in $\pset$ is at most $\eta$. This finishes the first stage of the algorithm. We now analyze its running time.

The algorithm consists of $O(\log k)$ iterations. Each iteration involves applying the algorithm from \Cref{cor: cut player}, whose running time is $O\left (k^{1+O(\eps)}\cdot (\log k)^{O(1/\eps^2)}\right )$, and the algorithm from \Cref{thm: matching player}, whose running time is $ \otilde (|E(G)|D'\hat \eta)\leq \otilde (|E(G)|D'\eta)$. 
Therefore, the total running time of the algorithm is: $$ \otilde\left (k^{1+O(\eps)}\cdot (\log k)^{O(1/\eps^2)}\right )+\otilde (|E(G)|D'\eta).$$

\paragraph{Stage 2: Pruning.}
In this stage, we apply the algorithm from \Cref{thm: expander pruning} to the graph $X$, with expansion parameter $\hat \phi=1/(\log k )^{\hat c/\eps}$, and the sequence $\sigma$ containing the edges of $F$ in an arbitrary order. Observe that the maximum vertex degree in $X$ is at most $\hat c\log k$, and $|F|\leq (\hat c\log k)\cdot z\leq (\hat c\log k)\cdot \hat \phi k/(20\hat c^2\log^2 k)\leq  \hat \phi k/(20\hat c\log k)$. Clearly, if $\Delta$ denotes the maximum vertex degree in $X$, then $|F|\leq \hat \phi |E(X)|/(10\Delta)$
holds.  Let $S'=S_{|F|}$ be the outcome of the algorithm, and let $X'$ be the graph obtained from $X\setminus F$ after the vertices of $S'$ are removed from it. Then, from \Cref{thm: expander pruning} graph $X'$ is a $\phi$-expander, for $\phi=\hat \phi/(6\Delta)=1/(\log k)^{O(1/\eps)}$. Let $\pset'\subseteq \pset$ be the set of paths corresponding to the edges in $X'$. Then paths in $\pset'$ define an embedding of $X'$ into $G$ with congestion at most $\eta$, where each path has length at most $D'$. Recall also that \Cref{thm: expander pruning} guarantees that $|S'|\leq 8|F|\Delta/\hat \phi\leq 8k/20$. Therefore, $|V(X')|\geq |S|/2$ must hold.

The running time of the second stage of the algorithm is bounded by $\Otilde(|F|\Delta^2/\hat \phi^2)\leq  \Otilde (k\log^2k/\hat \phi)\leq \Otilde 
\left (k\cdot (\log k)^{O(1/\eps)}\right )$.
The total running time of the whole algorithm is therefore bounded by: $$\otilde\left (k^{1+O(\eps)}(\log k)^{O(1/\eps^2)}\right )+\otilde (|E(G)|D'\eta).$$
In order to complete the proof of \Cref{lem: smaller pseudocut or expander}, it is now enough to prove \Cref{thm: matching player}, which we do next.


\subsubsection{Proof of \Cref{thm: matching player}}\label{subsec: matching-player-pseudocut}
The proof proceeds by iteratively applying the following claim.

\begin{claim}\label{claim: paths or pseudocut}
	There is a deterministic algorithm, that, given a graph $G$ with integral lengths $\ell(e)\geq 1$ on its edges $e\in E(G)$, two disjoint subsets $A',B'$ of its vertices with $|A'|\leq |B'|$, together with parameters $D'>0$ and $\eta'>0$, computes one of the following:
	
	\begin{itemize}
		\item either a collection $\pset'$ of at least $|A'|/2$ paths in $G$, where each path connects a distinct vertex of $A'$ to a distinct vertex of $B'$, and every path has length at most $D'$, and every edge in $G$ participates in at most $\eta'$ paths; or
		
		\item a collection $E'$ of at most $|A'|\cdot D'/(2\eta')$ edges of $G$, and two subsets $A''\subseteq A'$, $B''\subseteq B'$ of at least $|A'|/2$ vertices each, such that, in graph $G\setminus E'$, $\dist(A'',B'')> D'$.
	\end{itemize}
	
	The running time of the algorithm is $\otilde(|E(G)|D'\eta')$.
\end{claim}

We provide the proof of \Cref{claim: paths or pseudocut} below, after we complete the proof of \Cref{thm: matching player} using it.
We start with $\pset=\emptyset$, set $A'=A,B'=B$, and then iterate, as long as $|A'|>z$. In every iteration, we apply the algorithm from  \Cref{claim: paths or pseudocut} to the input graph $G$, parameter $D'$, and $\eta'=\eta/\log \kappa$. If the outcome of the algorithm is a set $E'$ of edges, of cardinality at most 
$\frac{|A'|\cdot D'}{2\eta'}\leq \frac{|A|\cdot D'\log k}{2\eta}$, and two subsets $A''\subseteq A'$, $B''\subseteq B'$ of at least $|A'|/2$ vertices each, such that, in graph $G\setminus E'$, $\dist(A'',B'')> D'$, then we set $A^*=A'',B^*=B''$, and terminate the algorithm with the output $A^*,B^*$ and $E'$. Note that, since $|A'|\geq z$ held, we are guaranteed that $|A^*|,|B^*|\geq z/2$, as required. Otherwise, we obtain  a collection $\pset'$ of at least $|A'|/2$ paths in $G$, where each path connects a distinct vertex of $A'$ to a distinct vertex of $B'$, and every path has length at most $D'$, and every edge in $G$ participates in at most $\eta'$ paths in $\pset'$. We add the paths in $\pset'$ to the set $\pset$, and we delete from $A'$ and $B'$ vertices that serve as endpoints to paths in $\pset'$, terminating the current iterations.

If any iteration terminates with the set $E'$ of edges, then the algorithm produces the required output. Therefore, we assume from now on that in every iteration, new paths are added to set $\pset$. Since the cardinality of the vertex set $A'$ must decrease by at least factor $2$ in every iteration,  the number of iterations is bounded by $\log \kappa$. Since each set $\pset'$ of paths computed by the algorithm from \Cref{claim: paths or pseudocut} causes edge-congestion at most $\eta'=\eta/\log\kappa$, the paths in the final set $\pset$ cause edge-congestion at most $\eta$. Clearly, the length of every path in $\pset$ is at most $D'$, and the endpoints of all paths in $\pset$ are distinct. Moreover, when the algorithm terminates, $|A'|\leq z$ holds, so we are guaranteed that $|\pset|\geq \kappa-z$.
Since the running time of every iteration is  $\otilde(|E(G)|D'\eta')=\Otilde(|E(G)|D'\eta/\log \kappa)$, and the total number of iterations is bounded by $\log \kappa$, the total running time is at most  $\otilde(|E(G)|D'\eta)$.

In order to complete the proof of  \Cref{thm: matching player}, it is now enough to prove  \Cref{claim: paths or pseudocut}, which we do next.

\begin{proofof}{\Cref{claim: paths or pseudocut}}
	For convenience, we denote $|A'|=\kappa$.
		We construct a new graph $H$: start with the  graph $G$, and add a source vertex $s$ that connects to every vertex in $A'$ with an edge of length $1$; similarly, add a destination vertex $t$, that connects to every vertex in $B'$ with an edge of length $1$. All other edge lengths remain unchanged in $H$. We use the standard Even-Shiloach Tree data structure  on graph $H$, with source vertex $s$, and distance threshold $D'+2$. Initialize $\pset'=\emptyset$. Additionally, for every edge $e\in E(G)$, we maintain a value $\gamma(e)$ -- the number of paths in $\pset'$ that contain the edge $e$. Initially, $\gamma(e)=0$ for all $e\in E(G)$.
		
		While the distance from $s$ to $t$ is less than $D'+2$, choose any path $P$ in $H$ connecting $s$ to $t$ of length at most $D'+2$. Let $P'$ be the path obtained from $P$ after we delete its endpoints, so $\ell(P')\leq D'$. Add path $P'$ to $\pset'$, and update, for every edge $e\in E(P')$, the value $\gamma(e)$. If $\gamma(e)$ exceeds $\eta'$, then we remove edge $e$ from graph $H$. Additionally, we remove from $H$ the fist and the last edge of $P$ (the edges that are incident to $s$ and $t$). We then continue to the next iteration. The total update time of the data structure is $ \otilde (|E(G)|D')$, and the total running time of the algorithm, that includes selecting the paths and deleting edges from $G$ as required, is bounded by $\otilde(|E(G)|D'\eta')$, since an edge may participate in at most $\eta'$ paths in $\pset'$. We now consider two cases. First, if $|\pset'|\geq  |A'|/2$ at the end of the algorithm, then we terminate the algorithm, and return the set $\pset'$ of paths. Clearly, the paths in $\pset'$ cause edge-congestion at most $\eta'$, and the length of every path is at most $D'$.
		
		From now on, we assume that, when the algorithm terminates, $|\pset'|<|A'|/2$ holds. Let $E'$ denote the set of all edges of graph $G$ that were deleted from $H$. Let $A''\subseteq A',B''\subseteq B'$ be the sets of vertices that do not serve as endpoints of paths in $\pset'$. Then $\dist_{G\setminus E'}(A'',B'')> D'$ (as otherwise we could have continued the algorithm). Moreover, an edge belongs to $E'$ only if $\eta'$ paths in $\pset'$ contain this edge. Since each path in $\pset'$ contains at most $D'$ edges, (as the length of each pedge is at least $1$), $|E'|\leq |\pset'|\cdot D'/\eta'\leq |A'|\cdot D'/(2\eta')$.
\end{proofof}

%% file: expander-APSP.tex
\subsection{Proof of \Cref{thm: expander APSP}}\label{subsec: expander APSP proof}

The proof is almost identical to that of \cite{APSP-old}. Our main tool is the following theorem that allows us to embed a small expander into a large expander.
The theorem is almost identical to Theorem 3.8 in \cite{APSP-old}, except that we need a slightly different tradeoff between various parameters and the running time. We include the proof in \Cref{subsubsection: embedding expander into smaller expander} for completeness.

\begin{lemma} [Analogue of Theorem 3.8 in \cite{APSP-old}]
	\label{lem:embed}
	There is a constant $c$ and a deterministic algorithm that, given an $n$-vertex
	graph $X$ that is a $\phi$-expander with maximum vertex degree $\Delta$, and a set $T\subseteq V(X)$ of its vertices whose cardinality is $k$, together with a parameter $0<\eps<1$,
	such that $\phi\leq 1/2^{\Omega(1/\eps)}$, computes a graph $X'$ with $V(X')=T$ that has maximum vertex degree at most $c\log k$, such that graph $X'$ is a $\hat \phi$-expander, for $\hat \phi=1/(\log k )^{c/\eps}$. The algorithm also computes an embedding $\pset$ of $X'$ into $X$, such that the paths in $\pset$ have length at most $\frac{c\Delta \log n}{\phi}$ and cause congestion at most $\frac{c\Delta^2\log^3n}{\phi^2}$. 
	The running time of the algorithm is $\otilde(\Delta^4n/\phi^3)+O\left (k^{1+O(\eps)}\cdot (\log k)^{O(1/\eps^2)}\right )$.
\end{lemma}

The proof of \Cref{lem:embed} is deferred to \Cref{subsubsection: embedding expander into smaller expander}. We now turn to complete the proof of  \Cref{thm: expander APSP} using it.

Throughout the algorithm, we denote $n=V(X)$. We use the parameters $\gamma=\ceil{n^{\eps}}$ and $r=\floor{1/\eps}$.
The idea is to construct and maintain a hierarchy $X_0,X_1,\ldots,X_r$ of expanders, where $X_0=X$, and for $i>0$, $|V(X_i)|=\ceil{n/\gamma^i}$. For all $i$, we will ensure that $V(X_i)\subseteq V(X_{i-1})$, and we will maintain an embedding of $X_i$ into $X_{i-1}$ via short paths that cause small edge-congestion.

Specifically, we will use the following parameters. We let $c$ be the constant from \Cref{lem:embed}, that we can assume to be a large enough constant. We set $\phi_0=\phi$, and, for $1\leq i\leq r$, $\phi_i=\hat \phi=1/(\log n )^{c/\eps}$. We also denote $\Delta_0=\Delta$, and for all $1\leq i\leq r$, $\Delta_i=c\log n$. We will ensure that for all $0\leq i\leq r$, $X_i$ is a $\phi_i$-expander, whose maximum vertex degree is at most $\Delta_i$. We also define parameters $\ell_1=c\Delta \log n /\phi$, and for $1< i\leq r+1$, $\ell_i=\hat \ell=c\Delta_{i-1}\log n/\phi_{i-1}=c^2\log^2 n/\hat \phi$. Additionally, we define parameters $\eta_1=c\Delta^2\log^3n/\phi^2$, and for $1< i\leq r+1$, $\eta_i=c\Delta_{i-1}^2\log^3n/\phi_{i-1}^2=c^3\log^5n/\hat \phi^2$. 

For all $1\leq i<r$, we will maintain an embedding $\pset_i$ of $X_{i}$ into $X_{i-1}$, where the paths in $\pset_i$ have length at most $\ell_i$ each, and cause total congestion at most $\eta_i$. Additionally, we will maintain an \EST data structure $\tau_i$ in graph $X'_i$, that is obtained from $X_i$ by adding a source vertex $s_i$, and connecting it to all vertices in $V(X_{i+1})$. The tree is rooted at $s_i$, and its depth bound is $D_i=4\ell_i$.
Lastly, for every edge $e\in E(X_{i-1})$, we will maintain a list $J_i(e)$ of all edges $e'\in E(G_{i})$, such that the embedding of $e'$ in $\pset_{i}$ contains the edge $e$; recall that $|J(e)|\leq \eta_i$ must hold. Whenever edge $e$ is deleted from graph $X_{i-1}$, this will trigger the deletion of all edges in its list $J_i(e)$ from graph $X_{i}$.
We use the algorithm from \Cref{thm: expander pruning} in order to maintain, for every expander $X_i$, the set $S_i$ of ``pruned-out'' vertices. When set $S_i$ becomes too large, we re-initialize the graphs $X_i,X_{i+1},\ldots,X_r$, and all the corresponding data structure.

The outcome of the algorithm (the vertex set $S$ that we maintain) is the set $S=S_0$ of vertices -- the vertices that we prune out of the main expander $X_0=X$.

\subsection*{Initializing the Data Structures}

At the beginning of the algorithm, we run procedure $\initexp(1)$, that constructs expander $X_1$, its embedding $\pset_1$ into $X_0$, and the lists $J_1(e)$ for edges $e\in E(X_0)$. The procedure then recursively calls to $\initexp(2)$, that constructs the data structures for higher levels. The algorithm for procedure $\initexp(i)$ is described in \Cref{alg: init-exp}. We note that it is identical to the algorithm of \cite{APSP-old}. We emphasize that procedure $\initexp(i)$ is only called when graph $X_{i-1}$ is defined, with $|V(X_{i-1})|\geq n/(2\gamma^{i-1})$, and we are guaranteed that $X_{i-1}$ is a $\phi_{i-1}/(6\Delta_{i-1})$-expander, whose maximum vertex degree is at most $\Delta_{i-1}$. From \Cref{obs: short paths in exp}, every pair of vertices of $X_{i-1}$ is then guaranteed to have a path of length at most $\frac{8\Delta_{i-1} \log n}{\phi_{i-1}}\leq \ell_i\leq D_i/2$, so throughout the algorithm, tree $\tau_{i-1}$ spans all vertices of $X_{i-1}$.

We note that we will ensure that, over the course of the algorithm, for all $1\leq i\leq r$, the size of set $S_i$  never exceeds $n/(2\gamma^i)$.

\program{Algorithm $\initexp(i)$}{alg: init-exp}{
	Assumption: graph $X_{i-1}$ is defined; $|V(X_{i-1})|\geq n/(2\gamma^{i-1})$; graph $X_{i-1}$ is a $\phi_{i-1}/(6\Delta_{i-1})$-expander, with maximum vertex degree is at most $\Delta_{i-1}$.
\begin{itemize}
\item If $i=r+1$, then initialize an \EST $\tau_{r}$ in graph $X_r$, rooted at an arbitrary
vertex, with distance threshold $D_{r+1}$; return. 
\item Let $V_i$ be an arbitrary subset of $V(X_{i-1})$, of cardinality $\ceil{n/\gamma^i}$.

\item Apply the algorithm from \Cref{lem:embed}, to graph $X_{i-1}$, with $T=V_i$, to compute an expander $X_{i}$ with vertex set $V_i$, and its embedding $\pset_{i}$ into $X_{i-1}$, so that $X_{i}$ is a $\phi_{i}$-expander, with maximum vertex degree at most $\Delta_{i}$; the paths in $\pset_i$ cause edge-congestion at most $\eta_i$ and have length at most $\ell_i$ each.

\item For every edge $e\in E(X_{i-1})$, initialize the list $J_i(e)$ of all edges of $X_{i}$ whose embedding  path in $\pset_{i}$ contains $e$.

\item Initialize the expander pruning algorithm from \Cref{thm: expander pruning}
on graph $X_{i}$, that will maintain a pruned vertex set $S_{i}\subseteq V(X_{i})$. 

\item Initialize an ES-tree $\tau_{i-1}$ in graph $X'_{i-1}$ that is obtained from $X_{i-1}$ by adding a source vertex $s_{i-1}$ and connecting it to all vertices in $V(X_{i})$. The tree  $\tau_{i-1}$ is rooted at $s_{i-1}$ and has depth threshold $D_i$.

\item Call $\initexp(i+1)$. 
\end{itemize}
}

\subsection*{Maintaining the Data Structures}
For all $0\leq i< r$, 
we denote by  $X_{i+1}^{(0)}$ the expander graph created by Procedure $\initexp(i)$. For all $t>0$, we denote by $X_{i+1}^{(t)}$ the graph that is obtained from  $X_{i+1}^{(0)}$ after $t$ edge deletions from $X$. 
As $t$ increases, our algorithm maintains the graph $X_{i+1}=X_{i+1}^{(t)}\setminus S_{i+1}$.
By \Cref{thm: expander pruning}, as long as $t\leq \phi_{i+1}|E(X_{i+1})|/(20\Delta_{i+1})$, graph $X_{i+1}$ remains a $(\phi_{i+1}/(6\Delta_{i+1}))$-expander, and $|V(X_{i+1})|\geq |V(X_{i+1}^{(0)})|/2\geq n/(2\gamma^i)$.

When some edge $e$ is deleted from graph $X$, we call  Algorithm $\mathtt{DeleteEdge}(0,e)$. The algorithm may recursively call to procedure $\mathtt{DeleteEdge}(i,e')$ for other expander graphs $X_i$ and edges $e'$. The algorithm  $\mathtt{DeleteEdge}(i,e)$  is shown in \Cref{alg:core delete}. We assume that edge $e$ lies in graph $X_i$.

\begin{figure}
	\fbox{
	\begin{minipage}{5.8 in}
		\begin{center}
			{\sc Algorithm $\mathtt{DeleteEdge}(i,e)$} 
		\end{center}
	
	Assumption: edge $e$ lies in graph $X_i$.
	
		\begin{itemize}
		\item If $i=r$, delete $e$ from graph $X_r$. Recompute the \EST
		$\tau_r$ in graph $X_{r}$, rooted at any vertex, with depth threshold $D_{r+1}$; return.
		
		\item Delete edge $e$ from graph $X_{i}$ and from the data structure $\tau_i$. 
		
		\item Update the pruned-out vertex set $S_{i}$ using the algorithm from \Cref{thm: expander pruning}, and update the tree $\tau_{i-1}$, by deleting edges $(s,x)$ for every vertex $x$ that was added to $S_i$. 
		
		\item If the total number of edge deletions from $X_i^{(0)}$ exceeds $\phi_i|E(X_{i}^{(0)})|/(20\Delta_i)$, call $\initexp(i-1)$; return. (Note: here we only count edges that were deleted from $X$ as part of input edge-deletion sequence, and we do not include edges that are incident to vertices of $S_i$).
		
		\item Let $Z_{i}^{new}$ denote the set of edges that were just removed from
		$X_{i}$. That is, $Z_{i}^{new}$ contains the edge $e$ and all edges incident
		to vertices that were just added to $S_{i}$. 
		\item For each edge $e\in Z_{i}^{new}$, for every edge $e'\in J_i(e)$, call $\mathtt{Delete}(i+1,e')$.
		\end{itemize}
	\end{minipage}
\vspace{2mm}
}
\caption{Algorithm $\mathtt{DeleteEdge}(i,e)$ \label{alg:core delete}}
\end{figure}

\subsection*{Total Update Time}

For all $0\leq i\leq r$, we denote $n_i=\ceil{n/\gamma^i}$. Recall that $|V(X_i)|\leq n_i$.
We bound the total update time of the algorithm in the following lemma.

\begin{lemma}\label{lem: oracle update time}
	The total update time of the algorithm is at most $O\left (\frac{n^{1+O(\eps)}\Delta^5(\log n)^{O(1/\eps^2)}}{\phi^5}\right ) $. 
\end{lemma}

\begin{proof}
	Fix an index $1\leq i\leq r$. We partition the execution of the algorithm into \emph{level-$i$ stages}, where each level-$i$ stage starts when $\initexp(i)$ is called (that is, graph $X_i$ is constructed from scratch), and terminates just before the subsequent call to $\mathtt{InitializeExpander}(i)$. Recall that, over the course of a level-$i$ stage, at most $\phi_i|E(X_{i}^{(0)})|/(20\Delta_i)$ edges are deleted from the graph $X_i^{(0)}$. We now bound the running time that is needed in order to initialize and maintain the level-$i$ data structure over the course of a single level-$i$ stage. This includes the following:
	
	\begin{itemize}
		\item Constructing expander $X_{i}$ and its embedding $\pset_{i}$ into graph  $X_{i-1}$, using the algorithm from  \Cref{lem:embed}; the running time is bounded by:

		\[\otilde\left(\frac{\Delta_{i-1}^4n_{i-1}}{\phi_{i-1}^3}+ n_i^{1+O(\eps)}\cdot (\log n)^{O(1/\eps^2)}\right );\]
		

		\item Initializing the lists $J_{i}(e)$ for edges $e\in E(X_{i-1})$: the time to initialize all such lists is bounded by the time needed to compute the embedding $\pset_i$.

		\item Initializing and maintaining the \EST  $\tau_{i-1}$: the running time is $\Otilde(|E(X_{i-1})|\cdot D_{i-1})\leq \Otilde(n_{i-1}\Delta_{i-1} \ell_{i-1})\leq  \Otilde\left (n_{i-1}\Delta^2_{i-1}/\phi_{i-1} \right )$.

		\item Running the algorithm  from \Cref{thm: expander pruning} for expander pruning on the expander $X_i$. Since a single level-$i$ stage may involve the deletion of up to $k=\phi_i|E(X_{i}^{(0)})|/(20\Delta_i)$ edges from $X_i$, the running time is bounded by:

		\[ \Otilde \left (\frac{\Delta_i^2}{\phi_i^2}\cdot \frac{\phi_i|E(X_{i}^{(0)})|}{20\Delta_i} \right )=\Otilde \left(\frac{n_i\Delta_i^2}{\phi_i}\right ). \]

		\item The remaining work, executed by $\mathtt{DeleteEdge}(i-1,e)$, for every edge $e$ that is deleted from graph $X_{i-1}$ (including edges incident to the vertices of the pruned out set $S_{i-1}$), requires $O(\eta_i)$ time per edge, with total time $O(|E(X_{i-1}^{(0)})|\cdot \eta_i)\leq  O\left (n_{i-1}\Delta_{i-1} \eta_i\right )=\Otilde \left(n_{i-1}\Delta_{i-1}^3/\phi_{i-1}^2\right )$.
	\end{itemize}
	
	Therefore, the total time that is needed in order to initialize and maintain the level-$i$ data structure over the course of a single level-$i$ stage is bounded by:
	
	\[\otilde\left (\frac{\Delta_{i-1}^4n_{i-1}}{\phi_{i-1}^3}+n_i^{1+O(\eps)}\cdot (\log n)^{O(1/\eps^2)}+ \frac{n_i\Delta_i^2}{\phi_i}\right ).\]

	Note that we did not include in this running time the time required for maintaining level-$(i+1)$ data structures, that is, calls to $\mathtt{InitializeExpander}(i+1)$ and $\mathtt{Delete}(i,e)$.
	
	Next, we bound the total number of level-$i$ stages. Consider some index $1< i'\leq r$, and consider a single level-$i'$ stage. 
	Clearly, the total number of edges  that are incident to the pruned-out vertices in $S_{i'}$ is bounded by $|E(X_{i'}^{(0)})|\leq O(n_{i'}\Delta_{i'})$.
	As the embedding $\pset_{i'+1}$ of $X_{i'+1}$ into $X_{i'}$ has congestion at most
	$\eta_{i'+1}\leq \Otilde(\Delta_{i'}^2/\phi_{i'}^2)$, this can cause at most $\otilde\left(  n_{i'}\Delta_{i'}^3/\phi_{i'}^2 \right )$ edge-deletions
	from graph $X_{i'+1}^{(0)}$. As a single level-$(i'+1)$ stage requires $k_{i'+1} \geq \Omega\left (\phi_{i'+1}|E(X_{i'+1}^{(0)})|/\Delta_{i'+1}\right )=\Omega(\phi_{i'+1}n_{i'+1})$ edge-deletions from $G_{i'+1}^{(0)}$, the number of level-$(i'+1)$ stages that are contained in a single level-$i'$ stage is bounded by:

	\[\otilde\left (\frac{n_{i'}\Delta_{i'}^3/\phi_{i'}^2}{n_{i'+1} \phi_{i'+1}}\right )=\otilde\left (\frac{n^{\eps}\Delta_{i'}^3}{\phi_{i'}^2\cdot \phi_{i'+1}}\right ).\]

	Since we only need to support at most $\phi |E(X)|/(20\Delta)$ edge deletions from the original graph $X$, there is only a single level-$0$ stage. 
	Recalling that $\phi_0=\phi$, and for all $1\leq i\leq r$, $\phi_i=\hat \phi=1/(\log n )^{O(1/\eps)}$, and that $\Delta_0=\Delta$, and for all $1\leq i\leq r$, $\Delta_i=O(\log n)$, we get that for all $1\leq i\leq r$, the total number of level-$i$ stages is bounded by:

\[  O\left (  \frac{n^{\eps i}(\log n)^{O(i)}\Delta^3}{\hat \phi^{3i}\cdot \phi^2}  \right ),   \]
	
	while the running time of a single level-$i$ phase, for $i>1$ is bounded by:

	\[\begin{split}
	 &\otilde\left (\frac{\Delta_{i-1}^4n_{i-1}}{\phi_{i-1}^3}+n_i^{1+O(\eps)}\cdot (\log n)^{O(1/\eps^2)}+ \frac{n_i\Delta_i}{\phi_i}\right )
	 \\&\leq \otilde\left (\frac{n}{n^{(i-1)\eps}\hat \phi^3}+n^{(1-i\eps)(1+O(\eps))}\cdot (\log n)^{O(1/\eps^2)}+ \frac{n}{n^{i\eps}\hat \phi}\right )\\
	 &\leq \otilde\left (\frac{n^{1+O(\eps)}\cdot (\log n)^{O(1/\eps^2)}}{n^{i\eps}\hat \phi^3}\right ).
	\end{split}   \]
	
	For $i=1$, the running time of a single level-$1$ phase is bounded by:

		\[ \otilde\left (\frac{\Delta^4n}{\phi^3}+n^{(1+O(\eps))(1-\eps)}\cdot (\log n)^{O(1/\eps^2)}+ \frac{n}{\hat \phi n^{\eps}}\right ).  \]

	Therefore, for every $1< i\leq r$, the total running time for maintaining level-$i$ data structure  is bounded by: 
	
	\[ \otilde\left (\frac{n^{1+O(\eps)}(\log n)^{O(1/\eps^2)}}{n^{i\eps}\hat \phi^3}\right )  \cdot  O\left (  \frac{n^{\eps i}(\log n)^{O(i)}\Delta^3}{\hat \phi^{3i}\cdot \phi^2}  \right )
	\leq \otilde\left (\frac{n^{1+O(\eps)}(\log n)^{O(1/\eps^2)}\Delta^3}{\hat \phi^{3i+3}\cdot \phi^2}\right ) 
	 \]
	 
	 Since $i\leq r \leq 1/\eps$, and $\hat \phi=1/(\log n )^{O(1/\eps)}$, this is bounded by: $O\left (\frac{n^{1+O(\eps)}\Delta^3(\log n)^{O(1/\eps^2)}}{\phi^2}\right ) $.
	
	Lastly, the total running time for maintaining the level-$1$ data structure is bounded by:

		\[ \otilde\left (\frac{\Delta^4n}{\phi^3}+ \frac{n^{1+O(\eps)}\cdot  (\log n)^{O(1/\eps^2)}}{\hat \phi^2}\right ) \cdot \Otilde \left (  \frac{n^{\eps }\Delta^3}{\hat \phi^{3}\cdot \phi^2}  \right ) 
		\leq \otilde\left (\frac{n^{1+O(\eps)}\Delta^7(\log n)^{O(1/\eps^2)}}{\phi^5}  \right ) .
		\]

	Summing this up over all $1\leq i\leq r$, we get that the total update time of the algorithm is at most $O\left (\frac{n^{1+O(\eps)}\Delta^7(\log n)^{O(1/\eps^2)}}{\phi^5}\right ) $.
\end{proof}

\subsection*{Responding to  $\exspquery$}

Lastly, we provide an algorithm for responding to queries $\exspquery$. Recall that, given a pair of vertices $x,y\in V(X)\setminus S$, the goal is to return a simple $x$-$y$  path $P$ in $X\setminus S$ of length at most $O\left( \Delta^2(\log n)^{O(1/\eps^2)}/\phi\right )$, in time $O(|E(P)|)$. 
We call algorithm $\mathtt{Query}(0,x,y)$, that is described in  \Cref{alg:exp query}, which recursively calls $\mathtt{Query}(i,x',y')$ for $i>0$. The idea of the algorithm is simple: we use the \EST $\tau_0$ in graph $X_0$ in order to compute two paths: one path connecting $x$ to some vertex $x'\in V(X_1)$, and one path connecting $y$ to some vertex $y'\in V(X_1)$, and then recursively call 
$\mathtt{Query}(1,x',y')$ to obtain a short path connecting $x'$ to $y'$ in $X_1$. We then use the embedding $\pset_1$ of $X_1$ into $G_0$ in order to convert the resulting path into an $x'$-$y'$ path in $X_0$. The final path connecting $x$ to $y$ is obtained by concatenating the resulting three paths.

\program{Algorithm $\mathtt{Query}(i,x,y)$}{alg:exp query}{
	Assumption: vertices $x,y$ lie in graph $X_i$.
	\begin{enumerate}
		\item If $i=r$, return the unique  $x$-$y$ path in three $\tau_{r}$, of length at most $D_{r+1}$. 
		\item Compute, in tree $\tau_{i}$, a unique path $Q_{x,x'}$ connecting $x$ to some vertex $x'\in V(X_{i+1})$, and a unique  path $Q_{y',y}$ connecting $y$ to some vertex $y'\in V(X_{i+1})$. The length of each path is bounded by $D_{i+1}$.
		\item If $x'=y'$, set $R_{x',y'}=\emptyset$; otherwise set $R_{x',y'}=\mathtt{Query}(i+1,x',y')$. 
		\item Let $Q_{x',y'}$ be a path in graph $X_i$ obtained by concatenating, for all edges $e' \in R_{x',y'}$, the corresponding path $P(e')\in \pset_{i+1}$ from the embedding of $X_{i+1}$ into $X_i$; recall that the length of each such path $P(e')$ is bounded by $\ell_{i+1}$.
		\item Return  the path $Q_{u,v}$, obtained by concatenating three paths: $Q_{x,x'}, Q_{x',y'}$, and $ Q_{y',y}$. 
	\end{enumerate}
}

The following lemma summarizes the guarantees of the algorithm for processing short-path queries.

\begin{lemma}\label{lem: oracle nonsimple path}
	For all $0\leq i\leq r$, algorithm $\mathtt{Query}(i,x,y)$, given a pair $x,y$ of vertices in graph $X_i$, returns a path connecting $x$ to $y$ in graph $X_i$ of length at most $L_i=(48c^2\log^2 n/\hat \phi)^{r-i+1}$ for $i>0$, and of length at most $L_0=(48c^2\log^2 n/\hat \phi)^{r+1}\cdot \Delta^2/\phi$ for $i=0$.
\end{lemma}

\begin{proof}
The proof is by induction on $i$, where the base case is $i=r$. Recall that, as observed already, one corollary from \Cref{lem:embed}, every pair $x,y$ of vertices of $X_r$ have a path of length at most  $D_{r+1}=4\ell_{r+1} \leq 4\hat \ell=4c^2\log^2 n/\hat \phi$ connecting them in $X_r$. The path returned by algorithm $\mathtt{Query}(r,x,y)$ has length at most $2D_{r+1}\leq 8c^2\log^2 n/\hat \phi\leq L_r$.

Assume now that the claim is true for some value $i+1$; we now prove it for value $i$. Consider algorithm $\mathtt{Query}(i,x,y)$. We are guaranteed that the algorithm computes 
a path $Q_{x,x'}$ connecting $x$ to some vertex $x'\in V(X_{i+1})$, and a  path $Q_{y',y}$ connecting $y$ to some vertex $y'\in V(X_{i+1})$, in graph $X_i$, where the lengths of both paths are bounded by $D_{i+1}$.
The path $R_{x',y'}=\mathtt{Query}(i+1,x',y')$ has length at most $L_{i+1}$, and so the resulting path $Q_{x',y'}$ has length at most $\ell_{i+1}\cdot L_{i+1}$. Therefore, altogether, the final $x$-$y$ path computed by the algorithm has length at most:

\[2D_{i+1}+\ell_{i+1}\cdot L_{i+1}=8\ell_{i+1}+L_{i+1}\ell_{i+1}.  \]

For $i\geq 1$, $\ell_{i+1}=c^2\log^2 n/\hat \phi$, while $\Delta_{i}=c\log n$, and so the length of the path is bounded by:

\[8\ell_{i+1}+L_{i+1}\ell_{i+1}\leq 2L_{i+1}\cdot c^2\log^2 n/\hat \phi\leq (2c^2\log^2n/\hat \phi)\cdot (48c^2\log^2 n/\hat \phi)^{r-i}\leq (48c^2\log^2 n/\hat \phi)^{r-i+1}.\]

For $i=0$, $\ell_1=c\Delta \log n /\phi$ and $\Delta_{0}=\Delta$, so the length of the path is bounded by:

\[8c\log n/\phi+L_{1}\cdot c\Delta\log n/\phi\leq (48c\Delta^2\log n/\phi)\cdot (48c^2\log^2 n/\hat \phi)^{r-1}\leq (48c^2\log^2 n/\hat \phi)^{r}\cdot \Delta^2/\phi.\]

	 From the  algorithm's description it is immediate to verify that the running time is bounded by $O(|E(Q)|)$, where $Q$ is the returned path.
\end{proof}

Recalling that $r=\floor{1/\eps}$
and $\hat \phi=1/(\log n )^{O(1/\eps)}$, we obtain the following immediate corollary.

\begin{corollary}\label{cor: query}
Given any pair of vertices $x,y\in V(G)\setminus S$, algorithm $\mathtt{Query}(0,x,y)$
returns a  (possibly non-simple)  $x$-$y$ path $Q$ in $X\setminus S$, of length at most $O\left( \Delta^2(\log n)^{O(1/\eps^2)}/\phi\right )$, in time $O(|E(Q)|)$.
\end{corollary}

\subsubsection{Proof of \Cref{lem:embed}}\label{subsubsection: embedding expander into smaller expander}
The proof of the lemma is very similar to the proof of \Cref{lem: expander or cut}. The algorithm employs the cut-matching game outlined in \Cref{subsec: pseudocut or expander}. The main difference is that the Matching Player is implemented via the following theorem, that was proved in \cite{detbalanced}, and its corollary.

\begin{theorem}\label{thm: matching player reg} (Theorem 3.2 in \cite{detbalanced}, builds on similar result of \cite{fast-vertex-sparsest})
	There is a deterministic algorithm, that, given an $n$-vertex graph $G=(V,E)$ with maximum vertex degree $\Delta$, and two disjoint subsets $A,B$ of its vertices, with $|A|\leq |B|$, and integers $z\geq 0$, $\ell\geq 32\Delta\log n$, computes one of the following:
	
	\begin{itemize}
			\item either a collection $\pset'$ of paths, each of which connects a distinct vertex of $A$ to a distinct vertex of $B$, with $|\pset'|\geq |A|-z$, such that the length of every path in $\pset'$ is at most $\ell$, and the paths in $\pset'$ cause congestion at most $\ell^2$; or
		
		\item a cut $(X,Y)$ in $G$, with $|X|,|Y|\geq z/2$, and $\phi_G(X,Y)\leq 72\Delta\log n/\ell$.
	\end{itemize}
	
	The running time of the algorithm is $\tilde O(\ell^3|E(G)|+\ell^2n)$.
\end{theorem}

We obtain the following immediate corollary of \Cref{thm: matching player reg}.

\begin{corollary}\label{cor: matching player reg}
		There is a deterministic algorithm, that, given an $n$-vertex graph $X=(V,E)$ with maximum vertex degree $\Delta$, that is a $\phi$-expander, and two disjoint subsets $A,B$ of its vertices, with $|A|\leq |B|$, computes a collection $\pset'$ of paths, where each path in $\pset'$ connects a distinct vertex from $A$ to a distinct vertex from $B$, and $|\pset'|= |A|$, such that the paths in $\pset'$ cause  congestion at most $O\left (\frac{\Delta^2\log^2n}{\phi^2}\right )$, and every path has length at most $O\left (\frac{\Delta \log n}{\phi}\right )$.
	The running time of the algorithm is $\otilde(\Delta^4n/\phi^3)$.
	\end{corollary}

\begin{proof}
	We set the parameter $\ell=\frac{150\Delta\log n}{\phi}$, and apply the algorithm from \Cref{thm: matching player reg} to graph $X$, vertex sets $A$ and $B$, and the parameter $\ell$, setting $z=0$. Notice that the algorithm may not return a cut $(X,Y)$ with $\phi_G(X,Y)\leq 72\Delta\log n/\ell<\phi$, because graph $X$ is a $\phi$-expander. Therefore, it must return a collection $\pset$ of paths, with each path connecting a distinct vertex from $A$ to a distinct vertex from $B$, and $|\pset|= |A|$. The length of each path in $\pset$ is bounded by $\ell\leq O\left (\frac{\Delta \log n}{\phi}\right )$, and the paths in $\pset$ cause congestion at most  $\ell^2\leq O\left (\frac{\Delta^2\log^2n}{\phi^2}\right )$. 
	The running time of the algorithm is bounded by $\tilde O(\ell^3|E(G)|+\ell^2n)\leq \otilde(\Delta^4n/\phi^3)$.
\end{proof}

The remainder of the proof of \Cref{lem:embed} is almost identical to the proof of \Cref{lem: expander or cut}. 
We first need to take care of the special case when $\eps<c''/\log k$, where $c''$ is the constant from \Cref{cor: cut player}. In this case, $k<2^{c''/\eps}$ holds, and, since we have assumed that $\phi\leq 1/2^{\Omega(1/\eps)}$, we can assume that $\phi\leq 1/k$. In this case, we let $X'$ be an arbitrary constant-degree $\phi'$-expander, where $\phi'=\Omega(1)$, with $V(X')=T$. In order to embed this expander into $X$, we select, for every edge $e=(x,y)\in E(X')$, a shortest $x$-$y$ path $P(e)$ connecting $x$ to $y$ in $X$; from \Cref{obs: short paths in exp}, the length of $P(e)$ is at most $O\left (\frac{\Delta \log n}{\phi}\right )$. We then let $\pset=\set{P(e)\mid e\in E(X')}$ be the embedding of $X'$ into $X$. Clearly, the congestion of this embedding is bounded by $k\leq 1/\phi$. The running time of this algorithm is bounded by $O(k|E(X)|)\leq O(n\Delta/\phi)$. Therefore, we assume from now on that $\eps\geq c''/\log k$.

We start with a graph $X'$, whose vertex set is $T$, and edge set is empty. We initialize $\pset=\emptyset$. For convenience, denote $|T|=k$; assume for now that $k$ is an even integer. We then perform at most $O(\log k)$ iterations. The $i$th iteration is performed as follows. First, we apply the algorithm from \Cref{cor: cut player} to graph $X'$. Assume for now that the algorithm produces  a partition $(A_i,B_i)$ of $V(X')$ with $|A_i|,|B_i|\geq k/4$, $|A_i|\leq |B_i|$, and $|E_{X'}(A_i,B_i)|\leq k/100$. We let $(A_i',B_i')$ be any partition of $V(X')$ with $|A_i'|=|B_i'|$ and $A_i\subseteq A_i'$. Next, we apply the algorithm from \Cref{cor: matching player reg} to graph $X'$ and vertex sets $A_i',B_i'$. The algorithm returns a collection $\pset_i$ of paths, where each path in $\pset_i$ connects a distinct vertex from $A_i'$ to a distinct vertex from $B_i'$ and $|\pset_i|= |A_i'|$, such that the paths in $\pset'$ cause  congestion at most $O\left (\frac{\Delta^2\log^2n}{\phi^2}\right )$, and every path has length at most $O\left (\frac{\Delta \log n}{\phi}\right )$. We let $M_i$ be the perfect matching between vertices of $A_i'$ and vertices of $B_i'$, where $(x,y)\in M_i$ iff some path in $\pset_i$ connects $x$ to $y$. We add the paths in $\pset_i$ to $\pset$, and continue to the next iteration.

Finally, assume that the algorithm from  \Cref{cor: cut player} returns a subset $S\subseteq V(X')$ of at least $k/2$ vertices of $X'$, such that graph $X[S]$ is a $\phi'$-expander, for $\phi'=1/(\log k )^{O(1/\eps)}$. Let $i$ be the index of the current iteration. We then set $A_i'=V(X)\setminus S$ and $B'_i=S$, and apply the algorithm from \Cref{thm: matching player} to graph $X'$, and vertex sets $A_i',B_i'$. We continue exactly as before, obtaining a matching $M_i$ between the vertices of $A'_i$ and the vertices of $B_i'$, of cardinality $|A'_i|$, and its corresponding routing $\pset_i$. We let $M_i$ be a partial matching between vertices of $A_i'$ and vertices of $B_i'$, where $(x,y)\in M_i'$ iff some path in $\pset_i$ connects $x$ to $y$, so $|M_i|=|A'_i|$. We then add the edges of $M_i$ to graph $X'$, add the paths in $\pset_i$ to set $\pset$, and terminate the algorithm, returning the graph $X'$ and its embedding $\pset$.

From \Cref{thm: CMG}, the total number of iterations in the algorithm is bounded by $O(\log k)$, and the final graph $X'$ is a $\hat \phi$-expander, for $\hat \phi=1/(\log k )^{O(1/\eps)}$.
The paths in $\pset$ provide an embedding of the graph $X'$ into $X$, where the length of every path in $\pset$ is at most $O\left (\frac{\Delta \log n}{\phi}\right )$. Since the number of iterations is at most $O(\log k)$, and each set $\pset_i$ of paths causes edge-congestion at most  $O\left (\frac{\Delta^2\log^2n}{\phi^2}\right )$, the total edge-congestion caused by the paths in $\pset$ is at most  $O\left (\frac{\Delta^2\log^3n}{\phi^2}\right )$. 

Recall that we have assumed that $|T|$ is an even integer. If this is not the case, then we let $t\in T$ be any vertex, and we run the same algorithm as above, replacing the set $T$ of terminals with $T'=T\setminus\set{t}$. Let $X'$ be the resulting $\hat \phi$-expander, with $V(X')=T'$, and let $\pset$ be its embedding into $X$. Let $t'\in T'$ ba any vertex, and let $P'$ be the shortest path connecting $t$ to $t'$ in graph $X$. From \Cref{obs: short paths in exp}, the length of $P'$ is at most $O(\Delta \log n/\phi)$. We then obtain the final expander $X''$ from $X'$, by inserting the vertex $t$ and the edge $(t,t')$ into it. The set $\pset\cup \set{P'}$ of paths defines an embedding of $X''$ into $X$ with edge-congestion bounded by $O\left (\frac{\Delta^2\log^3n}{\phi^2}\right )$, with path-lengths bounded by $O\left (\frac{\Delta \log n}{\phi}\right )$. It is easy to verify that graph $X''$ remains a $\phi''$-expander, for $\phi''=1/(\log k )^{O(1/\eps)}$.
It now remains to analyze the running time of the algorithm.

The algorithm consists of $O(\log k)$ iterations. Each iteration involves applying the algorithm from \Cref{cor: cut player}, whose running time is $O\left (k^{1+O(\eps)}\cdot (\log k)^{O(1/\eps^2)}\right )$, and the algorithm from \Cref{cor: matching player reg}, whose running time is $\otilde(\Delta^4n/\phi^3)$. Therefore, the total running time of the algorithm is $\otilde(\Delta^4n/\phi^3)+O\left (k^{1+O(\eps)}\cdot (\log k)^{O(1/\eps^2)}\right )$.

%% file: APSP-p3-v4.bbl
\newcommand{\etalchar}[1]{$^{#1}$}
\def\cprime{$'$} \def\cprime{$'$}
\begin{thebibliography}{KKOV07}

\bibitem[ABCP98]{neighborhood-cover1}
Baruch Awerbuch, Bonnie Berger, Lenore Cowen, and David Peleg.
\newblock Near-linear time construction of sparse neighborhood covers.
\newblock {\em SIAM Journal on Computing}, 28(1):263--277, 1998.

\bibitem[ACT14]{abraham2014fully}
Ittai Abraham, Shiri Chechik, and Kunal Talwar.
\newblock Fully dynamic all-pairs shortest paths: Breaking the {O (n)} barrier.
\newblock In {\em LIPIcs-Leibniz International Proceedings in Informatics},
  volume~28. Schloss Dagstuhl-Leibniz-Zentrum fuer Informatik, 2014.

\bibitem[AHK12]{AroraHK12}
Sanjeev Arora, Elad Hazan, and Satyen Kale.
\newblock The multiplicative weights update method: a meta-algorithm and
  applications.
\newblock {\em Theory of Computing}, 8(1):121--164, 2012.

\bibitem[AP90]{neighborhood-cover2}
Baruch Awerbuch and David Peleg.
\newblock Sparse partitions.
\newblock In {\em Proceedings [1990] 31st Annual Symposium on Foundations of
  Computer Science}, pages 503--513. IEEE, 1990.

\bibitem[BBG{\etalchar{+}}20]{new-spanner}
Aaron Bernstein, Jan van~den Brand, Maximilian~Probst Gutenberg, Danupon
  Nanongkai, Thatchaphol Saranurak, Aaron Sidford, and He~Sun.
\newblock Fully-dynamic graph sparsifiers against an adaptive adversary.
\newblock {\em arXiv preprint arXiv:2004.08432}, 2020.

\bibitem[BC16]{BernsteinChechik}
Aaron Bernstein and Shiri Chechik.
\newblock Deterministic decremental single source shortest paths: beyond the
  {O(mn)} bound.
\newblock In {\em Proceedings of the forty-eighth annual ACM symposium on
  Theory of Computing}, pages 389--397. ACM, 2016.

\bibitem[BC17]{BernsteinChechikSparse}
Aaron Bernstein and Shiri Chechik.
\newblock Deterministic partially dynamic single source shortest paths for
  sparse graphs.
\newblock In {\em Proceedings of the Twenty-Eighth Annual ACM-SIAM Symposium on
  Discrete Algorithms}, pages 453--469. SIAM, 2017.

\bibitem[Ber16]{bernstein16}
Aaron Bernstein.
\newblock Maintaining shortest paths under deletions in weighted directed
  graphs.
\newblock {\em SIAM Journal on Computing}, 45(2):548--574, 2016.

\bibitem[Ber17]{Bernstein}
Aaron Bernstein.
\newblock Deterministic partially dynamic single source shortest paths in
  weighted graphs.
\newblock In {\em LIPIcs-Leibniz International Proceedings in Informatics},
  volume~80. Schloss Dagstuhl-Leibniz-Center for Computer Science, 2017.

\bibitem[BHG{\etalchar{+}}20]{newest-spanner}
Thiago Bergamaschi, Monika Henzinger, Maximilian~Probst Gutenberg,
  Virginia~Vassilevska Williams, and Nicole Wein.
\newblock New techniques and fine-grained hardness for dynamic near-additive
  spanners.
\newblock {\em arXiv preprint arXiv:2010.10134}, 2020.

\bibitem[BHS07]{BaswanaHS07}
Surender Baswana, Ramesh Hariharan, and Sandeep Sen.
\newblock Improved decremental algorithms for maintaining transitive closure
  and all-pairs shortest paths.
\newblock {\em J. Algorithms}, 62(2):74--92, 2007.

\bibitem[BKS12]{BaswanaKS12}
Surender Baswana, Sumeet Khurana, and Soumojit Sarkar.
\newblock Fully dynamic randomized algorithms for graph spanners.
\newblock {\em {ACM} Trans. Algorithms}, 8(4):35:1--35:51, 2012.

\bibitem[BR11]{BernsteinR11}
Aaron Bernstein and Liam Roditty.
\newblock Improved dynamic algorithms for maintaining approximate shortest
  paths under deletions.
\newblock In {\em Proceedings of the Twenty-Second Annual {ACM-SIAM} Symposium
  on Discrete Algorithms, {SODA} 2011, San Francisco, California, USA, January
  23-25, 2011}, pages 1355--1365, 2011.

\bibitem[CGL{\etalchar{+}}19]{detbalanced}
Julia Chuzhoy, Yu~Gao, Jason Li, Danupon Nanongkai, Richard Peng, and
  Thatchaphol Saranurak.
\newblock A deterministic algorithm for balanced cut with applications to
  dynamic connectivity, flows, and beyond.
\newblock {\em CoRR}, abs/1910.08025, 2019.

\bibitem[Che18]{chechik}
Shiri Chechik.
\newblock Near-optimal approximate decremental all pairs shortest paths.
\newblock In {\em 2018 IEEE 59th Annual Symposium on Foundations of Computer
  Science (FOCS)}, pages 170--181. IEEE, 2018.

\bibitem[CK19]{fast-vertex-sparsest}
Julia Chuzhoy and Sanjeev Khanna.
\newblock A new algorithm for decremental single-source shortest paths with
  applications to vertex-capacitated flow and cut problems.
\newblock In {\em Proceedings of the 51st Annual ACM SIGACT Symposium on Theory
  of Computing}, pages 389--400, 2019.

\bibitem[CS21]{APSP-old}
Julia Chuzhoy and Thatchaphol Saranurak.
\newblock Deterministic algorithms for decremental shortest paths via layered
  core decomposition.
\newblock In {\em Proceedings of the 2021 ACM-SIAM Symposium on Discrete
  Algorithms (SODA)}, pages 2478--2496. SIAM, 2021.

\bibitem[CZ20]{chechikLowStretch}
Shiri Chechik and Tianyi Zhang.
\newblock Dynamic low-stretch spanning trees in subpolynomial time.
\newblock In {\em Proceedings of the Fourteenth Annual ACM-SIAM Symposium on
  Discrete Algorithms}, pages 463--475. SIAM, 2020.

\bibitem[DHZ00]{DorHZ00}
Dorit Dor, Shay Halperin, and Uri Zwick.
\newblock All-pairs almost shortest paths.
\newblock {\em {SIAM} J. Comput.}, 29(5):1740--1759, 2000.

\bibitem[Din06]{Dinitz}
Yefim Dinitz.
\newblock Dinitz' algorithm: The original version and {Even's} version.
\newblock In {\em Theoretical computer science}, pages 218--240. Springer,
  2006.

\bibitem[ES81]{EvenS}
Shimon Even and Yossi Shiloach.
\newblock An on-line edge-deletion problem.
\newblock {\em Journal of the ACM (JACM)}, 28(1):1--4, 1981.

\bibitem[FG19]{ForsterG19}
Sebastian Forster and Gramoz Goranci.
\newblock Dynamic low-stretch trees via dynamic low-diameter decompositions.
\newblock In {\em Proceedings of the 51st Annual {ACM} {SIGACT} Symposium on
  Theory of Computing, {STOC} 2019, Phoenix, AZ, USA, June 23-26, 2019}, pages
  377--388, 2019.

\bibitem[FGH20]{ForsterGH20}
Sebastian Forster, Gramoz Goranci, and Monika Henzinger.
\newblock Dynamic maintenance of low-stretch probabilistic tree embeddings with
  applications.
\newblock {\em CoRR}, abs/2004.10319, 2020.

\bibitem[FHN14a]{HenzingerKN14_focs}
Sebastian Forster, Monika Henzinger, and Danupon Nanongkai.
\newblock Decremental single-source shortest paths on undirected graphs in
  near-linear total update time.
\newblock In {\em 55th {IEEE} Annual Symposium on Foundations of Computer
  Science, {FOCS} 2014, Philadelphia, PA, USA, October 18-21, 2014}, pages
  146--155, 2014.

\bibitem[FHN14b]{HenzingerKN14_soda}
Sebastian Forster, Monika Henzinger, and Danupon Nanongkai.
\newblock A subquadratic-time algorithm for decremental single-source shortest
  paths.
\newblock In {\em Proceedings of the Twenty-Fifth Annual {ACM-SIAM} Symposium
  on Discrete Algorithms, {SODA} 2014, Portland, Oregon, USA, January 5-7,
  2014}, pages 1053--1072, 2014.

\bibitem[Fle00]{Fleischer00}
Lisa Fleischer.
\newblock Approximating fractional multicommodity flow independent of the
  number of commodities.
\newblock {\em {SIAM} J. Discrete Math.}, 13(4):505--520, 2000.

\bibitem[GK98]{GK98}
Naveen Garg and Jochen K{\"{o}}nemann.
\newblock Faster and simpler algorithms for multicommodity flow and other
  fractional packing problems.
\newblock In {\em 39th Annual Symposium on Foundations of Computer Science,
  {FOCS} '98, November 8-11, 1998, Palo Alto, California, {USA}}, pages
  300--309, 1998.

\bibitem[GVY95]{GVY}
N.~Garg, V.V. Vazirani, and M.~Yannakakis.
\newblock Approximate max-flow min-(multi)-cut theorems and their applications.
\newblock {\em SIAM Journal on Computing}, 25:235--251, 1995.

\bibitem[GWN20]{GutenbergW20}
Maximilian~Probst Gutenberg and Christian Wulff-Nilsen.
\newblock Deterministic algorithms for decremental approximate shortest paths:
  Faster and simpler.
\newblock In {\em Proceedings of the Fourteenth Annual ACM-SIAM Symposium on
  Discrete Algorithms}, pages 2522--2541. SIAM, 2020.

\bibitem[HdLT01]{dynamic-connectivity}
Jacob Holm, Kristian de~Lichtenberg, and Mikkel Thorup.
\newblock Poly-logarithmic deterministic fully-dynamic algorithms for
  connectivity, minimum spanning tree, 2-edge, and biconnectivity.
\newblock {\em J. ACM}, 48(4):723--760, July 2001.

\bibitem[HK95]{HenzingerKing}
Monika~Rauch Henzinger and Valerie King.
\newblock Fully dynamic biconnectivity and transitive closure.
\newblock In {\em Foundations of Computer Science, 1995. Proceedings., 36th
  Annual Symposium on}, pages 664--672. IEEE, 1995.

\bibitem[HKN16]{henzinger16}
Monika Henzinger, Sebastian Krinninger, and Danupon Nanongkai.
\newblock Dynamic approximate all-pairs shortest paths: Breaking the o(mn)
  barrier and derandomization.
\newblock {\em SIAM Journal on Computing}, 45(3):947--1006, 2016.

\bibitem[HKNS15]{HenzingerKNS15}
Monika Henzinger, Sebastian Krinninger, Danupon Nanongkai, and Thatchaphol
  Saranurak.
\newblock Unifying and strengthening hardness for dynamic problems via the
  online matrix-vector multiplication conjecture.
\newblock In {\em Proceedings of the forty-seventh annual ACM symposium on
  Theory of computing}, pages 21--30, 2015.

\bibitem[KKOV07]{KhandekarKOV2007cut}
Rohit Khandekar, Subhash Khot, Lorenzo Orecchia, and Nisheeth~K Vishnoi.
\newblock On a cut-matching game for the sparsest cut problem.
\newblock {\em Univ. California, Berkeley, CA, USA, Tech. Rep.
  UCB/EECS-2007-177}, 6(7):12, 2007.

\bibitem[K{\L}19]{reliable-hubs}
Adam Karczmarz and Jakub {\L}{{a}}cki.
\newblock Reliable hubs for partially-dynamic all-pairs shortest paths in
  directed graphs.
\newblock {\em arXiv preprint arXiv:1907.02266}, 2019.

\bibitem[KMP12]{KelnerMP12}
Jonathan~A. Kelner, Gary~L. Miller, and Richard Peng.
\newblock Faster approximate multicommodity flow using quadratically coupled
  flows.
\newblock In {\em Proceedings of the 44th Symposium on Theory of Computing
  Conference, {STOC} 2012, New York, NY, USA, May 19 - 22, 2012}, pages 1--18,
  2012.

\bibitem[KRV09]{KRV}
Rohit Khandekar, Satish Rao, and Umesh Vazirani.
\newblock Graph partitioning using single commodity flows.
\newblock {\em Journal of the ACM (JACM)}, 56(4):19, 2009.

\bibitem[{\L}N20]{lkacki2020near}
Jakub {\L}acki and Yasamin Nazari.
\newblock Near-optimal decremental approximate multi-source shortest paths.
\newblock {\em arXiv preprint arXiv:2009.08416}, 2020.

\bibitem[LR99]{LR}
F.~T. Leighton and S.~Rao.
\newblock Multicommodity max-flow min-cut theorems and their use in designing
  approximation algorithms.
\newblock {\em Journal of the ACM}, 46:787--832, 1999.

\bibitem[Mad10]{Madry10_stoc}
Aleksander Madry.
\newblock Faster approximation schemes for fractional multicommodity flow
  problems via dynamic graph algorithms.
\newblock In {\em Proceedings of the 42nd {ACM} Symposium on Theory of
  Computing, {STOC} 2010, Cambridge, Massachusetts, USA, 5-8 June 2010}, pages
  121--130, 2010.

\bibitem[RZ11]{RodittyZ11}
Liam Roditty and Uri Zwick.
\newblock On dynamic shortest paths problems.
\newblock {\em Algorithmica}, 61(2):389--401, 2011.

\bibitem[RZ12]{rodittyZ2}
Liam Roditty and Uri Zwick.
\newblock Dynamic approximate all-pairs shortest paths in undirected graphs.
\newblock {\em SIAM Journal on Computing}, 41(3):670--683, 2012.

\bibitem[SW19]{expander-pruning}
Thatchaphol Saranurak and Di~Wang.
\newblock Expander decomposition and pruning: Faster, stronger, and simpler.
\newblock In {\em Proceedings of the Thirtieth Annual {ACM-SIAM} Symposium on
  Discrete Algorithms, {SODA} 2019, San Diego, California, USA, January 6-9,
  2019}, pages 2616--2635, 2019.

\bibitem[TZ01]{TZ}
M.~Thorup and U.~Zwick.
\newblock Approximate distance oracles.
\newblock {\em Annual ACM Symposium on Theory of Computing}, 2001.

\bibitem[WW18]{williams2018subcubic}
Virginia~Vassilevska Williams and R~Ryan Williams.
\newblock Subcubic equivalences between path, matrix, and triangle problems.
\newblock {\em Journal of the ACM (JACM)}, 65(5):1--38, 2018.

\bibitem[Zwi98]{Zwick98}
Uri Zwick.
\newblock All pairs shortest paths in weighted directed graphs-exact and almost
  exact algorithms.
\newblock In {\em Proceedings 39th Annual Symposium on Foundations of Computer
  Science (Cat. No. 98CB36280)}, pages 310--319. IEEE, 1998.

\end{thebibliography}
